%% file: mainTW.tex
\documentclass[a4paper,11pt]{article}

\input{preamble.tex}

\begin{document}

\title{On Treewidth and Stable Marriage}

\author{
Sushmita Gupta\thanks{University of Bergen, Bergen, Norway. \texttt{sushmita.gupta@ii.uib.no}}
 \and  Saket Saurabh\thanks{The Institute of Mathematical Sciences, HBNI, Chennai, India. \texttt{saket@imsc.res.in}}
 \and Meirav Zehavi\thanks{University of Bergen, Bergen, Norway. \texttt{meirav.zehavi@ii.uib.no}} 
}

\maketitle

\thispagestyle{empty}

\begin{abstract} 
\input{abstract.tex}
\end{abstract}

\newpage
\pagestyle{plain}
\setcounter{page}{1}

\input{intro.tex}

\input{prelims.tex}

\input{overview.tex}

\input{w1hardness.tex}

\input{xp.tex}

\input{fpt.tex}

\input{seth.tex}

\input{conclusion.tex}

\bibliographystyle{siam}
\bibliography{references-RW}
%\bibliography{references}

\appendix

\section{Parameterized Complexity}\label{sec:PC}

Let $\Pi$ be an \NPH\ problem. In the framework of Parameterized Complexity, each instance of $\Pi$ is associated with a {\em parameter} $k$. Here, the goal is to confine the combinatorial explosion in the running time of an algorithm for $\Pi$ to depend only on $k$. Formally, we say that $\Pi$ is {\em fixed-parameter tractable (\FPT)} if any instance $(I, k)$ of $\Pi$ is solvable in time $f(k)\cdot |I|^{\OO(1)}$, where $f$ is an arbitrary function of $k$. A weaker request is that for every fixed $k$, the problem $\Pi$ would be solvable in polynomial time. Formally, we say that $\Pi$ is {\em slice-wise polynomial (\XP)} if any instance $(I, k)$ of $\Pi$ is solvable in time $f(k)\cdot |I|^{g(k)}$, where $f$ and $g$ are arbitrary functions of $k$. Nowadays, Parameterized Complexity supplies a rich toolkit to design \FPT\ and \XP\ algorithms.

Parameterized Complexity also provides methods to show that a problem is unlikely to be \FPT. The main technique is the one of parameterized reductions analogous to those employed in classical complexity. Here, the concept of \WO-hardness replaces the one of \NP-hardness, and for reductions we need not only construct an equivalent instance in \FPT\ time, but also ensure that the size of the parameter in the new instance depends only on the size of the parameter in the original one. If there exists such a reduction transforming a problem known to be \WOH\ to another problem $\Pi$, then the problem $\Pi$ is \WO-hard as well. Central \WOH-problems include, for example, deciding whether a nondeterministic single-tape Turing machine accepts within $k$ steps, {\sc Clique} parameterized be solution size, and {\sc Independent Set} parameterized by solution size. To show that a problem $\Pi$ is not \XP\ unless \textsf{P}=\NP, it is sufficient to show that there exists a fixed $k$ such $\Pi$ is \NPH. Then, the problem is said to be \paraH.

To obtain (essentially) tight conditional lower bounds for the running times of algorithms, we rely on the well-known {\em Exponential-Time Hypothesis (\ETH)} and {\em Strong Exponential-Time Hypothesis (\SETH)} \cite{DBLP:journals/jcss/ImpagliazzoP01,DBLP:journals/jcss/ImpagliazzoPZ01,DBLP:conf/iwpec/CalabroIP09}. To formalize the statements of \ETH\ and \SETH, first recall that  given a formula $\varphi$ in conjuctive normal form (CNF) with $n$ variables and $m$ clauses, the task of {\sc CNF-SAT} is to decide whether there is a truth assignment to the variables that satisfies $\varphi$. In the {\sc $p$-CNF-SAT} problem, each clause is restricted to have at most $p$ literals. First, \ETH\ asserts that {\sc 3-CNF-SAT} cannot be solved in time $\OO(2^{o(n)})$. Second, \SETH\ asserts that for every fixed $\epsilon<1$, there exists a (large) integer $p=p(\epsilon)$ such that {\sc $p$-CNF-SAT} cannot be solved in time $\OO((2-\epsilon)^{n})$. We remark that \ETH\ implies \FPT$\neq$\WO, and that \SETH\ implies \ETH. More information on Parameterized Complexity, \ETH\ and \SETH\ can be found in \cite{DBLP:series/txcs/DowneyF13,DBLP:books/sp/CyganFKLMPPS15}.

\section{Treewidth}\label{sec:tw}

Treewidth is a structural parameter indicating how much a graph resembles a tree. Formally,

\begin{definition}\label{def:treewidth}
A \emph{tree decomposition} of a graph $G$ is a pair $(T,\beta)$ of a tree $T$
and $\beta:V(T) \rightarrow 2^{V(G)}$, such that
\vspace{-0.5em}
\begin{enumerate}
\itemsep0em 
\item\label{item:twedge} for any edge $\{x,y\} \in E(G)$ there exists a node $v \in V(T)$ such that $x,y \in \beta(v)$, and
\item\label{item:twconnected} for any vertex $x \in V(G)$, the subgraph of $T$ induced by the set $T_x = \{v\in V(T): x\in\beta(v)\}$ is a non-empty tree.
\end{enumerate}
The {\em width} of $(T,\beta)$ is $\max_{v\in V(T)}\{|\beta(v)|\}-1$. The {\em treewidth} of $G$ is the minimum width over all tree decompositions of $G$.
\end{definition}

We use a standard form of a tree decomposition that simplifies the design of DP algorithms.

\begin{definition}
A tree decomposition $(T,\beta)$ of a graph $G$ is {\em nice} if for the root $r$ of $T$, $\beta(r)=\emptyset$, and each node $v\in V(T)$ is of one of the following types.
\vspace{-0.5em}
\begin{itemize}
\itemsep0em 
\item {\bf Leaf}: $v$ is a leaf in $T$ and $\beta(v)=\emptyset$.
\item {\bf Forget}: $v$ has one child, $u$, and there is a vertex $x\in\beta(u)$ such that $\beta(v)=\beta(u)\setminus\{x\}$.
\item {\bf Introduce}: $v$ has one child, $u$, and there is a vertex $x\in\beta(v)$ such that $\beta(v)\setminus\{x\}=\beta(u)$.
\item {\bf Join}: $v$ has two children, $u$ and $w$, and $\beta(v)=\beta(u)=\beta(w)$.
\end{itemize}
\end{definition}

For $v \in V(T)$, we say that $\beta(v)$ is the \emph{bag} of $v$, and $\gamma(v)$ denotes the union of the bags of $v$ and the descendants of $v$ in $T$. According to standard practice in Parameterized Complexity with respect to problems parameterized by $\tw$, we assume that every input instance is given to us along with a tree decomposition (of the appropriate graph, primal or $G_\Pi$) of width $\tw$. (Otherwise, such a decomposition can be computed using the means described in \cite{DBLP:books/sp/CyganFKLMPPS15,DBLP:series/txcs/DowneyF13}.) Given a tree decomposition $(T,\beta)$, Bodlaender \cite{DBLP:journals/siamcomp/Bodlaender96} showed how to construct a {\em nice} tree decomposition of the same width as $(T,\beta)$. Thus, from now onwards, when we design our algorithms, we assume that we have a nice tree decomposition of the appropriate graph of width $\tw$.

%Finally, we state the following well-known proposition.

%\begin{proposition}[folklore]\label{lem:tree}Given a graph $G$ with tree decomposition $(T,\beta)$, and a connected subgraph $S$ of $G$, the subgraph of $T$ induced by $T_{V(S)}=\{v\in V(T): V(S)\cap\beta(v)\neq\emptyset\}$ is~a~tree.
%\end{proposition}

\end{document}

%% file: abstract.tex
{\sc Stable Marriage} is a fundamental problem to both computer science and economics. Four well-known \NPH\ optimization versions of this problem are the {\sc Sex-Equal Stable Marriage (SESM)}, {\sc Balanced Stable Marriage (BSM)}, {\sc max-Stable Marriage with Ties (max-SMT)} and {\sc min-Stable Marriage with Ties (min-SMT)} problems.
In this paper, we analyze these problems from the viewpoint of Parameterized Complexity. We conduct the first study of these problems with respect to the parameter treewidth. First, we study the treewidth $\tw$ of the primal graph. We establish that all four problems are \WOH. In particular, while it is easy to show that all four problems admit algorithms that run in time $n^{\OO(\tw)}$, we prove that all of these algorithms are likely to be essentially optimal. Next, we study the treewidth $\tw$ of the rotation digraph. In this context, the {\sc max-SMT} and {\sc min-SMT} are not defined. For both {\sc SESM} and {\sc BSM}, we design (non-trivial) algorithms that run in time $2^{\tw}n^{\OO(1)}$. Then, for both {\sc SESM} and {\sc BSM}, we also prove that unless \SETH\ is false, algorithms that run in time $(2-\epsilon)^{\tw}n^{\OO(1)}$ do not exist for any fixed $\epsilon>0$. We thus present a comprehensive, complete picture of the behavior of central optimization versions of {\sc Stable Marriage} with respect to treewidth.

%% file: intro.tex
\newcommand{\il}[1]{\todo[backgroundcolor=Blue!80!white!10, inline]{#1}\xspace}

\newcommand{\mar}{marriage\xspace}
\newcommand{\etal}{{\it et al.}}

\newcommand{\maxsmt}{{\sc max-SMT}\xspace}
\newcommand{\minsmt}{{\sc min-SMT}\xspace}

\section{Introduction}\label{sec:intro}
\vspace{-0.55em}

{\em Matching under preferences} is a rich topic central to both economics and computer science, which has been consistently and intensively studied for over several decades. One of the main reasons for interest in this topic stems from the observation that it is extremely relevant to a wide variety of practical applications modeling situations where the objective is to {\em match} agents to other agents (or to resources). In the most general setting, a matching is defined as an allocation (or assignment) of agents to resources that satisfies some predefined criterion of compatibility/acceptability. Here, the (arguably) best known model is the {\em two-sided model}, where the agents on one side are referred to as {\it men}, and the agents on the other side are referred to as {\it women}. A few illustrative examples of real life situations where this model is employed in practice include matching hospitals to residents, students to colleges, kidney patients to donors and users to servers in a distributed Internet service. At the heart of all of these applications lies the fundamental {\sc Stable Marriage (SM)} problem. In particular, the Nobel Prize in Economics was awarded to Shapley and Roth in 2012 ``for the theory of stable allocations and the practice of market design.'' Moreover, several books have been dedicated to the study of {\sc SM} as well as optimization versions of this classical problem \cite{DBLP:books/daglib/0066875,opac-b1092346,DBLP:books/ws/Manlove13}. 

In this paper, we conduct a comprehensive study of four well-known \NPH\ optimization versions of {\sc SM}, namely the {\sc Sex-Equal SM (SESM)}, {\sc Balanced SM (BSM)}, {\sc max-SM with Ties (max-SMT)} and {\sc min-SM with Ties (min-SMT)} problems, from the viewpoint of Parameterized Complexity. Readers unfamiliar with the definitions of these problems are referred to Section \ref{sec:prelims}. The sizes of the solutions to all of these problems can often be as large as the instances themselves. Furthermore, as these problems are \NPH\ when preference lists are restricted to have a fixed constant length, the maximum length of a preference list is not a sensible parameter (see Appendix \ref{sec:PC}). Thus, we parameterize these problems by the treewidth of the primal graph as well as the treewidth of the rotation digraph. Arguably, the parameter treewidth is the most natural one. Moreover, parameterization by treewidth is a standard practice in Parameterized Complexity. Indeed, from a practical point of view, networks often tend to resembles trees, and from a theoretical point of view, treewidth is often a parameter with respect to which it is possible to derive ``fast'' parameterized algorithms. Accordingly, books on Parameterized Complexity devote several complete chapters solely to the study of treewidth (see \cite{DBLP:books/sp/CyganFKLMPPS15,DBLP:series/txcs/DowneyF13,Niedermeierbook06,DBLP:series/txtcs/FlumG06}). Nevertheless, our work is the first to study the parameterized complexity of optimization versions of {\sc SM} with respect to treewidth, although {\sc SM} is a basic problem to both economics and computer science. In this sense, our work fills a fundamental knowledge gap. We obtain {\em tight} upper and (conditional) lower bounds for the running times of algorithms for all of the problems that we study. Moreover, each set of results (\WO-hardness, \XP-algorithms, \FPT-algorithms and tight conditional lower bounds under \SETH) is derived in a novel systematic way that may apply not only to other optimization versions of {\sc SM}, but also to other parameterization of the problems studied in the paper.

\vspace{-0.7em}
\subsection{Related Work}
\vspace{-0.5em}

For a broad discussion of optimization variants of {\sc SM}, the reader is referred to the books \cite{DBLP:books/daglib/0066875,opac-b1092346,DBLP:books/ws/Manlove13} or surveys such as~\cite{Survey:Iwama-Miyazaki}. Here, we only briefly overview some relevant literature.

%\smallskip
\myparagraph{Sex-Equality.} The {\em egalitarian measure}, which minimizes the {\em sum} of the amount of dissatisfaction of men and that amount of dissatisfaction of women, is arguably the simplest notion of the quality of a stable matching. A stable matching minimizing this measure is known as an {\it egalitarian stable matching}, which is notably computable in polynomial time \cite{DBLP:journals/jacm/IrvingLG87}. Gusfield and Irving \cite{DBLP:books/daglib/0066875} noted that in the context of egalitarian stable matchings, it may be the case that members of one sex are considerably better off than the members of the opposite sex. Knuth discussed an example \cite[pg 56]{opac-b1092346} that had 10 stable matchings, each having an egalitarian function value of 20. However, the sex-equality function value (see Section \ref{sec:prelimsOptimization}) ranged from -12 to +12, with the optimal being 0. This motivated the definition of a {\it sex-equal stable matching} as a stable matching that minimizes the sex-equality function value. Gusfield and Irving \cite{DBLP:books/daglib/0066875} then asked whether there is a polynomial-time algorithm for \sesm. Kato \cite{Kato93} was the first to show that {\sc SESM} is \NPH\ by showing a reduction from {\sc Partially Ordered Knapsack}, one of the classical problems listed by Garey and Johnson \cite{GJ79}. McDermid and Irving \cite{DBLP:journals/algorithmica/McDermidI14} proved that given an instance of {\sc SM}, where the length of each preference list is at most 3, deciding whether or not there exists a stable matching whose sex-equal function value is 0, is \NPC. Contrastingly, if the length of each preference list is at most 2 on one side while the length of each preference lists on the other side may be unbounded, then \sesm\ is solvable in time $\mathcal{O}(n^{3})$, where $n$ denotes the number of men/women in the instance. 

%\begin{table}
%\begin{center}
  %\begin{tabular}{lcccc cc lcccc | }
   %% \toprule
    %\multicolumn{5}{l}{Men's preferences}& & \multicolumn{5}{l}{Women's preferences} \\
       %$m_{1}$ : & $w_{1}$ & $w_{2}$ & $w_{3}$ &$w_{4}$ & &  $w_{1}$ : & $m_{4}$ & $m_{3}$ & $m_{2}$ &$m_{1}$\\ 
       %$m_{2}$: & $w_{2}$ & $w_{1}$ & $w_{4}$ & $w_{3}$ & & $w_{2}$ : & $m_{3}$ & $m_{4}$ & $m_{1}$ & $m_{2}$  \\
        %$m_{3}$: & $w_{3}$ & $w_{4}$ & $w_{1}$ & $w_{2}$ & & $w_{3}$ : & $m_{2}$ & $m_{1}$ & $m_{4}$ & $m_{4}$ \\
        %$m_{4}$: & $w_{4}$ & $w_{3}$ & $w_{2}$ & $w_{1}$ & & $w_{4}$ : & $m_{1}$ & $m_{2}$ & $m_{3}$ & $m_{4}$ \\
      %% \hline
  %\end{tabular}
  %\label{Table:Knuth}
  %\caption{Knuth's example exhibiting 10 stable matchings}
%\end{center}
  %\end{table}

McDermid and Irving~\cite{DBLP:journals/algorithmica/McDermidI14} also studied exact exponential-time algorithms for \NPH\ special cases of \sesm. Specifically, if the preference lists on one side have length at most $\ell$ and
 there is no upper bound on the length of the preference lists on the other side, then they showed that given any $\epsilon>0$, \sesm\ can be solved in time  $(2^{\alpha n}+2^{\beta})\cdot n^{\OO(1)}$, where $\alpha =(5-2\sqrt{4})(\ell-2+\epsilon)$ and $\beta=(\ell-1)/2\epsilon$. For small enough values of $\epsilon$, the time complexity is close to $\mathcal{O}(1.0726^{n})$ for $\ell=3$, $\mathcal{O}(1.1504^{n})$ for $\ell=4$ and $\mathcal{O}(1.2339^{n})$ for $\ell=5$. 
Curiously, Romero-Medina \cite{Romero-Medina98} gave an exact algorithm for \sesm\ where there are no restrictions on the preference lists and claimed without proof that the algorithm runs in polynomial time. While Romero-Medina did not cite Kato's work \cite{Kato93}, the claim is an obvious contradiction to Kato's proof of NP-hardness (unless NP=P). Also, the time complexity of Romero-Medina's algorithm is likely to be worse than McDermid and Irving's when the length of the preference lists are bounded. 

In the context of the approximability of \sesm, for a given instance of {\sc SM}, let $\mu_{M}$ and $\mu_{W}$ denote the man-optimal and woman-optimal stable matchings (see Section \ref{sec:prelimsSM}), respectively. We define $\Lambda=\max\{|\delta(\mu_{M})|, |\delta(\mu_{W})|\}$, where $\delta(\mu)$ denotes the sex-equal function value of the matching $\mu$. Iwama \etal~\cite{DBLP:journals/talg/IwamaMY10} gave a polynomial-time algorithm that finds a {\it near-optimal} solution to \sesm. Formally, they showed that given some fixed $\epsilon>0$, there is an $\OO(n^{3+1/\epsilon})$ time algorithm that returns a stable matching $\mu$ such that $-\epsilon \Lambda \leq \delta(\mu)\leq \epsilon \Lambda$, or returns that no such $\mu$ exists. Furthermore, they exhibited an instance with two near-optimal stable matchings that have different egalitarian-cost measure. This prompted the authors to define the {\sc Minimum Egalitarian Sex-Equal Stable Marriage} problem. For this problem, shown to be \NPH, they gave a polynomial time algorithm whose approximation ratio is smaller than 2.

We note that \sesm\ was also studied for genetic and ant colony-based algorithms \cite{Nakamura95,VienViet07}. Quite recently, an empirical study on \sesm\ was undertaken by Giannakopoulos \etal~\cite{Giannakopoulos15}.

%\smallskip
\myparagraph{Balance.}
The {\sc BSM} problem was introduced in the influential work of Feder \cite{Feder95} on stable matchings. Intuitively, it is defined to be a stable matching in the instance that is desirable to both the sexes (see Section \ref{sec:prelimsOptimization}), i.e. it simultaneously minimizes the dissatisfaction of both sexes. It has been noted that it is not trivial to construct an instance of {\sc SM} in which no balanced stable matching is a sex-equal matching, and vice versa \cite{DBLP:books/ws/Manlove13}. However, one such example (of infinitely many instances), attributable to Eric McDermid, has been discussed by Manlove in \cite[pg 109]{DBLP:books/ws/Manlove13}. Feder \cite{Feder95} proved that this problem is \NPH\ and that it admits a 2-approximation algorithm. Later, it was shown that this problem also admits a $(2-1/\ell)$-approximation algorithm where $\ell$ is the maximum size of a set of acceptable partners \cite{DBLP:books/ws/Manlove13}. O'Malley \cite{omalleyThesis} phrased the {\sc BSM} problem in terms of constraint programming. McDermid, as reported in \cite{DBLP:books/ws/Manlove13}, proved that the measure of balance of {\sc BSM} is incomparable to the measure of fairness of {\sc SESM}. Finally, in the thesis of McDermid \cite{mcThesis} and the conclusion of McDermid and Irving \cite{DBLP:journals/algorithmica/McDermidI14}, the authors expressed interest in future studies of the {\sc BSM} problem with respect to treewidth.

%\smallskip
\myparagraph{Maximum/Minimum Cardinality.} When the preference lists have {\em ties}, there are three different notions of stability: super stability, strong stability and weak stability (see \cite{IRVING1994,Manlove2002}). Our work is centered around weak stability. In the presence of ties, the existence of a (weakly) stable matching is guaranteed; simply break the ties arbitrarily and run the Gale-Shapley algorithm \cite{gale62a} on the resulting instance. A stable matching in the new instance is (weakly) stable in the original instance. However, the breaking of ties affects the size of the stable matching produced. Thus, the size of the stable matchings are no longer exactly the same as in the case where preference lists are strict (see Section \ref{sec:prelimsSM}). This engenders the study of the computation of a {\it maximum (minimum) cardinality stable matching}, known as the \maxsmt (\minsmt)~problem, which capture scenarios where we would like to maximize (minimize) available resources. 

Irving \etal~\cite{IIMMM02j} showed that both \maxsmt and \minsmt are \NPH\ even if the inputs are restricted to have ties for only one sex, preference lists are of bounded length, and there is symmetry in the preference lists. Irving \etal~\cite{IrvingMaloveOMalley09} showed that \maxsmt is solvable in polynomial time if the length of the preference lists of one sex is at most 2, and the length of the preferences of the other sex is unbounded. Furthermore, it is shown that \maxsmt is not $\alpha$-approximable unless $\textsf{P}=\NP$, for some $\alpha>1$, even if each man's preference list is of length at most 3, and each woman's preference list is of length at most 4. Given the large ``gap'' in the computational complexity of these results, perhaps it is not surprising that this has led to the study of \maxsmt from the perspectives of approximation and parameterized~complexity.

In the context of approximation algorithms, note that it is easy to obtain a factor $2$-approximation---break the ties arbitrarily, and return some stable matching in the resulting instance. A breakthrough was achieved by Iwama \etal~\cite{DBLP:journals/talg/IwamaMY10}, who  obtained a factor $1.875$-approximation algorithm for \maxsmt using a local search technique. Kir\'{a}ly~\cite{kiraly2011} improved upon this result, and introduced a simple effective technique of promotion to break ties in a modified Gale-Shapley algorithm. In particular, he improved the approximation ratio to $5/3=1.66$ for  \maxsmt, and to $1.5$ for the one-sided ties version of \maxsmt, that is, the preference lists have to be strict on one side, while ties are permitted in the preference lists of the other side. For the one-sided version of \maxsmt, Huang and Kavitha~\cite{HuangK15} gave an approximation algorithm with factor $1.4667$, and Dean and Jalasutram \cite{MatchUp15} gave an approximation algorithm with factor $1.4615$. The best known approximation algorithm for \maxsmt is a factor $1.5$-approximation algorithm in \cite{Mcdermid09}. References to additional works addressing approximation algorithms and inapproximabilty results for \maxsmt\ can be found in \cite{MatchUp15}.

Marx and  Schlotter~\cite{MarxS10} studied \maxsmt with several parameters: (i) the maximum number of ties in an instance ($\kappa_1$); (ii) the maximum length of ties in an instance  ($\kappa_2$); (iii) the total length of the ties in an instance ($\kappa_3$).  The authors showed that \maxsmt is \WOH\ parameterized by $\kappa_1$, and \FPT parameterized by $\kappa_3$. Furthermore,  since it was shown that  \maxsmt is \NPH\ even when the length of each tie is at most 2, we do not hope to have an algorithm with running time $f(\kappa_2)n^{g(\kappa_2)}$, for any functions $f$ and $g$ that depend only on $k$.

%, such as the the maximum number of ties in an instance, the maximum length of ties in an instance, and  the total length of the ties in an instance. They also showed that \SMTI is {\sf W}-hard parameterized by the first and third of these parameters. 

Relatively less work has been done for \minsmt. Beyond what has been mentioned earlier, we would like to mention that lower bounds on the approximability of \minsmt has been discussed in Yanagisawa's Masters thesis \cite{Yanagisawa-Th} and in the paper \cite{Halldorsson03}. 
Finally, we remark that experimental approaches have also been undertaken to study \maxsmt. Munera \etal~\cite{munera2015solving} gave an algorithm based on local search. Moreover, Gent and Prosser~\cite{gent2002empirical} formulated the problem as a constrained optimization problem and gave an algorithm via constrained programming for both decision and optimization version. 

\vspace{-0.7em}
\subsection{Our Contribution}
\vspace{-0.5em}

Our contribution can be summarized in three main theorems, namely Theorems \ref{thm:w1hardIntro}, \ref{thm:fptIntro} and \ref{thm:sethIntro}. The proof of Theorem \ref{thm:xpIntro} is straightforward (unlike the other theorems), yet we present it for the sake of completeness. The approaches employed to establish each of these theorems are discussed in detail in Section \ref{sec:overview}. 
The principles underlying each of these approaches are quite general, and therefore they can be applicable to other parameterizations of these problems as well as to other problems related to {\sc SM}. Here, we only present the statements of our findings. Our first set of results analyzes the parameterized complexity of the {\sc SESM}, {\sc BSM}, {\sc max-SMT} and {\sc min-SMT} problems with respect to the treewidth of the primal graph.

\begin{theorem}\label{thm:w1hardIntro}
The {\sc SESM}, {\sc BSM}, {\sc max-SMT} and {\sc min-SMT} problems are all \WOH\ with respect to $\tw$, the treewidth of the primal graph. Moreover, unless \ETH\ fails, none of these problems can be solved in time $f(\tw)\cdot n^{o(\tw)}$ for any function $f$ that depends only on $\tw$.
\end{theorem}

Next, we observe that it is straightforward to derive \XP-algorithms whose running times, in light of Theorem \ref{thm:w1hardIntro}, are {\em essentially tight}.

\begin{theorem}\label{thm:xpIntro}
The {\sc SESM}, {\sc BSM}, {\sc max-SMT} and {\sc min-SMT} problems are all solvable in time $n^{\OO(\tw)}$, where $\tw$ is the treewidth of the primal graph.
\end{theorem}

Due to the barrier posed by Theorem \ref{thm:w1hardIntro}, we next turn to analyze the treewidth of the rotation digraph. Here, we only study {\sc SESM} and {\sc BSM}, as the rotation digraph is not defined in the context of {\sc max-SMT} and {\sc min-SMT}. On the positive side, we establish the following theorem.

\begin{theorem}\label{thm:fptIntro}
The {\sc SESM} and {\sc BSM} problems are both solvable in time $2^{\tw}\cdot n^6$, where $\tw$ is the treewidth of the rotation digraph.
\end{theorem}

Finally, we prove that unless \SETH\ fails, Theorem \ref{thm:fptIntro} {\em pinpoints precisely} the running times of \FPT-algorithms for both {\sc SESM} and {\sc BSM}.

\begin{theorem}\label{thm:sethIntro}
Unless \SETH\ fails, neither {\sc SESM} nor {\sc BSM} is solvable in time $(2-\epsilon)^{\tw}\cdot n^{\OO(1)}$ for any fixed $\epsilon>0$, where $\tw$ is the treewidth of the rotation digraph.
\end{theorem}

Thus, we present a comprehensive, complete picture of the behavior of central optimization versions of the {\sc SM} problem with respect to the parameter treewidth of both the primal graph and the rotation digraph. (We remark that along the way, we thus also resolve open problems posed by  McDermid \cite{mcThesis} and McDermid and Irving \cite{DBLP:journals/algorithmica/McDermidI14}.)

%% file: prelims.tex
% !TEX root = mainTW.tex

\vspace{-0.7em}
\section{Preliminaries}\label{sec:prelims}
\vspace{-0.5em}

Standard graph-theoretic terms not explicitly defined here can be found in \cite{DBLP:books/daglib/0030488}, and for standard notions in Parameterized Complexity, refer to Appendices \ref{sec:PC} and \ref{sec:tw}. Given a non-negative integer $n$, we use $[n]_0$ and $[n]$ to denote the sets $\{0,1,\ldots,n\}$ and $\{1,2,\ldots,n\}$, respectively. Given a function $f: A\rightarrow B$, $\domain(f)$ and $\image(f)$ denote the domain and the image of $f$, respectively.

\vspace{-0.7em}
\subsection{Stable Marriage}\label{sec:prelimsSM}
\vspace{-0.5em}

In the classic {\sc Stable Marriage (SM)} problem, the input consists of a set of men, $M$, and a set of women, $W$. The set of agents (men and women) is denoted by $A=M\cup W$. The total number of agents, $|A|$, is denoted by $n$. Each man (woman) has a {\em preference list}, which is a list ranking a subset of $W$ ($M$). More precisely, each man $m\in M$ is assigned a subset $W'\subseteq W$ and an {\em injective} function $\pos_m: W'\rightarrow [|W'|]$. Symmetrically, each woman $w\in W$ is assigned a subset $M'\subseteq M$ and an {\em injective} function $\pos_w: M'\rightarrow [|M'|]$.\footnote{Throughout our paper, we do not assume that each person must rank {\em all} people of the opposite sex. That is, we deal with the general case where preference lists may be incomplete. For emphasis, some papers add the letter ``I'' to the abbreviation {\sc SM}, but for the sake of brevity, we avoid this addition.} For all $m\in M$ and $w\in W$, it holds that $w\in\domain(\pos_m)$ if and only if $m\in\domain(\pos_w)$. 
The case where for every agent $a\in M\cup W$, the function $\pos_a$ may {\em not} be injective, is known as the {\sc SM with Ties (SMT)} problem. In this generalization of {\sc SM}, for every agent $a\in M\cup W$, the image of $\pos_a$ is restricted to be of the form $[t]$ for $t\in\mathbb{N}$.
The formulation of the objectives of {\sc SM} and {\sc SMT} relies on the notion of {\em stability}.

\begin{definition}
Given $M'\subseteq M$ and an injective function $\mu : M'\rightarrow W$, we say that a pair $(m,w)$ of a man $m\in M$ and a woman $w\in W$ such that $w\in\domain(\pos_m)$ is a {\em blocking pair of $\mu$} if {\bf (i)} $m\notin\domain(\mu)$ and $w\notin\image(\mu)$, or {\bf (ii)} $m\notin\domain(\mu)$ and $\pos_w(m) < \pos_w(\mu^{-1}(w))$, or {\bf (iii)} $\pos_m(w) < \pos_m(\mu(m))$ and $w\notin\image(\mu)$, or {\bf (iv)} $\pos_m(w) < \pos_m(\mu(m))$ and $\pos_w(m) < \pos_w(\mu^{-1}(w))$.
\end{definition} 

\begin{definition}
Given $M'\subseteq M$ and an injective function $\mu : M'\rightarrow W$, we say that $\mu$ is a {\em stable matching} if for every $m'\in M'$, $\mu(m')\in\domain(\pos_{m'})$, and $\mu$ has no blocking pair.
\end{definition} 

Roughly speaking, a stable matching is a matching between a subset of men and a subset of women such that there does not exist a pair of a man and a woman who prefer each other to their matched partners (if at all such partners exist). To simplify our presentation, we use the notation $(m,w)\in\mu$ to indicate that $m\in\domain(\mu)$ and it holds that $\mu(m)=w$. Moreover, we let $\cal S$ denote the set of all stable matchings. In the seminal paper \cite{gale62a}, Gale and Shapley showed that there always exists at least one stable matching.

\begin{proposition}[\cite{gale62a}]\label{lem:nonEmptyS}
The set ${\cal S}$ is non-empty.
\end{proposition}

Thus, the objective of {\sc SM} and {\sc SMT} is to {\em find} a stable matching. However, there can be an exponential number of stable matchings \cite{DBLP:books/daglib/0066875}. Notably, Gale and Sotomayor \cite{GaleMarilda85} showed that in the absence of ties, stable matchings do not differ in which men and women they match.

\begin{proposition}[\cite{GaleMarilda85}]\label{lem:matchSame}
In {\sc SM}, for all $\mu,\mu'\in{\cal S}$, $\domain(\mu)=\domain(\mu')$ and $\image(\mu)=\image(\mu')$.
\end{proposition}

In the absence of ties, we denote $M^\star=\{m\in M: m\in\bigcap_{\mu\in{\cal S}}\domain(\mu)\}$ and $W^\star=\{w\in W: w\in\bigcap_{\mu\in{\cal S}}\image(\mu)\}$. Note that by Proposition \ref{lem:matchSame}, we have that $M^\star=\{m\in M: m\in\bigcup_{\mu\in{\cal S}}\domain(\mu)\}$ and $W^\star=\{w\in W: w\in\bigcup_{\mu\in{\cal S}}\image(\mu)\}$. Finally, we denote $A^\star=M^\star\cup W^\star$.

We also recall the notions of man- and woman-optimal stable matchings.

\begin{proposition}[\cite{gale62a}]\label{lem:menOptimal}
In {\sc SM}, there is exactly one stable matching $\mu$, denoted by $\mu_M=\mu_\emptyset$, that minimizes $\sum_{(m,w)\in \mu}\pos_m(w)$. Symmetrically, there is exactly one stable matching $\mu$, denoted by $\mu_W$, that maximizes $\sum_{(m,w)\in \mu}\pos_m(w)$. Both $\mu_M$ and $\mu_W$ can be found in time $\OO(n^2)$.
\end{proposition}

We remark that  $\mu_M$ is also the unique stable matching that maximizes $\sum_{(m,w)\in \mu}\pos_w(m)$, and $\mu_W$ is also the unique stable matching that minimizes $\sum_{(m,w)\in \mu}\pos_w(m)$ \cite{gale62a}. From now onwards, by Proposition \ref{lem:menOptimal}, we assume that we have $\mu_M$ and $\mu_W$ at hand. 

Finally, the {\em primal graph} $H$ of an instance of {\sc SM} or {\sc SMI} is the bipartite graph whose vertex set is $\{v_a: a\in A\}$ and whose edge set is $\{\{v_m,v_w\}: m\in\domain(\pos_w)\}$.

\vspace{-0.7em}
\subsection{SM and SMT: Optimization}\label{sec:prelimsOptimization}
\vspace{-0.5em}

\myparagraph{Sex-Equality and Balance.} We first discuss the scenario where ties are forbidden, that is, we first present two optimization versions of {\sc SM}. Here, the set of stable matchings, $\cal S$, can be viewed as a {\em spectrum} where the two extremes are the man-optimal stable matching and the woman-optimal stable matching. Naturally, it is desirable to analyze matchings that lie somewhere in the middle, being {\em fair towards both sides} or {\em desirable by both sides}. Deciding which notion best describes an appropriate outcome depends on the specific situation at hand. Here, the quantity $\pos_{a}(\mu(a))$ is viewed as the ``satisfaction'' of $a$ in $\mu$, where a smaller value signifies a greater amount of satisfaction. For a stable matching $\mu\in{\cal S}$, the total satisfaction of men from $\mu$ is $\sat_M(\mu)=\sum_{(m,w)\in \mu}\pos_m(w)$, and the total satisfaction of women from $\mu$ is $\sat_W(\mu)=\sum_{(m,w)\in \mu}\pos_w(m)$.

In the {\sc Sex-Equal Stable Marriage (SESM)} problem, we seek a stable matching that is {\em fair towards both sides} by minimizing the difference between their individual amounts of satisfaction. The formulation of this problem relies on the notion of the sex-equality measure:

\begin{definition}\label{def:sexEqMeasure}
The {\em sex-equality measure} is the function $\delta:{\cal S}\rightarrow\mathbb{Z}$ such that for all $\mu\in{\cal S}$,
\vspace{-0.5em}
\[\delta(\mu) = \sat_M(\mu) - \sat_W(\mu).\]
\end{definition}

The best value that this measure attains is $\displaystyle{\Delta = \min_{\mu\in{\cal S}}|\delta(\mu)|}$. The objective of the {\sc SESM} problem is to find a stable matching $\mu$ such that $\delta(\mu)=\Delta$.\footnote{We only compute $\Delta$. By backtracking our computations, it is possible to construct a stable matching $\mu$ such that $\delta(\mu)=\Delta$. This remark is also relevant to our algorithms for {\sc BSM}, {\sc max-SMT} and {\sc min-SMT}.}

In {\sc Balanced Stable Marriage}, the objective is to find a stable matching that is {\em desirable by both sides}. Here, we rely on the notion of the {\em balance measure}, which is defined as follows.

\begin{definition}\label{def:balMeasure}
The {\em balance measure} is the function $\bal:{\cal S}\rightarrow\mathbb{Z}$ such that for all $\mu\in{\cal S}$,
\vspace{-0.5em}
\[\bal(\mu) = \max\{\sat_M(\mu), \sat_W(\mu)\}.\]
\end{definition}

The best value that this measure attains is $\displaystyle{\Bal = \min_{\mu\in{\cal S}}|\bal(\mu)|}$. The task of {\sc BSM} is to find a stable matching $\mu$ such that $\bal(\mu)=\Bal$. At first sight, this measure might seem conceptually similar to the previous one, but in fact, the two measures are quite different. Indeed, {\sc BSM} examines the amount of dissatisfaction of each party {\em individually}, and attempts to minimize the worse one among the two. This problem fits the scenario where each party is selfish in the sense that it wishes to minimize its own dissatisfaction irrespective of the dissatisfaction of the other party. Here, our goal is to find a matching desirable by both parties by ensuring that each individual amount of dissatisfaction does not exceed some threshold. In some situations, the minimization of $\bal(\mu)$ may indirectly also minimize $\delta(\mu)$, and vice versa, yet in general, this is not true. Indeed, it is known how to construct a family of instances where there does {\em not} exist any matching that is both a sex-equal stable matching and a balanced stable~matching \cite{DBLP:books/ws/Manlove13}.

\smallskip
\myparagraph{Maximum/Minimum Size.} We now present two optimization versions of {\sc SMT}. Here, the two (arguably) most natural objectives are to maximize or minimize the size of the outputted stable matching as it might be desirable to maintain stability while either maximizing or minimizing the use of available ``resources''. 
These objectives define the well-known {\sc max-SMT} and {\sc min-SMT} problems. Formally, given an instance of {\sc SMT}, the task of {\sc max-SMT} is to find a stable matching of maximum size, and the task of {\sc min-SMT} is to find a stable matching of minimum size. Here, the size of a matching is simply its number of matched pairs. We remark that due to Proposition \ref{lem:matchSame}, the study of both of these problems only makes sense in the presence~of~ties.

\vspace{-0.7em}
\subsection{Rotation Digraph}\label{sec:prelimsRotDig}
\vspace{-0.5em}

First, we stress that the rotation digraph is a notion defined only in the context of {\sc SM}, that is, in the absence of ties. Let us start by defining a rotation, which is an operation that transforms one stable matching to another. For this purpose, given $\mu\in{\cal S}$ and $m\in M^\star$, we let $s_\mu(m)$ denote the first woman $w$ succeeding $\mu(m)$ in $m$'s preference list, such that $w\notin\image(\mu)$ or $w$ prefers $m$ over $\mu^{-1}(w)$ (if such a woman exists). Now, a rotation is defined as follows.

\begin{definition}\label{def:rotation}
Let $\mu\in{\cal S}$. A {\em $\mu$-rotation} is an ordered sequence of pairs $\rho=((m_0,w_0),(m_1,w_1),$ $\ldots,(m_{r-1},w_{r-1})$ for some $r\in\{2,3,\ldots,n\}$, such that for all $i\in[r-1]_0$, $(m_i,w_i)\in\mu$ and $w_{(i+1)\mathrm{mod}\ r}=s_\mu(m_i)$.
For all $i\in[r-1]_0$, we say that $\rho$ {\em involves} $m_i$ and $w_i$.
\end{definition}

When $\mu$ is immaterial or clear from context, the term {\em rotation} replaces the term {\em $\mu$-rotation}. Given a $\mu$-rotation $\rho$, the {\em elimination of $\rho$} is the operation that modifies $\mu$ by matching each $m_i$ with $w_{(i+1)\mathrm{mod}\ r}$ rather than $w_i$. This operation results in a stable matching \cite{DBLP:journals/siamcomp/IrvingL86}. Let $R$ denote the set of all sequences of pairs $\rho$ for which there exists $\mu\in{\cal S}$ such that $\rho$ is a $\mu$-rotation. It is known that $|R|\leq n^2$. Moreover, for all $\rho\in R$, the agents involved in $\rho$ belong to $A^\star$ \cite{DBLP:books/daglib/0066875}.

\begin{proposition}[\cite{DBLP:journals/siamcomp/IrvingL86}]
Let $\mu\in{\cal S}$. There is a unique subset of $R$, denoted by $R(\mu)$, such that starting with $\mu_\emptyset$, there is an order in which the rotations in $R(\mu)$ can be eliminated to obtain $\mu$.
\end{proposition}

Irving and Leather \cite{DBLP:journals/siamcomp/IrvingL86} studied the {\em rotation poset} $\Pi=(R,\prec)$. Here, $\prec$ is a partial order on $R$ such that $\rho\prec\rho'$ if and only if for every $\mu\in{\cal S}$ such that $\rho'\in R(\mu)$, $\rho\in R(\mu)$ as well. We say that $R'\subseteq R$ is a {\em closed set} if there does not exist $\rho\in R\setminus R'$ and $\rho'\in R'$ such that $\rho\prec\rho'$. Moreover, given $R'\subseteq R$, we let $\cl(R')$ denote the smallest closed set that contains $R'$. We also say that an order in which the rotations in $R'$ are eliminated is {\em $\prec$-compatible} if for all $\rho,\rho'\in R$ such that $\rho\prec\rho'$, $\rho$ is eliminated before $\rho'$.
Roughly speaking, the rotation poset describes how every stable matching can be derived from the man-optimal stable matching. More precisely,

\begin{proposition}[\cite{DBLP:journals/siamcomp/IrvingL86}]\label{lem:smCorrRots}
Let $R'\subseteq R$ be a closed set. Starting with $\mu_\emptyset$, eliminating the rotations in $R'$ in any  $\prec$-compatible order  is valid---at each step, where our current stable matching is some $\mu$, the rotation that we eliminate next is a $\mu$-rotation. Moreover, all $\prec$-compatible orders in which one eliminates the rotations in $R'$ result in the same stable matching.
\end{proposition}

Given $R'\subseteq R$, let $\mu_{R'}$ denote the stable matching $\mu$ such that $R(\mu)=\cl(R')$. By Proposition \ref{lem:smCorrRots}, this notation is well defined. Moreover, given $a\in A$, let $R(a)$ denote the set of all rotations that involve $a$. Our \FPT\ algorithms will crucially rely on the following proposition.

\begin{proposition}[\cite{DBLP:books/daglib/0066875}]\label{lem:totalOrder}
For any $m\in M$, $\prec$ is a total order on $R(m)$.
\end{proposition}

Irving, Leather and Gusfield \cite{DBLP:journals/jacm/IrvingLG87} studied digraphs that are a compact representation of $\Pi$. Specifically, we say that a digraph is the {\em rotation digraph of $\Pi$}, denoted by $D_\Pi$, if it is the directed acyclic graph (DAG) of minimum size whose transitive closure is isomorphic to $\Pi$.

\begin{proposition}[\cite{DBLP:journals/jacm/IrvingLG87}]\label{prop:rotDig}
The rotation digraph $D_\Pi$ can be computed in time $\OO(n^2)$.
\end{proposition}

We let $G_\Pi$ denote the underlying undirected graph of $D_\Pi$. In light of Proposition \ref{prop:rotDig}, when we design our algorithms, we may assume that we have $D_\Pi$ and $G_\Pi$ at hand.

%% file: overview.tex
\vspace{-0.7em}
\section{Overview}\label{sec:overview}
\vspace{-0.5em}

In this section, we explain the main ingredients underlying our results. % Here, we focus on common themes in their proofs.

\smallskip
\myparagraph{W[1]-Hardness Results.} Our first set of results, given in Section \ref{sec:w1hardness}, establishes Theorem \ref{thm:w1hardIntro}. The source problem of the reductions developed in this section is the {\sc Multicolored Clique} problem. Our four constructions share common features, which we believe to be relevant to other reductions meant to prove that optimization versions of {\sc SM} and {\sc SMT} are \WOH\ with respect to various structural parameters. First, all of our reductions introduce the same sets of ``basic agents''. Roughly speaking, we introduce {\em two} basic men to represent each color class, {\em two} basic women to represent each vertex, and {\em one} basic woman to represent each edge. Second, the preference lists of the basic men representing color classes are set in a special form, to which we refer as the form of a {\em leader}. Informally, the preference lists of the two men representing a color class are {\em distorted} mirror images of one another. More precisely, each man among these two men ranks his own set of women representing vertices of the appropriate color class, where if we view the two women representing the same vertex as the same woman, then it is seen that the order in which one ranks these women is opposite to the order in which the other one ranks them. Both of these men ``embed'' in their preference lists women representing edges between women representing vertices, and both of them prefer a woman representing vertex $v$ over all women representing edges incident to $v$. The third common feature is that all of our four reductions then proceed to introduce {\em similar} sets of agents that are meant to construct vertex and edge selection gadgets. In particular, in all of our reductions, the only interaction between vertex selector gadgets and edge selector gadgets is via the special men whose preference lists are of the form of a leader. Moreover, all of our reductions introduce quite similar definitions of ``enriched'' sets of agents who {\em locally} interact with basic agents. Hence, in all of our reductions, we ensure in a somewhat similar manner that basic agents, excluding those special men whose preference lists are of the form of a leader, can replace partners in a manner that does not enforce too many other ``close-by'' men to change partners as well in order to maintain stability.

The principles described above are quite general as we prove that they are useful for two natural ``types'' of optimization problems that may a priori seem different. The first type is the one where the challenge lies in the  {\em output} that is enforced to comply with a {\em satisfaction target value}, and the second type is the one where the challenge lies in the {\em input} that is generalized to include ties while the output only needs to be either large enough or small enough. When we examine each type of problems separately, it is revealed that our reductions to {\sc SESM} and {\sc BSM} share many other similar ideas, and the same holds true when we compare our reductions to {\sc max-SMT} and {\sc min-SMT}. In fact, our reduction to {\sc BSM} is a modification of our reduction to {\sc SESM} where we carefully plug-in different numbers of ``dummy agents'' into the preference lists of some ``central'' agents to manipulate the different men to choose partners in a coordinate manner without relying on stability, but merely on subtle analysis of the satisfaction value.

For a concrete illustrative example, let us give a high-level overview of other elements incorporated in one of our reductions, namely, the reduction to {\sc SESM}. Here, we begin by introducing the so called Original Vertex Selector gadget, which is a structure where one man representing a color class selects one woman representing a vertex of that color class as his partner. The selection of the partners of the women representing all other vertices of that color class is done locally by introducing an enriched set of men. Then, we introduce the Mirror Vertex Selector gadget, which handles the other man representing the same color class in an {\em almost} symmetric manner. By embedding dummy agents into the preference lists of the basic agents involved in Original and Mirror Vertex Selector gadgets, we can ensure that in all of these gadgets together, a predetermined number of men would be matched to their most preferred women. We remark that the ability to set such a predetermined value is crucially dependent on the fact that each color class is represented by {\em two} gadgets. Next, to ensure that the two gadgets representing a color class are consistent (in the sense that the two men representing the same color class are matched to women representing the same vertex), we introduce a new special gadget per vertex $v$, called the Consistency gadget. This gadget consists of only four agents, where one of the men in this gadget ranks both women (outside the gadget) that represent the vertex $v$. Here, we need to ensure that among all of the gadgets representing the same color class, the ``configuration'' of exactly one gadget will be such that its men will not attain their best partner. Here, the gadgets cannot interact directly by introducing, for example, new common agents that agents of different Consistency gadgets would rank, as we need to ensure that the treewidth of the primal graph of the output instance is small. Hence, to coordinate between these gadgets, we rely on carefully chosen numbers of dummy agents that are added to the preference lists of their agents. The numbers involved in this gadget are of a different magnitude than those involved in the Original and Mirror Vertex Selector gadgets. However, we cannot also assign each color class a number of a different magnitude, as then we would end up assigning numbers of magnitudes such as $n^k$, which means that we would need to insert $n^k$ dummy agents and hence the construction would not be done in ``FPT time''. Nevertheless, we are able to overcome this difficulty by using a simple equation that has a unique solution of the form that we want, and using the appropriate coefficients as a guide for the number of dummy agents to be inserted.\footnote{By using a different equation, we are able to reuse our construction in the context of {\sc BSM}, and hence it seems like the applicability of our construction is quite broad (where one only needs to be able tune the equation according to the target measure at hand).}

Afterwards, we introduce the Edge Selector gadget. Here, we define one gadget per edge, which involves the woman representing that edge. Such a gadget indicates that an edge has been selected by being in the configuration where the woman representing the edge has {\em not} attained her best partner. Notice that unlike the previously discussed gadgets, which are analyzed from the perspective of the men, these gadgets are analyzed from the perspective of the women. In particular, while in the Consistency gadgets the total satisfaction of women is forced to be low, here the total satisfaction of the women would be forced to be high. However, we stress that this difference is not employed to attain a certain target value, but it is used to control which matchings are stable and to avoid introducing any form of direct interaction between Edge Selector gadgets and Consistency gadgets (which is necessary to ensure that the treewidth of the output is small). In particular, having set up all of the previously mentioned gadgets, we still encounter a significant imbalance between the satisfaction of men and women in the stable matchings of the form that we would like to represent solutions. However, this issue is easily handled by introducing ``garbage collector'' agents which counterweight this imbalance properly.

\smallskip
\myparagraph{\XP-Algorithms.} Our second set of results, given in Section \ref{sec:xp}, establishes Theorem \ref{thm:xpIntro}. The proof of this theorem is based on a standard application of the method of dynamic programming over nice tree decompositions, and it is sketched in this paper only for the sake of completeness.

\smallskip
\myparagraph{\FPT-Algorithms.} Our third set of results, given in Section \ref{sec:fpt}, leads to the establishment of Theorem \ref{thm:fptIntro}. For this purpose, we present an approach that deviates from standard applications of DP over tree decompositions. First, the proof of its correctness integrates new insights into the structure of rotation digraphs that might be of independent interest. To formulate these insights, we introduce new notions that may be adapted to tackle other optimization versions of {\sc SM}.
Second, while in standard DP elements that have not yet been examined determine only how to extend/modify partial solutions, in our DP such elements (rotations) are part of the partial solutions themselves. Thus, we need to design a delicate mechanism that maintains consistency between the manner in which we handled rotations in the bags below the current one, and the manner in which we anticipate handling rotations in other bags.
Third, we face difficulties stemming from the fact that to update partial solutions, we need to change the assignments of women to men that correspond to these partial solutions, yet the information we have at hand does not directly reveal the assignments. By associating a directed path with each man, and tracing the manner in which the path ``enters and leaves'' every bag of the tree decomposition, we are able to deduce sufficient information on the assignments. We remark that this solution introduces yet another difficulty, namely, the need to store and maintain ``illegal scores'' that are associated with assignments of several men to the same woman.

We proceed with a more detailed description of some ingredients of our approach. First, to unify the principles underlying our approach, we introduce a problem more generic than {\sc SESM} and {\sc BSM},\footnote{By analyzing {\sc SESM} and {\sc BSM} separately, we also show how to speed-up the generic algorithm.} where the objective is to determine, for all $t_M,t_W\in\mathbb{N}$, whether there exists $\mu\in{\cal S}$ such that both $\sat_M(\mu)=t_M$ and $\sat_W(\mu)=t_W$. 
To describe our approach compactly, we note that a {\em state} is a pair $(v,R')$ where $v\in V(T)$ and $R'\subseteq \beta(v)$. Moreover, $\mu\in{\cal S}$ is said to be {\em compatible with} a state $(v,R')$ if $R(\mu)\cap\beta(v)=R'$. Now, to design our algorithm, we first observe that each man can be associated with a directed path of $D_\Pi$ whose vertex-set is a superset of $R(m)$. With respect to a given state, we then proceed to introduce special ``entry'' and ``exit'' points for each man, based on the directed path associated with him. These special points allow us to further identify a subpath of the path of each man that captures the current most updated information that we have about his partner, where the internal vertices (which are rotations) of the paths encompass all changes that might occur when we analyze future states. We are then able to define exactly which men have been already settled with a partner. Next, again with respect to a given state, but also with respect to a stable matching $\mu$, we assign a woman to each man. This woman may be either the same one that $\mu$ assigns to $m$ (if $m$ is settled) or some specific woman whose choice is based on the most updated information we could extract. We remark that such an assignment, due to inherent uncertainties at intermediate steps of the computation, may assign several men to the same woman. Having assigned tentative partners to men, we are able to introduce definitions related to tentative amounts of satisfaction.

The description above sets the background for the study of our algorithm. The algorithm itself is very short and easily implementable. The computation simply fills a table that consists of Boolean entries of the form {\sf N}$[v,R',t_M,t_W]$, where $v\in V(T)$, $R'\subseteq\beta(v)$ and $t_M,t_W\in [n^2]$. Here, {\sf N}$[v,R',t_M,t_W]=1$ if and only if there exists $\mu\in{\cal S}$ that is compatible with the state $(v,R')$ and where the ``tentative'' amounts of satisfaction of men and women are $t_M$ and $t_W$, respectively. We stress that these amounts of satisfaction are {\em not} $\sat_M(\mu)$ and $\sat_W(\mu)$ (such amounts simply {\em cannot} be computed when we handle the entry {\sf N}$[v,R',t_M,t_W]$ since we do not have enough information at hand at that point to extract them). In the computation of an introduce node, in particular, we need to correct our tentative amounts of satisfaction. Having defined the algorithm, the technical part of the analysis begins. Here, we need to carefully analyze each type of node of the nice tree decomposition, and prove that the definitions we have set up as background indeed allow us to trace the paths of the men correctly, and to obtain precise amounts of satisfaction at the end. In this context, we present a sequence of lemmata (for each type of node) that verify consistencies between types of men, partners and tentative amounts of satisfaction deduced for the current node and for the child(ren) of the current node.

\smallskip
\myparagraph{Lower Bounds Based on \SETH.} Our last set of results, given in Section \ref{sec:seth}, establishes of Theorem \ref{thm:sethIntro}. Here, by plugging in coefficients associated with two different equations to (essentially) the same construction, we are able to handle both {\sc SESM} and {\sc BSM}. Let us now give a high-level overview of the reduction to {\sc SESM}. We remark that some ideas relevant to our \WO-hardness results also underlie the proofs of our \SETH-based results, yet here we also introduce several new ideas on top of them. The source problem of our reduction is the {\sc $s$-Sparse $p$-CNF-SAT} problem (for some appropriate choice of $s$ and $p$), which is the special case of {\sc CNF-SAT} where the size of each clause is at most $p$ and there are at most $sn$ clauses in total.

We begin by partitioning the set of all clauses into a ``large'' number of (pairwise-disjoint) small sets of clauses, where the size of each small set of clauses is fixed according to $s$ and $p$. The necessity of having this partition stems from our need to ensure that the number of dummy agents that we need to insert into preference lists of other agents would not be too large (yet we would still need an exponential number of dummy agents). For the sake of clarity of explanation, let us think of each small set of clauses as a color class. For each color class $t$, we enumerate all truth assignments that satisfy all of the clauses of that color $t$. Each such truth assignment is represented by two sets of variables, the {\em true set} and the {\em false set}, where the first contains all those variables that the assignment sets to true {\em and} which appear in clauses of color $t$, and the second contains all those other variables appearing in clauses of color $t$. Each variable $x_t$ is represented by two basic agents called $m_t$ and $w_t$, each true set associated with color class $i$, indexed by $j$ in that color class, is represented by two basic agents called $m^i_j$ and $w^i_j$, and each false assignment associated with color class $i$, indexed by $j$ in that color class (where the two sets indexed $j$ in the same color class correspond to the same truth assignment), is represented by two basic agents called $\overline{m}^i_j$ and $\overline{w}^i_j$. In addition to these agents, we also employ two ``garbage collector'' agents who are meant to counterweight imbalance in satisfaction of men and women that is present in ``desirable configurations'' of the gadgets described below.

For each variable $x_t$, the Variable Selector gadget consists of four agents, the basic agents $m_t$ and $w_t$ as well as two ``enriched'' agents, $\widehat{m}_t$ and $\widehat{w}_t$. The enriched agents do not rank any agent outside the gadget, and their sole purpose is to allow the gadget to encode two internal configurations, one where $m_t$ is matched to $w_t$ (while $\widehat{m}_t$ is matched to $\widehat{w}_t$), and the other where $m_t$ is matched to  $\widehat{w}_t$ while $w_t$ is matched to $\widehat{m}_t$. Having two such local agents also enables us to allow a basic agent $a$ to rank some other basic agent $b$ and yet $a$ and $b$ would never be matched to one another in any stable matching.
The first configuration indicates that $x_t$ should be assigned false, while the second one indicates that $x_t$ should be assigned true. The man $m_t$ is defined to prefer all women representing false sets that contain $x_t$ to the woman $\widehat{w}_t$, while the woman $w_t$ is defined to prefer all men representing true sets that contain $x_t$ to the man $m_t$.

Next, for every true set, we introduce the Truth Selector gadget, which in addition to $m^i_j$ and $w^i_j$, consists of the two enriched agents $\widehat{m}^i_j$ and $\widehat{w}^i_j$. Here, the configuration where $m^i_j$ is matched to $w^i_j$ indicates that the truth assignment is selected, and the configuration where $m^i_j$ is not matched to $w^i_j$ indicates that the truth assignment is not selected. To ensure that for each color class, {\em exactly one} truth assignment would be selected, we insert dummy agents into the preference lists of the basic agents of Truth Selector gadgets whose numbers correspond to the coefficients of a certain equation. While the number color classes is large, it is still significantly smaller than $n$, and thus the magnitute of the numbers involved is small enough for our purpose. 
The man $m^i_j$ is defined to prefer $w_t$  to $\widehat{w}^i_j$ for all women $w_t$ that rank $m^i_j$, thus ensuring that if the truth assignment in which he is involved is selected, all of the appropriate Variable Selector gadgets would have to be selected as well to maintain stability. Moreover, $m^i_j$ is also defined to prefer all women $\overline{w}^i_k$ for $k\neq j$ over $\widehat{w}^i_j$ (the purpose of this setting would be clarified below).

Finally, for every false set, we introduce the False Selector gadget, which in addition to $\overline{m}^i_j$ and $\overline{w}^i_j$, consists of the two enriched agents $\widehat{\overline{m}}^i_j$ and $\widehat{\overline{w}}^i_j$. Here, unlike the previous gadgets, the configuration where $m^i_j$ is matched to $w^i_j$ indicates that the truth assignment is {\em not} selected, and the configuration where $m^i_j$ is not matched to $w^i_j$ indicates that the truth assignment is selected. The woman $\overline{w}^i_j$ is defined to prefer $m_t$ to $\widehat{\overline{m}}^i_j$ for all men $m_t$ that rank $\overline{w}^i_j$, thus ensuring that if the truth assignment in which she is involved is selected, all of the appropriate Variable Selector gadgets would have to {\em not} be selected to maintain stability. Moreover, $\overline{w}^i_j$ is also defined to prefer all men $m^i_k$ where $k\neq j$ to $\widehat{\overline{m}}^i_j$. This ensure that if the current false set is selected, the only true set of the same color class that can be selected is the one that is complementary to this false set. We remark that to further conveniently control the selection of false sets, we also insert dummy agents into the preference lists of their basic agents.

Having defined the reduction, we proceed to precisely characterize {\em every} stable matching of the output. Then, we are also able to {\em precisely identify which set of rotations gives rise to which stable matching}. Overall, we can then explicitly construct the rotation digraph of the output. In our proof, for the sake of simplicity, we actually construct a supergraph of the rotation digraph. We are thus able to show that the rotation digraph of our instance is simply a DAG with three layers such that the middle layer contains exactly $n$ vertices, and that if we remove the middle layer from that graph, the subgraph that remains is just a collection of ``small'' connected components---roughly speaking, each such connected component would be a representation of one color class. Hence, we are able to conclude that the treewidth (of the underlying undirected graph) of the rotation digraph is not much larger than $n$.

%% file: w1hardness.tex
\vspace{-0.7em}
\section{Primal Graph: W[1]-Hardness}\label{sec:w1hardness}
\vspace{-0.5em}

In this section, we prove Theorem \ref{thm:w1hardIntro}, based on the approach described in Section~\ref{sec:overview}. Throughout this section, the notation $\tw$ refers to the treewidth of the primal graph. The source of our reduction is {\sc Multicolored Clique}, which is defined as follows. The input of {\sc Mulicolored Clique} consists of a graph $G$, a positive integer $k$, and a partition $(V^1,V^2,\ldots,V^k)$ of $V(G)$, where for all $i,j\in[k]$, $|V^i|=|V^j|$. Here, the parameter is $k$. For every index $i\in[k]$, the set $V^i$ is called {\em color class $i$}. The task is to decide whether $G$ contains a clique that consists of exactly one vertex from each color class. We denote $n=|V(G)|$ and $p=n/k$. Moreover, for every color class $i\in[k]$, we denote $V^i=\{v^i_1,v^i_2,\ldots,v^i_p\}$, and for every two color classes $i,j\in[k]$ where $i<j$, we denote $E^{i,j}=\{\{u,v\}\in E(G): u\in V^i, v\in V^j\}$ and $q^{i,j}=|E^{i,j}|$. Accordingly, we denote $E^{i,j}=\{e^{i,j}_1,e^{i,j}_2,\ldots,e^{i,j}_{q^{i,j}}\}$. For every vertex $v\in V(G)$, we denote the set of edges incident to $v$ in $G$ by $E(v)$. We assume w.l.o.g.~that $|E(G)|\geq n$. Furthermore, we implicitly assume that $|E(G)|$ is significantly larger than $k$ (else the problem is solvable by a parameterized algorithm). For our purpose, it is sufficient to assume that $|E(G)|>10^k$.

\begin{proposition}[\cite{DBLP:journals/tcs/FellowsHRV09,DBLP:books/sp/CyganFKLMPPS15}]\label{prop:multiClique}
The {\sc Multicolored Clique} problem is \WOH\ with respect to $k$. Moreover, unless \ETH\ fails, {\sc Multicolored Clique} cannot be solved in time $f(k)\cdot n^{o(k)}$ for any function $f$ that depends only on $k$.
\end{proposition}

Each section below is devoted to one reduction.

\input{w1hardnessSESM.tex}

\input{w1hardnessBSM.tex}

\input{w1hardnessMaxSMT.tex}

\input{w1hardnessMinSMT.tex}

%% file: w1hardnessSESM.tex
\subsection{Sex Equal Stable Marriage}\label{sec:w1SESM}

First, we prove that {\sc SESM} is \WOH, and that unless \ETH\ fails, {\sc SESM} cannot be solved in time $f(\tw)\cdot n^{o(\tw)}$ for any function $f$ that depends only on $\tw$.

\subsubsection{Reduction}

Let $I=(G,(V^1,V^2,\ldots,V^k))$ be an instance of {\sc Multicolored Clique}. We now describe how to construct an instance $\red_{SESM}(I)=(M,W,\{\pos_m\}|_{m\in M},\{\pos_w\}|_{w\in W})$ of {\sc SESM}.

\medskip
\myparagraph{Basic Agents.} We introduce the following sets of basic agents, which would be part of all of our reductions.
\begin{itemize}
\itemsep0em 
\item $M_{\bas}=\{m^1,m^2,\ldots,m^k\}$. Each man $m^i$ would be the basic vertex that represents color class $i$.
\item $\widehat{M}_{\bas}=\{\widehat{m}^1,\widehat{m}^2,\ldots,\widehat{m}^k\}$. Each man $\widehat{m}^i$ would be the basic vertex that is the mirror of the vertex $m^i$.
\item For every color class $i\in[k]$, $W^i_{\bas}=\{w^i_1,w^i_2,\ldots,w^i_p\}$. Each woman $w^i_j$ would be the basic vertex that represents the selection of $v^i_j$.
\item For every color class $i\in[k]$, $\widehat{W}^i_{\bas}=\{\widehat{w}^i_1,\widehat{w}^i_2,\ldots,\widehat{w}^i_p\}$. Each woman $\widehat{w}^i_j$ would be the basic vertex that is the mirror of the vertex $w^i_j$.
\item For every two color class $i,j\in[k]$ where $i<j$, $W^{i,j}_{\bas}=\{w^{i,j}_1,w^{i,j}_2,\ldots,w^{i,j}_{q^{i,j}}\}$. Each woman $w^{i,j}_t$ would be the basic vertex that represents the selection of $e^{i,j}_t$.
\end{itemize}

Moreover, in all of our reductions, the preference lists of the men in $M_{\bas}\cup \widehat{M}_{\bas}$ would be of the following form, which we call the form of a {\bf leader} since, roughly speaking, women representing edges incident to some vertex $v$ would follow the woman representing $v$ with respect to their positions in preference lists. Formally, for every color class $i\in[k]$, the preference lists of $m^i$ and $\widehat{m}^i$ satisfy the four following conditions.
	\begin{enumerate}
	\item $\domain(\pos_{m^i})\cap \left(\bigcup_{j=1}^k(W^j_{\bas}\cup \widehat{W}^j_{\bas}\cup (\bigcup_{\ell=1}^k W^{\ell,j}_{\bas}))\right) = W^i_{\bas}\cup W^{i,j}_{\bas}$;\\ $\domain(\pos_{\widehat{m}^i})\cap \left(\bigcup_{j=1}^k(W^j_{\bas}\cup \widehat{W}^j_{\bas}\cup (\bigcup_{\ell=1}^k W^{\ell,j}_{\bas}))\right) = \widehat{W}^i_{\bas}\cup W^{i,j}_{\bas}$.
	\item $\pos_{m^i}(w^i_1) < \pos_{m^i}(w^i_2) < \cdots < \pos_{m^i}(w^i_t)$; $\pos_{\widehat{m}^i}(w^i_t) < \pos_{\widehat{m}^i}(w^i_{t-1}) < \cdots < \pos_{\widehat{m}^i}(w^i_1)$.
	\item For every index $t\in[j]$ where $i<t$ and edge $e^{i,t}_{\ell}\in E(v^i_j)$, $\pos_{m^i}(w^i_j) < \pos_{m^i}(w^{i,t}_{\ell})$, and if $j\leq t-1$, then $\pos_{m^i}(w^{i,t}_{\ell}) < \pos_{m^i}(w^i_{j+1})$. The internal ordering of all women representing edges in $E(v^i_j)$ is arbitrary.
	\item For every index $t\in[j]$ where $i<t$ and edge $e^{i,t}_{\ell}\in E(v^i_j)$, $\pos_{\widehat{m}^i}(w^i_j) < \pos_{\widehat{m}^i}(w^{i,t}_{\ell})$, and if $j\geq 2$, then $ \pos_{\widehat{m}^i}(w^{i,t}_{\ell}) < \pos_{\widehat{m}^i}(w^i_{j-1})$. The internal ordering of all women representing edges in $E(v^i_j)$ is arbitrary.
	\end{enumerate}

We are now ready to present the gadgets employed by our reduction.

\medskip
\myparagraph{Original Vertex Selector.} For every color class $i\in[k]$, our first set of gadgets introduces the following set of new men: $M^i_{\enr}=\{m^i_2,m^i_3,\ldots,m^i_p\}$. Here, the abbreviation ``enr'' stands for the word ``enriched''. Each man $m^i_j$ would be the basic vertex most preferred by the woman $w^i_j$. The preference lists of the new men are defined as follows. For every index $j\in\{2,3,\ldots,p\}$, set $\domain(m^i_j)=\{w^i_j,w^i_{j-1}\}$, $\pos_{m^i_j}(w^i_j)=1$ and $\pos_{m^i_j}(w^i_{j-1})=1$. Moreover, for every index $j\in\{2,3,\ldots,p\}$, set $\pos_{w^i_j}(m^i) < \pos_{w^i_j}(m^i_j)$. Finally, for every index $j\in[p-1]$, set $\pos_{w^i_j}(m^i_{j+1}) < \pos_{w^i_j}(m^i)$.

Notice that so far, we have finished defining exactly which agents in $M_{\bas}\cup (\bigcup_{i=1}^k W^i_{\bas}\cup M^i_{\enr})$ rank which other agents in this set as well as what is the order in which they rank them (although we have not yet finished defining the preference lists of some of the agents in this set). For an illustration of the Original Vertex Selector gadget, the reader is referred to Fig.~\ref{fig:sesmW1Fig1}.

\begin{figure}[t!]\centering
\fbox{\includegraphics[scale=0.9]{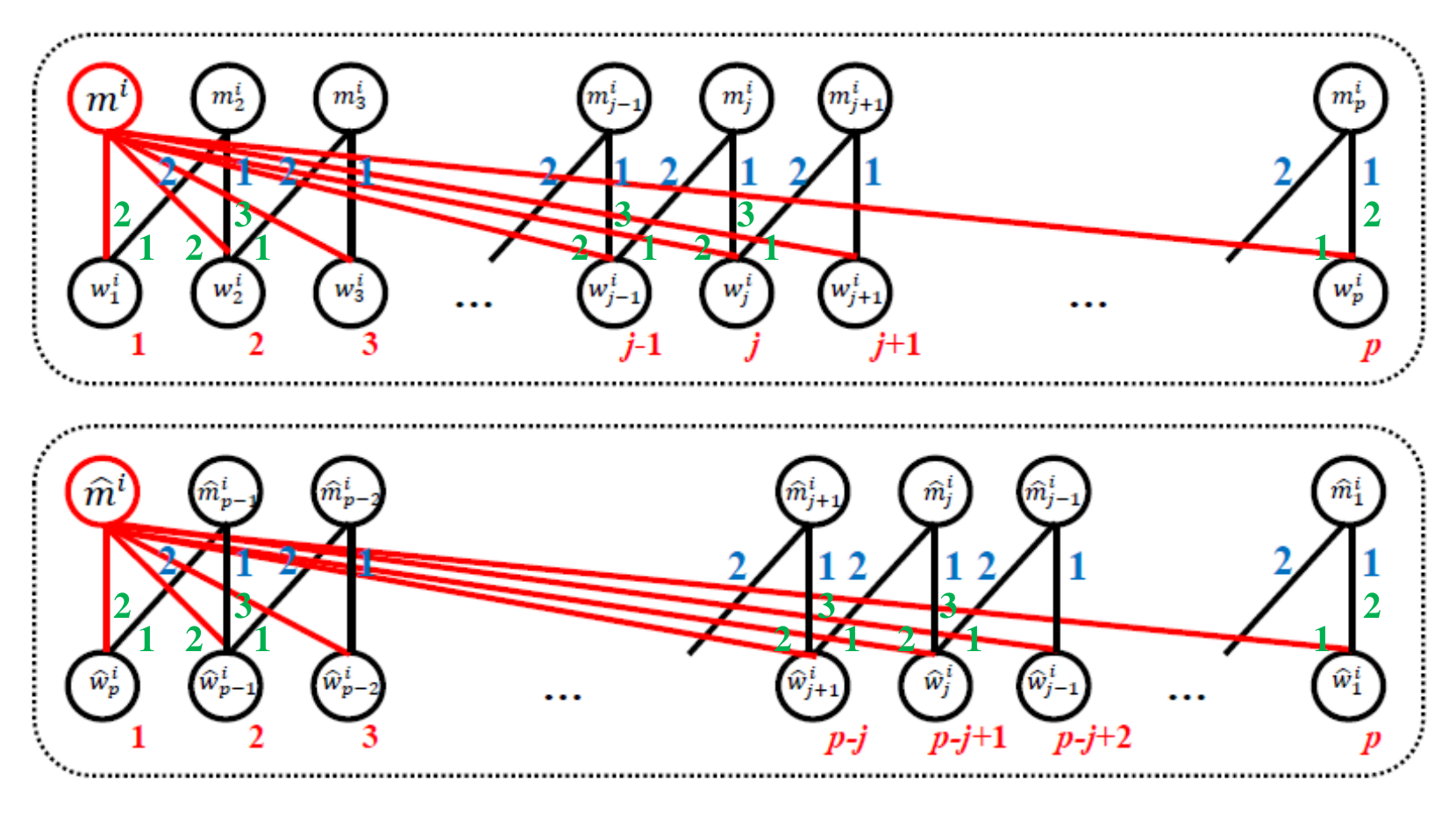}}
\caption{The Original and Mirror Vertex Selector gadgets. Colorful numbers indicate the {\em relative} order in which agents rank only those other agents shown in the figure.}\label{fig:sesmW1Fig1}
\end{figure}

\medskip
\myparagraph{Mirror Vertex Selector.} For every color class $i\in[k]$, our second set of gadgets introduces the following set of new men: $\widehat{M}^i_{\enr}=\{\widehat{m}^i_{1},\widehat{m}^i_{2},\ldots,\widehat{m}^i_{p-1}\}$. Each man $\widehat{m}^i_j$ would be the basic vertex most preferred by the woman $\widehat{w}^i_j$. The preference lists of the new men are defined as follows. For every index $j\in[p-1]$, set $\domain(\widehat{m}^i_j)=\{\widehat{w}^i_j,\widehat{w}^i_{j+1}\}$, $\pos_{\widehat{m}^i_j}(\widehat{w}^i_j)=1$ and $\pos_{\widehat{m}^i_j}(\widehat{w}^i_{j+1})=1$. Moreover, for every index $j\in[p-1]$, set $\pos_{\widehat{w}^i_j}(\widehat{m}^i) < \pos_{\widehat{w}^i_j}(\widehat{m}^i_j)$. Finally, for every index $j\in\{2,3,\ldots,p\}$, set $\pos_{\widehat{w}^i_j}(\widehat{m}^i_{j-1}) < \pos_{\widehat{w}^i_j}(\widehat{m}^i)$.

Thus, we have so far finished defining exactly which agents in $M_{\bas}\cup (\bigcup_{i=1}^k W^i_{\bas}\cup M^i_{\enr})\cup \widehat{M}_{\bas}\cup (\bigcup_{i=1}^k \widehat{W}^i_{\bas}\cup \widehat{M}^i_{\enr})$ rank which other agents in this set as well as what is the order in which they rank them. For an illustration of the Mirror Vertex Selector gadget, the reader is referred to Fig.~\ref{fig:sesmW1Fig1}.

\medskip
\myparagraph{Consistency.} We would next like to ensure that for all $i,j\in[k]$, $m^i$ would matched to $w^i_j$ if and only if $\widehat{m}^i$ would be matched to $\widehat{w}^i_j$. For this purpose, we insert the Consistency gadgets. Here, for every color class $i\in[k]$, we introduce two sets of new men, $\widetilde{M}^i=\{\widetilde{m}^i_1,\widetilde{m}^i_2,\ldots,\widetilde{m}^i_p\}$ and $\overline{M}^i=\{\overline{m}^i_1,\overline{m}^i_2,\ldots,\overline{m}^i_p\}$, and two sets of new women, $\widetilde{W}^i=\{\widetilde{w}^i_1,\widetilde{w}^i_2,\ldots,\widetilde{w}^i_p\}$ and $\overline{W}^i=\{\overline{w}^i_1,\overline{w}^i_2,\ldots,\overline{w}^i_p\}$. For all $i\in[k]$ and $j\in[p]$, we set the preference lists of $\overline{m}^i_j$, $\overline{m}^i_j$, $\widetilde{w}^i_j$ and $\overline{w}^i_j$ as~follows. First, let us set the preference list of $\widetilde{m}^i_j$.
\begin{itemize}
\item The intersection of $\domain(\pos_{\widetilde{m}^i_1})$ with the set of all agents, excluding the happy agents defined later, is exactly $\{\widetilde{w}^i_1, \widehat{w}^i_1,\overline{w}^i_1\}$, and $\pos_{\widetilde{m}^i_1}(\widetilde{w}^i_1) < \pos_{\widetilde{m}^i_1}(\widehat{w}^i_1) < \pos_{\widetilde{m}^i_1}(\overline{w}^i_1)$. 

\item If $2\leq j\leq p-2$, then excluding happy agents, the intersection of $\domain(\pos_{\widetilde{m}^i_j})$ with the set of all agents, excluding the happy agents defined later, is exactly $\{\widetilde{w}^i_j, w^i_j, \widehat{w}^i_j,\overline{w}^i_j\}$, and $\pos_{\widetilde{m}^i_j}(\widetilde{w}^i_j) < \pos_{\widetilde{m}^i_j}(w^i_j) < \pos_{\widetilde{m}^i_j}(\widehat{w}^i_j) < \pos_{\widetilde{m}^i_j}(\overline{w}^i_j)$. 

\item The intersection of $\domain(\pos_{\widetilde{m}^i_p})$ with the set of all agents, excluding the happy agents defined later, is exactly $\{\widetilde{w}^i_p, w^i_p, \overline{w}^i_p\}$, and $\pos_{\widetilde{m}^i_p}(\widetilde{w}^i_p) < \pos_{\widetilde{m}^i_t}(w^i_p) < \pos_{\widetilde{m}^i_p}(\overline{w}^i_p)$.
\end{itemize}
Second, let us set the preference lists of $\overline{m}^i_j$, $\widetilde{w}^i_j$ and $\overline{w}^i_j$.

\begin{itemize}
\item $\domain(\pos_{\overline{m}^i_j}) = \{\overline{w}^i_j,\widetilde{w}^i_j\}$, $\pos_{\overline{m}^i_j}(\overline{w}^i_j)=1$ and $\pos_{\overline{m}^i_j}(\widetilde{w}^i_j)=2$.

\item The intersection of $\domain(\pos_{\widetilde{w}^i_j})$ with the set of all agents, excluding the happy agents defined later, is exactly $\{\overline{m}^i_j,\widetilde{m}^i_j\}$, and $\pos_{\overline{w}^i_j}(\overline{m}^i_j)<\pos_{\overline{w}^i_j}(\widetilde{m}^i_j)$.

\item $\domain(\pos_{\overline{w}^i_j}) = \{\widetilde{m}^i_j,\overline{m}^i_j\}$, $\pos_{\overline{w}^i_j}(\widetilde{m}^i_j)=1$ and $\pos_{\overline{w}^i_j}(\overline{m}^i_j)=2$.
\end{itemize}

For all $i\in[k]$ and $j\in[p]$, we are now ready to explicitly define the preference lists of $w^i_j$ and $\widehat{w}^i_j$ as follows. First, let us define the preference list of $w^i_j$.
\begin{itemize}
\item  $\domain(\pos_{w^i_1})=\{m^i_2,m^i\}$, $\pos_{w^i_1}(m^i_{2})=1$ and $\pos_{w^i_1}(m^i)=2$.
\item If $2\leq j\leq p-1$, then excluding happy agents, the intersection of $\domain(\pos_{w^i_j})$ with the set of all agents is exactly $\{m^i_{j+1}, m^i, \widetilde{m}^i_j, m^i_j\}$, and $\pos_{w^i_j}(m^i_{j+1})<\pos_{w^i_j}(m^i)<\pos_{w^i_j}(\widetilde{m}^i_j)<\pos_{w^i_j}(m^i_j)$.
\item The intersection of $\domain(\pos_{w^i_p})$ with the set of all agents, excluding the happy agents defined later, $\{m^i, \widetilde{m}^i_p, m^i_p\}$, and $\pos_{w^i_p}(m^i)<\pos_{w^i_p}(\widetilde{m}^i_p)<\pos_{w^i_p}(m^i_p)$.
\end{itemize}
Second, let us define the preference list of $\widehat{w}^i_j$.

\begin{itemize}
\item $\domain(\pos_{\widehat{w}^i_p})=\{\widehat{m}^i_{p-1}, \widehat{m}^i\}$, $\pos_{\widehat{w}^i_p}(\widehat{m}^i_{p-1})=1$ and $\pos_{\widehat{w}^i_p}(\widehat{m}^i)=2$.
\item Else if $2\leq j\leq p-1$, then excluding happy agents, the intersection of $\domain(\pos_{\widehat{w}^i_j})$ with the set of all agents is exactly $\{\widehat{m}^i_{j-1}, \widehat{m}^i, \widetilde{m}^i_j, \widehat{m}^i_j\}$, and $\pos_{\widehat{w}^i_j}(\widehat{m}^i_{j-1})<\pos_{\widehat{w}^i_j}(\widehat{m}^i)<\pos_{\widehat{w}^i_j}(\widetilde{m}^i_j)<\pos_{\widehat{w}^i_j}(\widehat{m}^i_j)$.
\item The intersection of $\domain(\pos_{\widehat{w}^i_1})$ with the set of all agents, excluding the happy agents defined later, is exactly $\{\widehat{m}^i, \widetilde{m}^i_1, \widehat{m}^i_1\}$, and $\pos_{\widehat{w}^i_1}(\widehat{m}^i)<\pos_{\widehat{w}^i_1}(\widetilde{m}^i_1)<\pos_{\widehat{w}^i_1}(\widehat{m}^i_1)$.
\end{itemize}

For an illustration of a Consistency gadget, the reader is referred to Fig.~\ref{fig:sesmW1Fig2}.

\begin{figure}[t!]\centering
\fbox{\includegraphics[scale=0.9]{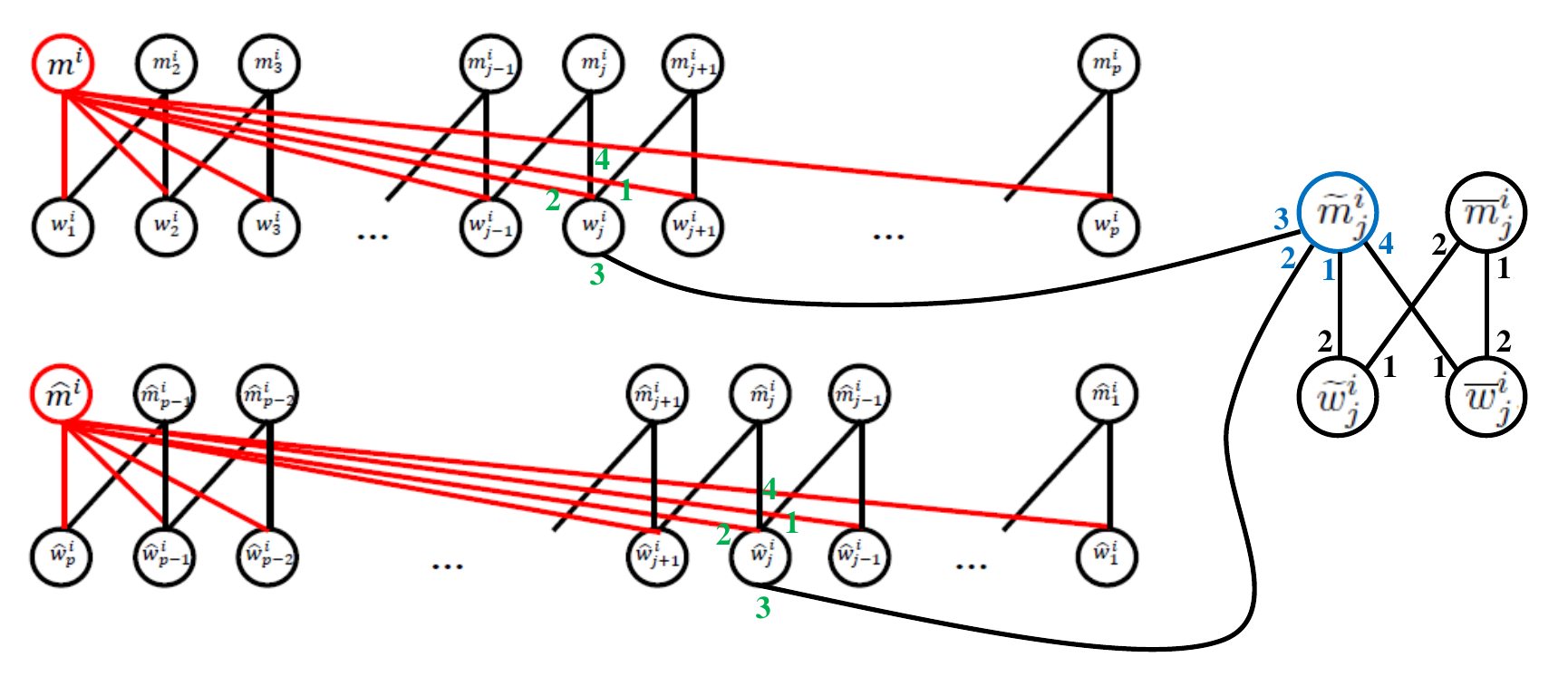}}
\caption{The Consistency gadget. Colorful numbers indicate the {\em relative} order in which agents rank only those other agents shown in the figure.}\label{fig:sesmW1Fig2}
\end{figure}

\medskip
\myparagraph{Edge Selector.} For every two color classes $i,j\in[k]$ where $i<j$, our last set of gadgets introduces the following three sets of new agents: $M^{i,j}_{\enr}=\{m^{i,j}_1,m^{i,j}_2,\ldots,m^{i,j}_{q^{i,j}}\}$, $\overline{M}^{i,j}=\{\overline{m}^{i,j}_1,\overline{m}^{i,j}_2,\ldots,\overline{m}^{i,j}_{q^{i,j}}\}$ and $\overline{W}^{i,j}=\{\overline{w}^{i,j}_1,\overline{w}^{i,j}_2,\ldots,\overline{w}^{i,j}_{q^{i,j}}\}$. For every $t\in [q^{i,j}]$, the preference lists of the new agents, $m^{i,j}_t$, $\overline{m}^{i,j}_t$ and $\overline{w}^{i,j}_t$, are defined as follows.
\begin{itemize}
\item The intersection of $\domain(\pos_{m^{i,j}_t})$ with the set of all agents, excluding the happy agents defined later, is exactly
$\{w^{i,j}_t,\overline{w}^{i,j}_t\}$, and $\pos_{m^{i,j}_t}(w^{i,j}_t)<\pos_{m^{i,j}_t}(\overline{w}^{i,j}_t)$.
\item $\domain(\pos_{\overline{m}^{i,j}_t})=\{\overline{w}^{i,j}_t,w^{i,j}_t\}$, $\pos_{\overline{m}^{i,j}_t}(\overline{w}^{i,j}_t)=1$ and $\pos_{\overline{m}^{i,j}_t}(w^{i,j}_t)=2$.
\item $\domain(\pos_{\overline{w}^{i,j}_t})=\{m^{i,j}_t,\overline{m}^{i,j}_t\}$, $\pos_{\overline{w}^{i,j}_t}(m^{i,j}_t)=1$ and $\pos_{\overline{w}^{i,j}_t}(\overline{m}^{i,j}_t)=2$.
\end{itemize}

For all $i,j\in[k]$ where $i<j$ and $t\in[q^{i,j}]$, we are now ready to define also the preference list of $w^{i,j}_t$ up to the insertion of happy agents that are defined later. The intersection of $\domain(\pos_{w^{i,j}_t})$ with the set of all agents, excluding the happy agents defined later, is exactly $\{\overline{m}^{i,j}_t,m^i,\widehat{m}^i,m^j,\widehat{m}^j,m^{i,j}_t\}$, and $\pos_{w^{i,j}_t}(\overline{m}^{i,j}_t) < \pos_{w^{i,j}_t}(m^i) < \pos_{w^{i,j}_t}(\widehat{m}^i) < \pos_{w^{i,j}_t}(m^j) < \pos_{w^{i,j}_t}(\widehat{m}^j) < \pos_{w^{i,j}_t}(m^{i,j}_t)$.

For an illustration of an Edge Selector gadget, the reader is referred to Fig.~\ref{fig:sesmW1Fig3}.

\begin{figure}[t!]\centering
\fbox{\includegraphics[scale=0.81]{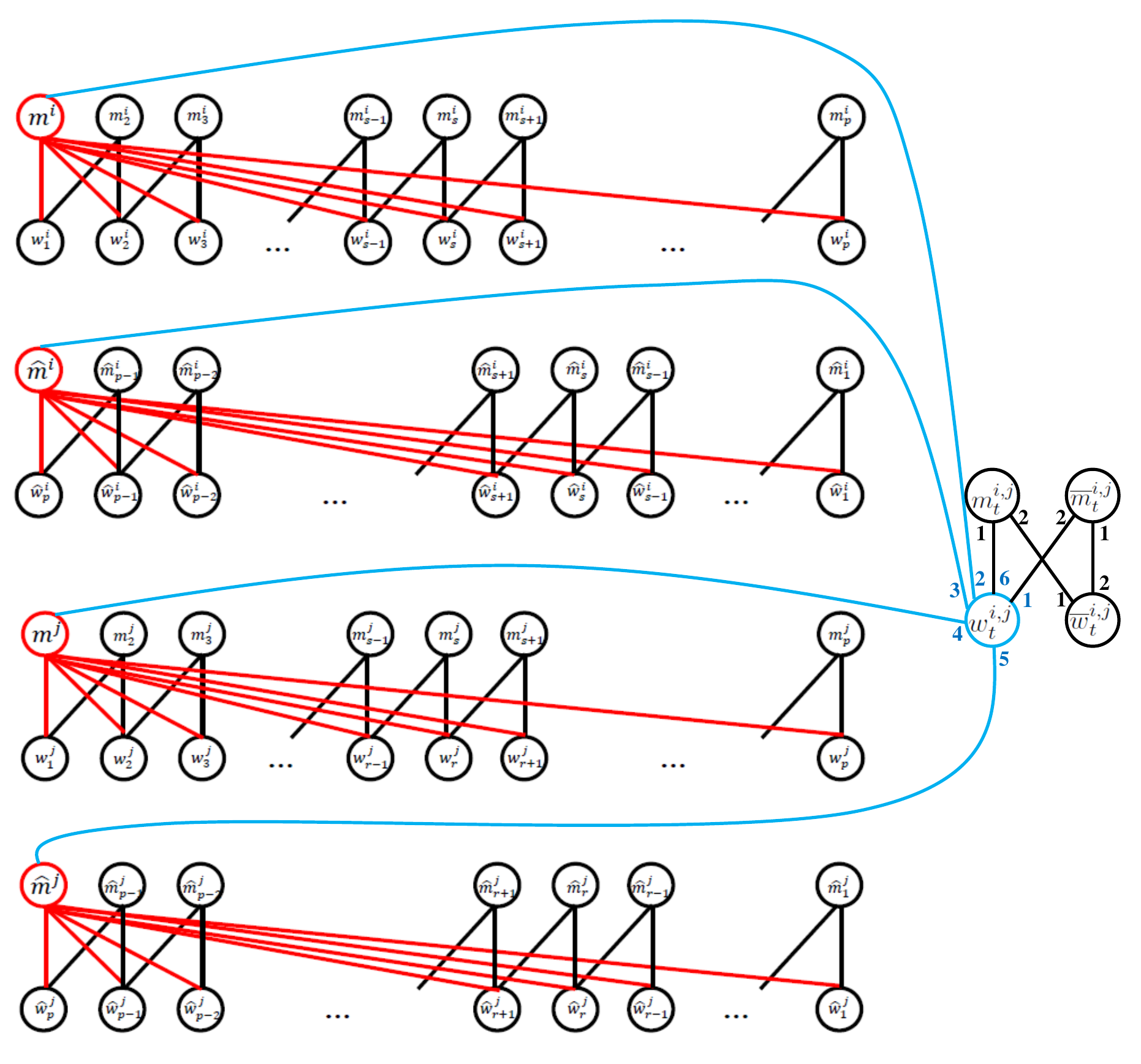}}
\caption{The Edge Selector gadget. Colorful numbers indicate the {\em relative} order in which agents rank only those other agents shown in the figure.}\label{fig:sesmW1Fig3}
\end{figure}

\myparagraph{Happy Pairs.} To be able to control the measure of sex-equality, we introduce agents whose sole purpose is to serve as ``fillers'' of preference lists of other agents. For this purpose, we rely on the following notion of a happy pair.

\begin{definition}\label{def:happy}
A {\em happy pair} is a pair $(m,w)$ of a man $m\in M$ and a woman $w\in W$ such that $\pos_m(w)=\pos_w(m)=1$. A {\em happy agent} is an agent that belongs to a happy pair.
\end{definition}

We introduce $\alpha'=\alpha 4^k|E(G)|^{10}$ new happy pairs, denoted by $(m^1_{\hap},w^1_{\hap}),(m^2_{\hap},w^2_{\hap}),\ldots,$ $(m^\alpha_{\hap},w^{\alpha'}_{\hap})$, where
\[\begin{array}{ll}
\alpha = & -9k +3pk +4k^2 - (pk-k+2)|E(G)|\\
& - (|E(G)|-2{k\choose 2})|E(G)|^{10} - (2^k-1)|E(G)|^{30} + (p-1)(2^k-1)|E(G)|^{40}.
\end{array}\]
Initially, the preference list of every such happy agent $a$ contains only one agent, the one that belongs to the same happy pair as $a$. In what follows, we insert happy agents into the preference lists of previously defined agents. We implicitly assume that when we insert a happy agents $a$ into the preference of an agent $b$, the agent $b$ is appended to the end of the preference list of $s$. Note that in this manner, the agent $a$ together with the agents at the top of the preference list of $a$ remain a happy pair.

Now, we insert happy women into the preference lists of men of the forms $m^i$, $\widehat{m}^i$, $\widetilde{m}^i_j$ and $m^{i,j}_t$. Whenever we state below that we insert some set of arbitrarily chosen happy agents to the preference list of some agent, we suppose that these happy agents are chosen from the set of happy agents such that  neither them nor their most preferred partners have already been inserted to the preference list of any other agent. It would be clear that the number $\alpha'$ is sufficiently large to allow such selection.

First, for every $i\in[k]$, we explicitly define the preference list of $m^i$ as follows. We set $\domain(\pos_{m^i})$ to consist of the union of $W^i_{\bas}\cup W^{i,j}_{\bas}$ and a set of arbitrarily chosen $|E(G)|^{20}(p-1)+\displaystyle{\sum_{j=1}^p}(|E(G)|-|E^{i,j}|)$ happy women. Now, we set the preference list to satisfy the following conditions.
\begin{enumerate}
\item For all $j\in[p]$, $\pos_{m^i}(w^i_j)=1+(1+|E(G)|+|E(G)|^{20})(j-1)$.
\item For all $j\in[p]$, we insert all of the women $w^{i,j}_t$ such that $e^{i,j}_t\in E(v^i_j)$ into positions $2+(1+|E(G)|+|E(G)|^{20})(j-1),3+(1+|E(G)|+|E(G)|^{10})(j-1),\ldots,1+|E(v^i_j)|+(1+|E(G)|+|E(G)|^{20})(j-1)$ (the choice of which of these women occupies which of these positions is arbitrary).
\item All of the positions that have not been occupied by the conditions above are occupied by the happy women (the choice of which happy woman occupies which vacant position is arbitrary).
\end{enumerate}

Second, for every $i\in[k]$, we explicitly define the preference list of $\widehat{m}^i$ as follows. We set $\domain(\pos_{\widehat{m}^i})$ to consist of the union of $\widehat{W}^i_{\bas}\cup \widehat{W}^{i,j}_{\bas}$ and a set of arbitrarily chosen $|E(G)|^{20}(p-1)+\displaystyle{\sum_{j=1}^p}(|E(G)|-|E^{i,j}|)$ happy women. Now, we set the preference list to satisfy the following conditions.
\begin{enumerate}
\item For all $j\in[p]$, $\pos_{\widehat{m}^i}(w^i_j)=1+(1+|E(G)|+|E(G)|^{20})(p-j)$.
\item For all $j\in[p]$, we insert all of the women $w^{i,j}_t$ such that $e^{i,j}_t\in E(v^i_j)$ into positions $2+(1+|E(G)|+|E(G)|^{20})(p-j),3+(1+|E(G)|+|E(G)|^{20})(p-j),\ldots,1+|E(v^i_j)|+(1+|E(G)|+|E(G)|^{20})(j-1)$ (the choice of which of these women occupies which of these positions is arbitrary).
\item All of the positions that have not been occupied by the conditions above are occupied by the happy women (the choice of which happy woman occupies which vacant position is arbitrary).
\end{enumerate}

Third, for every $i\in[k]$, we explicitly define the preference list of $\widetilde{m}^i_1$ as follows. We set $\domain(\pos_{\widetilde{m}^i_1})$ to consist of the union of $\{\widetilde{w}^i_1,\widehat{w}^i_1,\overline{w}^i_1\}$ and a set of arbitrarily chosen $|E(G)|^{30}+1$ happy women. Define $\pos_{\widetilde{m}^i_1}(\widetilde{w}^i_1)=1$, $\pos_{\widetilde{m}^i_1}(\widehat{w}^i_1)=2$ and $\pos_{\widetilde{m}^i_1}(\overline{w}^i_1)=4+|E(G)|^{30}$. All of the positions that have not been occupied above are occupied by the happy women (the choice of which happy woman occupies which vacant position is arbitrary).
For every $i\in[k]$ and $j\in[p]$ where $2\leq j\leq p-1$, we explicitly define the preference list of $\widetilde{m}^i_1$ as follows. We set $\domain(\pos_{\widetilde{m}^i_j})$ to consist of the union of $\{\widetilde{w}^i_j,w^i_j,\widehat{w}^i_j,\overline{w}^i_j\}$ and a set of arbitrarily chosen $2^{i-1}|E(G)|^{30}$ happy women. Define $\pos_{\widetilde{m}^i_j}(\widetilde{w}^i_j)=1$, $\pos_{\widetilde{m}^i_j}(w^i_j)=2$, $\pos_{\widetilde{m}^i_j}(\widehat{w}^i_j)=3$ and $\pos_{\widetilde{m}^i_j}(\overline{w}^i_j)=4+2^{i-1}|E(G)|^{30}$. All of the positions that have not been occupied above are occupied by the happy women (the choice of which happy woman occupies which vacant position is arbitrary).
For every $i\in[k]$, we explicitly define the preference list of $\widetilde{m}^i_p$ as follows. We set $\domain(\pos_{\widetilde{m}^i_p})$ to consist of the union of $\{\widetilde{w}^i_p,w^i_p,\overline{w}^i_p\}$ and a set of arbitrarily chosen $2^{k-1}|E(G)|^{30}+1$ happy women. Define $\pos_{\widetilde{m}^i_p}(\widetilde{w}^i_p)=1$, $\pos_{\widetilde{m}^i_p}(w^i_p)=2$ and $\pos_{\widetilde{m}^i_p}(\overline{w}^i_p)=4+2^{k-1}|E(G)|^{30}$. All of the positions that have not been occupied above are occupied by the happy women (the choice of which happy woman occupies which vacant position is arbitrary).

Fourth, for every $i,j\in[k]$ where $i<j$ and $t\in[q^{i,j}]$, we explicitly define the preference list of $m^{i,j}_t$ as follows. We set $\domain(\pos_{m^{i,j}_t})$ to consist of the union of $\{w^{i,j}_t,\overline{w}^{i,j}_t\}$ and a set of arbitrarily chosen $|E(G)|^{10}$ happy women. Define $\pos_{m^{i,j}_t}(w^{i,j}_t)=1$ and $\pos_{m^{i,j}_t}(\overline{w}^{i,j}_t)=2+|E(G)|^{10}$. All of the positions that have not been occupied above are occupied by the happy women (the choice of which happy woman occupies which vacant position is arbitrary).

We next insert happy men into the preference lists of women of the forms $w^i_j$ where $j\neq 1$, $\widehat{w}^i_j$ where $j\neq p$, $\widetilde{w}^i_j$ and $w^{i,j}_t$. First, for every $i\in[k]$ and $j\in\{2,3,\ldots,p-1\}$, we explicitly define the preference list of $w^i_j$ as follows. We set $\domain(\pos_{w^i_j})$ to consist of the union of $\{m^i_{j+1},m^i,\widetilde{m}^i_j,m^i_j\}$ and a set of arbitrarily chosen $|E(G)|^{20}$ happy men. Define $\pos_{w^i_j}(m^i_{j+1})=1$, $\pos_{w^i_j}(m^i)=2$, $\pos_{w^i_j}(\widetilde{m}^i_j)=3$ and $\pos_{w^i_j}(m^i_j)=4+|E(G)|^{20}$. All of the positions that have not been occupied above are occupied by the happy men (the choice of which happy man occupies which vacant position is arbitrary). Moreover, for every $i\in[k]$, we explicitly define the preference list of $w^i_p$ as follows. We set $\domain(\pos_{w^i_p})$ to consist of the union of $\{m^i,\widetilde{m}^i_p,m^i_p\}$ and a set of arbitrarily chosen $|E(G)|^{20}+1$ happy men. Define $\pos_{w^i_p}(m^i)=2$, $\pos_{w^i_p}(\widetilde{m}^i_p)=3$ and $\pos_{w^i_p}(m^i_p)=4+|E(G)|^{20}$. All of the positions that have not been occupied above are occupied by the happy men (the choice of which happy man occupies which vacant position is arbitrary).

Second, for every $i\in[k]$ and $j\in\{2,3,\ldots,p-1\}$, we explicitly define the preference list of $\widehat{w}^i_j$ as follows. We set $\domain(\pos_{\widehat{w}^i_j})$ to consist of the union of $\{\widehat{m}^i_{j-1},\widehat{m}^i,\widetilde{m}^i_j,\widehat{m}^i_j\}$ and a set of arbitrarily chosen $|E(G)|^{20}$ happy men. Define $\pos_{\widehat{w}^i_j}(\widehat{m}^i_{j-1})=1$, $\pos_{\widehat{w}^i_j}(\widehat{m}^i)=2$, $\pos_{\widehat{w}^i_j}(\widetilde{m}^i_j)=3$ and $\pos_{\widehat{w}^i_j}(\widehat{m}^i_j)=4+|E(G)|^{20}$. All of the positions that have not been occupied above are occupied by the happy men (the choice of which happy man occupies which vacant position is arbitrary). Moreover, for every $i\in[k]$, we explicitly define the preference list of $\widehat{w}^i_1$ as follows. We set $\domain(\pos_{\widehat{w}^i_p})$ to consist of the union of $\{\widehat{m}^i,\widetilde{m}^i_1,\widehat{m}^i_1\}$ and a set of arbitrarily chosen $|E(G)|^{20}+1$ happy men. Define $\pos_{\widehat{w}^i_1}(m^i)=2$, $\pos_{\widehat{w}^i_1}(\widetilde{m}^i_1)=3$ and $\pos_{\widehat{w}^i_1}(m^i_1)=4+|E(G)|^{20}$. All of the positions that have not been occupied above are occupied by the happy men (the choice of which happy man occupies which vacant position is arbitrary).

Third, for every $i\in[k]$ and $j\in[p]$, we explicitly define the preference list of $\widetilde{w}^i_j$ as follows. We set $\domain(\pos_{\widetilde{w}^i_j})$ to consist of the union of $\{\overline{m}^i_j,\widetilde{m}^i_j\}$ and a set of arbitrarily chosen $2^{k-i}|E(G)|^{40}$ happy men. Define $\pos_{\widetilde{w}^i_j}(\overline{m}^i_j)=1$ and $\pos_{\widetilde{w}^i_j}(\widetilde{m}^i_j)=2+2^{k-i}|E(G)|^{40}$. All of the positions that have not been occupied above are occupied by the happy men (the choice of which happy man occupies which vacant position is arbitrary).

Fourth, for every $i,j\in[k]$ where $i<j$ and $t\in[q^{i,j}]$, we explicitly define the preference list of $w^{i,j}_t$ as follows. We set $\domain(\pos_{w^{i,j}_t})$ to consist of the union of $\{\overline{m}^{i,j}_t,m^i,\widehat{m}^i,m^j,\widehat{m}^j,m^{i,j}_t\}$ and a set of arbitrarily chosen $|E(G)|^{10}$ happy men. Define $\pos_{w^{i,j}_t}(\overline{m}^{i,j}_t)=1$, $\pos_{w^{i,j}_t}(m^i)=2$, $\pos_{w^{i,j}_t}(\widehat{m}^i)=3$, $\pos_{w^{i,j}_t}(m^j)=4$, $\pos_{w^{i,j}_t}(\widehat{m}^j)=5$ and $\pos_{w^{i,j}_t}(m^{i,j}_t)=6+|E(G)|^{10}$. All of the positions that have not been occupied above are occupied by the happy men (the choice of which happy man occupies which vacant position is arbitrary).

\medskip
\myparagraph{Garbage Collector.} Finally, we introduce one new man, $m^\star$, and one new woman, $w^\star$. The preference list of $\mu^\star$ first contains all happy women, $w^1_{\hap},w^2_{\hap},\ldots,w^\alpha_{\hap}$, in some arbitrary order, and afterwards it contains the woman $w^\star$. The preference list of $w^\star$ is simply defined to contain only the man $m^\star$.

\subsubsection{Treewidth}

We begin the analysis of the reduction by bounding the treewidth of the resulting primal graph.

\begin{lemma}\label{lem:twSESM}
Let $I$ be an instance of {\sc Multicolored Clique}. Then, the treewidth of $\red_{SESM}(I)$ is bounded by $2k+\OO(1)$.
\end{lemma}

\begin{proof}
Let $P$ be the primal graph of $\red_{SESM}(I)$, and let $\widehat{P}$ denote the graph obtained from $P$ by the removal of all of (the vertices that represent) men in $M_{\bas}\cup\widehat{M}_{\bas}$. Note that since $|M_{\bas}\cup\widehat{M}_{\bas}|=2k$, to prove that the treewidth of $P$ is bounded by $2k+\OO(1)$, it is sufficient to prove that the treewidth of every connected component of $\widehat{P}$ is bounded by $\OO(1)$. Indeed, given a tree decomposition of every connected component of $\widehat{P}$ of with $\OO(1)$, we can construct a tree decomposition of $P$ of width $2k+\OO(1)$ by simply inserting all of the men in $M_{\bas}\cup\widehat{M}_{\bas}$ into every bag of each of the tree decompositions and then arbitrarily connecting the tree decompositions to obtain a single tree rather than a forest. Let $P'$ denote the graph obtained from $P'$ be removing all of the happy agents. Note that all of the happy agents are either leaves themselves in $P'$ or vertices of degree 2 that are incident to happy agents that are leaves in $P'$. Given a tree decomposition $(T,\beta)$ of $P'$, we can construct a tree decomposition of $\widehat{P}$ of either the same width or width $\OO(1)$ as follows. We assign to each vertex $v$ of $P'$, which is adjacent to some $x$ happy agents, an arbitrarily chosen node $u$ whose bag contains $v$, insert $x$ new leaves to $T$ which are each adjacent to $u$, and defining the bag of each of these leaves to contain $v$, a distinct happy agent adjacent to $v$ and the agent most preferred by this happy agent. For each pair of happy agents not yet inserted, we create a new node whose bag contains only these two agents, and attach this node as a leaf to some arbitrarily chosen node of $T$.

By the arguments above, it is sufficient to show that the treewidth of $P'$ is upper bounded by $\OO(1)$. First, for all $i,j\in[k]$ where $i<j$ and $t\in[q^{i,j}]$, we have that $\{m^{i,j}_t,w^{i,j}_t,\widetilde{m}^{i,j}_t\}$ is the entire vertex set of a connected component of $P'$. Since this connected component is simply a cycle, its treewidth is 2. 

We next note that for all $i\in[k]$, we have that $X^i=M^i_{\enr}\cup\widehat{M}^i_{\enr}\cup W^i_{\bas}\cup \widehat{W}^i_{\bas}\cup \widetilde{M}^i\cup \overline{M}^i\cup \widetilde{W}^i\cup \overline{W}^i$ is the entire vertex set of a connected component of $P'$. For this connected component, which we denote by $C^i$. we explicitly define a tree decomposition $(T^i,\beta^i)$ as follows. The tree $T^i$ is simply a path on $p$ vertices, denoted by $T=u_1-u_2-\cdots-u_p$. We define $\beta^i(u_1)=\{w^i_1,m^i_2,\widehat{w}^i_1,\widehat{m}^i_1,\widetilde{m}^i_1,\overline{m}^i_1,\widetilde{w}^i_1,\overline{w}^i_1\}$ and $\beta^i(u_p)=\{w^i_p,m^i_p,\widehat{w}^i_p,\widehat{m}^i_{p-1},\widetilde{m}^i_p,\overline{m}^i_p,\widetilde{w}^i_p,\overline{w}^i_p\}$. For all $j\in\{2,\ldots,p-1\}$, we define $\beta^i(u_j)=\{w^i_j,m^i_j,m^i_{j+1},\widehat{w}^i_j,\widehat{m}^i_j,\widehat{m}^i_{j-1},\widetilde{m}^i_j,\overline{m}^i_j,\widetilde{w}^i_j,\overline{w}^i_j\}$. Note that the size of each bag of is upper bounded by $10=\OO(1)$. Moreover, each agent in $X^i$ belongs to the bags of at most two nodes of $T^i$, and these two nodes are adjacent. Lastly, for all $j\in[p]$, all endpoints of edges incident to either $w^i_j$ or $\widehat{w}^i_j$ belong to the bag $\beta^i(u_j)$, and every edge of $C^i$ has an endpoint in $W^i_{\bas}\cup\widehat{W}_{\bas}$. Thus, $(T^i,\beta^i)$ is indeed a tree decomposition of $C^i$ of width $\OO(1)$. 

Finally, the treewidth of the connected component consisting only of $\{m^\star,w^\star\}$ is clearly 1. Observe that we have indeed considered every connected component of $P'$, and thus we conclude the proof of the lemma.
\end{proof}

\subsubsection{Correctness}

\myparagraph{Forward Direction.} We first show how given a solution of an instance $I$ of {\sc Multicolored Clique}, we can construct a stable matching $\mu$ of $\red_{SESM}(I)$ whose sex-equality measure is exactly 0. For this purpose, we introduce the following definition. Here, when we denote a vertex-set by $U=\{v^1_{\ell_1},v^2_{\ell_2},\ldots,v^t_{\ell_k}\}$, we implicitly assume that for every $i\in[k]$, $v^i_{\ell_i}\in V^i$. Moreover, when we denote an edge-set by $W=\{e^{i,j}_{\ell_{i,j}}: i,j\in[k], i<j\}$, we implicitly assume that for every $i,j\in[k]$ where $i<j$, $e^{i,j}_{\ell_{i,j}}$ is an edge whose endpoints are a vertex in $V^i$ and a vertex in $V^j$.

\begin{definition}\label{def:cliquetoSESM}
Let $I=(G,(V^1,V^2,\ldots,V^k))$ be a \yesinstance\ of {\sc Multicolored Clique}, and let $U=\{v^1_{\ell_1},v^2_{\ell_2},\ldots,v^t_{\ell_k}\}$ and $W=\{e^{i,j}_{\ell_{i,j}}: i,j\in[k], i<j\}$ denote the vertex and edge sets, respectively, of a multicolored clique $C$ of $G$. 
Then, the matching $\mu^C_{SESM}$ of $\red_{SESM}(I)$ is defined as follows.
\begin{itemize}
\item For all $i\in[k]$: $\mu^C_{SESM}(m^i)=w^i_{\ell_i}$ and $\mu^C_{SESM}(\widehat{m}^i)=\widehat{w}^i_{\ell_i}$.
\item For all $i\in[k]$ and $j\in\{2,3,\ldots,\ell_i\}$: $\mu^C_{SESM}(m^i_j)=w^i_{j-1}$.
\item For all $i\in[k]$ and $j\in\{\ell_i+1,\ell_i+2,\ldots,p\}$: $\mu^C_{SESM}(m^i_j)=w^i_j$.
\item For all $i\in[k]$ and $j\in\{\ell_i,\ell_i+1,\ldots,p-1\}$: $\mu^C_{SESM}(\widehat{m}^i_j)=\widehat{w}^i_{j+1}$.
\item For all $i\in[k]$ and $j\in[\ell_i-1]$: $\mu^C_{SESM}(\widehat{m}^i_j)=\widehat{w}^i_j$.
\item For all $i\in[k]$: $\mu^C_{SESM}(\widetilde{m}^i_{\ell_i})=\overline{w}^i_{\ell_i}$ and $\mu^C_{SESM}(\overline{m}^i_{\ell_i})=\widetilde{w}^i_{\ell_i}$.
\item For all $i\in[k]$ and $j\in[p]$ such that $j\neq\ell_i$: $\mu^C_{SESM}(\widetilde{m}^i_j)=\widetilde{w}^i_j$ and $\mu^C_{SESM}(\overline{m}^i_j)=\overline{w}^i_j$.
\item For all $i,j\in[k]$ where $i<j$: $\mu^C_{SESM}(m^{i,j}_{\ell_{i,j}})=w^{i,j}_{\ell_{i,j}}$ and $\mu^C_{SESM}(\overline{m}^{i,j}_{\ell_{i,j}})=\overline{w}^{i,j}_{\ell_{i,j}}$.
\item For all $i,j\in[k]$ where $i<j$ and $t\in[q^{i,j}]$ such that $t\neq\ell_{i,j}$: $\mu^C_{SESM}(m^{i,j}_t)=\overline{w}^{i,j}_t$ and $\mu^C_{SESM}(\overline{m}^{i,j}_t)=w^{i,j}_t$.
\item For all $i\in[\alpha']$: $\mu^C_{SESM}(m^i_{\hap})=w^i_{\hap}$.
\item $\mu^C_{SESM}(m^\star)=w^\star$.
\end{itemize}
\end{definition}

Observe that $\mu$ matches all agents of $\red_{SESM}(I)$. Let us first argue that $\mu^C_{SESM}$ is a stable matching.

\begin{lemma}\label{lem:w1SESMforwardSM}
Let $I=(G,(V^1,V^2,\ldots,V^k))$ be a \yesinstance\ of {\sc Multicolored Clique}. Let $C$ be a multicolored clique of $G$. Then, $\mu^C_{SESM}$ is a stable matching of $\red_{SESM}(I)$.
\end{lemma}

\begin{proof}
First, notice that for every $i\in[\alpha']$, $m^i_{\hap}$ is matched to the woman he prefers the most, and therefore this man cannot belong to any blocking pair. Thus, since the preference list of $\mu^\star$ only contains happy women in addition to $w^\star$, we also have that $\mu^\star$ cannot belong to any blocking pair. For all $i\in[k]$ and $j\in\{2,3,\ldots,p\}$ such that $j>\ell_i$, $m^i_j$ is matched to the woman he prefers the most, and for all $i\in[k]$ and $j\in[p-1]$ such that $j<\ell_i$, $\widehat{m}^i_j$ is matched to the woman he prefers the most, thus these men cannot belong to any blocking pair. Moreover, for all $i\in[k]$ and $j\in[p]$ such that $j\neq\ell_i$: $\widetilde{m}^i_j$ and $\overline{m}^i_j$ are matched to the women they prefer the most, and thus these men cannot belong to any blocking pair as well. We also note that for all $i,j\in[k]$ where $i<j$, $m^{i,j}_{\ell_{i,j}}$ and $\overline{m}^{i,j}_{\ell_{i,j}}$ are matched to the women they prefer the most, and thus these men cannot belong to any blocking pair as well.

Next, we analyze each of the remaining men in $M$ separately.
\begin{itemize}
\item For all $i\in[k]$, recall that $\mu^C_{SESM}(m^i)=w^i_{\ell_i}$. The only women who $m^i$ prefers over $w^i_{\ell_i}$, who are not happy women, are those that belong to the sets $X=\{w^i_j: j\in[\ell_i-1]\}$ and $Y=\{w^{i,j}_t: j\in[k], j>i$, there exists $s\in[\ell_i-1]$ such that $e^{i,j}_t$ is incident to $v^i_s$ in $G\}$. However, for all $w^i_j\in X$, $\pos_{w^i_j}(\mu^C_{SESM}(w^i_j))=\pos_{w^i_j}(m^i_{j+1})<\pos_{w^i_j}(m^i)$, and for all $w^{i,j}_t\in Y$, $\pos_{w^{i,j}_t}(\mu^C_{SESM}(w^{i,j}_t))=\pos_{w^{i,j}_t}(m^{i,j}_t)<\pos_{w^{i,j}_t}(m^i)$. Thus, $m^i$ cannot belong to any blocking pair. Symmetrically, we derive that for all $i\in[k]$, $\widehat{m}^i$ also cannot belong to any blocking pair. Note that to ensure that both $m^i$ and $\widehat{m}^i$ do not belong to any blocking pair, we crucially rely on the definition of their preference lists to be of the form of a leader.

\item For all $i\in[k]$ and $j\in\{2,3,\ldots,p\}$ such that $j\leq \ell_i$, recall that $\mu^C_{SESM}(m^i_j)=w^i_{j-1}$. The only woman who $m^i_j$ prefers over $w^i_{j-1}$ is $w^i_j$. However, $w^i_j$ is matched to either $m^i$ or $m^i_{j+1}$, who she prefers over $m^i_j$. Thus, $m^i_j$ also cannot belong to a blocking pair. Symmetrically, we derive that for all $i\in[k]$ and $j\in[p-1]$ such that $j\geq \ell_i$, $\widehat{m}^i_j$ also cannot belong to a blocking pair.

\item For all $i\in[k]$, recall that $\mu^C_{SESM}(\widetilde{m}^i_{\ell_i})=\overline{w}^i_{\ell_i}$. The only women who $\widetilde{m}^i_{\ell_i}$ prefers over $\overline{w}^i_{\ell_i}$, who are not happy women, are $w^i_{\ell_i}$ (if $\ell_i\neq 1$),$\widehat{w}^i_{\ell_i}$ (if $\ell_i\neq p$) and $\widetilde{w}^i_{\ell_i}$. However, $w^i_{\ell_i}$ and $\widehat{w}^i_{\ell_i}$ are matched to $m^i$ and $\widehat{m}^i$, respectively, who they prefer over $\widetilde{m}^i_{\ell_i}$. Moreover, $\widetilde{w}^i_{\ell_i}$ is matched to $\overline{m}^i_{\ell_i}$, who she prefers over $\widetilde{m}^i_{\ell_i}$.
Thus, $\widetilde{m}^i_{\ell_i}$ cannot belong to any blocking pair.

\item For all $i\in[k]$, recall that $\mu^C_{SESM}(\overline{m}^i_{\ell_i})=\widetilde{w}^i_{\ell_i}$. The only woman who $\overline{m}^i_{\ell_i}$ prefers over $\widetilde{w}^i_{\ell_i}$ is $\overline{w}^i_{\ell_i}$. However, $\overline{w}^i_{\ell_i}$ is matched to $\widetilde{m}^i_{\ell_i}$, who she prefers over $\overline{m}^i_{\ell_i}$. Thus, $\overline{m}^i_{\ell_i}$ cannot belong to a blocking pair.

\item For all $i,j\in[k]$ where $i<j$ and $t\in[q^{i,j}]$ such that $t\neq\ell_{i,j}$, recall that $\mu^C_{SESM}(m^{i,j}_t)=\overline{w}^{i,j}_t$. The only woman who $m^{i,j}_t$ prefers over $\overline{w}^{i,j}_t$, who is not a happy woman, is $w^{i,j}_t$. However, $w^{i,j}_t$ is matched to $\overline{m}^{i,j}_t$, who she prefers over $m^{i,j}_t$. Thus, $m^{i,j}_t$ cannot belong to a blocking pair.

\item For all $i,j\in[k]$ where $i<j$ and $t\in[q^{i,j}]$ such that $t\neq\ell_{i,j}$, recall that $\mu^C_{SESM}(\overline{m}^{i,j}_t)=w^{i,j}_t$. The only woman who $\overline{m}^{i,j}_t$ prefers over $w^{i,j}_t$ is $\overline{w}^{i,j}_t$. However, $\overline{w}^{i,j}_t$ is matched to $m^{i,j}_t$, who she prefers over $\overline{m}^{i,j}_t$. Thus, $\overline{m}^{i,j}_t$ cannot belong to a blocking pair.
\end{itemize}

This concludes the proof of the lemma.
\end{proof}

In light of Lemma \ref{lem:w1SESMforwardSM}, the measure $\delta(\mu^C_{SESM})$ is well defined. We proceed to analyze this measure with the following lemma.

\begin{lemma}\label{lem:w1SESMforwardMeasure}
Let $I=(G,(V^1,V^2,\ldots,V^k))$ be a \yesinstance\ of {\sc Multicolored Clique}. Let $C$ be a multicolored clique of $G$. Then, $\delta(\mu^C_{SESM})=0$.
\end{lemma}

\begin{proof}
Let $U=\{v^1_{\ell_1},v^2_{\ell_2},\ldots,v^t_{\ell_k}\}$ and $W=\{e^{i,j}_{\ell_{i,j}}: i,j\in[k], i<j\}$ denote the vertex and edge sets, respectively, of the multicolored clique $C$. We first analyze the positions of women in the preference lists of their matched partners.
\begin{itemize}
\item For all $i\in[k]$, $\pos_{m^i}(w^i_{\ell_i})=1+(1+|E(G)|+|E(G)|^{20})(\ell_i-1)$ and $\pos_{\widehat{m}^i}(\widehat{w}^i_{\ell_i})=1+(1+|E(G)|+|E(G)|^{20})(p-\ell_i)$.
\item For all $i\in[k]$ and $j\in\{2,3,\ldots,p\}$ such that $j\leq \ell_i$: $\pos_{m^i_j}(w^i_{j-1})=2$.
\item For all $i\in[k]$ and $j\in\{2,3,\ldots,p\}$ such that $j>\ell_i$: $\pos_{m^i_j}(w^i_j)=1$.
\item For all $i\in[k]$ and $j\in[p-1]$ such that $j\geq \ell_i$: $\pos_{\widehat{m}^i_j}(\widehat{w}^i_{j+1})=2$.
\item For all $i\in[k]$ and $j\in[p-1]$ such that $j<\ell_i$: $\pos_{\widehat{m}^i_j}(\widehat{w}^i_j)=1$.
\item For all $i\in[k]$: $\pos_{\widetilde{m}^i_{\ell_i}}(\overline{w}^i_{\ell_i})=4+2^{i-1}|E(G)|^{30}$ and $\pos_{\overline{m}^i_j}(\widetilde{w}^i_j)=2$.
\item For all $i\in[k]$ and $j\in[p]$ such that $j\neq\ell_i$: $\pos_{\widetilde{m}^i_j}(\widetilde{w}^i_j)=1$ and $\pos_{\overline{m}^i_j}(\overline{w}^i_j)=1$.
\item For all $i\in[k]$, $j\in[p]$: $\pos_{m^{i,j}_{\ell_{i,j}}}(w^{i,j}_{\ell_{i,j}})=1$ and $\pos_{\overline{m}^{i,j}_{\ell_{i,j}}}(\overline{w}^{i,j}_{\ell_{i,j}})=1$.
\item For all $i\in[k]$, $j\in[p]$ and $t\in[q^{i,j}]$ such that $t\neq\ell_{i,j}$: $\pos_{m^{i,j}_t}(\overline{w}^{i,j}_t)=2+|E(G)|^{10}$ and $\pos_{\overline{m}^{i,j}_t}(w^{i,j}_t)=2$.
\item For all $i\in[\alpha']$: $\pos_{m^i_{\hap}}(w^i_{\hap})=1$.
\item $\pos_{m^\star}(w^\star)=\alpha+1$.
\end{itemize}

Thus, we have that the following equality holds.
\[\begin{array}{lll}
\sat_M(\mu^C_{SESM}) & = & \displaystyle{\sum_{i=1}^k\left(2 + (1+|E(G)|+|E(G)|^{20})(\ell_i-1) + (1+|E(G)|+|E(G)|^{20})(p-\ell_i)\right)}\\

&& + \displaystyle{\sum_{i=1}^k\sum_{j=2}^{\ell_i}2 + \sum_{i=1}^k\sum_{j=\ell_i+1}^{p}1 + \sum_{i=1}^k\sum_{j=\ell_i}^{p-1}2 + \sum_{i=1}^k\sum_{j=1}^{\ell_i-1}1}\\

&& + \displaystyle{\sum_{i=1}^k\left(4+2^{i-1}|E(G)|^{30} + 2\right) + \sum_{i=1}^k\sum_{j\in[p],j\neq\ell_i}\left(1 +1\right)}\\

&& + \displaystyle{\sum_{i=1}^{k-1}\sum_{j=i+1}^k\left(1+1\right) + \sum_{i=1}^{k-1}\sum_{j=i+1}^k\sum_{t\in[q^{i,j}],t\neq\ell_{i,j}}\left(2+|E(G)|^{10} + 2\right)}\\

\medskip
&& + \alpha' + (\alpha+1)\\

& = & 2k + (1+|E(G)|+|E(G)|^{20})(p-1)k + 3kp - 3k + 6k + (2^k-1)|E(G)|^{30}\\

\medskip
&&  + 2(p-1)k + k(k-1) + (4+|E(G)|^{10})(|E(G)|-{k \choose 2}) + \alpha' + \alpha + 1\\

& = & 1 + 3k + 6pk -k^2 + (pk-k+4)|E(G)|\\
&& + (p-1)k|E(G)|^{20} + (|E(G)|-{k \choose 2})|E(G)|^{10} + (2^k-1)|E(G)|^{30} + \alpha' + \alpha.
\end{array}\]

Second, we analyze the positions of men in the preference lists of their matched partners.
\begin{itemize}
\item For all $i\in[k]$: $\pos_{w^i_{\ell_i}}(m^i)=2$ and $\pos_{\widehat{w}^i_{\ell_i}}(\widehat{m}^i)=2$.
\item For all $i\in[k]$ and $j\in[\ell_i-1]$: $\pos_{w^i_j}(m^i_{j+1})=1$.
\item For all $i\in[k]$ and $j\in[p]$ such that $j>\ell_i$: $\pos_{w^i_j}(m^i_{j+1})=4+|E(G)|^{20}$.
\item For all $i\in[k]$ and $j\in[p]$ such that $j>\ell_i$: $\pos_{\widehat{w}^i_j}(\widehat{m}^i_{j-1})=1$.
\item For all $i\in[k]$ and $j\in[\ell_i-1]$: $\pos_{w^i_j}(m^i_{j+1})=4+|E(G)|^{20}$.
\item For all $i\in[k]$: $\pos_{\widetilde{w}^i_{\ell_i}}(\overline{m}^i_{\ell_i})=1$ and $\pos_{\overline{w}^i_{\ell_i}}(\widetilde{m}^i_{\ell_i})=1$.
\item For all $i\in[k]$ and $j\in[p]$ such that $j\neq\ell_i$: $\pos_{\widetilde{w}^i_j}(\widetilde{m}^i_j)=2+2^{k-i}|E(G)|^{40}$ and $\pos_{\overline{w}^i_j}(\overline{m}^i_j)=2$.
\item For all $i\in[k]$, $j\in[p]$: $\pos_{w^{i,j}_{\ell_{i,j}}}(m^{i,j}_{\ell_{i,j}})=6+|E(G)|^{10}$ and $\pos_{\overline{w}^{i,j}_{\ell_{i,j}}}(\overline{m}^{i,j}_{\ell_{i,j}})=2$.
\item For all $i\in[k]$, $j\in[p]$ and $t\in[q^{i,j}]$ such that $t\neq\ell_{i,j}$: $\pos_{w^{i,j}_t}(\overline{m}^{i,j}_t)=1$ and $\pos_{\overline{w}^{i,j}_t}(m^{i,j}_t)=1$.
\item For all $i\in[\alpha']$: $\pos_{w^i_{\hap}}(m^i_{\hap})=1$.
\item $\pos_{w^\star}(m^\star)=1$.
\end{itemize}

Thus, we have that the following equality holds.
\[\begin{array}{lll}
\sat_W(\mu^C_{SESM}) & = & \displaystyle{\sum_{i=1}^k\left(2+2\right) + \sum_{i=1}^k\sum_{j=1}^{\ell_i-1}1 + \sum_{i=1}^k\sum_{j=\ell_i+1}^p(4+|E(G)|^{20}) + \sum_{i=1}^k\sum_{j=\ell_i+1}^p1}\\

&& + \displaystyle{\sum_{i=1}^k\sum_{j=1}^{\ell_i-1}(4+|E(G)|^{20}) + \sum_{i=1}^k\left(1+1\right) + \sum_{i=1}^k\sum_{j\in[p],j\neq\ell_i}\left(2+2^{k-i}|E(G)|^{40} + 2\right)}\\

\medskip
&& + \displaystyle{\sum_{i=1}^{k-1}\sum_{j=i+1}^k\left(6+|E(G)|^{10}+2\right) + \sum_{i=1}^{k-1}\sum_{j=i+1}^k\sum_{t\in[q^{i,j}],t\neq\ell_{i,j}}\left(1+1\right) + \alpha' + 1}\\

& = & 4k + (p-1)k + (p-1)k(4+|E(G)|^{20}) + 2k + 4(p-1)k\\

\medskip
&& + (p-1)(2^k-1)|E(G)|^{40} + {k\choose 2}(8+|E(G)|^{10}) + 2(|E(G)|-{k\choose 2}) + \alpha' + 1\\

& = & 1 - 6k + 9pk + 3k^2 + 2|E(G)|\\

&& + (p-1)k|E(G)|^{20} + {k\choose 2}|E(G)|^{10} + (p-1)(2^k-1)|E(G)|^{40} + \alpha'.
\end{array}\]

Thus, to derive that $\sat_M(\mu^C_{SESM})=\sat_W(\mu^C_{SESM})$, which would imply that $\delta(\mu^C_{SESM})=0$, we need to show that the following equality holds.
\[\begin{array}{ll}
& 3k + 6pk -k^2 + (pk-k+4)|E(G)| + (|E(G)|-{k \choose 2})|E(G)|^{10} + (2^k-1)|E(G)|^{30} + \alpha\\

= & - 6k + 9pk + 3k^2 + 2|E(G)| + {k\choose 2}|E(G)|^{10} + (p-1)(2^k-1)|E(G)|^{40}
\end{array}\]

However, recall that
\[\begin{array}{ll}
\alpha = & -9k +3pk +4k^2 - (pk-k+2)|E(G)|\\
& - (|E(G)|-2{k\choose 2})|E(G)|^{10} - (2^k-1)|E(G)|^{30} + (p-1)(2^k-1)|E(G)|^{40}.
\end{array}\]

This concludes the proof of the lemma.
\end{proof}

Combining Lemmata \ref{lem:w1SESMforwardSM} and \ref{lem:w1SESMforwardMeasure}, we derive the following corollary.

\begin{corollary}\label{cor:w1SESMforward}
Let $I$ be a \yesinstance\ of {\sc Multicolored Clique}. Then, for the instance $\red_{SESM}(I)$ of {\sc SESM}, $\Delta=0$.
\end{corollary}

This concludes the proof of the forward direction.

\medskip
\myparagraph{Reverse Direction.} Second, we prove that given an instance $I$ of {\sc Multicolored Clique}, if for the instance $\red_{SESM}(I)$ of {\sc SESM}, $\Delta=0$, then we can construct a solution for $I$. To this end, we first need to analyze the structure of stable matchings of $\red_{SESM}(I)$ whose sex-equality measure is $0$. Let us begin by proving the following two lemmata.

\begin{lemma}\label{lem:w1SESMAll}
Let $I$ be an of {\sc Multicolored Clique}. Every stable matching of $\red_{SESM}(I)$ matches all agents.
\end{lemma}

\begin{proof}
Observe that the matching $\mu$ that matches every man to the woman he prefers the most, except for $m^\star$ who is matched to $w^\star$, is a stable matching. Indeed, in this matching, it is clear that no man but $m^\star$ can belong to a blocking pair simply because there is no woman who such a man prefers over his current partner, and the man $m^\star$ cannot belong to a blocking pair since all of the women who he prefers over $w^\star$ are matched to their most preferred men. Thus, by Proposition \ref{lem:matchSame}, we deduce that every stable matching of $\red_{SESM}(I)$ matches all men. Since the number of men is equal to the number of women, we conclude the correctness of the lemma.
\end{proof}

\begin{lemma}\label{lem:w1SESMReverseHap}
Let $I$ be an of {\sc Multicolored Clique}. Every stable matching $\mu$ of $\red_{SESM}(I)$ satisfies the following conditions.
\begin{enumerate}
\item\label{item:w1SESMReverseHap1} For all $i\in[k]$, $\mu(m^i)\in W^i_{\bas}$ and $\mu(\widehat{m}^i)\in \widehat{W}^i_{\bas}$.
\item\label{item:w1SESMReverseHap2} For all $i\in[k]$ and $j\in[p]$, $\mu(\widetilde{m}^i_j)\in\{\widetilde{w}^i_j,\overline{w}^i_j\}$.
\item\label{item:w1SESMReverseHap3} For all $i\in[\alpha']$, $\mu(m^i_{\hap})=w^i_{\hap}$.
\item\label{item:w1SESMReverseHap4} $\mu(m^\star)=w^\star$.
\end{enumerate}
\end{lemma}

\begin{proof}
Let $\mu$ be a stable matching of $\red_{SESM}(I)$. Since for all $i\in[\alpha']$, $m^i_{\hap}$ and $w^i_{\hap}$ prefer each other over all other agents, they must be matched to one another (by $\mu$), else they would have formed a blocking pair. Then, since apart from happy women, the preference lists of $m^\star$ and $w^\star$ only contain each other, they also must be matched to one another, else they would have formed a blocking pair. We have thus proved the satisfaction of Conditions \ref{item:w1SESMReverseHap3} and \ref{item:w1SESMReverseHap4}.

Next, to prove the satisfaction of Condition \ref{item:w1SESMReverseHap2}, consider some man $\widetilde{m}^i_j\in \widetilde{M}^i_j$. Since Condition \ref{item:w1SESMReverseHap3} is satisfied, $\widetilde{m}^i_j$ is either unmatched or matched to a woman in $\{\widetilde{w}^i_j,w^i_j,\widehat{w}^i_j,\overline{w}^i_j\}$. Suppose, by way of contradiction, that $\widetilde{m}^i_j$ is either unmatched or matched to a woman in $\{w^i_j,\widehat{w}^i_j\}$. Then, since $\widetilde{m}^i_j$ prefers being matched to $\widetilde{w}^i_j$ over his current status, yet $\mu$ does not have any blocking pair, we have that $\widetilde{w}^i_j$ must be matched to (since she is the only unhappy man that she prefers over $\widetilde{m}^i_j$). Then, however $\overline{w}^i_j$ is necessarily left unmatched, and thus she forms a blocking pair together with $\overline{m}^i_j$, reaching a contradiction. 

Finally, to prove the satisfaction of Condition \ref{item:w1SESMReverseHap1}, consider some index $i\in[k]$. Since Condition \ref{item:w1SESMReverseHap3} is satisfied, $m^i$ is either unmatched or matched to a woman in $W^i_{\bas}\cup (\bigcup_{j=i+1}^kW^{i,j}_{\bas})$. Suppose, by way of contradiction, that $m^i$ is either unmatched or matched to a woman in $(\bigcup_{j=i+1}^kW^{i,j}_{\bas})$. By Lemma~\ref{lem:w1SESMAll}, the first possibility is not feasible, and therefore there exists $j\in[k]$ such that $m^i$ is matched to some woman $w^{i,j}_t\in W^{i,j}_t$. Since $w^{i,j}_t$ prefers $\overline{m}^{i,j}_t$ over $m^i$, yet $\mu$ does not have any blocking pair, we have that $\overline{m}^{i,j}_t$ must be matched to $\overline{w}^{i,j}_t$ (since she is the only unhappy man that she prefers over $w^{i,j}_t$). Then, however $m^{i,j}_t$ is necessarily left unmatched, and thus he forms a blocking pair together with $\overline{w}^{i,j}_t$, reaching a contradiction. Symmetrically, we derive that $\mu(\widehat{m}^i)\in \widehat{W}^i_{\bas}$. This concludes the proof of the lemma.
\end{proof}

To proceed with our proof of correctness of our reduction, we need to analyze the sex-equality measure. For this purpose, we use the notation defined below, which is well-defined due to Lemmata \ref{lem:w1SESMAll} and \ref{lem:w1SESMReverseHap}.

\begin{definition}\label{def:w1hardSESMABC}
Let $I$ be an instance of {\sc Multicolored Clique}, and let $\mu$ be a stable matching of $\red_{SESM}(I)$. For all $i\in[k]$, let $a(\mu,i)$ and $\widehat{a}(\mu,i)$ denote the indices $j\in[p]$ and $\widehat{j}\in[p]$, respectively, such that $\mu(m^i)=w^i_j$ and $\mu(\widehat{m}^i)=\widehat{w}^i_j$. Moreover, for all $i\in[k]$, denote $b(\mu,i)=|\{j\in[p]: \mu(\widetilde{m}^i_j)=\overline{w}^i_j\}|$. Finally, denote $c(\mu)=|\{(i,j,t)\in[k]\times[k]\times[q^{i,j}]: i<j, \mu(m^{i,j}_t)=\overline{w}^i_j\}|$.
\end{definition}

Next, we analyze the sex-equality measure.

\begin{lemma}\label{lem:w1SESMReverseSpecificMeasure}
Let $I$ be an instance of {\sc Multicolored Clique}, and let $\mu$ be a stable matching of $\red_{SESM}(I)$. Then, for some $-100|E(G)|^2\leq x\leq 100|E(G)|^2$, it holds that
\[\begin{array}{ll}
\sat_M(\mu)-\alpha' = & \displaystyle{\sum_{i=1}^k(p+a(\mu,i)-\widehat{a}(\mu,i)-1)|E(G)|^{20} + (c(\mu)-|E(G)|+2{k\choose 2})|E(G)|^{10}}\\
&+ \displaystyle{\left((\sum_{i=1}^kb(\mu,i)2^{i-1})-2^k+1\right)|E(G)|^{30} + (p-1)(2^k-1)|E(G)|^{40} + x}.
\end{array}\]

Moreover, for some $-10|E(G)|^2\leq y\leq 10|E(G)|^2$, it holds that
\[\begin{array}{l}
\sat_W(\mu)-\alpha'=\\
\displaystyle{\sum_{i=1}^k(p+\widehat{a}(\mu,i)-a(\mu,i)-1)|E(G)|^{20} + (|E(G)|-c(\mu))|E(G)|^{10} + \sum_{i=1}^k(p-b(\mu,i))2^{k-i}|E(G)|^{40} + y}.
\end{array}\]
\end{lemma}

\begin{proof}
On the one hand, by Lemmata \ref{lem:w1SESMAll} and \ref{lem:w1SESMReverseHap} and the definition of preference lists of the men of $\red_{SESM}(I)$, we have that
\[\begin{array}{lll}
\sat_M(\mu) & = & \displaystyle{\sum_{i=1}^k\pos_{m^i}(\mu(m^i)) + \sum_{i=1}^k\pos_{\widehat{m}^i}(\mu(\widehat{m}^i)) + \sum_{i=1}^k\sum_{j=1}^p\pos_{m^i_j}(\mu(m^i_j)) + \sum_{i=1}^k\sum_{j=1}^p\pos_{\widehat{m}^i_j}(\mu(\widehat{m}^i_j))}\\

&& + \displaystyle{\sum_{i=1}^k\sum_{j=1}^p\pos_{\widetilde{m}^i_j}(\mu(\widetilde{m}^i_j)) + \sum_{i=1}^k\sum_{j=1}^p\pos_{\overline{m}^i_j}(\mu(\overline{m}^i_j))}\\

&& + \displaystyle{\sum_{i=1}^{k-1}\sum_{j=i+1}^k\sum_{t=1}^{q^{i,j}}\pos_{m^{i,j}_t}(\mu(m^{i,j}_t)) + \sum_{i=1}^{k-1}\sum_{j=i+1}^k\sum_{t=1}^{q^{i,j}}\pos_{\overline{m}^{i,j}_t}(\mu(\overline{m}^{i,j}_t))} + \alpha'+\alpha+1\\

& = & \displaystyle{\sum_{i=1}^k\left(1+(1+|E(G)|+|E(G)|^{20})(a(\mu,i)-1)\right)}\\

&& + \displaystyle{\sum_{i=1}^k\left(1+(1+|E(G)|+|E(G)|^{20})(p-\widehat{a}(\mu,i))\right)}\\

&& + \displaystyle{\sum_{i=1}^k\sum_{j=2}^{a(\mu,i)}2 + \sum_{i=1}^k\sum_{j=a(\mu,i)+1}^{p}1} + \displaystyle{\sum_{i=1}^k\sum_{j=1}^{\widehat{a}(\mu,i)-1}1 + \sum_{i=1}^k\sum_{j=\widehat{a}(\mu,i)}^{p-1}2}\\

&& + \displaystyle{\sum_{i=1}^k\left(b(\mu,i)(6 + 2^{i-1}|E(G)|^{30}) + 2(p-b(\mu,i))\right)}\\

&& + \displaystyle{c(\mu)(4+|E(G)|^{10}) + 2(|E(G)|-c(\mu)) + \alpha'+\alpha + 1}.
\end{array}\]

Thus, we have that
\[\begin{array}{ll}
& \displaystyle{\alpha + \sum_{i=1}^k(p+a(\mu,i)-\widehat{a}(\mu,i)-1)|E(G)|^{20} + c(\mu)|E(G)|^{10} + \sum_{i=1}^kb(\mu,i)2^{i-1}|E(G)|^{30}} -10|E(G)|^2\\
\leq & \sat_M(\mu) -\alpha' \\
\leq & \displaystyle{\alpha + \sum_{i=1}^k(p+a(\mu,i)-\widehat{a}(\mu,i)-1)|E(G)|^{20} + c(\mu)|E(G)|^{10} + \sum_{i=1}^kb(\mu,i)2^{i-1}|E(G)|^{30}} + 10|E(G)|^2.
\end{array}\]

Substituting $\alpha$, we have that 
\[\begin{array}{ll}
& \displaystyle{\sum_{i=1}^k(p+a(\mu,i)-\widehat{a}(\mu,i)-1)|E(G)|^{20} + (c(\mu)-|E(G)|+2{k\choose 2})|E(G)|^{10}}\\
& + \displaystyle{\left((\sum_{i=1}^kb(\mu,i)2^{i-1})-2^k+1\right)|E(G)|^{30} + (p-1)(2^k-1)|E(G)|^{40} - 100|E(G)|^2}\\

\leq & \sat_M(\mu)-\alpha' \\

\leq & \displaystyle{\sum_{i=1}^k(p+a(\mu,i)-\widehat{a}(\mu,i)-1)|E(G)|^{20} + (c(\mu)-|E(G)|+2{k\choose 2})|E(G)|^{10}}\\
&+ \displaystyle{\left((\sum_{i=1}^kb(\mu,i)2^{i-1})-2^k+1\right)|E(G)|^{30} + (p-1)(2^k-1)|E(G)|^{40} + 100|E(G)|^2}.
\end{array}\]

On the other hand, by Lemmata \ref{lem:w1SESMAll} and \ref{lem:w1SESMReverseHap} and the definition of preference lists of the women of $\red_{SESM}(I)$, we have that
\[\begin{array}{lll}
\sat_W(\mu) & = & \displaystyle{\sum_{i=1}^k\pos_{\mu(m^i)}(m^i) + \sum_{i=1}^k\pos_{\mu(\widehat{m}^i)}(\widehat{m}^i) + \sum_{i=1}^k\sum_{j=1}^p\pos_{\mu(m^i_j)}(m^i_j) + \sum_{i=1}^k\sum_{j=1}^p\pos_{\mu(\widehat{m}^i_j)}(\widehat{m}^i_j)}\\

&& + \displaystyle{\sum_{i=1}^k\sum_{j=1}^p\pos_{\mu(\widetilde{m}^i_j)}(\widetilde{m}^i_j) + \sum_{i=1}^k\sum_{j=1}^p\pos_{\mu(\overline{m}^i_j)}(\overline{m}^i_j)}\\

&& + \displaystyle{\sum_{i=1}^{k-1}\sum_{j=i+1}^k\sum_{t=1}^{q^{i,j}}\pos_{\mu(m^{i,j}_t)}(m^{i,j}_t) + \sum_{i=1}^{k-1}\sum_{j=i+1}^k\sum_{t=1}^{q^{i,j}}\pos_{\mu(\overline{m}^{i,j}_t)}(\overline{m}^{i,j}_t)} + \alpha' + 1\\

& = & \displaystyle{4k + \sum_{i=1}^k(a(\mu,i)-1) + \sum_{i=1}^k(p-a(\mu,i))(4 + |E(G)|^{20}))}\\

&& + \displaystyle{\sum_{i=1}^k(p-\widehat{a}(\mu,i)) + \sum_{i=1}^k(\widehat{a}(\mu,i)-1)(4 + |E(G)|^{20}))}\\

&& + \displaystyle{\sum_{i=1}^k\left(2b(\mu,i) + (p-b(\mu,i))(4+2^{k-i}|E(G)|^{40})\right)}\\

&& + 2c(\mu) + (|E(G)|-c(\mu))(8+|E(G)|^{10}) + \alpha' + 1.
\end{array}\]

Thus, we have that
\[\begin{array}{ll}
& \displaystyle{(p+\widehat{a}(\mu,i)-a(\mu,i)-1)|E(G)|^{20} + (|E(G)|-c(\mu))|E(G)|^{10} + \sum_{i=1}^k(p-b(\mu,i))2^{k-i}|E(G)|^{40}}\\
& -10|E(G)|^2\\
\leq & \sat_W(\mu)-\alpha' \\
\leq & \displaystyle{(p+\widehat{a}(\mu,i)-a(\mu,i)-1)|E(G)|^{20} + (|E(G)|-c(\mu))|E(G)|^{10} + \sum_{i=1}^k(p-b(\mu,i))2^{k-i}|E(G)|^{40}}\\
& + 10|E(G)|^2.
\end{array}\]

This concludes the proof of the lemma.
\end{proof}

The following two lemmata prove further useful properties of the matching partners of agents in the context of a stable matching $\mu$ that satisfies $\delta(\mu)=0$.

\begin{lemma}\label{lem:w1SESMReverseTilde}
Let $I$ be an instance of {\sc Multicolored Clique}. Let $\mu$ be a stable matching of $\red_{SESM}(I)$ such that $\delta(\mu)=0$. Then, for all $i\in[k]$, there exists $j\in[p]$ such that $\mu(\widetilde{m}^i_j)=\overline{w}^i_j$ and for all $t\neq j$, $\mu(\widetilde{m}^i_j)=\widetilde{w}^i_j$.
\end{lemma}

\begin{proof}
The statement of the lemma is equivalent to the statement that for all $i\in[k]$, $b(\mu,i)=1$. Since $\delta(\mu)=0$, it holds that $\sat_M(\mu)-\alpha'=\sat_W(\mu)-\alpha'$. Hence, by Lemma \ref{lem:w1SESMReverseSpecificMeasure}, the two following equalities are satisfied.
\begin{itemize}
\item $\displaystyle{\left((\sum_{i=1}^kb(\mu,i)2^{i-1})-2^k+1\right)|E(G)|^{30} = 0}$.
\item $\displaystyle{(p-1)(2^k-1)|E(G)|^{40} = \sum_{i=1}^k(p-b(\mu,i))2^{k-i}|E(G)|^{40}}$.
\end{itemize}

Simplifying the equalities above, we derive that the two following equalities are satisfied.
\begin{itemize}
\item $\displaystyle{\sum_{i=1}^kb(\mu,i)2^{i-1} = 2^k-1}$.
\item $\displaystyle{\sum_{i=1}^kb(\mu,i)2^{k-i}= 2^k-1}$.
\end{itemize}

Note that for all $i\in[k]$, $b(\mu,i)\in[p]$. Let $\varphi$ be an assignment to the variables $b(\mu,i)$ that satisfies this condition as well as the two equalities above. We claim that $\varphi$ necessarily assigns 1 to all of these variables. This claim can be easily proven by induction on $k$. In the base case, where $k=1$, the first equality directly implies that $b(\mu,1)=1$. Now, suppose that $k\geq 2$ and that the claim holds for $k-1$. Then, first note that to satisfy the first equality, it must hold that $b(\mu,k)\leq 1$, while to satisfy the second one, it must holds that $b(\mu,k)\geq 1$. Thus, $b(\mu,k)=1$, which implies that the two following equalities are then satisfied.
\begin{itemize}
\item $\displaystyle{\sum_{i=1}^{k-1}b(\mu,i)2^{i-1} + 2^{k-1} = 2^k-1}$. That is, $\displaystyle{\sum_{i=1}^{k-1}b(\mu,i)2^{i-1} = 2^{k-1}-1}$
\item $\displaystyle{\sum_{i=1}^{k-1}b(\mu,i)2^{k-i} + 1 = 2^k-1}$. That is, $\displaystyle{\sum_{i=1}^{k-1}b(\mu,i)2^{(k-1)-i} = 2^{k-1}-1}$.
\end{itemize}

By the inductive hypothesis, we derive that for all $i\in[k-1]$, it also holds that $b(\mu,i)=1$. This concludes the proof of the lemma.
\end{proof}

\begin{lemma}\label{lem:w1SESMReverseColClass}
Let $I$ be an instance of {\sc Multicolored Clique}. Let $\mu$ be a stable matching of $\red_{SESM}(I)$ such that $\delta(\mu)=0$. Then, for all $i\in[k]$, there exists $j\in[p]$ such that $\mu(m^i)=w^i_j$ and $\mu(\widehat{m}^i)=\widehat{m}^i_j$.
\end{lemma}

\begin{proof}
Since $\delta(\mu)=0$, it holds that $\sat_M(\mu)-\alpha'=\sat_W(\mu)-\alpha'$. Hence, by Lemma \ref{lem:w1SESMReverseSpecificMeasure}, the following equality is satisfied.
\[\displaystyle{\sum_{i=1}^k(p+a(\mu,i)-\widehat{a}(\mu,i)-1)|E(G)|^{20}} = \displaystyle{\sum_{i=1}^k(p+\widehat{a}(\mu,i)-a(\mu,i)-1)|E(G)|^{20}}.\]

For all $i\in[k]$, denote $a'(\mu,i)=(p+1)-\widehat{a}(\mu,i)$. Then, by simplifying the equality above, we derive that the following equality is satisfied.
\[\displaystyle{\sum_{i=1}^k(a(\mu,i)+a'(\mu,i)) = (p+1)k}.\]

Fix some $i'\in[k]$. By Lemma \ref{lem:w1SESMReverseTilde}, there exists $j'\in[p]$ such that $\mu(\widetilde{m}^{i'}_{j'})=\overline{w}^{i'}_{j'}$. Observe that $\widetilde{m}^{i'}_{j'}$ prefers both $w^{i'}_{j'}$ and $\widehat{w}^{i'}_{j'}$ over $\overline{w}^{i'}_{j'}$, but since $\mu\in{\cal S}$, he forms a blocking pair with neither of them. Thus, we deduce that $w^{i'}_{j'}$ is matched to either $m^{i'}$ or $m^{i'}_{j'+1}$, and that $\widehat{w}^{i'}_{j'}$ is matched to either $\widehat{m}^{i'}$ or $\widehat{m}^{i'}_{j'-1}$. Hence, by Lemmata \ref{lem:w1SESMAll} and \ref{lem:w1SESMReverseHap}, we deduce that all of the women in $\{w^{i'}_j\in W^{i'}_{\bas}: j\geq j'\}$ are matched to men in $\{m^{i'}_j\in M^{i'}_{\enr}: j>j'\}\cup\{m^{i'}\}$, and all of the women in $\{\widehat{w}^{i'}_j\in \widehat{W}^{i'}_{\bas}: j\leq j'\}$ are matched to men in $\{\widehat{m}^{i'}_j\in \widehat{M}^{i'}_{\enr}: j<j'\}\cup\{\widehat{m}^{i'}\}$. In particular, this implies that $a(\mu,i')\geq j'$ and $\widehat{a}(\mu,i')\leq j$. The latter inequality is equivalent to $a'(\mu,i')\geq (p+1)-j$. We thus get that $a(\mu,i')+a'(\mu,i')\geq p+1$.

Since the choice of $i'$ above was arbitrary, we derive that for all $i\in[k]$, $a(\mu,i)+a'(\mu,i)\geq p+1$. However, since $\displaystyle{\sum_{i=1}^k(a(\mu,i)+a'(\mu,i)) = (p+1)k}$, we have that for all $i\in[k]$, $a(\mu,i)+a'(\mu,i)=p+1$. By substituting $a'(i,\cdot)$, we have that for all $i\in[k]$, $a(\mu,i)=\widehat{a}(\mu,i)$. This claim is equivalent to the statemtn of the lemma, and thus we conclude its proof.
\end{proof}

We are now ready to prove the correctness of the reverse direction.

\begin{lemma}\label{lem:w1SESMReverse}
Let $I=(G,(V^1,V^2,\ldots,V^k))$ be an instance of {\sc Multicolored Clique}. If for the instance $\red_{SESM}(I)$ of {\sc SESM}, $\Delta=0$, then $I$ is a \yesinstance\ of {\sc Multicolored Clique}.
\end{lemma}

\begin{proof}
Suppose that for the instance $\red_{SESM}(I)$ of {\sc SESM}, $\Delta=0$. Then, there exists a stable matching $\mu$ such that $\delta(\mu)=0$. Hence, by Lemma \ref{lem:w1SESMReverseSpecificMeasure}, the following equality is satisfied.
\[\displaystyle{(c(\mu)-|E(G)|+2{k\choose 2})|E(G)|^{10} = (|E(G)|-c(\mu))|E(G)|^{10}}.\]

Simplifying the equality above, we have that the following equality is satisfied.
\[\displaystyle{c(\mu) = |E(G)|-{k\choose 2}}.\]

Denote $W'=\{w^{i,j}_t\in \bigcup_{i,j\in[k], i<j}W^{i,j}_{\bas}: \mu(w^{i,j}_t)=m^{i,j}_t\}$. By the equality above and Lemmata \ref{lem:w1SESMAll} and \ref{lem:w1SESMReverseHap}, we have that $|W'|={k\choose 2}$. By Lemma \ref{lem:w1SESMReverseColClass}, for all $i\in[k]$, it is well defined to let $\ell_i$ denote the index $j\in[p]$ such that $\mu(m^i)=w^i_{\ell_i}$ and $\mu(\widehat{m}^i)=\widehat{w}^i_{\ell_i}$. Since every woman $w^{i,j}_t\in W'$ prefers both $m^i$ and $\widehat{m}^i$ over her matched partner, we have that $w^{i,j}_t$ is located before $w^i_{\ell_i}$ in the preference list of $m^i$ as well as before $\widehat{w}^i_{\ell_i}$ in the preference list of $\widehat{m}^i$. However, by the definition of the preference lists of $m^i$ and $\widehat{m}^i$, it must then hold that $e^{i,j}_t$ is an edge incident to $v^i_{\ell_i}$ in $G$. Hence, we have that $X=\{v^i_{\ell_i}: i\in[k]\}$ is a subset of $V(G)$ size $k$ such that every edge in $Y=\{e^{i,j}_t: w^{i,j}_t\}$ is incident to two vertices in $X$. Since $|Y|=|W'|={k\choose 2}$, we deduce that $X$ is the vertex set of a colorful clique of $G$. We thus conclude that $I$ is a \yesinstance\ of {\sc Multicolored Clique}.
\end{proof}

\medskip
\myparagraph{Summary.} Finally, we note that the reduction can be performed in time that is polynomial in the size of the output. That is, we have the following observation.

\begin{observation}\label{obs:w1SESMTime}
Let $I=(G,(V^1,V^2,\ldots,V^k))$ be an instance of {\sc Multicolored Clique}. Then, the instance $\red_{SESM}(I)$ of {\sc SESM} can be constructed in time $2^{\OO(k)}\cdot n^{\OO(1)}$. Here, $n=|V(G)|$.
\end{observation}

By Proposition \ref{prop:multiClique}, Lemma \ref{lem:twSESM}, Corollary \ref {cor:w1SESMforward}, Lemma \ref{lem:w1SESMReverse} and Observation \ref{obs:w1SESMTime}, we conclude that {\sc SESM} is \WOH. Moreover, unless \ETH\ fails, {\sc SESM} cannot be solved in time $f(\tw)\cdot n^{o(\tw)}$ for any function $f$ that depends only on $\tw$. Here, $n$ is the number of agents.

%% file: w1hardnessBSM.tex
\subsection{Balanced Stable Marriage}\label{sec:w1BSM}

Second, we prove that {\sc BSM} is \WOH, and that unless \ETH\ fails, {\sc BSM} cannot be solved in time $f(\tw)\cdot n^{o(\tw)}$ for any function $f$ that depends only on $\tw$. 

\subsubsection{Reduction}

Let $I=(G,(V^1,V^2,\ldots,V^k))$ be an instance of {\sc Multicolored Clique}. We now describe how to construct an instance $\red_{BSM}(I)=(M,W,\{\pos_m\}|_{m\in M},\{\pos_w\}|_{w\in W})$ of {\sc BSM}. The construction of $\red_{BSM}(I)$ is identical to $\red_{SESM}(I)$ up until the part where we introduce happy pairs. The preference lists the agents of the form $m^i$, $\widehat{m}^i$, $m^{i,j}_t$, $w^i_j$, $\widetilde{m}^i_j$, $m^{i,j}_t$, $w^{i,j}_t$ and $w^\star$ are defined exactly as in the case of {\sc SESM}. The preference list of $m^\star$ is defined as before with the exception that now rather than $\alpha$ happy women, it contains the following number $\widehat{\alpha}$ of happy women.
\[\begin{array}{lll}
\widehat{\alpha} &=& \alpha - (p-1)(2^k-1)|E(G)|^{40} + \frac{1}{3}(p-1)(4^k-1)|E(G)|^{40}\\
&= & -9k +3pk +4k^2 - (pk-k+2)|E(G)|\\
&& - (|E(G)|-2{k\choose 2})|E(G)|^{10} - (2^k-1)|E(G)|^{30} + \frac{1}{3}(p-1)(4^k-1)|E(G)|^{40}.
\end{array}\]

Thus, it only remains to define the preference lists of the agents of the form $\widetilde{w}^i_j$.
For every $i\in[k]$ and $j\in[p]$, we explicitly define the preference list of $\widetilde{w}^i_j$ as follows. We set $\domain(\pos_{\widetilde{w}^i_j})$ to consist of the union of $\{\overline{m}^i_j,\widetilde{m}^i_j\}$ and a set of arbitrarily chosen $4^{i-1}|E(G)|^{40}$ happy men. Define $\pos_{\widetilde{w}^i_j}(\overline{m}^i_j)=1$ and $\pos_{\widetilde{w}^i_j}(\widetilde{m}^i_j)=2+4^{i-1}|E(G)|^{40}$. All of the positions that have not been occupied above are occupied by the happy men (the choice of which happy man occupies which vacant position is arbitrary).

Finally, let us define
\[\begin{array}{ll}
\eta = & 1 + 3k + 6pk -k^2 + (pk-k+4)|E(G)|\\
& + (p-1)k|E(G)|^{20} + (|E(G)|-{k \choose 2})|E(G)|^{10} + (2^k-1)|E(G)|^{30} + \alpha'+\widehat{\alpha}\\
\end{array}\]

Note that
\[\begin{array}{ll}
\eta = & 1 - 6k + 9pk + 3k^2 + 2|E(G)|\\
& + (p-1)k|E(G)|^{20} + {k\choose 2}|E(G)|^{10} + (p-1)\frac{1}{3}(4^k-1)|E(G)|^{40} + \alpha'.
\end{array}\]

\subsubsection{Treewidth}

Let $I$ be an instance of {\sc Multicolored Clique}. Notice that the primal graph $H$ of $\red_{BSM}(I)=(M,W,\{\pos_m\}|_{m\in M},\{\pos_w\}|_{w\in W})$ is the same as the primal graph $H'$ of the instance of {\sc SESM} constructed in Section \ref{sec:w1BSM} with the exception that a different number of pendent paths on two vertices that are attached to each vertex of the form $\widetilde{w}^i_j$. Clearly, this observation implies that the treewidth of $H$ is the same as the treewidth of $H'$. We thus have the following version of Lemma \ref{lem:twSESM}.

\begin{lemma}\label{lem:twBSM}
Let $I$ be an instance of {\sc Multicolored Clique}. Then, the treewidth of $\red_{BSM}(I)$ is bounded by $2k+\OO(1)$.
\end{lemma}

\subsubsection{Correctness}

\myparagraph{Forward Direction.} Due to the manner in which we define $\red_{BSM}$ in the context of {\sc BSM}, we can again employ Definition \ref{def:cliquetoSESM}. More precisely, given a multicolored clique $C$ of an instance $I$ of {\sc Multicolored Clique}, we define $\mu^C_{BSM}$ exactly as $\mu^C_{SESM}$ (with the modification that we now match a different number of happy agents to one another). Consequently, we derive the appropriate version of Lemma \ref{lem:w1SESMforwardSM}.

\begin{lemma}\label{lem:w1BSMforwardSM}
Let $I=(G,(V^1,V^2,\ldots,V^k))$ be a \yesinstance\ of {\sc Multicolored Clique}. Let $C$ be a multicolored clique of $G$. Then, $\mu^C_{BSM}$ is a stable matching of $\red_{BSM}(I)$.
\end{lemma}

In light of Lemma \ref{lem:w1BSMforwardSM}, the measure $\bal(\mu^C_{BSM})$ is well defined. We proceed to analyze this measure with the following lemma.

\begin{lemma}\label{lem:w1BSMforwardMeasure}
Let $I=(G,(V^1,V^2,\ldots,V^k))$ be a \yesinstance\ of {\sc Multicolored Clique}. Let $C$ be a multicolored clique of $G$. Then, $\delta(\mu^C_{BSM})\leq\eta$.
\end{lemma}

\begin{proof}
First, since among the men, only the preference list of $m^\star$ has changed, where now it contains $\widehat{\alpha}$ happy women rather than $\alpha$ happy women, we have that $\sat_M(\mu^C_{BSM}) = \sat_M(\mu^C_{SESM})-\alpha+\widehat{\alpha}$. By the definition of $\eta$ and the arguments given in the proof of Lemma \ref{lem:w1SESMforwardMeasure}, we have that $\sat_M(\mu^C_{SESM})=\eta-(p-1)(2^k-1)|E(G)|^{40}+(p-1)\frac{1}{3}(4^k-1)|E(G)|^{40}=\eta-\alpha+\widehat{\alpha}$, which implies that $\sat_M(\mu^C_{BSM})=\eta$.

Second, note that since the preference lists of women of the form $\widetilde{w}^i_j$ have changed, the term 
\[\begin{array}{lll}
\displaystyle{\sum_{i=1}^k\sum_{j=1}^p\pos_{\widetilde{w}^i_j}(\mu^C_{SESM}(\widetilde{w}^i_j))} &=& k + \sum_{i=1}^k(p-1)\left(2+2^{i-1}|E(G)|^{40}\right)\\
&=& k + \displaystyle{\sum_{i=1}^k(p-1)\left(2+2^{k-i}|E(G)|^{40}\right)}\\
& = & k + 2(p-1)k + (p-1)(2^k-1)|E(G)|^{40}.
\end{array}\]
which is part of the analysis of $\sat_W(\mu^C_{SESM})$, is replaced by the term 

\[\begin{array}{lll}
\displaystyle{\sum_{i=1}^k\sum_{j=1}^p\pos_{\widetilde{w}^i_j}(\mu^C_{BSM}(\widetilde{w}^i_j))} &=& k + \sum_{i=1}^k(p-1)\left(2+2^{i-1}|E(G)|^{40}\right)\\
&=& k + \sum_{i=1}^k(p-1)\left(2+4^{i-1}|E(G)|^{40}\right)\\
& = & k + 2(p-1)k + (p-1)\frac{1}{3}(4^k-1)|E(G)|^{40}.
\end{array}\]
Since the preference lists of all other women remained the same, this is the only term that is changed. Thus, we have that $\sat_W(\mu^C_{BSM}) = \sat_W(\mu^C_{SESM}) - (p-1)(2^k-1)|E(G)|^{40} + (p-1)\frac{1}{3}(4^k-1)|E(G)|^{40}$. By Lemma \ref{lem:w1SESMforwardMeasure}, $\sat_M(\mu^C_{SESM})=\sat_W(\mu^C_{SESM})$, and therefore $\sat_W(\mu^C_{BSM})=\eta$. We thus conclude that $\delta(\mu^C_{BSM})\leq\eta$.
\end{proof}

Combining Corollary \ref{lem:w1BSMforwardSM} and Lemma \ref{lem:w1BSMforwardMeasure}, we derive the following corollary.

\begin{corollary}\label{cor:w1BSMforward}
Let $I$ be a \yesinstance\ of {\sc Multicolored Clique}. Then, for the instance $\red_{BSM}(I)$ of {\sc BSM}, $\Bal\leq\eta$.
\end{corollary}

This concludes the proof of the forward direction.

\medskip
\myparagraph{Reverse Direction.} Second, we prove that given an instance $I$ of {\sc Multicolored Clique}, if for the instance $\red_{BSM}(I)$ of {\sc BSM}, $\Bal=0$, then we can construct a solution for $I$. First, exactly as in the case of {\sc SESM}, the following two lemmata are true.

\begin{lemma}\label{lem:w1BSMAll}
Let $I$ be an of {\sc Multicolored Clique}. Every stable matching of $\red_{BSM}(I)$ matches all agents.
\end{lemma}

\begin{lemma}\label{lem:w1BSMReverseHap}
Let $I$ be an of {\sc Multicolored Clique}. Every stable matching $\mu$ of $\red_{BSM}(I)$ satisfies the following conditions.
\begin{enumerate}
\item For all $i\in[k]$, $\mu(m^i)\in W^i_{\bas}$ and $\mu(\widehat{m}^i)\in \widehat{W}^i_{\bas}$.
\item For all $i\in[k]$ and $j\in[p]$, $\mu(\widetilde{m}^i_j)\in\{\widetilde{w}^i_j,\overline{w}^i_j\}$.
\item For all $i\in[\alpha']$, $\mu(m^i_{\hap})=w^i_{\hap}$.
\item $\mu(m^\star)=w^\star$.
\end{enumerate}
\end{lemma}

Now, we reuse Definition \ref{def:w1hardSESMABC} in the context of {\sc BSM}, and turn to prove the appropriate adaptation of Lemma \ref{lem:w1SESMReverseSpecificMeasure} (specifically, we need to update the coefficient of $|E(G)|^{40}$ in the equality involving $\sat_W(\mu)$). 

\begin{lemma}\label{lem:w1BSMReverseSpecificMeasure}
Let $I$ be an instance of {\sc Multicolored Clique}, and let $\mu$ be a stable matching of $\red_{BSM}(I)$. Then, for some $-100|E(G)|^2\leq x\leq 100|E(G)|^2$, it holds that
\[\begin{array}{ll}
\sat_M(\mu)-\alpha' = & \displaystyle{\sum_{i=1}^k(p+a(\mu,i)-\widehat{a}(\mu,i)-1)|E(G)|^{20} + (c(\mu)-|E(G)|+2{k\choose 2})|E(G)|^{10}}\\
&+ \displaystyle{\left((\sum_{i=1}^kb(\mu,i)2^{i-1})-2^k+1\right)|E(G)|^{30} + \frac{1}{3}(p-1)(4^k-1)|E(G)|^{40} + x}.
\end{array}\]

Moreover, for some $-10|E(G)|^2\leq y\leq 10|E(G)|^2$, it holds that
\[\begin{array}{l}
\sat_W(\mu)-\alpha'=\\
\displaystyle{\sum_{i=1}^k(p+\widehat{a}(\mu,i)-a(\mu,i)-1)|E(G)|^{20} + (|E(G)|-c(\mu))|E(G)|^{10} + \sum_{i=1}^k(p-b(\mu,i))4^{i-1}|E(G)|^{40} + y}.
\end{array}\]
\end{lemma}

\begin{proof}
On the one hand, note that except for $m^\star$, the preference lists of the men of $\red_{BSM}(I)$ are the same as their preference lists in $\red_{SESM}(I)$. Hence, the first part of the lemma (that is, the equality concerning $\sat_M(\mu)$) is proven exactly as in the proof of Lemma \ref{lem:w1SESMReverseSpecificMeasure}, where due to the man $m^\star$, the term $(p-1)(2^k-1)|E(G)|^{40}$ is replaced by the term $\frac{1}{3}(p-1)(4^k-1)|E(G)|^{40}$.

On the other hand, by Lemmata \ref{lem:w1BSMAll} and \ref{lem:w1BSMReverseHap} and the definition of preference lists of the agents of $\red_{BSM}(I)$, we have that
\[\begin{array}{lll}
\sat_W(\mu) & = & \displaystyle{\sum_{i=1}^k\pos_{\mu(m^i)}(m^i) + \sum_{i=1}^k\pos_{\mu(\widehat{m}^i)}(\widehat{m}^i) + \sum_{i=1}^k\sum_{j=1}^p\pos_{\mu(m^i_j)}(m^i_j) + \sum_{i=1}^k\sum_{j=1}^p\pos_{\mu(\widehat{m}^i_j)}(\widehat{m}^i_j)}\\

&& + \displaystyle{\sum_{i=1}^k\sum_{j=1}^p\pos_{\mu(\widetilde{m}^i_j)}(\widetilde{m}^i_j) + \sum_{i=1}^k\sum_{j=1}^p\pos_{\mu(\overline{m}^i_j)}(\overline{m}^i_j)}\\

&& + \displaystyle{\sum_{i=1}^{k-1}\sum_{j=i+1}^k\sum_{t=1}^{q^{i,j}}\pos_{\mu(m^{i,j}_t)}(m^{i,j}_t) + \sum_{i=1}^{k-1}\sum_{j=i+1}^k\sum_{t=1}^{q^{i,j}}\pos_{\mu(\overline{m}^{i,j}_t)}(\overline{m}^{i,j}_t)} + \alpha' + 1\\

& = & \displaystyle{4k + \sum_{i=1}^k(a(\mu,i)-1) + \sum_{i=1}^k(p-a(\mu,i))(4 + |E(G)|^{20}))}\\

&& + \displaystyle{\sum_{i=1}^k(p-\widehat{a}(\mu,i)) + \sum_{i=1}^k(\widehat{a}(\mu,i)-1)(4 + |E(G)|^{20}))}\\

&& + \displaystyle{\sum_{i=1}^k\left(2b(\mu,i) + (p-b(\mu,i))(4+4^{i-1}|E(G)|^{40})\right)}\\

&& + 2c(\mu) + (|E(G)|-c(\mu))(8+|E(G)|^{10}) + \alpha' + 1.
\end{array}\]

Thus, we have that
\[\begin{array}{ll}
& \displaystyle{(p+\widehat{a}(\mu,i)-a(\mu,i)-1)|E(G)|^{20} + (|E(G)|-c(\mu))|E(G)|^{10} + \sum_{i=1}^k(p-b(\mu,i))4^{i-1}|E(G)|^{40}}\\
& -10|E(G)|^2\\
\leq & \sat_W(\mu)-\alpha' \\
\leq & \displaystyle{(p+\widehat{a}(\mu,i)-a(\mu,i)-1)|E(G)|^{20} + (|E(G)|-c(\mu))|E(G)|^{10} + \sum_{i=1}^k(p-b(\mu,i))4^{i-1}|E(G)|^{40}}\\
& + 10|E(G)|^2.
\end{array}\]

This concludes the proof of the lemma.
\end{proof}

Since the coefficient of $|E(G)|^{40}$ above has changed, we need to explicitly prove the following version of Lemma \ref{lem:w1SESMReverseTilde}. We remark that if this coefficient were to remain the same, the lemma would not have been correct, and therefore we had to perform the presented modification of the preference lists of $m^\star$ and women of the form $\widetilde{w}^i_j$.

\begin{lemma}\label{lem:w1BSMReverseTilde}
Let $I$ be an instance of {\sc Multicolored Clique}. Let $\mu$ be a stable matching of $\red_{BSM}(I)$ such that $\bal(\mu)\leq\eta$. Then, for all $i\in[k]$, there exists $j\in[p]$ such that $\mu(\widetilde{m}^i_j)=\overline{w}^i_j$ and for all $t\neq j$, $\mu(\widetilde{m}^i_j)=\widetilde{w}^i_j$.
\end{lemma}

\begin{proof}
The statement of the lemma is equivalent to the statement that for all $i\in[k]$, $b(\mu,i)=1$. Since $\Bal(\mu)=0$, it holds that $\sat_M(\mu)\leq\eta$ and $\sat_W(\mu)\leq\eta$. Hence, by Lemma \ref{lem:w1BSMReverseSpecificMeasure}, the two following inequalities are satisfied.
\begin{itemize}
\item $\displaystyle{\left((\sum_{i=1}^kb(\mu,i)2^{i-1})-2^k+1\right)|E(G)|^{30} \leq 0}$.
\item $\displaystyle{\sum_{i=1}^k(p-b(\mu,i))4^{i-1}|E(G)|^{40} \leq \frac{1}{3}(p-1)(4^k-1)|E(G)|^{40}}$.
\end{itemize}

Simplifying the inequalities above, we derive that the two following inequalities are satisfied.
\begin{enumerate}
\item\label{item:eq1} $\displaystyle{\sum_{i=1}^kb(\mu,i)2^{i-1} \leq 2^k-1}$.
\item\label{item:eq2} $\displaystyle{\frac{1}{3}(4^k-1) \leq \sum_{i=1}^kb(\mu,i)4^{i-1}}$.
\end{enumerate}

Note that for all $i\in[k]$, $b(\mu,i)\in[p]$. Let $\varphi$ be an assignment to the variables $b(\mu,i)$ that satisfies this condition as well as the two inequalities above. We claim that $\varphi$ necessarily assigns 1 to all of these variables. This claim can be proven by induction on $k$. In the base case, where $k=1$, the Equation \ref{item:eq1} directly implies that $b(\mu,1)\leq 1$, while Equation \ref{item:eq2} implies that $1=\frac{1}{3}(4^1-1)\leq b(\mu,1)$, and therefore $b(\mu,1)=1$.

We next suppose that $k\geq 2$ and that the claim holds for $k-1$. Then, first note that to satisfy Equality \ref{item:eq1}, it must hold that $b(\mu,k)\leq 1$. Suppose, by way of contradiction, that $b(\mu,k)=0$. By Equation \ref{item:eq1}, we have that $b(\mu,k-1)\leq 3$. Accordingly, we consider the following cases.
\begin{itemize}
\item Assume that $b(\mu,k-1)=3$. By Equation \ref{item:eq1}, $\displaystyle{\sum_{i=1}^{k-2}b(\mu,i)2^{i-1} \leq 2^k-1-3\cdot 2^{k-2} = \frac{1}{4}2^k-1}$. Hence, we have that
\[\begin{array}{ll}
\displaystyle{\sum_{i=1}^kb(\mu,i)4^{i-1}} &= \displaystyle{\frac{3}{16}2^{2k} +\sum_{i=1}^{k-2}b(\mu,i)2^{i-1}2^i}\\
& \leq \displaystyle{\frac{3}{16}2^{2k} + \left(\frac{1}{4}2^k-1\right)2^{k-2}}\\
& = \displaystyle{\frac{3}{16}2^{2k} + \frac{1}{16}2^{2k} - 2^{k-2}} \leq \displaystyle{\frac{1}{4}4^k-1}.
\end{array}\]
However, $\displaystyle{\frac{1}{4}4^k-1} < \frac{1}{3}(4^k-1)$, which contradicts the satisfaction of Equality \ref{item:eq2}.
\item Assume that $b(\mu,k-1)=2$. By Equation \ref{item:eq1}, $\displaystyle{\sum_{i=1}^{k-2}b(\mu,i)2^{i-1} \leq 2^k-1-2\cdot 2^{k-2} = \frac{1}{2}2^k-1}$. Hence, we have that
\[\begin{array}{ll}
\displaystyle{\sum_{i=1}^kb(\mu,i)4^{i-1}} &= \displaystyle{\frac{1}{8}2^{2k} +\sum_{i=1}^{k-2}b(\mu,i)2^{i-1}2^i}\\
& \leq \displaystyle{\frac{1}{8}2^{2k} + \left(\frac{1}{2}2^k-1\right)2^{k-2}}\\
& = \displaystyle{\frac{1}{8}2^{2k} + \frac{1}{8}2^{2k} - 2^{k-2}} \leq \displaystyle{\frac{1}{4}4^k-1}.
\end{array}\]
Again, this contradicts the satisfaction of Equality \ref{item:eq2}.
\item Assume that $b(\mu,k-1)=1$. By Equation \ref{item:eq1}, $\displaystyle{\sum_{i=1}^{k-2}b(\mu,i)2^{i-1} \leq 2^k-1-2^{k-2} = \frac{3}{4}2^k-1}$. Hence, we have that
\[\begin{array}{ll}
\displaystyle{\sum_{i=1}^kb(\mu,i)4^{i-1}} &= \displaystyle{\frac{1}{16}2^{2k} +\sum_{i=1}^{k-2}b(\mu,i)2^{i-1}2^i}\\
& \leq \displaystyle{\frac{1}{16}2^{2k} + \left(\frac{3}{4}2^k-1\right)2^{k-2}}\\
& = \displaystyle{\frac{1}{16}2^{2k} + \frac{3}{16}2^{2k} - 2^{k-2}} \leq \displaystyle{\frac{1}{4}4^k-1}.
\end{array}\]
Again, this contradicts the satisfaction of Equality \ref{item:eq2}.
\item Assume that $b(\mu,k-1)=0$. Then, by Equality \ref{item:eq1}, we have that
\[\begin{array}{ll}
\displaystyle{\sum_{i=1}^kb(\mu,i)4^{i-1}} &= \displaystyle{\sum_{i=1}^{k-2}b(\mu,i)2^{i-1}2^i}\leq \displaystyle{(2^k-1)2^{k-2}}\leq \displaystyle{\frac{1}{4}4^k-1}.
\end{array}\]
Again, this contradicts the satisfaction of Equality \ref{item:eq2}.
\end{itemize}

Thus, we derive that $b(\mu,k)=1$, which implies that the two following inequalities are then satisfied.
\begin{itemize}
\item $\displaystyle{\sum_{i=1}^{k-1}b(\mu,i)2^{i-1} + 2^{k-1} \leq 2^k-1}$. That is, $\displaystyle{\sum_{i=1}^{k-1}b(\mu,i)2^{i-1} \leq 2^{k-1}-1}$
\item $\displaystyle{\frac{1}{3}(4^k-1) \leq \sum_{i=1}^{k-1}b(\mu,i)4^{i-1} + 4^{k-1}}$. That is, $\displaystyle{\frac{1}{3}(4^{k-1}-1) \leq \sum_{i=1}^{k-1}b(\mu,i)4^{i-1}}$.
\end{itemize}

By the inductive hypothesis, we derive that for all $i\in[k-1]$, it also holds that $b(\mu,i)=1$. This concludes the proof of the lemma.
\end{proof}

Having proved Lemma \ref{lem:w1BSMReverseTilde}, we now turn the prove the following version of Lemma \ref{lem:w1SESMReverseColClass}.

\begin{lemma}\label{lem:w1BSMReverseColClass}
Let $I$ be an instance of {\sc Multicolored Clique}. Let $\mu$ be a stable matching of $\red_{BSM}(I)$ such that $\bal(\mu)\leq\eta$. Then, for all $i\in[k]$, there exists $j\in[p]$ such that $\mu(m^i)=w^i_j$ and $\mu(\widehat{m}^i)=\widehat{m}^i_j$.
\end{lemma}

\begin{proof}
Since $\bal(\mu)=0$, it holds that $\sat_M(\mu)\leq\eta$ and $\sat_W(\mu)\leq\eta$. By Lemma \ref{lem:w1BSMReverseTilde}, for all $i\in[k]$, $b(\mu,i)=1$.
Hence, by Lemma \ref{lem:w1SESMReverseSpecificMeasure}, to satisfy $\sat_M(\mu)\leq\eta$, the following inequality is satisfied.
\[\displaystyle{\sum_{i=1}^k(p+a(\mu,i)-\widehat{a}(\mu,i)-1)|E(G)|^{20} \leq (p-1)k|E(G)|^{20}}.\]
Moreover, by Lemma \ref{lem:w1SESMReverseSpecificMeasure}, to satisfy $\sat_W(\mu)\leq\eta$, the following inequality is satisfied.
\[\displaystyle{\sum_{i=1}^k(p+\widehat{a}(\mu,i)-a(\mu,i)-1)|E(G)|^{20} \leq (p-1)k|E(G)|^{20}}.\]
That is, the two following inequalities are satisfied.
\begin{itemize}
\item $\displaystyle{\sum_{i=1}^k(a(\mu,i)-\widehat{a}(\mu,i))\leq 0}$.
\item $\displaystyle{\sum_{i=1}^k(\widehat{a}(\mu,i)-a(\mu,i))\leq 0}$.
\end{itemize}

These two inequalities imply that $\displaystyle{\sum_{i=1}^ka(\mu,i)=\sum_{i=1}^k\widehat{a}(\mu,i)}$. For all $i\in[k]$, denote $a'(\mu,i)=(p+1)-\widehat{a}(\mu,i)$. Then,
\[\displaystyle{\sum_{i=1}^k(a(\mu,i)+a'(\mu,i)) = (p+1)k}.\]

Having established the last equality above, the proof proceeds exactly as the proof of Lemma \ref{lem:w1SESMReverseColClass}.
\end{proof}

We are now ready to prove the correctness of the reverse direction.

\begin{lemma}\label{lem:w1BSMReverse}
Let $I=(G,(V^1,V^2,\ldots,V^k))$ be an instance of {\sc Multicolored Clique}. If for the instance $\red_{BSM}(I)$ of {\sc BSM}, $\Bal=0$, then $I$ is a \yesinstance\ of {\sc Multicolored Clique}.
\end{lemma}

\begin{proof}
Suppose that for the instance $\red_{BSM}(I)$ of {\sc BSM}, $\Bal=0$. Then, there exists a stable matching $\mu$ such that $\bal(\mu)=0$. Hence, $\sat_M(\mu)\leq\eta$ and $\sat_W(\mu)\leq\eta$. By Lemmata \ref{lem:w1BSMReverseSpecificMeasure}, \ref{lem:w1BSMReverseTilde} and \ref{lem:w1BSMReverseColClass}, to satisfy $\sat_M(\mu)\leq\eta$, the following inequality is satisfied.
\[(c(\mu)-|E(G)|+2{k\choose 2})\leq {k\choose 2}.\]
Moreover, by Lemmata \ref{lem:w1BSMReverseSpecificMeasure}, \ref{lem:w1BSMReverseTilde} and \ref{lem:w1BSMReverseColClass}, to satisfy $\sat_W(\mu)\leq\eta$, the following inequality is satisfied.
\[(|E(G)|-c(\mu))\leq {k\choose 2}.\]
The first inequality is equivalent to $\displaystyle{c(\mu)\leq |E(G)|-{k\choose 2}}$, while the second inequality is equivalent to $\displaystyle{c(\mu)\geq |E(G)|-{k\choose 2}}$. Thus, we have that $\displaystyle{c(\mu)=|E(G)|-{k\choose 2}}$. Having established this last equality, the proof proceeds exactly as the proof of Lemma \ref{lem:w1SESMReverse}.
\end{proof}

\medskip
\myparagraph{Summary.} Finally, we note that the reduction can be performed in time that is polynomial in the size of the output. That is, we have the following observation.

\begin{observation}\label{obs:w1BSMTime}
Let $I=(G,(V^1,V^2,\ldots,V^k))$ be an instance of {\sc Multicolored Clique}. Then, the instance $\red_{BSM}(I)$ of {\sc BSM} can be constructed in time $2^{\OO(k)}\cdot n^{\OO(1)}$. Here, $n=|V(G)|$.
\end{observation}

By Proposition \ref{prop:multiClique}, Lemma \ref{lem:twBSM}, Corollary \ref{cor:w1BSMforward}, Lemma \ref{lem:w1BSMReverse} and Observation \ref{obs:w1BSMTime}, we conclude that {\sc BSM} is \WOH. Moreover, unless \ETH\ fails, {\sc BSM} cannot be solved in time $f(\tw)\cdot n^{o(\tw)}$ for any function $f$ that depends only on $\tw$. Here, $n$ is the number of agents.

%% file: w1hardnessMaxSMT.tex
\subsection{max-Stable Marriage with Ties}\label{sec:w1maxSMT}

Next, we prove that {\sc max-SMT} is \WOH, and that unless \ETH\ fails, {\sc max-SMT} cannot be solved in time $f(\tw)\cdot n^{o(\tw)}$ for any function $f$ that depends only on $\tw$.

\subsubsection{Reduction}

Let $I=(G,(V^1,V^2,\ldots,V^k))$ be an instance of {\sc Multicolored Clique}. We now describe how to construct an instance $\red_{max}(I)=(M,W,\{\pos_m\}|_{m\in M},\{\pos_w\}|_{w\in W})$ of {\sc max-SMT}. First, we again use the set of basic agents defined in Section \ref{sec:w1SESM}. Recall that for this set of basic agents, the preference lists of the men in $M_{\bas}\cup \widehat{M}_{\bas}$ are of the form of a leader. In the current reduction, this form completes the precise definition of the preference lists of the men in $M_{\bas}\cup \widehat{M}_{\bas}$, since for all $i\in[k]$, we now set the intersection of $\domain(\pos_{m^i})$ with $M$, the set of {\em all} men, to be exactly $W^i_{\bas}\cup W^{i,j}_{\bas}$, and we also set the intersection of $\domain(\pos_{\widehat{m}^i})$ with $M$ to be exactly $\widehat{W}^i_{\bas}\cup W^{i,j}_{\bas}$.

In our current reduction, we would have one set of gadgets called Combined Vertex Selector gadgets, rather than the three sets of Original Vertex Selector gadgets, Mirror Vertex Selector gadgets and Consistency gadgets that were introduced in Section \ref{sec:w1SESM}. Namely, we would be able to integrate all of the properties guaranteed by the later three sets of gadgets using only one set of gadgets in a manner that does not compromise the readability of the reduction. Afterwards, we would again introduce a set of Edge Selector gadgets. Here, however, the definition of happy pairs is not required, as the source of difficulty of {\sc max-SMT} lies in the existence of ties in preference lists and not in the necessity to achieve a certain budget such as $\Delta=0$. 

\medskip
\myparagraph{Combined Vertex Selector.} For every color class $i\in[k]$, our first set of gadgets introduces the following sets of new men: $M^i_{\enr}=\{m^i_2,m^i_3,\ldots,m^i_p\}$ and $\widehat{M}=\{\widehat{m}^i_1,\widehat{m}^i_2,\ldots,\widehat{m}^i_p\}$. Note that or all $i\in[k]$, $m^i_1$ is not defined, but $\widehat{m}^i_1$ is defined. We also introduce one new woman, called $w^i$. For all $j\in\{2,3,\ldots,p\}$, the preference list of $m^i_j$ is defined as follows. We set $\domain(\pos_{m^i_j})=\{w^i_{j-1},w^i_j\}$, and $\pos_{m^i_j}(w^i_{j-1})=\pos_{m^i_j}(w^i_j)=1$. 
Moreover, for all $j\in[p]$, the preference list of $\widehat{m}^i_j$ is defined as follows. We set $\domain(\pos_{\widehat{m}^i_j})=\{w^i,w^i_j,\widehat{w}^i_j\}$, $\pos_{\widehat{m}^i_j}(w^i)=1$, $\pos_{\widehat{m}^i_j}(w^i_j)=2$ and $\pos_{\widehat{m}^i_j}(\widehat{w}^i_j)=3$. This completes the definition of the preference lists of {\em all} men introduced so far.

Let us now also define the preference lists of the women participating in the Combined Vertex Selector gadget that handles color class $i\in[k]$. First, the preference list of $w^i_1$ is defined as follows. We set $\domain(\pos_{w^i_1})=\{m^i_2,\widehat{m}^i_1,m^i\}$, $\pos_{w^i_1}(m^i_2)=1$, $\pos_{w^i_1}(\widehat{m}^i_1)=2$ and $\pos_{w^i_1}(m^i)=3$. For all $j\in\{2,3,\ldots,p-1\}$, the preference list of $w^i_j$ is defined as follows. We set $\domain(\pos_{w^i_j})=\{m^i_j,m^i_{j+1},\widehat{m}^i_j,m^i\}$, $\pos_{w^i_j}(m^i_j)=\pos_{w^i_j}(m^i_{j+1})=1$, $\pos_{w^i_j}(\widehat{m}^i_j)=2$ and $\pos_{w^i_j}(m^i)=3$. The preference list of $w^i_p$ is defined as follows. We set $\domain(\pos_{w^i_p})=\{m^i_p,\widehat{m}^i_1\}$, $\pos_{w^i_p}(m^i_p)=1$, $\pos_{w^i_p}(\widehat{m}^i_p)=2$ and $\pos_{w^i_p}(m^i)=3$. Second, for all $j\in[p]$, the preference list of $\widehat{w}^i_j$ is defined by setting $\domain(\pos_{\widehat{w}^i_j})=\{\widehat{m}^i_j, \widehat{m}^i\}$, $\pos_{\widehat{w}^i_j}(\widehat{m}^i_j)=1$ and $\pos_{\widehat{w}^i_j}(\widehat{m}^i)=2$. Finally, the preference list of $w^i$ is defined as follows. We set $\domain(\pos_{w^i})=M^i_{\enr}$, where for all $j\in\{2,3\ldots,p\}$, $\pos_{w^i}(m^i_j)=1$.

An illustration of a Combined Vertex Selector gadget is given in Fig.~\ref{fig:w1MaxSMT1}.

\begin{figure}[t!]\centering
\fbox{\includegraphics[scale=0.9]{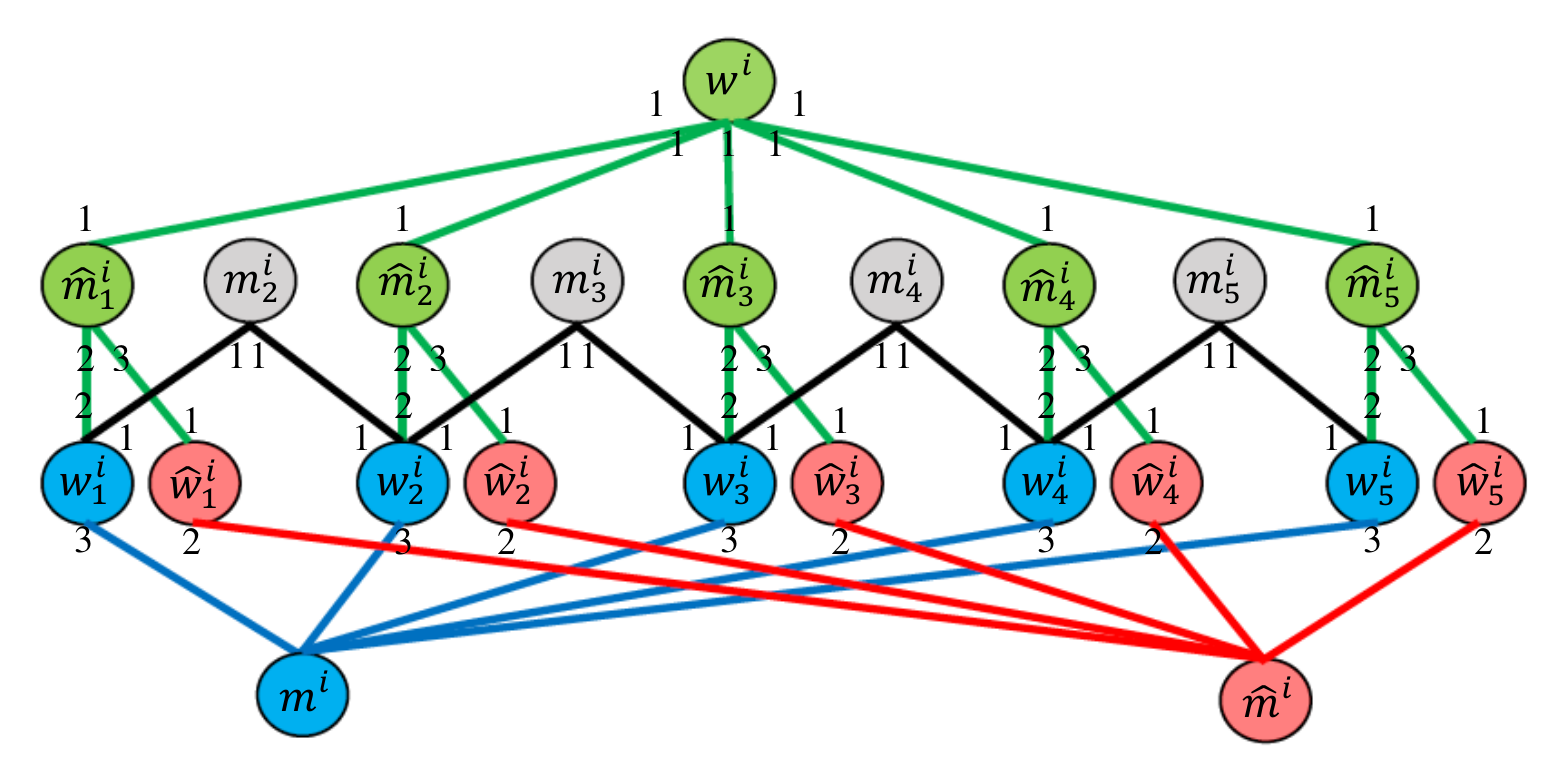}}
\caption{The Combined Vertex Selector gadget where $p=5$. Numbers indicate the positions of agents in preference lists.}\label{fig:w1MaxSMT1}
\end{figure}

\medskip
\myparagraph{Edge Selector.} For every two color classes $i,j\in[k]$ where $i<j$, we introduce the following three sets of new men: $M^{i,j}_{\enr}=\{m^{i,j}_1,m^{i,j}_2,\ldots,m^{i,j}_{q^{i,j}}\}$, $\overline{M}^{i,j}=\{\overline{m}^{i,j}_1,\overline{m}^{i,j}_2,\ldots,\overline{m}^{i,j}_{q^{i,j}}\}$ and $\widetilde{M}^{i,j}=\{\widetilde{m}^{i,j}_2,\widetilde{m}^{i,j}_3,\ldots,\widetilde{m}^{i,j}_{q^{i,j}}\}$. We also introduce two sets of new women: $\overline{W}^{i,j}=\{\overline{w}^{i,j}_1,\overline{w}^{i,j}_2,\ldots,\overline{w}^{i,j}_{q^{i,j}}\}$ and $\widetilde{M}^{i,j}=\{\widetilde{w}^{i,j}_2,\widetilde{w}^{i,j}_3,\ldots,\widetilde{w}^{i,j}_{q^{i,j}}\}$.

Let us now define the preference lists of the men participating in the Edge Selector gadget that handles color classes $i,j\in[k]$ where $i<j$. For all $t\in[q^{i,j}]$, the preference list of $m^{i,j}_t$ is defined by setting $\domain(\pos_{m^{i,j}_t})=\{w^{i,j}_t,\overline{w}^{i,j}_t\}$ and $\pos_{m^{i,j}_t}(w^{i,j}_t)=\pos_{\overline{m}^{i,j}_t}(w^{i,j}_t)=1$.
The preference list of $\overline{m}^{i,j}_1$ is set by defining $\domain(\pos_{\overline{m}^{i,j}_1})=\{\widetilde{w}^{i,j}_2,w^{i,j}_1\}$, $\pos_{\overline{m}^{i,j}_1}(\widetilde{w}^{i,j}_2)=1$ and $\pos_{\overline{m}^{i,j}_1}(w^{i,j}_1)=2$.
For all $t\in\{2,3,\ldots,q^{i,j}\}$, the preference list of $\overline{m}^{i,j}_t$ is set by defining $\domain(\pos_{\overline{m}^{i,j}_t})=\{\widetilde{w}^{i,j}_t,\widetilde{w}^{i,j}_{t+1},w^{i,j}_t\}$, $\pos_{\overline{m}^{i,j}_t}(\widetilde{w}^{i,j}_t)=\pos_{\overline{m}^{i,j}_t}(\widetilde{w}^{i,j}_{t+1})=1$ and $\pos_{\overline{m}^{i,j}_t}(w^{i,j}_t)=2$.
The preference list of $\overline{m}^{i,j}_p$ is set by defining $\domain(\pos_{\overline{m}^{i,j}_p})=\{\widetilde{w}^{i,j}_p,w^{i,j}_p\}$, $\pos_{\overline{m}^{i,j}_p}(\widetilde{w}^{i,j}_p)=1$ and $\pos_{\overline{m}^{i,j}_p}(w^{i,j}_p)=2$.
For all $t\in\{2,3,\ldots,q^{i,j}\}$, the preference list of $\widetilde{m}^{i,j}_t$ is defined by setting $\domain(\pos_{\widetilde{m}^{i,j}_t})=\{\overline{w}^{i,j}_{t-1},\overline{w}^{i,j}_t\}$ and $\pos_{\widetilde{m}^{i,j}_t}(\overline{w}^{i,j}_{t-1})=\pos_{\widetilde{m}^{i,j}_t}(\overline{w}^{i,j}_t)=1$.

We proceed by defining the preference lists of the women participating in the Edge Selector gadget that handles color classes $i,j\in[k]$ where $i<j$. For all $t\in[p]$, the preference list of $w^{i,j}_t$ is defined as follows. We set $\domain(\pos_{w^{i,j}_t})=\{m^{i,j}_t,m^i,\widehat{m}^i,m^j,\widehat{m}^j,\overline{m}^{i,j}_t\}$, $\pos_{w^{i,j}_t}(m^{i,j}_t)=1$, $\pos_{w^{i,j}_t}(m^i)=2$, $\pos_{w^{i,j}_t}(\widehat{m}^i)=2$, $\pos_{w^{i,j}_t}(m^j)=2$, $\pos_{w^{i,j}_t}(\widehat{m}^j)=2$ and $\pos_{w^{i,j}_t}(\overline{m}^{i,j}_t)=3$.
The preference list of $\overline{w}^{i,j}_1$ is set by defining $\domain(\pos_{\overline{w}^{i,j}_1})=\{\widetilde{m}^{i,j}_2,m^{i,j}_1\}$, $\pos_{\overline{w}^{i,j}_1}(\widetilde{m}^{i,j}_2)=1$ and $\pos_{\overline{w}^{i,j}_1}(m^{i,j}_1)=2$.
For all $t\in\{2,3,\ldots,q^{i,j}\}$, the preference list of $\overline{w}^{i,j}_t$ is set by defining $\domain(\pos_{\overline{w}^{i,j}_t})=\{\widetilde{m}^{i,j}_t,\widetilde{m}^{i,j}_{t+1},m^{i,j}_t\}$, $\pos_{\overline{w}^{i,j}_t}(\widetilde{m}^{i,j}_t)=\pos_{\overline{w}^{i,j}_t}(\widetilde{m}^{i,j}_{t+1})=1$ and $\pos_{\overline{w}^{i,j}_t}(m^{i,j}_t)=2$.
The preference list of $\overline{w}^{i,j}_p$ is set by defining $\domain(\pos_{\overline{w}^{i,j}_p})=\{\widetilde{m}^{i,j}_p,m^{i,j}_p\}$, $\pos_{\overline{w}^{i,j}_p}(\widetilde{m}^{i,j}_p)=1$ and $\pos_{\overline{w}^{i,j}_p}(m^{i,j}_p)=2$.
For all $t\in\{2,3,\ldots,q^{i,j}\}$, the preference list of $\widetilde{w}^{i,j}_t$ is defined by setting $\domain(\pos_{\widetilde{w}^{i,j}_t})=\{\overline{m}^{i,j}_{t-1},\overline{m}^{i,j}_t\}$ and $\pos_{\widetilde{w}^{i,j}_t}(\overline{m}^{i,j}_{t-1})=\pos_{\widetilde{w}^{i,j}_t}(\overline{m}^{i,j}_t)=1$.

An illustration of an Edge Selector gadget is given in Fig.~\ref{fig:w1MaxSMT2}.

\begin{figure}[t!]\centering
\fbox{\includegraphics[scale=0.9]{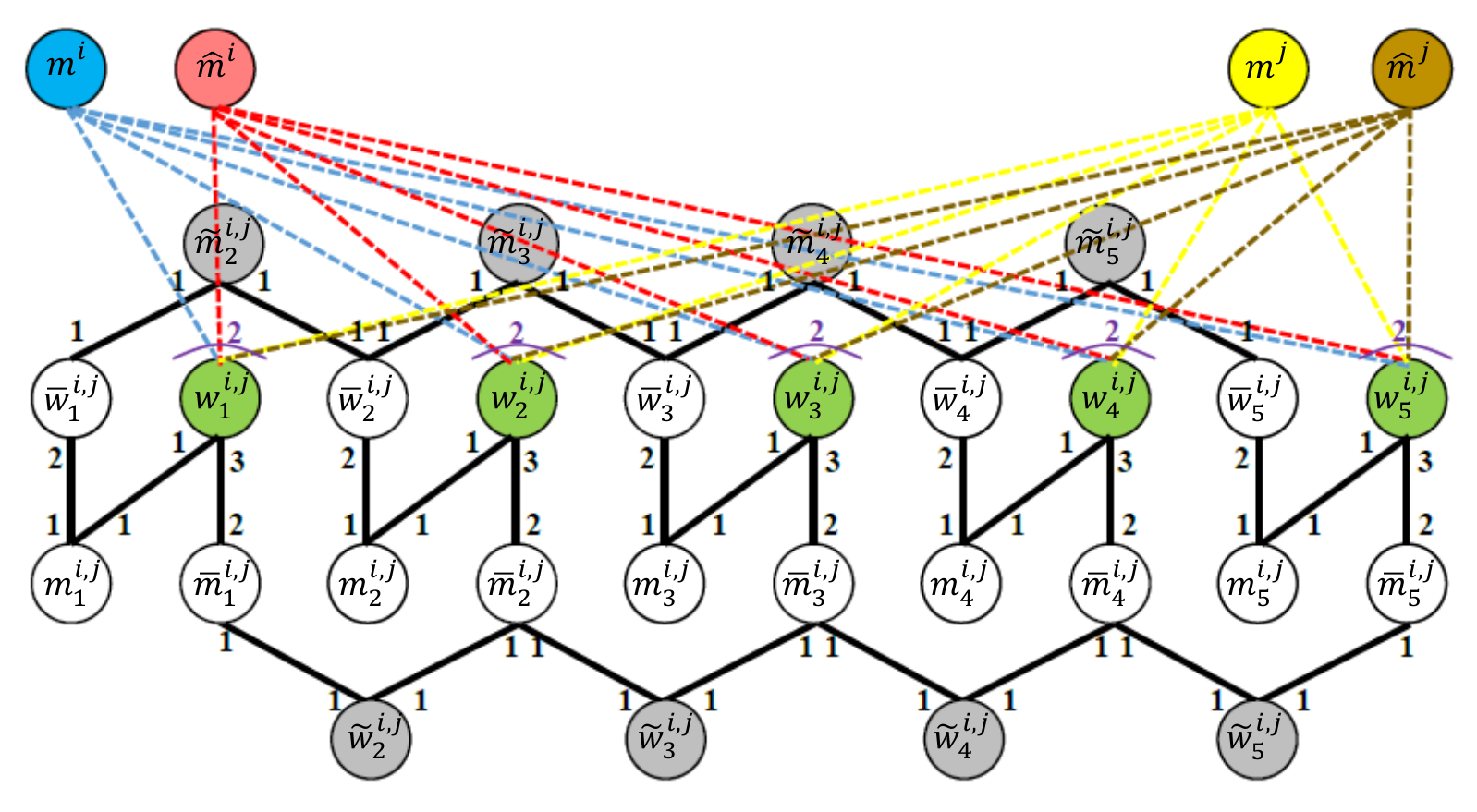}}
\caption{The Edge Selector gadget where $q^{i,j}=5$. Numbers indicate the positions of agents in preference lists.}\label{fig:w1MaxSMT2}
\end{figure}

\subsubsection{Treewidth}

We begin the analysis of the reduction by bounding the treewidth of the resulting primal graph.

\begin{lemma}\label{lem:twMAX}
Let $I$ be an instance of {\sc Multicolored Clique}. Then, the treewidth of $\red_{max}(I)$ is bounded by $2k+\OO(1)$.
\end{lemma}

\begin{proof}
Let $P$ be the primal graph of $\red_{max}(I)$, and let $P'$ denote the graph obtained from $P$ by the removal of all of (the vertices that represent) men in $M_{\bas}\cup\widehat{M}_{\bas}$. Note that $|M_{\bas}\cup\widehat{M}_{\bas}|=2k$. Hence, as in the proof of Lemma \ref{lem:twSESM}, to prove that the treewidth of $P$ is bounded by $2k+\OO(1)$, it is sufficient to prove that the treewidth of every connected component of $P'$ is bounded by $\OO(1)$. First, for all $i\in[k]$, we have that $M^i_{\enr}\cup\widehat{M}^i_{\enr}\cup W^i_{\bas}\cup \widehat{W}^i_{\bas}\cup \{w^i\}$ is the entire vertex set of a connected component of $P'$. Since after we remove $w^i$ from this connected component we obtain a tree (whose treewidth is 1), we deduce that the treewidth of this connected component is at most 2. 

Next, note that for all $i,j\in[k]$ where $i<j$ and $t\in[q^{i,j}]$, it holds that $X^{i,j}=M^{i,j}_{\enr}\cup\overline{M}^{i,j}\cup\widetilde{M}^{i,j}\cup W^{i,j}_{\bas}\cup\overline{W}^{i,j}\cup\widetilde{W}^{i,j}$ is the entire vertex set of a connected component of $P'$. For this connected component, which we denote by $C^{i,j}$, we explicitly define a tree decomposition $(T^{i,j},\beta^{i,j})$ as follows. Denote $q=q^{i,j}$. The tree $T^{i,j}$ is simply a path on $q$ vertices, denoted by $T=u_1-u_2-\cdots-u_q$. We define $\beta^{i,j}(u_1)=\{w^{i,j}_1,m^{i,j}_1,\overline{w}^{i,j}_1,\overline{m}^{i,j}_1,\widetilde{m}^{i,j}_2,\widetilde{w}^{i,j}_2\}$ and $\beta^{i,j}(u_q)=\{w^{i,j}_q,m^{i,j}_q,\overline{w}^{i,j}_q,\overline{m}^{i,j}_q,\widetilde{m}^{i,j}_q,\widetilde{w}^{i,j}_q\}$. For all $t\in\{2,\ldots,q-1\}$, we define $\beta^{i,j}(u_t)=\{w^{i,j}_t,m^{i,j}_t,\overline{w}^{i,j}_t,\overline{m}^{i,j}_t,\widetilde{m}^{i,j}_t,\widetilde{w}^{i,j}_t,\widetilde{m}^{i,j}_{t+1},\widetilde{w}^{i,j}_{t+1}\}$. Note that the size of each bag of is upper bounded by $10=\OO(1)$. Moreover, each agent in $X^{i,j}$ belongs to the bags of at most two nodes of $T^{i,j}$, and these two nodes are adjacent. Lastly, for all $j\in[p]$, all endpoints of edges incident to either $\overline{w}^{i,j}_t$ or $\overline{m}^{i,j}_t$ belong to the bag $\beta^{i,j}(u_t)$, and every edge of $C^{i,j}$ has an endpoint in $\overline{W}^{i,j}\cup\overline{M}^{i,j}$. Thus, $(T^i,\beta^i)$ is indeed a tree decomposition of $C^{i,j}$ of width $\OO(1)$.

We have thus considered every connected component of $P'$, and hence we conclude the proof of the lemma.
\end{proof}

\subsubsection{Correctness}

\myparagraph{Forward Direction.} We first show how given a solution of an instance $I$ of {\sc Multicolored Clique}, we can construct a stable matching $\mu$ of $\red_{max}(I)$ which matches all agents. For this purpose, we introduce the following definition.

\begin{definition}\label{def:cliquetoMAX}
Let $I=(G,(V^1,V^2,\ldots,V^k))$ be a \yesinstance\ of {\sc Multicolored Clique}, and let $U=\{v^1_{\ell_1},v^2_{\ell_2},\ldots,v^t_{\ell_k}\}$ and $W=\{e^{i,j}_{\ell_{i,j}}: i,j\in[k], i<j\}$ denote the vertex and edge sets, respectively, of a multicolored clique $C$ of $G$. Then, the matching $\mu^C_{max}$ of $\red_{max}(I)$ is defined as follows.
\begin{itemize}
\item For all $i\in[k]$: $\mu^C_{max}(m^i)=w^i_{\ell_i}$ and $\mu^C_{max}(\widehat{m}^i)=\widehat{w}^i_{\ell_i}$.
\item For all $i\in[k]$ and $j\in\{2,3,\ldots,\ell_i\}$: $\mu^C_{max}(m^i_j)=w^i_{j-1}$.
\item For all $i\in[k]$ and $j\in\{\ell_i+1,\ell_i+2,\ldots,p\}$ : $\mu^C_{max}(m^i_j)=w^i_j$.
\item For all $i\in[k]$: $\mu^C_{max}(\widehat{m}^i_{\ell_i})=w^i$.
\item For all $i\in[k]$ and $j\in[p]$ such that $j\neq\ell_i$: $\mu^C_{max}(\widehat{m}^i_j)=\widehat{w}^i_j$.
\item For all $i,j\in[k]$ where $i<j$: $\mu^C_{max}(m^{i,j}_{\ell_{i,j}})=\overline{w}^{i,j}_{\ell_{i,j}}$ and $\mu^C_{max}(\overline{m}^{i,j}_{\ell_{i,j}})=w^{i,j}_{\ell_{i,j}}$.
\item For all $i,j\in[k]$ where $i<j$ and $t\in[q^{i,j}]$ such that $t\neq\ell_{i,j}$: $\mu^C_{max}(m^{i,j}_t)=w^{i,j}_t$.
\item For all $i,j\in[k]$ where $i<j$ and $t\in\{2,3,\ldots,\ell_{i,j}\}$: $\mu^C_{max}(\widetilde{m}^{i,j}_t)=\overline{w}^{i,j}_{t-1}$ and $\mu^C_{max}(\overline{m}^{i,j}_{t-1})=\widetilde{w}^{i,j}_t$.
\item For all $i,j\in[k]$ where $i<j$ and $t\in\{\ell_{i,j}+1,\ell_{i,j}+2,\ldots,q^{i,j}\}$: $\mu^C_{max}(\widetilde{m}^{i,j}_t)=\overline{w}^{i,j}_t$ and $\mu^C_{max}(\overline{m}^{i,j}_t)=\widetilde{w}^{i,j}_t$.
\end{itemize}
\end{definition}

By Definition \ref{def:cliquetoMAX}, we directly derive the following observation.

\begin{observation}\label{obs:cliquetoMAX}
Let $I=(G,(V^1,V^2,\ldots,V^k))$ be a \yesinstance\ of {\sc Multicolored Clique}. Let $C$ be a multicolored clique of $G$. Then, $\mu^C_{max}$ matches all agents of $\red_{max}(I)$.
\end{observation}

\begin{lemma}\label{lem:w1MAXforward}
Let $I=(G,(V^1,V^2,\ldots,V^k))$ be a \yesinstance\ of {\sc Multicolored Clique}. Let $C$ be a multicolored clique of $G$. Then, $\mu^C_{MAX}$ is a stable matching of $\red_{max}(I)$.
\end{lemma}

\begin{proof}
First, for all $i\in[k]$, we claim that neither $m^i$ not $\widehat{m}^i$ can belong to a blocking pair. All the women in $W^i_{\bas}\setminus\{w^i_{\ell_i}\}$ are matched to men that they rank at position 1 in their preference lists, and therefore none of them can form a blocking pair with $m^i$. Similarly, all the women in $\widehat{W}^i_{\bas}\setminus\{\widehat{w}^i_{\ell_i}\}$ are matched to men that they rank at position 1 in their preference lists, and therefore none of them can form a blocking pair with $\widehat{m}^i$.  All of the other women in the preference lists of both $m^i$ and $\widehat{m}^i$ belong to $\bigcup_{j\in[k],j\neq i}W^{i,j}_{\bas}$. For all $j\in[k],j\neq i$, due to the fact that the preference lists of $m^i$ and $\widehat{m}^i$ are of the form of a leader and $e^{i,j}_{\ell_{i,j}}$ is incident to $v^i_{\ell_i}$ in $G$, we have that $m^i$ prefers $w^i_{\ell_i}$ over $w^{i,j}_{\ell_i}$ and that $\widehat{m}^i$ prefers $\widehat{w}^i_{\ell_i}$ over $w^{i,j}_{\ell_i}$. Moreover, for all $j\in[k],j\neq i$ and $t\in[q^{i,j}]$ such that $t\neq\ell_{i,j}$, we have that $w^{i,j}_t$ is matched to the man she prefers the most, and therefore she can form a blocking pair with neither $m^i$ nor $\widehat{m}^i$.

Second, notice that for all $i\in[k]$ and $j\in[p]$, $m^i_j$ is matched to a woman at position 1 in his preference list, and therefore he cannot belong to a blocking pair. Moreover, notice that for all $i\in[k]$, $w^i$ is matched to a man at position 1 in her preference list, and therefore she cannot belong to a blocking pair. Thus, since for all $i\in[k]$ and $j\in[p]$, the only woman that $\widehat{m}^i_j$ prefers over his matched partner is $w^i$, we have that $\widehat{m}^i_j$ cannot belong to a blocking pair.

Third, for all $i,j\in[k]$ where $i<j$ and $t\in[q^{i,j}]$, $m^{i,j}$ is matched to a woman at position 1 in his preference list, and therefore he cannot belong to a blocking pair. Moreover, for all $i,j\in[k]$ where $i<j$ and $t\in\{2,3,\ldots,q^{i,j}\}$, $\widetilde{w}^{i,j}_t$ is matched to a man at position 1 in her preference list, and therefore she cannot belong to a blocking pair. Hence, for $i,j\in[k]$ where $i<j$ and $t\in[q^{i,j}]$, all of the women that $\overline{m}^{i,j}_t$ prefers over his matched partner (if there are any such women) cannot belong to blocking pairs, and thus $\overline{m}^{i,j}_t$ cannot belong to a blocking pair. Finally, for all $i,j\in[k]$ where $i<j$ and $t\in\{2,3,\ldots,q^{i,j}\}$, $\widetilde{m}^{i,j}_t$ is matched to a woman at position 1 in his preference list, and therefore he cannot belong to a blocking pair.
\end{proof}

Combining Observation \ref{obs:cliquetoMAX} and Lemma \ref{lem:w1MAXforward}, we derive the following corollary.

\begin{corollary}\label{cor:w1MAXforward}
Let $I$ be a \yesinstance\ of {\sc Multicolored Clique}. Then, the instance $\red_{max}(I)$ of {\sc max-SMT} admits a stable matching that matches all agents.
\end{corollary}

This concludes the proof of the forward direction.

\medskip
\myparagraph{Reverse Direction.} Second, we prove that given an instance $I$ of {\sc Multicolored Clique}, if the instance $\red_{max}(I)$ of {\sc max-SMT} admits a stable matching that matches all agents, then we can construct a solution for $I$. To this end, we analyze the structure of stable matchings of $\red_{max}(I)$ that match all agents.

\begin{lemma}\label{lem:revw1MAXColor}
Let $I$ be an instance of {\sc Multicolored Clique}. Let $\mu$ be a stable matching of $\red_{max}(I)$ that matches all agents. Then, for all $i\in[k]$, there exists $j\in[p]$ such that $\mu(m^i)=w^i_j$ and $\mu(\widehat{m}^i)=\widehat{w}^i_j$.
\end{lemma}

\begin{proof}
Let $i\in[k]$ be some color class. First, notice that excluding $m^i$ and $\widehat{m}^i$, the Combined Vertex Selector gadget for color class $i$ contains exactly $2p+1$ women and $2p-1$ men. Thus, since $\mu$ matches all agents, both $m^i$ and $\widehat{m}^i$ must be matched to women that belong to this gadget. Since among the women in this gadget, $m^i$ only ranks those women in $W^i_{\bas}$ and $\widehat{m}^i$ only ranks those women in $\widehat{W}^i_{\bas}$, we have that there exist $j,\widehat{j}\in[p]$ such that $\mu(m^i)=w^i_j$ and $\mu(\widehat{m}^i)=\widehat{w}^i_{\widehat{j}}$. Since $\mu$ is a stable matching although $w^i_j$ prefers $\widehat{m}^i_j$ over $m^i$, we have that $\mu$ matches $\widehat{m}^i_j$ to $w^i$ as this is the only woman that over which $\widehat{m}^i_j$ does not prefer $w^i_j$. However, note that excluding $\widehat{m}^i_j$, the only man in the preference list of $\widehat{w}^i_j$ is $\widehat{m}^i$. Since $\widehat{w}^i_j$ is matched by $\mu$, we deduce that $j=\widehat{j}$. As the choice $i$ was arbitrary, we conclude that the lemma is correct.
\end{proof}

\begin{lemma}\label{lem:revw1MAXEdge}
Let $I$ be an instance of {\sc Multicolored Clique}. Let $\mu$ be a stable matching of $\red_{max}(I)$ that matches all agents. Then, for all $i,j\in[k]$ where $j<i$, there exists $t\in[q^{i,j}]$ such that $\mu(w^{i,j}_t)=\overline{m}^{i,j}_t$.
\end{lemma}

\begin{proof}
Let $i,j\in[k]$, $i<j$, be some two color classes. Note that $|\widetilde{W}^{i,j}|=p-1$ while $|\overline{M}^{i,j}|=p$. Thus, there exists $t\in[q^{i,j}]$ such that $\overline{m}^{i,j}_t$ is not matched to a woman in $\widetilde{W}^{i,j}$. Since the only other woman in the preference list of $\overline{m}^{i,j}_t$ is $w^{i,j}_t$ and $\mu$ matches all agents, we have that $\mu(w^{i,j}_t)=\overline{m}^{i,j}_t$. As the choices of $i$ and $j$ were arbitrary, we conclude that the lemma is correct.
\end{proof}

\begin{lemma}\label{lem:w1MAXReverse}
Let $I=(G,(V^1,V^2,\ldots,V^k))$ be an instance of {\sc Multicolored Clique}. If the instance $\red_{max}(I)$ of {\sc max-SMT} admits a stable matching that matches all agents, then $I$ is a \yesinstance\ of {\sc Multicolored Clique}.
\end{lemma}

\begin{proof}
Suppose that the instance $\red_{max}(I)$ of {\sc max-SMT} admits a stable matching $\mu$ that matches all agents. By Lemma~\ref{lem:revw1MAXColor}, for all $i\in[k]$, there exists $\ell_i\in[p]$ such that $\mu(m^i)=w^i_{\ell_i}$ and $\mu(\widehat{m}^i)=\widehat{w}^i_{\ell_i}$. Moreover, by Lemma~\ref{lem:revw1MAXEdge}, for all $i,j\in[k]$ where $j<i$, there exists $\ell_{i,j}\in[q^{i,j}]$ such that $\mu(w^{i,j}_{\ell_{i,j}})=\overline{m}^{i,j}_{\ell_{i,j}}$. Denote $U=\{v^1_{\ell_1},v^2_{\ell_2},\ldots,v^t_{\ell_k}\}$ and $W=\{e^{i,j}_{\ell_{i,j}}: i,j\in[k], i<j\}$. Note that we have proved that $|U|=k$ and $|W|\geq {k\choose 2}$. Since for all $i,j\in[k]$ where $j<i$, $w^{i,j}_{\ell_i}$ prefers both $m^i$ and $\widehat{m}^i$ over her matched partner, we have that $w^{i,j}_{\ell_i}$ is not located after $w^i_{\ell_i}$ in the preference list of $m^i$ as well as that it is not located after $\widehat{w}^i_{\ell_i}$ in the preference list of $\widehat{m}^i$. However, by the definition of the preference lists of $m^i$ and $\widehat{m}^i$, it must then hold that $e^{i,j}_{\ell_{i,j}}$ is an edge incident to $v^i_{\ell_i}$ in $G$. Hence, we derive that $U$ is the vertex set of a colorful clique of $G$. We thus conclude that $I$ is a \yesinstance\ of {\sc Multicolored Clique}.
\end{proof}

\medskip
\myparagraph{Summary.} Finally, we note that the reduction can be performed in time that is polynomial in the size of the output. That is, we have the following observation.

\begin{observation}\label{obs:w1MAXTime}
Let $I=(G,(V^1,V^2,\ldots,V^k))$ be an instance of {\sc Multicolored Clique}. Then, the instance $\red_{max}(I)$ of {\sc max-SMT} can be constructed in time polynomial in the size~of~$I$.
\end{observation}

By Proposition \ref{prop:multiClique}, Lemma \ref{lem:twMAX}, Corollary \ref {cor:w1MAXforward}, Lemma \ref{lem:w1MAXReverse} and Observation \ref{obs:w1MAXTime}, we conclude that {\sc max-SMT} is \WOH. Moreover, unless \ETH\ fails, {\sc max-SMT} cannot be solved in time $f(\tw)\cdot n^{o(\tw)}$ for any function $f$ that depends only on $\tw$. Here, $n$ is the number of agents.

%% file: w1hardnessMinSMT.tex
\subsection{min-Stable Marriage with Ties}\label{sec:w1minSMT}

Finally, we prove that {\sc min-SMT} is \WOH, and that unless \ETH\ fails, {\sc min-SMT} cannot be solved in time $f(\tw)\cdot n^{o(\tw)}$ for any function $f$ that depends only on $\tw$.

\subsubsection{Reduction}

Let $I=(G,(V^1,V^2,\ldots,V^k))$ be an instance of {\sc Multicolored Clique}. We now describe how to construct an instance $\red_{min}(I)=(M,W,\{\pos_m\}|_{m\in M},\{\pos_w\}|_{w\in W})$ of {\sc min-SMT}. First, we again use the set of basic agents defined in Section \ref{sec:w1SESM}. Recall that for this set of basic agents, the preference lists of the men in $M_{\bas}\cup \widehat{M}_{\bas}$ are of the form of a leader. Here, we also add two sets of new women: $W^i_{\enr}=\{w^1,w^2,\ldots,w^k\}$ and $\widehat{W}^i_{\enr}=\{\widehat{w}^1,\widehat{w}^2,\ldots,\widehat{w}^k\}$. For all $i\in[k]$, we set the intersection of $\domain(\pos_{m^i})$ with $M$, the set of {\em all} men, to be exactly $W^i_{\bas}\cup W^{i,j}_{\bas}\cup\{w^i\}$, and we also set the intersection of $\domain(\pos_{\widehat{m}^i})$ with $M$ to be exactly $\widehat{W}^i_{\bas}\cup W^{i,j}_{\bas}\cup\{\widehat{w}^i\}$. Moreover, we define $w^i$ and $\widehat{w}^i$ to be the least preferred women by $m^i$ and $\widehat{m}^i$, respectively.
We thus complete the precise definition of the preference lists of the men in $M_{\bas}\cup \widehat{M}_{\bas}$. The preference list of $w^i$ is defined to contains only $m^i$, and the preference list of $\widehat{w}^i$ is defined to contains only $\widehat{m}^i$. In what follows, as in the case of {\sc max-SMT}, we introduce two types of gadgets, namely, the Combined Vertex Selector gadget and the Edge Selector gadget.

\medskip
\myparagraph{Combined Vertex Selector.} For every color class $i\in[k]$, our first set of gadgets introduces the following sets of new men: $M^i_{\enr}=\{m^i_1,m^i_2,\ldots,m^i_p\}$ and $\widehat{M}=\{\widehat{m}^i_1,\widehat{m}^i_2,\ldots,\widehat{m}^i_p\}$. We also introduce one set of new women: $\overline{W}^i=\{\overline{w}^i_1,\overline{w}^i_2,\ldots,\overline{w}^i_p\}$. For all $j\in[p]$, the preference list of $m^i_j$ is defined by setting $\domain(\pos_{m^i_j})=\{w^i_j,\overline{w}^i_j\}$, $\pos_{m^i_j}(w^i_j)=1$ and $\pos_{m^i_j}(\overline{w}^i_j)=2$. Moreover, for all $j\in[p]$, the preference list of $\widehat{m}^i_j$ is defined by setting $\domain(\pos_{\widehat{m}^i_j})=\{\widehat{w}^i_j,\overline{w}^i_j\}$, $\pos_{\widehat{m}^i_j}(\widehat{w}^i_j)=1$ and $\pos_{\widehat{m}^i_j}(\overline{w}^i_j)=2$.

For all $j\in[p]$, the preference list of $w^i_j$ is defined by setting $\domain(\pos_{w^i_j})=\{m^i,m^i_j\}$, and $\pos_{w^i_j}(m^i)=\pos_{w^i_j}(m^i_j)=1$. Moreover, for all $j\in[p]$, the preference list of $\widehat{w}^i_j$ is defined by setting $\domain(\pos_{\widehat{w}^i_j})=\{\widehat{m}^i,m^i_j\}$, and $\pos_{\widehat{w}^i_j}(m^i)=\pos_{\widehat{w}^i_j}(\widehat{m}^i_j)=1$. Finally, the preference list of $\overline{w}^i_j$ is defined by setting $\domain(\pos_{\overline{w}^i_j})=\{m^i_j,\widehat{m}^i_j\}$, and $\pos_{\overline{w}^i_j}(m^i_j)=\pos_{\overline{w}^i_j}(\widehat{m}^i_j)=1$. An illustration of a Combined Vertex Selector gadget is given in Fig.~\ref{fig:w1MinSMT1}.

\begin{figure}[t!]\centering
\fbox{\includegraphics[scale=0.9]{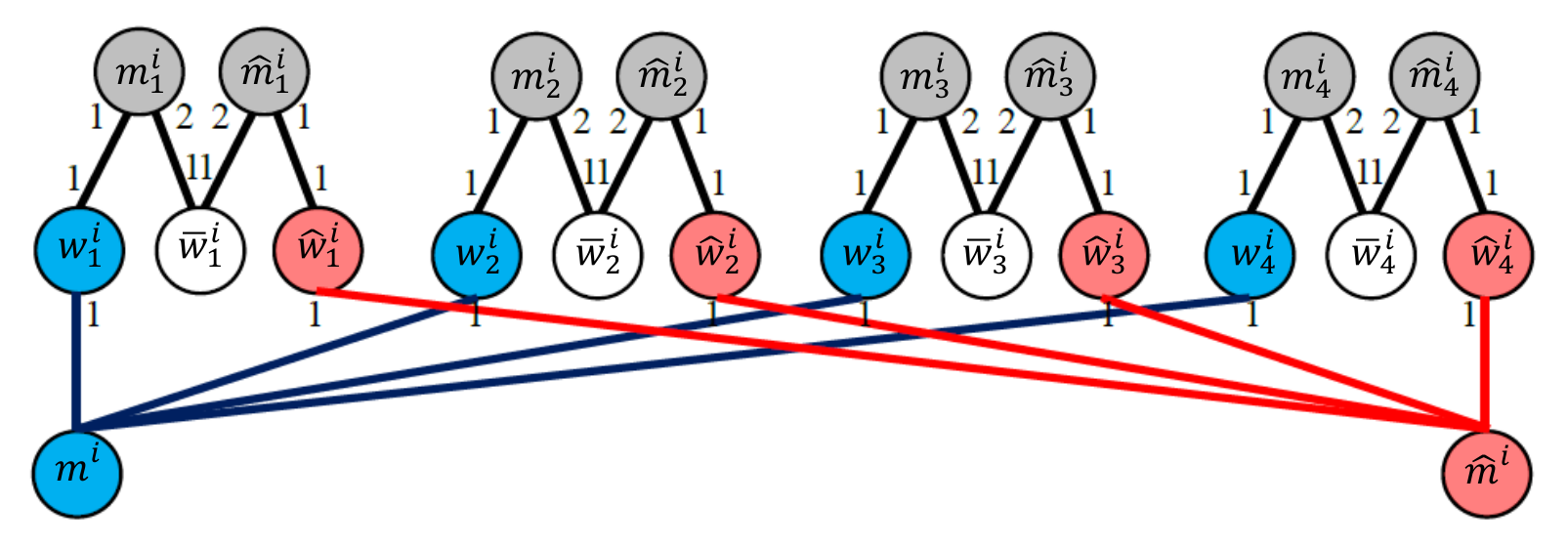}}
\caption{The Combined Vertex Selector gadget where $p=4$. Numbers indicate the positions of agents in preference lists.}\label{fig:w1MinSMT1}
\end{figure}

\medskip
\myparagraph{Edge Selector.} For every two color classes $i,j\in[k]$ where $i<j$, we introduce one set of new men: $M^{i,j}_{\enr}=\{m^{i,j}_1,m^{i,j}_2,\ldots,m^{i,j}_{q^{i,j}}\}$. We also introduce one new woman, called $w^{i,j}$. For all $t\in[q^{i,j}]$, the preference list of $m^{i,j}_t$ is defined by setting $\domain(\pos_{m^{i,j}_t})=\{w^{i,j},w^{i,j}_t\}$, $\pos_{m^{i,j}_t}(w^{i,j})=1$ and $\pos_{m^{i,j}_t}(w^{i,j}_t)=2$. For all $t\in[q^{i,j}]$, the preference list of $w^{i,j}_t$ is defined by setting $\domain(\pos_{w^{i,j}_t})=\{m^{i,j}_t,m^i,\widehat{m}^i,m^j,\widehat{m}^j\}$, $\pos_{w^{i,j}_t}(m^{i,j}_t)=1$ and $\pos_{w^{i,j}_t}(m^i)=\pos_{w^{i,j}_t}(\widehat{m}^i)=\pos_{w^{i,j}_t}(m^j)=\pos_{w^{i,j}_t}(\widehat{m}^j)=2$. The preference list of $w^{i,j}$ is defined by setting $\domain(\pos_{w^{i,j}})=M^{i,j}_{\enr}$ where for all $t\in[q^{i,j}]$, $\pos_{w^{i,j}}(m^{i,j}_t)=1$. An illustration of a Combined Vertex Selector gadget is given in Fig.~\ref{fig:w1MinSMT2}.

\begin{figure}[t!]\centering
\fbox{\includegraphics[scale=0.625]{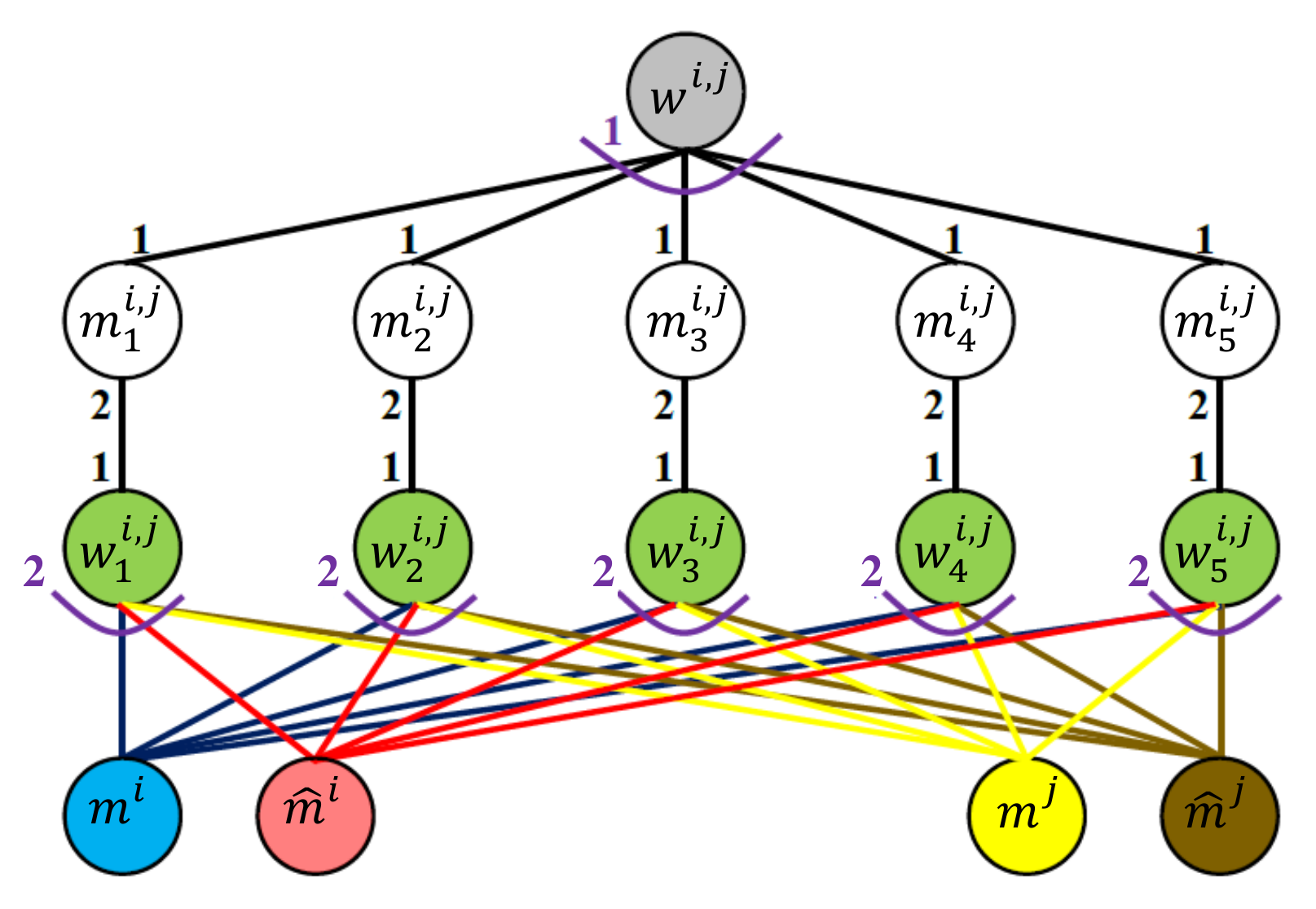}}
\caption{The Edge Selector gadget where $q^{i,j}=5$. Numbers indicate the positions of agents in preference lists.}\label{fig:w1MinSMT2}
\end{figure}

\subsubsection{Treewidth}

We begin the analysis of the reduction by bounding the treewidth of the resulting primal graph.

\begin{lemma}\label{lem:twMIN}
Let $I$ be an instance of {\sc Multicolored Clique}. Then, the treewidth of $\red_{min}(I)$ is bounded by $2k+\OO(1)$.
\end{lemma}

\begin{proof}
Let $P$ be the primal graph of $\red_{min}(I)$, and let $P'$ denote the graph obtained from $P$ by the removal of all of (the vertices that represent) men in $M_{\bas}\cup\widehat{M}_{\bas}$. Note that $|M_{\bas}\cup\widehat{M}_{\bas}|=2k$. Hence, as in the proof of Lemma \ref{lem:twSESM}, to prove that the treewidth of $P$ is bounded by $2k+\OO(1)$, it is sufficient to prove that the treewidth of every connected component of $P'$ is bounded by $\OO(1)$. However, $P'$ is a forest, and therefore the treewidth of each of its connected components is simply 1.
\end{proof}

\subsubsection{Correctness}

\myparagraph{Forward Direction.} We first show how given a solution of an instance $I$ of {\sc Multicolored Clique}, we can construct a stable matching $\mu$ of $\red_{min}(I)$ which matches all agents. For this purpose, we introduce the following definition.

\begin{definition}\label{def:cliquetoMIN}
Let $I=(G,(V^1,V^2,\ldots,V^k))$ be a \yesinstance\ of {\sc Multicolored Clique}, and let $U=\{v^1_{\ell_1},v^2_{\ell_2},\ldots,v^t_{\ell_k}\}$ and $W=\{e^{i,j}_{\ell_{i,j}}: i,j\in[k], i<j\}$ denote the vertex and edge sets, respectively, of a multicolored clique $C$ of $G$. Then, the matching $\mu^C_{min}$ of $\red_{min}(I)$ is defined as follows.
\begin{itemize}
\item For all $i\in[k]$: $\mu^C_{min}(m^i)=w^i_{\ell_i}$ and $\mu^C_{min}(\widehat{m}^i)=\widehat{w}^i_{\ell_i}$.
\item For all $i\in[k]$: $\mu^C_{min}(m^i_{\ell_i})=\overline{w}^i_{\ell_i}$ and $\widehat{m}^i_{\ell_i}$ is not matched.
\item For all $i\in[k]$ and $j\in[p]$ such that $j\neq\ell_i$: $\mu^C_{min}(m^i_j)=w^i_j$ and $\mu^C_{min}(\widehat{m}^i_j)=\widehat{w}^i_j$.
\item For all $i,j\in[k]$ where $i<j$: $\mu^C_{min}(m^{i,j}_{\ell_{i,j}})=w^{i,j}$.
\item For all $i,j\in[k]$ where $i<j$ and $t\in[q^{i,j}]$ such that $t\neq\ell_{i,j}$: $\mu^C_{min}(m^{i,j}_t)=w^{i,j}_t$.
\end{itemize}
\end{definition}

By Definition \ref{def:cliquetoMIN}, we directly derive the following observation.

\begin{observation}\label{obs:cliquetoMIN}
Let $I=(G,(V^1,V^2,\ldots,V^k))$ be a \yesinstance\ of {\sc Multicolored Clique}. Let $C$ be a multicolored clique of $G$. Then, the size of $\mu^C_{min}$ is $k+2|V(G)|+|E(G)|$.
\end{observation}

\begin{lemma}\label{lem:w1MINforward}
Let $I=(G,(V^1,V^2,\ldots,V^k))$ be a \yesinstance\ of {\sc Multicolored Clique}. Let $C$ be a multicolored clique of $G$. Then, $\mu^C_{min}$ is a stable matching of $\red_{min}(I)$.
\end{lemma}

\begin{proof}
First, for all $i\in[k]$, we claim that neither $m^i$ not $\widehat{m}^i$ can belong to a blocking pair. All the women in $W^i_{\bas}\setminus\{w^i_{\ell_i}\}$ are matched to men that they rank at position 1 in their preference lists, and therefore none of them can form a blocking pair with $m^i$. Similarly, all the women in $\widehat{W}^i_{\bas}\setminus\{\widehat{w}^i_{\ell_i}\}$ are matched to men that they rank at position 1 in their preference lists, and therefore none of them can form a blocking pair with $\widehat{m}^i$.  All of the other women in the preference lists of both $m^i$ and $\widehat{m}^i$ belong to $\bigcup_{j\in[k],j\neq i}W^{i,j}_{\bas}$. For all $j\in[k],j\neq i$, due to the fact that the preference lists of $m^i$ and $\widehat{m}^i$ are of the form of a leader and $e^{i,j}_{\ell_{i,j}}$ is incident to $v^i_{\ell_i}$ in $G$, we have that $m^i$ prefers $w^i_{\ell_i}$ over $w^{i,j}_{\ell_i}$ and that $\widehat{m}^i$ prefers $\widehat{w}^i_{\ell_i}$ over $w^{i,j}_{\ell_i}$. Moreover, for all $j\in[k],j\neq i$ and $t\in[q^{i,j}]$ such that $t\neq\ell_{i,j}$, we have that $w^{i,j}_t$ is matched to the man she prefers the most, and therefore she can form a blocking pair with neither $m^i$ nor $\widehat{m}^i$.

Second, notice that for all $i\in[k]$ and $j\in[p]$ such that $j\neq\ell_i$, $m^i_j$ and $\widehat{m}^i_j$ are matched to women at position 1 in their preference lists, and therefore they cannot belong to any blocking pair. Moreover, notice that for all $i\in[k]$, $w^i_{\ell_i}$, $\widehat{w}^i_{\ell_i}$ and $\overline{w}^i_{\ell_i}$ are matched to men at position 1 in their preference lists, and therefore they cannot belong to any blocking pair. Since the only women in the preference lists of $m^i_j$ and $\widehat{m}^i_j$ are these three women, we have that there two men cannot belong to any blocking pair as well.

Third, for all $i,j\in[k]$ where $i<j$ and $t\in[q^{i,j}]$, note that $w^{i,j}$ is matched to a man at position 1 in her preference list, and therefore she cannot belong to any blocking pair. Hence, every man in $M^{i,j}_{\enr}$ is matched to a woman that he ranks ranks either first or second, and the (unique) woman that he ranks first is $w^{i,j}$. Therefore, no man in $M^{i,j}_{\enr}$ can belong to a blocking pair.
\end{proof}

Combining Observation \ref{obs:cliquetoMIN} and Lemma \ref{lem:w1MINforward}, we derive the following corollary.

\begin{corollary}\label{cor:w1MINforward}
Let $I$ be a \yesinstance\ of {\sc Multicolored Clique}. Then, the instance $\red_{min}(I)$ of {\sc min-SMT} admits a stable matching of size $k+2|V(G)|+|E(G)|$.
\end{corollary}

This concludes the proof of the forward direction.

\medskip
\myparagraph{Reverse Direction.} Second, we prove that given an instance $I$ of {\sc Multicolored Clique}, if the instance $\red_{min}(I)$ of {\sc min-SMT} admits a stable matching of size at most $k+2|V(G)|+|E(G)|$, then we can construct a solution for $I$. To this end, we analyze the structure of stable matchings of $\red_{min}(I)$.

\begin{lemma}\label{lem:w1MINReverse1}
Let $I$ be an instance of {\sc Multicolored Clique}. Let $\mu$ be a stable matching of $\red_{min}(I)$. Then, $\mu$ matches all of the men in $\left(\bigcup_{i=1}^kM^i_{\bas}\cup \widehat{M}^i_{\bas}\right)\cup \left(\bigcup_{i=1}^{k-1}\bigcup_{j=i+1}^kM^{i,j}_{\enr}\right)$.
\end{lemma}

\begin{proof}
For each men in $\left(\bigcup_{i=1}^kM^i_{\bas}\cup \widehat{M}^i_{\bas}\right)\cup \left(\bigcup_{i=1}^{k-1}\bigcup_{j=i+1}^kM^{i,j}_{\enr}\right)$, there exists a woman such this man is the unique man at position 1 in her list. Indeed, for all $i\in[k]$ and or all men $m^i\in M^i$ and $\widehat{m}^i$, these women are $w^i$ and $\widehat{w}^i$, respectively. Moreover, for all $i,j\in[k]$ where $i<j$ and for every man $m^{i,j}_t\in M^{i,j}_{\enr}$, this woman is $w^{i,j}_t$. Thus, every man in $\left(\bigcup_{i=1}^kM^i_{\bas}\cup \widehat{M}^i_{\bas}\right)\cup \left(\bigcup_{i=1}^{k-1}\bigcup_{j=i+1}^kM^{i,j}_{\enr}\right)$ is matched by $\mu$, else he would have formed a blocking pair with the aforementioned woman.
\end{proof}

\begin{lemma}\label{lem:w1MINReverse2}
Let $I$ be an instance of {\sc Multicolored Clique}. Let $\mu$ be a stable matching of $\red_{min}(I)$. Then, for all $i\in[k]$, either $\mu$ matches all of the men in $M^i_{\enr}\cup \widehat{M}^i_{\enr}$ or the two following conditions hold: {\bf (i)} $\mu$ matches $p-1$ of the men in $M^i$, and {\bf (ii)} there exists $j\in[p]$ such that both $\mu(m^i)=w^i_j$ and $\mu(\widehat{m})=\widehat{w}^i_j$.
\end{lemma}

\begin{proof}
Consider some color class $i\in[k]$. If $\mu$ matches all of the men in $M^i_{\enr}\cup \widehat{M}^i_{\enr}$, then we are done. Thus, suppose that there exists a man in this set that $\mu$ does not match. Due to symmetry between $M^i_{\bas}$ and $\widehat{M}^i_{\bas}$, we can assume w.l.o.g.~that there exists $j\in[p]$ such that $\mu$ does not match $m^i_j$. Since $\mu$ is a stable matching and $m^i_j$ ranks $w^i_j$, it must hold that this woman is matched to the only other man that she ranks, who is $m^i$. That is, $\mu(w^i_j)=m^i$. Moreover, since $m^i_j$ ranks $\overline{w}^i_j$, it must hold that this woman is matched to the only other man that she ranks, who is $\widehat{m}^i_j$. That is, $\mu(\overline{w}^i_j)=\widehat{m}^i_j$. However, $\widehat{m}^i_j$ prefers $\widehat{w}^i_j$ over $\overline{w}^i_j$. Thus, $\widehat{w}^i_j$ is matched to the only other man that she ranks, who is $\widehat{m}^i$. That is, $\mu(\widehat{w}^i_j)=\widehat{m}^i$. Since the choice of $i$ was arbitrary, we conclude that the lemma is correct.
\end{proof}

\begin{lemma}\label{lem:w1MINReverse3}
Let $I$ be an instance of {\sc Multicolored Clique}. Let $\mu$ be a stable matching of $\red_{min}(I)$. Then, for all $i,j\in[k]$ where $i<j$, there exists $t\in[q^{i,j}]$ such that $w^{i,j}_t$ is either unmatched or matched to one man among $m^i$ and $\widehat{m}^i$. 
\end{lemma}

\begin{proof}
Consider some two color classes $i,j\in[k]$ where $i<j$. Since $w^{i,j}$ is the unique woman at position 1 in the preference list of all men in $M^{i,j}_{\enr}$, she is matched by $\mu$. Hence, the number of men in $M^{i,j}_{\enr}$ that are not matched by $\mu$ to $w^{i,j}$ is exactly $q^{i,j}-1$. However, the number of women in $W^{i,j}_{\bas}$ is exactly $q^{i,j}$. Thus, by the definition of the preference lists of the women in $W^{i,j}_{\bas}$, we deduce that there exists $t\in[q^{i,j}]$ such that $w^{i,j}_t$ is either unmatched or matched to one man among $m^i$ and $\widehat{m}^i$. Since the choices of $i$ and $j$ were arbitrary, we conclude that the lemma is correct.
\end{proof}

From Lemmata \ref{lem:w1MINReverse1}, \ref{lem:w1MINReverse2} and \ref{lem:w1MINReverse3}, we derive the following result.

\begin{lemma}\label{lem:w1MINReverse4}
Let $I$ be an instance of {\sc Multicolored Clique}. Let $\mu$ be a stable matching of $\red_{min}(I)$ of size at most $k+2|V(G)|+|E(G)|$. Then, for all $i\in[k]$, there exists $j\in[p]$ such that both $\mu(m^i)=w^i_j$ and $\mu(\widehat{m})=\widehat{w}^i_j$. Moreover, for all $i,j\in[k]$ where $i<j$, there exists $t\in[q^{i,j}]$ such that $w^{i,j}_t$ is unmatched. 
\end{lemma}

\begin{proof}
First, note that the size of a matching is equal to the number of men that it matches. Hence, $\mu$ matches at most $k+2|V(G)|+|E(G)|$ men. By Lemma \ref{lem:w1MINReverse1}, all of the men in $\left(\bigcup_{i=1}^kM^i_{\bas}\cup \widehat{M}^i_{\bas}\right)\cup \left(\bigcup_{i=1}^{k-1}\bigcup_{j=i+1}^kM^{i,j}_{\enr}\right)$ are matched by $\mu$. Since there are $2k+|E(G)|$ men in this set, and the set of all other men is $\bigcup_{i=1}^k{M^i\cup\widehat{M}^i}$ whose size is $2|V(G)|$, we have that $\mu$ not matches at least $k$ men from $\bigcup_{i=1}^k{M^i\cup\widehat{M}^i}$. Recall that $|V(G)|=pk$. Thus, by Lemma \ref{lem:w1MINReverse2}, this scenario is only possible if for all $i\in[k]$, there exists $j\in[p]$ such that both $\mu(m^i)=w^i_j$ and $\mu(\widehat{m})=\widehat{w}^i_j$. Now, consider some indices $i,j\in[k]$ where $i<j$. Note that since $w^{i,j}$ is the unique woman at position 1 in the preference list of all men in $M^{i,j}_{\enr}$, she is matched by $\mu$. Hence, the number of men in $M^{i,j}_{\enr}$ that are not matched by $\mu$ to $w^{i,j}$ is exactly $q^{i,j}-1$. However, the number of women in $W^{i,j}_{\bas}$ is exactly $q^{i,j}$, and since we have argued that the men in $M^i_{\bas}\cup \widehat{M}^i_{\bas}$ are matched to women in $W^i_{\bas}\cup \widehat{W}^i_{\bas}$, the only men that these women can be matched to are those that belong to $M^{i,j}_{\enr}$. We thus deduce that there exists $t\in[q^{i,j}]$ such that $w^{i,j}_t$ is unmatched.
\end{proof}

We are now ready to prove the correctness of the reverse direction.

\begin{lemma}\label{lem:w1MINReverse}
Let $I=(G,(V^1,V^2,\ldots,V^k))$ be an instance of {\sc Multicolored Clique}. If the instance $\red_{min}(I)$ of {\sc min-SMT} admits a stable matching of size at most $k+2|V(G)|+|E(G)|$, then $I$ is a \yesinstance\ of {\sc Multicolored Clique}.
\end{lemma}

\begin{proof}
Suppose that the instance $\red_{min}(I)$ of {\sc min-SMT} admits a stable matching $\mu$ of size at most $k+2|V(G)|+|E(G)|$.
By Lemma~\ref{lem:w1MINReverse4}, for all $i\in[k]$, there exists $\ell_i\in[p]$ such that $\mu(m^i)=w^i_{\ell_i}$ and $\mu(\widehat{m}^i)=\widehat{w}^i_{\ell_i}$. Moreover, by Lemma~\ref{lem:w1MINReverse4}, for all $i,j\in[k]$ where $j<i$, there exists $\ell_{i,j}\in[q^{i,j}]$ such that $\mu(w^{i,j}_{\ell_{i,j}})$ is unmatched. Denote $U=\{v^1_{\ell_1},v^2_{\ell_2},\ldots,v^t_{\ell_k}\}$ and $W=\{e^{i,j}_{\ell_{i,j}}: i,j\in[k], i<j\}$. Note that we have proved that $|U|=k$ and $|W|\geq {k\choose 2}$. Since for all $i,j\in[k]$ where $j<i$, $w^{i,j}_{\ell_i}$ prefers both $m^i$ and $\widehat{m}^i$ over being unmatched, we have that $w^{i,j}_{\ell_i}$ is not located after $w^i_{\ell_i}$ in the preference list of $m^i$ as well as that it is not located after $\widehat{w}^i_{\ell_i}$ in the preference list of $\widehat{m}^i$. However, by the definition of the preference lists of $m^i$ and $\widehat{m}^i$, it must then hold that $e^{i,j}_{\ell_{i,j}}$ is an edge incident to $v^i_{\ell_i}$ in $G$. Hence, we derive that $U$ is the vertex set of a colorful clique of $G$. We thus conclude that $I$ is a \yesinstance\ of {\sc Multicolored Clique}.
\end{proof}

\medskip
\myparagraph{Summary.} Finally, we note that the reduction can be performed in time that is polynomial in the size of the output. That is, we have the following observation.

\begin{observation}\label{obs:w1MINTime}
Let $I=(G,(V^1,V^2,\ldots,V^k))$ be an instance of {\sc Multicolored Clique}. Then, the instance $\red_{min}(I)$ of {\sc min-SMT} can be constructed in time polynomial in the size~of~$I$.
\end{observation}

By Proposition \ref{prop:multiClique}, Lemma \ref{lem:twMIN}, Corollary \ref {cor:w1MINforward}, Lemma \ref{lem:w1MINReverse} and Observation \ref{obs:w1MINTime}, we conclude that {\sc min-SMT} is \WOH. Moreover, unless \ETH\ fails, {\sc min-SMT} cannot be solved in time $f(\tw)\cdot n^{o(\tw)}$ for any function $f$ that depends only on $\tw$. Here, $n$ is the number of agents.

%% file: xp.tex
\section{Primal Graph: XP Algorithms}\label{sec:xp}

In this section, we sketch the proof of Theorem \ref{thm:xpIntro}. Each subsection below is devoted to one problem. Throughout this section, $\tw$ is used to denote the treewidth of the primal graph, and $n=|A|$ denotes the number of agents.

\subsection{Sex Equal Stable Marriage}\label{sec:xpSESM}

We first show that {\sc SESM} can be solved in time $n^{\OO(\tw)}$. To this end, let $I=(M,W,\{\pos_m\}|_{m\in M},$ $\{\pos_w\}|_{w\in W})$ be an instance of {\sc SESM}, and let $(T,\beta)$ denote a nice tree decomposition of width $\tw$ of the primal graph. First, in light of Proposition \ref{lem:matchSame}, by employing the algorithm given by Proposition \ref{lem:menOptimal}, we compute the subset $A^\star\subseteq A$ of agents who are matched by every stable matching, which is also the set of agents who are matched by at least one stable matching.

Based on the method of DP, we introduce a table {\sf N}. Each entry of the table {\sf N} is of the form {\sf N}$[v,f,t]$, where $v\in V(T)$, $f: \beta(v)\cap A^\star\rightarrow A^\star$ is an injective function such that for all $a\in\domain(f)$, $a$ and $f(a)$ are agents of opposite sex, and $t\in\{-n^2,-n^2+1,\ldots,n^2\}$. To formulate our objective, we rely on the following definition, which is also relevant to instances of {\sc SMT}, that is, in the presence of ties. We remark that this definition will be reused when we design our other \XP\ algorithms.
\begin{definition}\label{def:xpPartial}
Let $(M,W,\{\pos_m\}|_{m\in M},$ $\{\pos_w\}|_{w\in W})$ be an instance of {\sc SMT}, and let $(T,\beta)$ denote a tree decomposition of width $\tw$ of the primal graph. Given a node $v\in V(T)$, a set $U\subseteq\beta(v)$ and a function $f: U\rightarrow A$, we say an injective function $g$ is {\em $(v,U,f)$-stable} if $\domain(g)\subseteq\gamma(v)$, $\image(g)\subseteq A$, and the following conditions are satisfied.
\begin{itemize}
\item For all $a\in\beta(v)$, if $a\notin U$ then $a\notin\domain(g)$, and otherwise $g(a)=f(a)$.
\item There do not exist $m\in\gamma(v)\cap M$ and $w\in\gamma(v)\cap W$ whose matching is {\em inconsistent}, that is, at least one of the following conditions is satisfied: {\bf (i)} $g(m)=w$ and either $w\notin\domain(g)$ or $g(w)\neq m$; {\bf (ii)} $g(w)=m$ and either $m\notin\domain(g)$ or $g(m)\neq w$.
\item There do not exist $m\in\gamma(v)\cap M$ and $w\in\gamma(v)\cap W$ that strictly prefer being matched to one another over their ``status'' with respect to $g$, that is, both of the following conditions are satisfied: {\bf (i)} either $m\notin\domain(g)$ or $m$ prefers $w$ over $g(m)$; {\bf (ii)} either $w\notin\domain(g)$ or $w$ prefers $m$ over $g(w)$.
\end{itemize}
\end{definition}

We say that {\sf N} is {\em computed correctly} if for each entry {\sf N}$[v,f,t]$, it holds that {\sf N}$[v,f,t]\in\{0,1\}$, and {\sf N}$[v,f,t]=1$ if and only if there exists a $(v,\beta(v)\cap A^\star,f)$-stable function $g$ such that $\sum_{m\in\domain(g)\cap M}(\pos_m(g(m)) - \pos_{g(m)}(m))=t$. Note that for all $\mu\in{\cal S}$, it holds that $\sat_M(\mu),\sat_W(\mu)$ $\in[n^2]$ (see also Lemma \ref{lem:boundTarget} in Section \ref{sec:dpTable}). Thus, if {\sf N} is computed correctly, the output is simply the smallest absolute value $t'$ for which there exists an integer $t\in\{-n^2,-n^2+1,\ldots,n^2\}$ such that $t'=|t|$ and {\sf N}$[\rootT(T),f: \emptyset\rightarrow\emptyset,t]=1$.

For the sake of completeness, we proceed to present the computation of the entries of {\sf N}. Since the proof of correctness is straightforward, it is omitted. We process the entries of {\sf N} by traversing the tree $T$ in post-order. The order in which we process entries corresponding to the same node $v\in V(T)$ is arbitrary. In the basis, where $v$ is a leaf and thus $\beta(v)=\emptyset$, we simply set {\sf N}$[v,f,t]$ to be 1 if and only if $t=0$. Next, we compute entries where $v$ is a forget node, an introduce node or a join node.

\smallskip
\myparagraph{Forget Node.} Let $u$ denote the child of $v$ in $T$, and $a$ denote the agent in $\beta(u)\setminus\beta(v)$. Let ${\cal F}$ denote the set of all functions $f':\beta(u)\cap A^\star\rightarrow A^\star$ that are identical of $f$ when restricted to $\beta(v)\cap A^\star$. Then, {\sf N}$[v,f,t]$ is set to be 1 if and only if there exists $f'\in{\cal F}$ such that {\sf N}$[u,f',t]=1$.

\smallskip
\myparagraph{Introduce Node.} Let $u$ be the child of $v$ in $T$, and $a$ be the agent in $\beta(v)\setminus\beta(u)$. Let $f'$ be the restriction of $f$ to $\beta(u)\cap A^\star$. Consider the following cases, recalling that the terms ``inconsistency'' and ``status'' are defined as in Definition \ref{def:xpPartial}.
\begin{itemize}
\itemsep0em 
\item If there exists $a'\in\beta(v)$ such that the matching of $a$ and $a'$ is inconsistent: {\sf N}$[v,f,t]=0$.
\item Else if there exists $a'\in\beta(v)$ such that $a$ and $a'$ prefer being matched to one another over their status with respect to $f$: {\sf N}$[v,f,t]=0$.
\item Else if $a$ is a man: {\sf N}$[v,f,t]=${\sf N}$[u,f',t-(\pos_a(f(a))-\pos_{f(a)}(a))]$.
\item Otherwise: {\sf N}$[v,f,t]=${\sf N}$[u,f',t]$.
\end{itemize}

\myparagraph{Join Node.} Let $u$ and $s$ denote the children of $v$ in $T$. For the sake of efficiency, we compute all entries of the form {\sf N}$[v,f,\cdot]$ simultaneously. First, we initialize each such entry to 0. Now, we compute $T(u,f)=\{\widehat{t}\in\{-n^2,-n^2+1,\ldots,n^2\}: \mathrm{\sf N}[u,f,\widehat{t}]=1\}$ and $T(s,f)=\{\widehat{t}\in\{-n^2,-n^2+1,\ldots,n^2\}: \mathrm{\sf N}[s,f,\widehat{t}]=1\}$. Then, for all $\widehat{t}\in T(u,f)$ and $(t^*,f)\in T(s,f)$ such that $t=\widehat{t}+t^*-\sum_{m\in\domain(f)\cap M\cap\beta(v)}(\pos_m(f(m)) - \pos_{f(m)}(m))\in[n^2]$, we set $\mathrm{\sf N}[v,f,t]=1$. (An entry may be set to $1$ multiple times.)

\subsection{Balanced Stable Marriage}\label{sec:xpBSM}

We now show that {\sc BSM} can be solved in time $n^{\OO(\tw)}$. To this end, let $I=(M,W,\{\pos_m\}|_{m\in M},$ $\{\pos_w\}|_{w\in W})$ be an instance of {\sc BSM}, and let $(T,\beta)$ denote a nice tree decomposition of width $\tw$ of the primal graph. We compute the same set $A^\star$ as in the case of {\sc SESM}, and also introduce a table {\sf N} with the same set of entries {\sf N}$[v,f,t]$, except that now $t$ is restricted to belong to the set $[n^2]_0$. Here, we say that {\sf N} is {\em computed correctly} if for each entry {\sf N}$[v,f,t]$, it holds that {\sf N}$[v,f,t]\in[n^2]_0\cup\{\nil\}$, and for all $i\in[n^2]_0$, {\sf N}$[v,f,t]=i$ if and only if there exists a $(v,\beta(v)\cap A^\star,f)$-stable function $g$ such that $\sum_{m\in\domain(g)\cap M}\pos_m(g(m))=t$ and $\sum_{w\in\domain(g)\cap W}\pos_w(g(w))=i$, and there does not exist a $(v,\beta(v)\cap A^\star,f)$-stable function $g'$ such that $\sum_{m\in\domain(g')\cap M}\pos_m(g'(m))=t$ and $\sum_{w\in\domain(g')\cap W}\pos_w(g'(w))<i$.

We process the entries of {\sf N} by traversing the tree $T$ in post-order. The order in which we process entries corresponding to the same node $v\in V(T)$ is arbitrary. In the basis, where $v$ is a leaf, we simply set {\sf N}$[v,f,t]$ to be 0 if $t=0$ and to $\nil$ otherwise. Next, we compute entries where $v$ is a forget node, an introduce node or a join node.

\smallskip
\myparagraph{Forget Node.} Let $u$ denote the child of $v$ in $T$, and $a$ denote the agent in $\beta(u)\setminus\beta(v)$. Let ${\cal F}$ denote the set of all functions $f':\beta(u)\cap A^\star\rightarrow A^\star$ that are identical of $f$ when restricted to $\beta(v)\cap A^\star$. Then, {\sf N}$[v,f,t]$ is set to be $i\in[n^2]_0$ if and only if there exists $f'\in{\cal F}$ such that {\sf N}$[u,f',t]=i$, and there does not exist $\widehat{f}\in{\cal F}$ such that {\sf N}$[u,\widehat{f},t]<i$.

\smallskip
\myparagraph{Introduce Node.} Let $u$ be the child of $v$ in $T$, and $a$ be the agent in $\beta(v)\setminus\beta(u)$. Let $f'$ be the restriction of $f$ to $\beta(u)\cap A^\star$. Consider the following cases.
\begin{itemize}
\itemsep0em 
\item If there exists $a'\in\beta(v)$ such that the matching of $a$ and $a'$ is inconsistent: {\sf N}$[v,f,t]=\nil$.
\item Else if there exists $a'\in\beta(v)$ such that $a$ and $a'$ prefer being matched to one another over their status with respect to $f$: {\sf N}$[v,f,t]=\nil$.
\item Else if $a$ is a man: {\sf N}$[v,f,t]=${\sf N}$[u,f',t-\pos_a(f(a))]$.
\item Otherwise: {\sf N}$[v,f,t]=${\sf N}$[u,f',t]+\pos_{a}(f(a))$.
\end{itemize}

\myparagraph{Join Node.} Let $u$ and $s$ denote the children of $v$ in $T$. For the sake of efficiency, we compute all entries of the form {\sf N}$[v,f,\cdot]$ simultaneously. First, we initialize each such entry to $\nil$. Now, we compute $T(u,f)=\{\widehat{t}\in[n^2]_0: \mathrm{\sf N}[u,f,\widehat{t}]\neq\nil\}$ and $T(s,f)=\{\widehat{t}\in[n^2]_0: \mathrm{\sf N}[s,f,\widehat{t}]\neq\nil\}$. Then, for all $\widehat{t}\in T(u,f)$ and $(t^*,f)\in T(s,f)$ such that $t=\widehat{t}+t^*-\sum_{m\in\domain(f)\cap M\cap\beta(v)}\pos_m(f(m))\in[n^2]_0$, we set $\mathrm{\sf N}[v,f,t]$ to be the minimum between the previous value this entry stored and $\mathrm{\sf N}[u,f,\widehat{t}] + \mathrm{\sf N}[s,f,t^*] - \sum_{w\in\domain(f)\cap W\cap\beta(v)}\pos_w(f(w))$.

\subsection{max-SMT}\label{sec:xpmaxSMT}

We proceed to show that {\sc max-SMT} can be solved in time $n^{\OO(\tw)}$. To this end, let $I=(M,W,\{\pos_m\}|_{m\in M},$ $\{\pos_w\}|_{w\in W})$ be an instance of {\sc max-SMT}, and let $(T,\beta)$ denote a nice tree decomposition of width $\tw$ of the primal graph.  Again, we introduce a table {\sf N}. Here, each entry of {\sf N} is of the form {\sf N}$[v,U,f]$, where $v\in V(T)$, $U\subseteq\beta(v)$, and $f: U\rightarrow A$ is an injective function such that for all $a\in\domain(f)$, $a$ and $f(a)$ are agents of opposite sex. We say that {\sf N} is {\em computed correctly} if for each entry {\sf N}$[v,U,f]$, it holds that {\sf N}$[v,U,f]\in[n]_0\cup\{\nil\}$, and for all $i\in[n]_0$, {\sf N}$[v,U,f]=i$ if and only if there exists a $(v,U,f)$-stable function $g$ such that matches exactly $i$ men, and there does not exist a $(v,U,f)$-stable function $g'$ that matches more than $i$ men.

We process the entries of {\sf N} by traversing the tree $T$ in post-order. The order in which we process entries corresponding to the same node $v\in V(T)$ is arbitrary. In the basis, where $v$ is a leaf, we simply set {\sf N}$[v,U,f]$ to be 0. Next, we compute entries where $v$ is a forget node, an introduce node or a join node.

\smallskip
\myparagraph{Forget Node.} Let $u$ denote the child of $v$ in $T$, and $a$ denote the agent in $\beta(u)\setminus\beta(v)$. Let ${\cal F}$ denote the set containing $f$ as well as all functions $f': U\cup\{a\}\rightarrow A$ that are identical of $f$ when restricted to $U$. Then, {\sf N}$[v,U,f]$ is set to be $i\in[n]_0$ if and only if there exists $f'\in{\cal F}$ such that {\sf N}$[u,\domain(f'),f']=i$, and there does not exist $\widehat{f}\in{\cal F}$ such that {\sf N}$[u,\domain(\widehat{f}),\widehat{f}]>i$.

\smallskip
\myparagraph{Introduce Node.} Let $u$ be the child of $v$ in $T$, and $a$ be the agent in $\beta(v)\setminus\beta(u)$. Let $f'$ be the restriction of $f$ to $\beta(u)\cap U$. Consider the following cases.
\begin{itemize}
\itemsep0em 
\item If there exists $a'\in\beta(v)$ such that the matching of $a$ and $a'$ is inconsistent: {\sf N}$[v,U,f]=\nil$.
\item Else if there exists $a'\in\beta(v)$ such that $a$ and $a'$ prefer being matched to one another over their status with respect to $f$: {\sf N}$[v,U,f]=\nil$.
\item Else if $a$ is a man and $a\in\domain(f)$: {\sf N}$[v,U,f]=${\sf N}$[u,\domain(f'),f']+1$.
\item Otherwise: {\sf N}$[v,U,f]=${\sf N}$[u,\domain(f'),f']$.
\end{itemize}

\myparagraph{Join Node.} Let $u$ and $s$ denote the children of $v$ in $T$. Then, we set {\sf N}$[v,U,f]=\mathrm{\sf N}[u,U,f]+\mathrm{\sf N}[s,U,f]-|U\cap M|$.

\subsection{min-SMT}\label{sec:xpminSMT}

We proceed to show that {\sc min-SMT} can be solved in time $n^{\OO(\tw)}$. To this end, let $I=(M,W,\{\pos_m\}|_{m\in M},$ $\{\pos_w\}|_{w\in W})$ be an instance of {\sc min-SMT}, and let $(T,\beta)$ denote a nice tree decomposition of width $\tw$ of the primal graph. We introduce a table {\sf N} with the same set of entries as in the case of {\sc max-SMT}. We say that {\sf N} is {\em computed correctly} if for each entry {\sf N}$[v,U,f]$, it holds that {\sf N}$[v,U,f]\in[n]_0\cup\{\nil\}$, and for all $i\in[n]_0$, {\sf N}$[v,U,f]=i$ if and only if there exists a $(v,U,f)$-stable function $g$ such that matches exactly $i$ men, and there does not exist a $(v,U,f)$-stable function $g'$ that matches less than $i$ men.

We process the entries of {\sf N} by traversing the tree $T$ in post-order. The order in which we process entries corresponding to the same node $v\in V(T)$ is arbitrary. In the basis, where $v$ is a leaf, we simply set {\sf N}$[v,U,f]$ to be 0. Next, we compute entries where $v$ is a forget node, an introduce node or a join node.

\smallskip
\myparagraph{Forget Node.} Let $u$ denote the child of $v$ in $T$, and $a$ denote the agent in $\beta(u)\setminus\beta(v)$. Let ${\cal F}$ denote the set containing $f$ as well as all functions $f': U\cup\{a\}\rightarrow A$ that are identical of $f$ when restricted to $U$. Then, {\sf N}$[v,U,f]$ is set to be $i\in[n]_0$ if and only if there exists $f'\in{\cal F}$ such that {\sf N}$[u,\domain(f'),f']=i$, and there does not exist $\widehat{f}\in{\cal F}$ such that {\sf N}$[u,\domain(\widehat{f}),\widehat{f}]<i$.

\smallskip
\myparagraph{Introduce Node.} This computation is identical to the one presented for the case of {\sc max-SMT}.

\smallskip
\myparagraph{Join Node.} This computation is again identical to the one presented for the case of {\sc max-SMT}.

%% file: fpt.tex
% !TEX root = mainTW.tex

\section{Rotation Digraph: FPT Algorithms}\label{sec:fpt}

In this section, we prove Theorem \ref{thm:fptIntro}, based on the approach described in Section~\ref{sec:overview}. First, in Section \ref{sec:genSM}, we solve a problem that we call {\sc Generic SM (GSM)}, which is defined as follows. Given an instance of {\sc SM} as input, the objective of {\sc GSM} is to determine, for all $t_M,t_W\in\mathbb{N}$, whether there exists $\mu\in{\cal S}$ such that both $\sat_M(\mu)=t_M$ and $\sat_W(\mu)=t_W$. Here, the running time would be $\OO(2^{\tw}\cdot n^{10})$, where $\tw$ is the treewidth of $G_\Pi$. Given this algorithm as a black box, it is straightforward to solve both {\sc SESM} and {\sc BSM} in time $\OO(2^{\tw}\cdot n^{10})$. In Sections \ref{sec:fptSESM} and \ref{sec:fptBSM}, we further show that simple modifications to this algorithm imply that both {\sc SESM} and {\sc BSM} can also be solved in time $\OO(2^{\tw}\cdot n^6)$.

\subsection{Generic Stable Marriage}\label{sec:genSM}

The presentation of our algorithm for {\sc GSM}, which we call \alg, is structured as follows. First, we set-up terminology that allows us to handle the obstacles described in Section \ref{sec:overview}, namely the delicate interplay between ``past'' and ``future'' imposed by our partial solutions and their inability to reveal in a direct manner, for any given man, who is the woman with whom it is matched.
Then, we turn to present \alg, which is based on DP. Afterwards, we analyze the running time of this algorithm. Finally, we reach the most technical part of this section, which is the proof of correctness of \alg. We remark that \alg\ is very short and easily implementable---the proof of the correctness is the part that is technically involved.

\subsubsection{Terminology}\label{sec:dpTerminology}

We say that a pair $(v,R')$ of a node $v\in V(T)$ and a set $R'\subseteq \beta(v)$ is a {\em state}. We define notions with respect to states, and when the states are clear from context, we avoid mentioning them explicitly. Each state captures an equivalence class, formally defined as follows. 

\begin{definition}
A stable matching $\mu$ is {\em compatible with} a state $(v,R')$ if $R(\mu)\cap\beta(v)=R'$.
\end{definition}

Accordingly, we let ${\cal S}(v,R')\subseteq{\cal S}$ denote the set of all stable matchings compatible with $(v,R')$. Next, we proceed by identifying, for every man in $M^\star$, a directed path that captures the manner in which he may swap partners along the execution of our algorithm.

\begin{lemma}\label{lem:manPath}
Let $m\in M^\star$. Then, $D_\Pi$ contains a directed path whose vertex-set is a superset of $R(m)$, and such a path can be found in time $\OO(n^2)$.
\end{lemma}

\begin{proof}
By Proposition \ref{lem:totalOrder}, since the transitive closure of $D_\Pi$ is isomorphic to $\Pi$, we have that $D_\Pi$ contains a directed path whose vertex-set is a superset of $R(m)$. Moreover, by iterating over $R=V(D_\Pi)$ and identifying the rotations involving $m$, we compute $R(m)$. Since $D_\Pi$ is a DAG with $\OO(n^2)$ vertices and arcs (Proposition \ref{prop:rotDig}), a topological order of $D_\Pi$ can be computed in time $\OO(n^2)$. Then, since $D_\Pi$ is a DAG, by running breadth-first search (BFS) from the first vertex in $R(m)$ (according to the topological order), we identify a directed path in $D_\Pi$ whose vertex-set is a superset of $R(m)$. Overall, the running time is bounded by $\OO(n^2)$.
\end{proof}

For each man $m\in M^\star$, we use the computation in Lemma~\ref{lem:manPath} to obtain some directed path in $D_\Pi$ whose vertex-set is a superset of $R(m)$. We remark that in case $R(m)=\emptyset$, we can simply let $P(m)$ be the empty path. We say that a man $m\in M^\star$ is {\em relevant to $(v,R')$} if $\beta(v)\cap V(P(m))\neq\emptyset$. Given such $m$, we proceed by identifying vertices on $P(m)$ that determine how to handle a given state $(v,R')$. Roughly speaking, these vertices are ``exist and entry points'' located in $\beta(v)$ that help identifying the woman assigned to~$m$ (see Fig.~\ref{fig:FPT1}).

\begin{figure}[t!]\centering
\fbox{\includegraphics[scale=0.7]{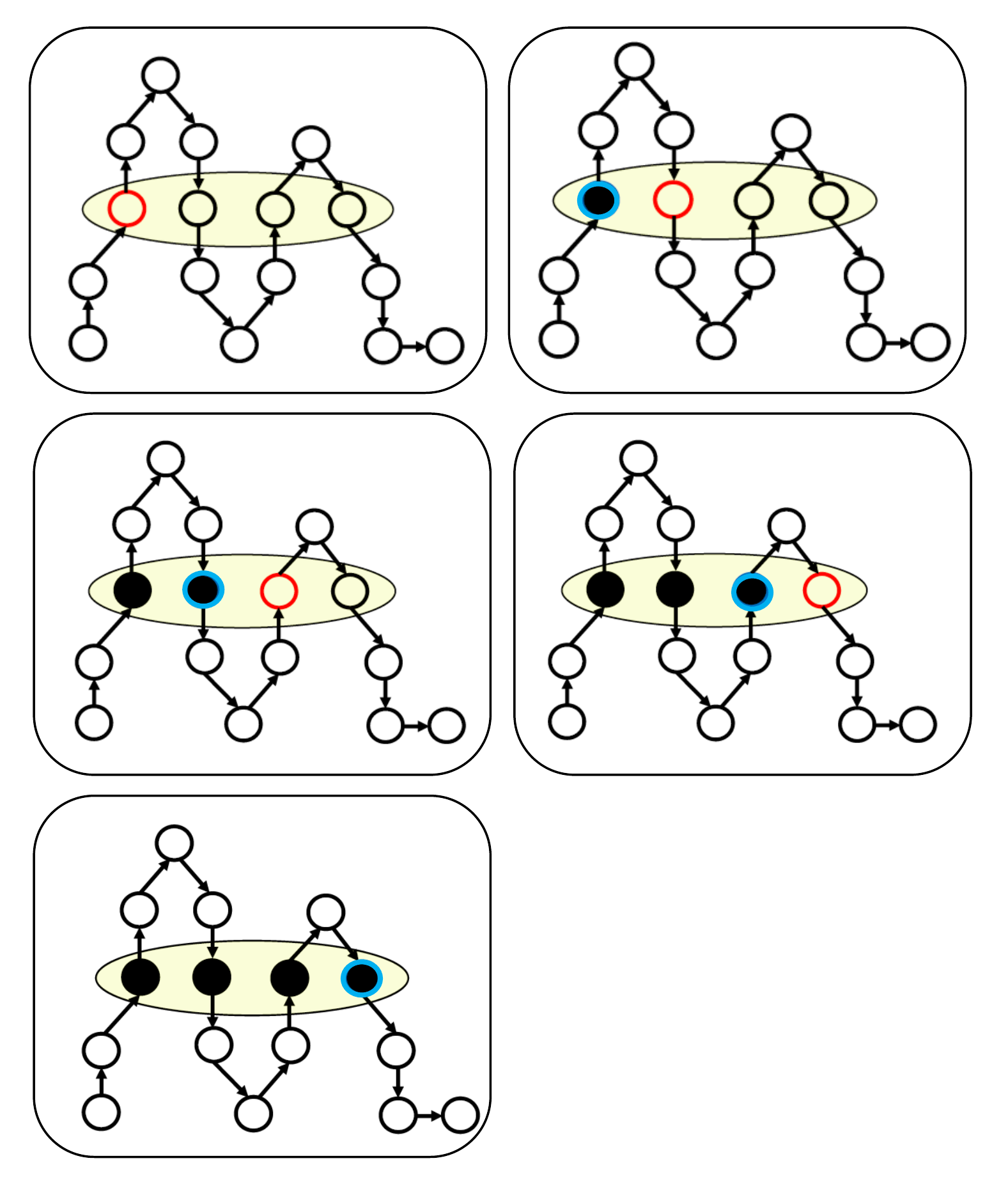}}
\caption{Illustrations of Definition \ref{def:lf}. The drawn vertices form a directed path of one man, $m$. The yellow shape captures the bag $\beta(v)$, and the vertices in $R'$ are drawn in black. The vertices $\ell_{v,R'}(m)$ and $f_{v,R'}(m)$ are drawn in blue and red circles, respectively. In the first illustration, $\ell_{v,R'}(m)=\nil$, an in the last illustration, $f_{v,R'}(m)=\nil$.}\label{fig:FPT1}
\end{figure}

\begin{definition}\label{def:lf}
Let $(v,R')$ be a state, and let $m\in M^\star$ be relevant to $(v,R')$. Then, $\ell_{v,R'}(m)$ denotes the last vertex on $P(m)$ that belongs to $R'$ (if such a vertex does not exist, $\ell_{v,R'}(m)=\nil$). Accordingly, the vertex $f_{v,R'}(m)$ is defined as follows. 
\begin{itemize}
\itemsep0em 
\item If $\ell_{v,R'}(m)\neq\nil$, then $f_{v,R'}(m)$ denotes the first vertex on $P(m)$ that succeeds $\ell_{v,R'}(m)$ and which belongs to $\beta(v)\setminus R'$ (if such a vertex does not exist, $f_{v,R'}(m)=\nil$).
\item Otherwise, $f_{v,R'}(m)$ denotes the first vertex on $P(m)$ that belongs to $\beta(v)$.
\end{itemize}
\end{definition}

Having $\ell_{v,R'}(m)$ and $f_{v,R'}(m)$ at hand, we proceed to identify the subpath of $P(m)$ whose internal vertices are rotations such that when we handle the state $(v,R')$, it is not known which of these rotations should be eliminated (see Fig.~\ref{fig:FPT2}).

\begin{definition}\label{def:lfPath}
Let $(v,R')$ be a state, and let $m\in M^\star$ be relevant to $(v,R')$. Then, the path $P_{v,R'}(m)$ is defined as follows.
\begin{itemize}
\itemsep0em 
\item If $\ell_{v,R'}(m),f_{v,R'}(m)\neq\nil$, then $P_{v,R'}(m)$ denotes the subpath of $P(m)$ that starts at $\ell_{v,R'}(m)$ and ends at $f_{v,R'}(m)$.
\item If $\ell_{v,R'}(m)\neq\nil$ and $f_{v,R'}(m)=\nil$, then $P_{v,R'}(m)$ denotes the subpath of $P(m)$ that starts at $\ell_{v,R'}(m)$ and ends at the last vertex of $P(m)$.
\item If $\ell_{v,R'}(m)=\nil$ and $f_{v,R'}(m)\neq\nil$, then $P_{v,R'}(m)$ denotes the subpath of $P(m)$ that starts at the first vertex of $P(m)$ and ends at $f_{v,R'}(m)$.
\end{itemize}
\end{definition}

\begin{figure}[t!]\centering
\fbox{\includegraphics[scale=0.7]{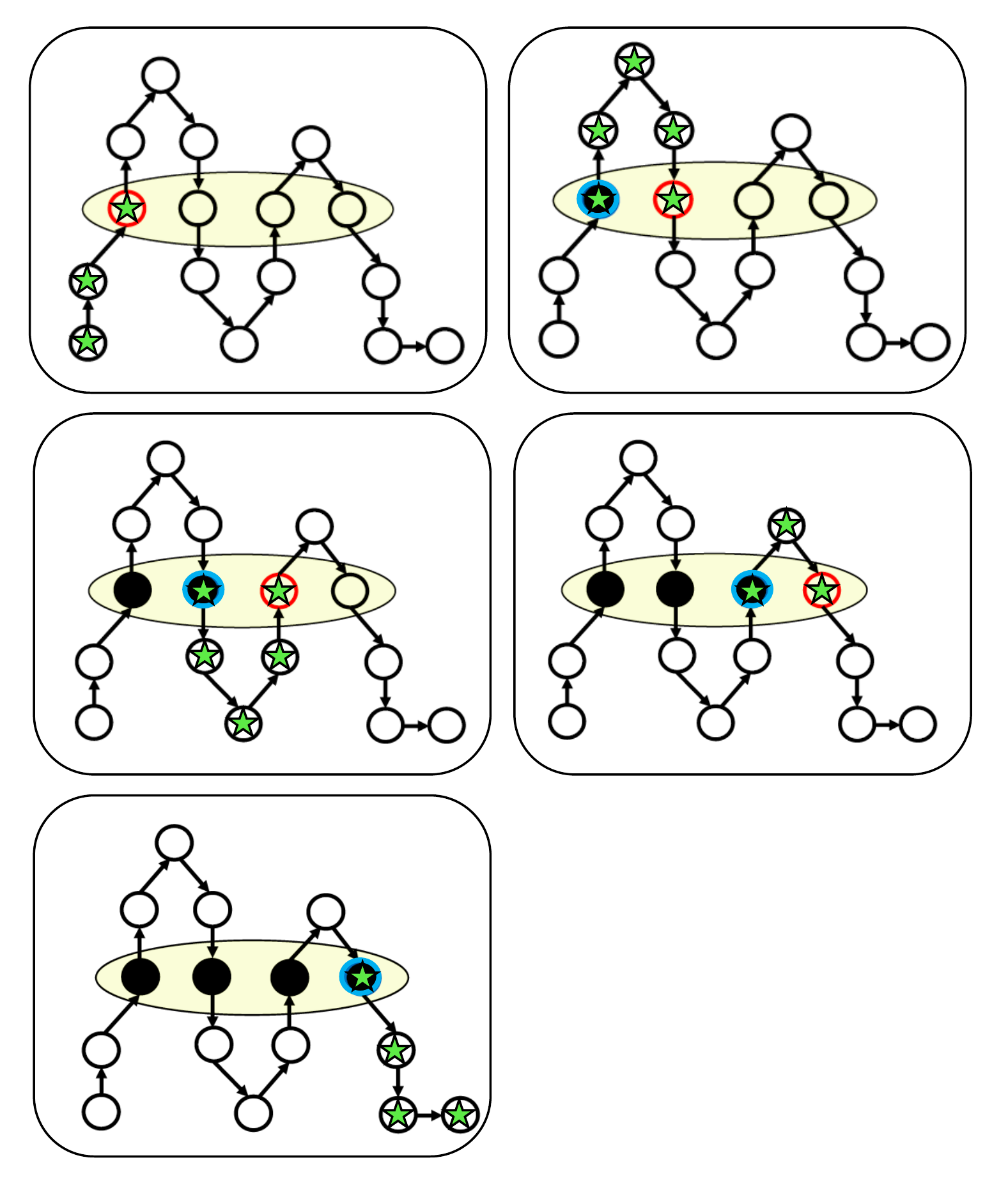}}
\caption{Illustrations of Definition \ref{def:lfPath}. The vertices of $P_{v,R'}(m)$ are marked by green stars.}\label{fig:FPT2}
\end{figure}

Notice that for a relevant man $m\in M^\star$, at least one among $\ell_{v,R'}(m)$ and $f_{v,R'}(m)$ is not $\nil$, and therefore $P_{v,R'}(m)$ is well defined. Next, we assign a type to each man $m\in M^\star$, which roughly indicates whether we have already chosen the ``final'' partner of $m$. To this end, we first recall a well-known proposition \, after which we prove a simple lemma.

\begin{proposition}[Folklore]\label{lem:tree}
Given a graph $G$ with tree decomposition $(T,\beta)$, and a connected subgraph $S$ of $G$, the subgraph of $T$ induced by $T_{V(S)}=\{v\in V(T): V(S)\cap\beta(v)\neq\emptyset\}$ is~a~tree.
\end{proposition}

\begin{lemma}\label{lem:fullClassify}
Let $(v,R')$ be a state, and let $m\in M^\star$ such that $V(P(m))\setminus\gamma(v),~ V(P(m))\cap \gamma(v)\neq\emptyset$, (recall that $\gamma(v)$ is the union of all the bags of $v$ and the descendants of $v$ in $T$). Then, $m$ is relevant.
%$V(P(m))\setminus\gamma(v),V(P(m))\setminus(R\setminus\gamma(v))\neq\emptyset$. Then, $m$ is relevant.
\end{lemma}

\begin{proof}
Since the underlying graph of $P(m)$ is connected, Proposition \ref{lem:tree} implies that $T_{V(P(m))}=\{v\in V(T): V(P(m))\cap\beta(v)\neq\emptyset\}$ induces a tree. Thus, since $V(P(m))\setminus\gamma(v)$ and $V(P(m))\cap \gamma(v)\neq\emptyset$, it holds that $m$ is relevant.
\end{proof}

By Lemma \ref{lem:fullClassify}, the next definition indeed assigns a type to each man in $M^\star$.

\begin{definition}\label{def:manType}
Given a state $(v,R')$ and a man $m\in M^\star$, the {\em type of $m$ with respect to $(v,R')$} is defined as follows.
\begin{enumerate}
\itemsep0em 
\item\label{item:manType1} If either {\bf (i)} $V(P(m))\subseteq \gamma(v)$, or {\bf (ii)} $m$ is relevant and $V(P_{v,R'}(m))\subseteq\gamma(v)$, then $m$ is {\em settled with respect to $(v,R')$}.
\item\label{item:manType2} Otherwise, either {\bf (i)} $V(P(m))\cap \gamma(v)=\emptyset$, or {\bf (ii)} $m$ is relevant and $V(P_{v,R'}(m))\setminus\gamma(v)\neq\emptyset$. Then, $m$ is {\em unsettled with respect to $(v,R')$}.

%$\emptyset\neq V(P(m))\subseteq R\setminus\gamma(v)$, or {\bf (ii)} $m$ is relevant and $V(P_{v,R'}(m))\setminus\gamma(v)\neq\emptyset$. Then, $m$ is {\em unsettled with respect to $(v,R')$}.
\end{enumerate}
\end{definition}

For example, let us consider Fig.~\ref{fig:FPT2}. The depicted man is relevant, and supposing that all vertices drawn below the bag indeed belong to bags assigned to descendants of $v$, the man is settled only in the three illustrations on the left of the figure. We proceed to verify that Definition \ref{def:manType} is consistent in the sense that each man in $M^\star$ has a unique type.

\begin{lemma}\label{lem:consistClassify}
Given a state $(v,R')$ and a man $m\in M^\star$, the type of $m$ with respect to $(v, R')$ cannot be both settled and unsettled.
\end{lemma}

\begin{proof}
Clearly, Conditions \ref{item:manType1}(i) and \ref{item:manType2}(i) cannot hold simultaneously, and Conditions \ref{item:manType1}(ii) and \ref{item:manType2}(ii) cannot hold simultaneously. Since $V(P_{v,R'}(m))\subseteq V(P(m))$, Conditions \ref{item:manType1}(i) and \ref{item:manType2}(ii) cannot hold simultaneously. If Condition \ref{item:manType2}(i) holds, then $m$ is not relevant; thus, Conditions \ref{item:manType2}(i) and \ref{item:manType1}(ii) cannot hold simultaneously.
\end{proof}

To prove that our algorithm is correct, it will be useful to view Condition \ref{item:manType2} in Definition \ref{def:manType} through the prism of the following result.

\begin{lemma}\label{lem:helpClassify}
Let $(v,R')$ be a state, and let $m\in M^\star$ be relevant. If it does not hold that $V(P_{v,R'}(m))\subseteq\gamma(v)$, then it holds that $\emptyset\neq V(P_{v,R'}(m))\cap\gamma(v)\subseteq\{\ell_{v,R'}(m),f_{v,R'}(m)\}$.
\end{lemma}

\begin{proof}
Assume that it does not hold that $V(P_{v,R'}(m))\subseteq\gamma(v)$. Let $\widehat{P}$ denote the subpath of $P_{v,R'}(m)$ that excludes $\ell_{v,R'}(m)$ and $f_{v,R'}(m)$. By Definition~\ref{def:lf}, $V(\widehat{P})\cap\beta(v)=\emptyset$. Since the underlying graph of $\widehat{P}$ is connected, Proposition \ref{lem:tree} implies that $T_{V(\widehat{P})}=\{v\in V(T): V(\widehat{P})\cap\beta(v)\neq\emptyset\}$ induces a tree. Thus, either $V(\widehat{P})\subseteq \gamma(v)$ or $\emptyset\neq V(\widehat{P})\subseteq R\setminus\gamma(v)$. Since it does not hold that $V(P_{v,R'}(m))\subseteq\gamma(v)$, we conclude that $\emptyset\neq V(\widehat{P})\subseteq R\setminus\gamma(v)$, which implies that $\emptyset\neq V(P_{v,R'}(m))\cap\gamma(v)\subseteq\{\ell_{v,R'}(m),f_{v,R'}(m)\}$.
\end{proof}

%We proceed to verify that Definition \ref{def:manType} is consistent in the sense that each man in $M^\star$ has a unique type.
%\begin{lemma}\label{lem:consistClassify}
%Given a state $(v,R')$ and a man $m\in M^\star$, the type of $m$ with respect to $(v, R')$ cannot be both settled and unsettled.
%\end{lemma}
%
%\begin{proof}
%Clearly, Conditions \ref{item:manType1}(i) and \ref{item:manType2}(i) cannot hold simultaneously, and Conditions \ref{item:manType1}(ii) and \ref{item:manType2}(ii) cannot hold simultaneously. Since $V(P_{v,R'}(m))\subseteq V(P(m))$, Conditions \ref{item:manType1}(i) and \ref{item:manType2}(ii) cannot hold simultaneously. If Condition \ref{item:manType2}(i) holds, then $m$ is not relevant; thus, Conditions \ref{item:manType2}(i) and \ref{item:manType1}(ii) cannot hold simultaneously.
%\end{proof}

Having identified the type of a man, we are able to extract its partner with respect to a stable matching, given the information available when handling some specific state.

\begin{definition}\label{def:manPartner}
Let $(v,R')$ be a state, $\mu\in{\cal S}(v,R')$, and $m\in M^\star$. Then, the {\em partner of $m$ with respect to $(v,R',\mu)$}, denoted by $\partner_{v,R',\mu}(m)$, is defined as follows.
\begin{enumerate}
\itemsep0em 
\item If $m$ is settled, then $\partner_{v,R',\mu}(m)=\mu(m)$. (By Proposition \ref{lem:matchSame}, $\mu(m)$ is well defined.)
\item Otherwise, $m$ is unsettled. Then, $\partner_{v,R',\mu}(m)=\mu_{R'}(m)$.
\end{enumerate}
\end{definition}

Note that the partner of a man $m$ is not determined solely by $\mu$, but it depends on $(v,R')$. Intuitively, as our partial solutions will not be ``completely known'', it makes sense that some information should be extracted from $(v,R')$. We remark that the manner in which we extract this information, and in particular, the precise definition of $\partner_{v,R',\mu}(m)$, is a crucial element in the proof of  correctness of our algorithm.

We are now ready to define how to ``measure the quality of our partial solutions'' (since our partial solutions are not completely known, we cannot use Definition \ref{def:sexEqMeasure}). Observe that the score $\beta$ does not necessarily reflect an assignment of at most {\em one} man to every woman.

\begin{definition}\label{def:deltaState}
Given a state $(v,R')$, $\mu\in{\cal S}(v,R')$ and $m\in M^\star$, define
\[\displaystyle{\alpha_{v,R'}(\mu,m) = \pos_m(\partner_{v,R',\mu}(m))},\ \mathrm{and}\ \displaystyle{\lambda_{v,R'}(\mu,m) = \pos_{\partner_{v,R',\mu}(m)}(m)}.\]
Accordingly, given a state $(v,R')$ and $\mu\in{\cal S}(v,R')$, define
\[\displaystyle{\alpha_{v,R'}(\mu) = \sum_{m\in M^\star}\alpha_{v,R'}(\mu,m)},\ \mathrm{and}\ \displaystyle{\lambda_{v,R'}(\mu) = \sum_{m\in M^\star}\lambda_{v,R'}(\mu,m)}.\]
\end{definition}

Definition \ref{def:deltaState} directly implies the following observation.

\begin{observation}\label{obs:focusPartner}
Let $(v,R')$ and $(u,\widehat{R})$ be states, $\mu\in{\cal S}(v,R')$ and $\widehat{\mu}\in{\cal S}(u,\widehat{R})$. If for all $m\in M^\star$, it holds that $\partner_{v,R',\mu}(m)=\partner_{u,\widehat{R},\widehat{\mu}}(m)$, then $\alpha_{v,R'}(\mu)=\alpha_{u,\widehat{R}}(\widehat{\mu})$ and $\lambda_{v,R'}(\mu)=\lambda_{u,\widehat{R}}(\widehat{\mu})$.
\end{observation}

Moreover, by Observation \ref{obs:focusPartner} and Definitions \ref{def:deltaState} and \ref{def:manPartner}, we have the following observation.

\begin{observation}\label{obs:focusSettle}
Let $(v,R')$ and $(u,\widehat{R})$ be states, and $\mu\in{\cal S}(v,R')\cap{\cal S}(u,\widehat{R})$ such that each $m\in M^\star$ has the same type with respect to $(v,R')$ and $(u,\widehat{R})$, and if $m$ is unsettled then $\mu_{R'}(m)=\mu_{\widehat{R}}(m)$. Then, $\alpha_{v,R'}(\mu)=\alpha_{u,\widehat{R}}(\widehat{\mu})$ and $\lambda_{v,R'}(\mu)=\lambda_{u,\widehat{R}}(\widehat{\mu})$.
\end{observation}

Given a state $(v,R')$ and $t\in\{-n^2,-n^2+1,\ldots,n^2-1,n^2\}$, denote ${\cal S}(v,R',t_M,t_W)=\{\mu\in{\cal S}(v,R'): \alpha_{v,R'}(\mu)=t_M, \lambda_{v,R'}(\mu)=t_W\}$.

\subsubsection{The DP Table}\label{sec:dpTable}

We let {\sf N} denote our DP table. Each entry of the table {\sf N} is of the form {\sf N}$[v,R',t_M,t_W]$, where $v\in V(T)$, $R'\subseteq\beta(v)$ and $t_M,T_W\in [n^2]$. The following definition addresses the purpose of these entries.

\begin{definition}\label{def:computeCorrectly}
We say that {\sf N} is {\em computed correctly} if for each entry {\sf N}$[v,R',t_M,t_W]$, it holds that 
{\sf N}$[v,R',t_M,t_W]\in\{0,1\}$, and {\sf N}$[v,R',t_M,t_W]=1$ if and only if ${\cal S}(v,R',t_M,t_W)\neq\emptyset$.
\end{definition}

Next, we prove why it is sufficient to compute {\sf N} correctly.

\begin{lemma}\label{lem:boundTarget}
For any $\mu\in{\cal S}$, it holds that $\sat_M(\mu),\sat_W(\mu)\in[n^2]$.
\end{lemma}

\begin{proof}
For all $m\in M^\star$ and $w\in\domain(\pos_m)$, it holds that $1\leq\pos_m(w)\leq n$, and for all $w\in W^\star$ and $m\in\domain(\pos_w)$, it holds that $1\leq\pos_w(m)\leq n$. Thus, by the definitions of $\sat_M$ and $\sat_W$, we have that $\sat_M(\mu),\sat_W(\mu)\in[n^2]$.
\end{proof}

\begin{lemma}\label{lem:Delta}
Suppose that {\sf N} is computed correctly. For all $t_M,t_W\in[n^2]$, there exists $\mu\in{\cal S}$ such that $\sat_M(\mu)=t_M$ and $\sat_W(\mu)=t_W$ if and only if {\sf N}$[\rootT(T),\emptyset,t_M,t_W]=1$.
\end{lemma}

\begin{proof}
Since $\beta(\rootT(T))=\emptyset$, it holds that ${\cal S}={\cal S}(\rootT(T),$ $\emptyset)$. Furthermore, since $\gamma(\rootT(T))=R$, it holds that every $m\in M^\star$ is settled with respect to $(\rootT(T),\emptyset)$. Therefore, every $\mu\in{\cal S}$ satisfies $\alpha_{\rootT(T),\emptyset}(\mu)=\sat_M(\mu)$ and $\lambda_{\rootT(T),\emptyset}(\mu)=\sat_W(\mu)$. We thus have that for all $t_M,t_W\in[n^2]$, it holds that ${\cal S}(\rootT(T),\emptyset,t_M,t_W)\neq\emptyset$ if and only if there exists $\mu\in{\cal S}$ such that $\sat_M(\mu)=t_M$ and $\sat_W(\mu)=t_W$. Since {\sf N} is computed correctly, we conclude the proof.
\end{proof}

In light of Lemmata \ref{lem:boundTarget} and \ref{lem:Delta}, to prove that {\sc GSM} is solvable in time $\OO(2^{tw}\cdot n^{10})$, it is sufficient that we show that it is possible to compute {\sf N} correctly in time $\OO(2^{\mathrm{\it tw}}\cdot n^{10})$. The rest of Section \ref{sec:fptSESM} focuses on the proof of this claim.

\subsubsection{Computation}\label{sec:dpComputation}

We process the entries of {\sf N} by traversing the tree $T$ in post-order. The order in which we process entries corresponding to the same node $v\in V(T)$ is arbitrary. Thus, the basis corresponds to entries {\sf N}$[v,R',t_M,t_W]$ where $v$ is a leaf, and the steps correspond to entries where $v$ is a forget node, an introduce node or a join node.

\medskip
\myparagraph{Leaf Node.} In the basis, where $v$ is a leaf, we have that $\beta(v)=\emptyset$. We consider two cases.
\begin{enumerate}
\itemsep0em 
\item If $t_M=\alpha_{v,\emptyset}(\mu_{\emptyset})$ and $t_W=\lambda_{v,\emptyset}(\mu_{\emptyset})$: {\sf N}$[v,\emptyset,t_M,t_W]=1$.
\item Otherwise: {\sf N}$[v,\emptyset,t_M,t_W]=0$.
\end{enumerate}

% \noindent The computations of entries {\sf N}$[v,R',t]$ where $v$ is a forget node, an introduce node or a join node are given below. 

%\medskip
\myparagraph{Forget Node.} Let $u$ denote the child of $v$ in $T$, and $\rho$ denote the vertex in $\beta(u)\setminus\beta(v)$. Then, {\sf N}$[v,R',t_M,t_W]=\max\{\mathrm{\sf N}[u,R',t_M,t_W],\mathrm{\sf N}[u,R'\cup\{\rho\},t_M,t_W]\}$.

\medskip
\myparagraph{Introduce Node.} Let $u$ be the child of $v$ in $T$, and $\rho$ be the vertex in $\beta(v)\setminus\beta(u)$. Consider the following cases.
\begin{enumerate}
\item If $(\cl(R')\cap\beta(v))\setminus R'\neq\emptyset$: {\sf N}$[v,R',t_M,t_W]=0$.
\item\label{case:introEasy} Else if $\rho\notin R'$: $M_\rho=\emptyset$; {\sf N}$[v,R',t_M,t_W]=\mathrm{\sf N}[u,R',t_M,t_W]$.
\item\label{case:introSep} Otherwise ($\rho\in R'$): Given $m\in M^\star$, let $\rho_m$ denote the last vertex on $P(m)$ that belongs to $\cl(\rho)$, where if such a vertex does not exist, $\rho_m=\nil$.
Denote $M_\rho=\{m\in M^\star: \rho_m\in R\setminus\gamma(u),\ell_{v,R'}(m)\in\cl(\rho_m)\cup\{\nil\}\}$ (see Fig.~\ref{fig:FPT3}). Roughly speaking, $M_\rho$ is the set of men for which we identified a partner that we need to replace at the computation of the current entry. For each $m\in M_\rho$, denote $\diff^\alpha_{v,R'}(m)=\alpha_{v,R'}(\mu_{R'},m)-\alpha_{u,R'\setminus\{\rho\}}(\mu_{R'\setminus\{\rho\}},m)$. Accordingly, denote $\diff^\alpha(M_\rho)=\sum_{m\in M_\rho}\diff^\alpha_{v,R'}(m)$. Symmetrically, for $m\in M_\rho$, denote $\diff^\lambda_{v,R'}(m)=\lambda_{v,R'}(\mu_{R'},m)-\lambda_{u,R'\setminus\{\rho\}}(\mu_{R'\setminus\{\rho\}},m)$. Accordingly, denote $\diff^\lambda(M_\rho)=\sum_{m\in M_\rho}\diff^\lambda_{v,R'}(m)$. Now, {\sf N}$[v,R',t_M,t_W]$ is computed as follows.
\vspace{-0.5em}
\[\mathrm{\sf N}[v,R',t_M,t_W] = \mathrm{\sf N}[u,R'\setminus\{\rho\},t_M-\diff^\alpha(M_\rho),t_W-\diff^\lambda(M_\rho)].\]
\end{enumerate}
%We are now left with Case \ref{case:introSep}. Denote $M_\rho=\{m\in M^\star: m$ is relevant, $\rho=\ell_{v,R'}(m)\}$. For each $m\in M_\rho$, denote $\diff_{v,R'}(m)=\delta_{v,R'}(m)-\delta_{v,R'\setminus\{\rho\}}(m)$. Accordingly, denote $\diff(M_\rho)=\sum_{m\in M_\rho}\diff(m)$. Now, {\sf N}$[v,R',t]$ is computed as follows.
%\[\mathrm{\sf N}[v,R',t] = \mathrm{\sf N}[u,R'\setminus\{\rho\},t+\diff(M_\rho)].\]

%\smallskip
\myparagraph{Join Node.} Let $u$ and $w$ denote the children of $v$ in $T$. For the sake of efficiency, we compute all entries of the form {\sf N}$[v,R',\cdot,\cdot]$ simultaneously. First, we initialize each such entry to 0. Now, we compute $T(u,R')=\{(\widehat{t}_M,\widehat{t}_W)\in[n^2]\times[n^2]: \mathrm{\sf N}[u,R',\widehat{t}_M,\widehat{t}_W]=1\}$ and $T(w,R')=\{(\widehat{t}_M,\widehat{t}_W)\in[n^2]\times[n^2]: \mathrm{\sf N}[w,R',\widehat{t}_M,\widehat{t}_W]=1\}$. Then, for all $(\widehat{t}_M,\widehat{t}_W)\in T(u,R')$ and $(t^*_M,t^*_W)\in T(w,R')$ such that $t_M=\widehat{t}_M+t^*_M-\alpha_{v,R'}(\mu_{R'})\in[n^2]$ and $t_W=\widehat{t}_W+t^*_W-\lambda_{v,R'}(\mu_{R'})\in[n^2]$, we set $\mathrm{\sf N}[v,R',t_M,t_W]=1$. (An entry may be set to $1$ multiple times.)

\begin{figure}[t!]\centering
\fbox{\includegraphics[scale=0.7]{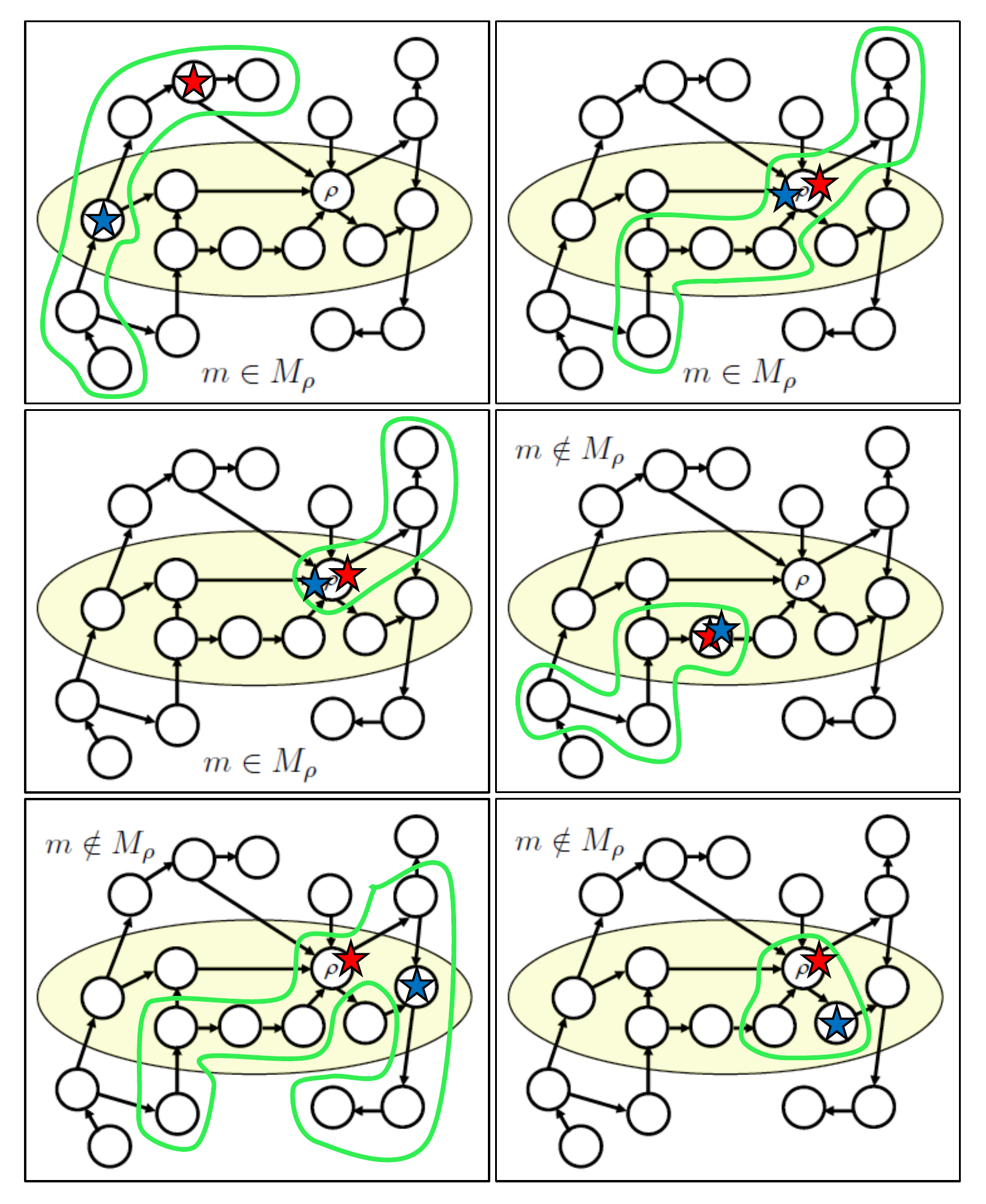}}
\caption{Illustrations of $M_\rho$. The yellow shape captures the bag $\beta(v)$, and it is assumed that $R'=\beta(v)$. In each illustration, a directed path $P(m)$ of a man $m$ is highlighted in green, the vertices $\rho_m$ and $\ell_{v,R'}(m)$ are markeed by red and blue stars, respectively, and it is specified whether $m\in M_\rho$.}\label{fig:FPT3}
\end{figure}

\subsubsection{Time Complexity}\label{sec:time}

The table {\sf N} contains at most $|V(T)|\cdot 2^{\mathrm{\it tw}}\cdot n^4 = \OO(2^{\mathrm{\it tw}}\cdot n^6)$ entries. The computation of an entry {\sf N}$[v,R',t_M,t_W]$ where $v$ is a leaf node, a forget node or an introduce node, is performed in time $\OO(1)$. Moreover, for every join node $v\in V(T)$ and $R'\subseteq R$, the total time to compute all $\OO(n^4)$ entries of the form {\sf N}$[v,R',\cdot,\cdot]$ is bounded by $\OO(n^8)$. Therefore, the overall running time of our algorithm is bounded by $\OO(2^{\mathrm{\it tw}}\cdot n^{10})$.

\subsubsection{Correctness}

In this subsection, we prove that {\sf N} has been computed correctly (see Definition~\ref{def:computeCorrectly}). By our discussions in Sections~\ref{sec:dpTable} and~\ref{sec:time}, this proof concludes that {\sc GSM} is solvable in time $\OO(2^{tw}\cdot n^{10})$. The proof is by induction on the order of the computation.

\medskip
\myparagraph{Basis.}
For the basis of the induction, consider an entry {\sf N}$[v,R',t_M,t_W]$ such that $v$ is a leaf in $T$. In this case, $\beta(v)=\emptyset$ and therefore $R'=\emptyset$. Then, the following lemma holds.

\begin{lemma}\label{lem:basis}
$\mu_{\emptyset}\in {\cal S}(v,\emptyset)$, and for all $\mu\in{\cal S}(v,\emptyset)$, $\alpha_{v,\emptyset}(\mu)=\alpha_{v,\emptyset}(\mu_{\emptyset})$ and $\lambda_{v,\emptyset}(\mu)=\lambda_{v,\emptyset}(\mu_{\emptyset})$.
\end{lemma}

\begin{proof}
Since $R(\mu_{\emptyset})=\beta(v)=\emptyset$, it holds that $\mu_{\emptyset}\in {\cal S}(v,\emptyset)$. Moreover, since $\gamma(v)=\emptyset$, every $m\in M^\star$ is either unsettled with respect to $(v,\emptyset)$ or satisfies $V(P(m))=\emptyset$. Thus, for all $\mu\in{\cal S}(v,\emptyset)$ and $m\in M^\star$, it holds that $\partner_{v,\emptyset,\mu}(m)=\mu_\emptyset(m)$. By Observation \ref{obs:focusPartner}, we have that for all $\mu\in{\cal S}(v,\emptyset)$, it holds that $\lambda_{v,\emptyset}(\mu)=\lambda_{v,\emptyset}(\mu_{\emptyset})$.
\end{proof}

By Lemma \ref{lem:basis}, we have that for all $t_M,t_W\in[n^2]$, it holds that ${\cal S}(v,R',t_M,t_W)\neq\emptyset$ if and only if $t_M=\alpha_{v,\emptyset}(\mu_{\emptyset})$ and $t_W=\lambda_{v,\emptyset}(\mu_{\emptyset})$. We thus conclude that the basis is correct.

\medskip
\myparagraph{Inductive Hypothesis.}
Next, we consider some entry {\sf N}$[v,R',t_M,t_W]$ such that $v$ is not a leaf in $T$, and we assume that for every child $u$ of $v$, all entries of the form {\sf N}$[u,\cdot,\cdot,\cdot]$ have been computed correctly (see Definition~\ref{def:computeCorrectly}). We analyze the cases where $v$ is a forget node, an introduce node and a join node separately~below.

\medskip
\myparagraph{Forget Node.}
By the inductive hypothesis, to prove that our computation is correct, it is sufficient to show that ${\cal S}(v,R',t_M,t_W)\neq\emptyset$ if and only if ${\cal S}(u,R',t_M,t_W)\cup{\cal S}(u,R'\cup\{\rho\},t_M,t_W)\neq\emptyset$. We start with the following claim.

\begin{lemma}\label{lem:forgetType}
Let $m\in M^\star$. Then, the type of $m$ is the same with respect to both $(v,R')$ and $(u,R')$. Moreover, either the type of $m$ is the same with respect to both $(v,R')$ and $(u,R'\cup\{\rho\})$ or $(\cl(\rho)\cap\beta(v))\setminus R'\neq\emptyset$.
\end{lemma}

\begin{proof}
Since $\gamma(v)=\gamma(u)$, if $V(P(m))\subseteq \gamma(v)$, then $m$ is settled with respect to $(v,R'),(u,R')$ and $(u,R'\cup\{\rho\})$. Moreover, if $V(P(m)) \cap\gamma(v)=\emptyset$, then $m$ is unsettled with respect to  $(v,R'),(u,R')$ and $(u,R'\cup\{\rho\})$.

%$\emptyset\neq V(P(m))\subseteq R\setminus\gamma(v)$, then $m$ is unsettled with respect to  $(v,R'),(u,R')$ and $(u,R'\cup\{\rho\})$.

Next, suppose that $m$ is relevant and $V(P_{v,R'}(m))\subseteq\gamma(v)$. We first show that the type of $m$ is the same with respect to both $(v,R')$ and $(u,R')$. Since $\beta(u)=\beta(v)\cup\{\rho\}$, we have that $\ell_{u,R'}(m)=\ell_{v,R'}(m)$. Furthermore, if $f_{u,R'}(m)\neq f_{v,R'}(m)$, then either $f_{v,R'}(m)=\nil$ or $f_{u,R'}(m)=\rho$ is located before $f_{v,R'}(m)$ on $P(m)$. Therefore, we necessarily have that $V(P_{u,R'}(m))\subseteq V(P_{v,R'}(m))$. Thus, since $V(P_{v,R'}(m))\subseteq\gamma(v)$ and $\gamma(v)=\gamma(u)$, we have that $V(P_{u,R'}(m))\subseteq\gamma(u)$. This implies $m$ is settled with respect to $(v,R')$ and $(u,R')$. Now, suppose that $\cl(\rho)\cap\beta(v)\subseteq  R'$. We show that the type of $m$ is the same with respect to $(v,R')$ and $(u,R'\cup\{\rho\})$. Since $\beta(u)=\beta(v)\cup\{\rho\}$, if $\ell_{u,R'\cup\{\rho\}}(m)\neq \ell_{v,R'}(m)$, then either $\ell_{v,R'}(m)=\nil$ or $\ell_{u,R'\cup\{\rho\}}(m)=\rho$ is located after $\ell_{v,R'}(m)$ on $P(m)$. Since $\cl(\rho)\cap\beta(u)\subseteq  R'\cup\{\rho\}$, $f_{u,R'\cup\{\rho\}}(m)=f_{v,R'}(m)$. Therefore, we necessarily have that $V(P_{u,R'\cup\{\rho\}}(m))\subseteq V(P_{v,R'}(m))$. Since $V(P_{v,R'}(m))\subseteq\gamma(v)$ and $\gamma(v)=\gamma(u)$, this implies that $V(P_{u,R'\cup\{\rho\}}(m))\subseteq\gamma(u)$. Thus, $m$ is settled with respect to $(v,R')$ and $(u,R'\cup\{\rho\})$. 

Finally, suppose that $m$ is relevant and $V(P_{v,R'}(m))\setminus\gamma(v)\neq\emptyset$. By Lemma \ref{lem:helpClassify}, we have that $V(P_{v,R'}(m))\cap\gamma(v)\subseteq\{\ell_{v,R'}(m),f_{v,R'}(m)\}$, which implies that $\rho\notin V(P_{v,R'}(m))$.
Therefore, $\ell_{u,R'}(m)=\ell_{v,R'}(m)$ and $f_{u,R'}(m)=f_{v,R'}(m)$, which implies that $P_{u,R'}(m)=P_{v,R'}(m)$. Since $\gamma(v)=\gamma(u)$, we deduce that $V(P_{u,R'}(m))\setminus\gamma(u)\neq\emptyset$, and therefore $m$ is unsettled with respect to both $(v,R')$ and $(u,R')$. Now, suppose that $\cl(\rho)\cap\beta(v)\subseteq  R'$. Then, since $V(P_{v,R'}(m))\cap\gamma(v)\subseteq\{\ell_{v,R'}(m),f_{v,R'}(m)\}$, we have that $\ell_{u,R'\cup\{\rho\}}(m)=\ell_{v,R'}(m)$ and $f_{u,R'\cup\{\rho\}}(m)=f_{v,R'}(m)$. Again, this leads to the conclusion that $m$ is unsettled with respect to $(u,R'\cup\{\rho\})$.
\end{proof}

First, suppose that ${\cal S}(v,R',t_M,t_W)\neq\emptyset$. Then, there exists $\mu\in{\cal S}(v,R',t_M,t_W)$. The conclusion that ${\cal S}(u,R',t_M,t_W)\cup{\cal S}(u,R'\cup\{\rho\},t_M,t_W)\neq\emptyset$ is implied by the following lemma.

\begin{lemma}
Let $\mu\in{\cal S}(v,R',t_M,t_W)$. If $\rho\notin R(\mu)$, then $\mu\in{\cal S}(u,R',t_M,t_W)$, and otherwise $\mu\in{\cal S}(u,R'\cup\{\rho\},t_M,t_W)$.
\end{lemma}

\begin{proof}
Assume that $\rho\notin R(\mu)$. Since $\mu\in{\cal S}(u,R')$, it holds that $R(\mu)\cap\beta(v)=R'$. Thus, since $\beta(u)=\beta(v)\cup\{\rho\}$ and $\rho\notin R(\mu)$, we have that $R(\mu)\cap\beta(u)=R'$. Therefore, $\mu\in{\cal S}(u,R')$. Next, assume that $\rho\in R(\mu)$. Since $\mu\in{\cal S}(u,R')$, it holds that $R(\mu)\cap\beta(v)=R'$ and $\cl(\rho)\cap\beta(v)\subseteq R'$. Thus, since $\beta(u)=\beta(v)\cup\{\rho\}$ and $\rho\in R(\mu)$, we have that $R(\mu)\cap\beta(u)=R'\cup\{\rho\}$.

It remains to show that if $\rho\notin R(\mu)$, then $\alpha_{u,R'}(\mu)=t_M$ and $\lambda_{u,R'}(\mu)=t_W$, and otherwise $\alpha_{u,R'\cup\{\rho\}}(\mu)=t_M$ and $\lambda_{u,R'\cup\{\rho\}}(\mu)=t_W$.
Note that $\alpha_{v,R'}(\mu)=t_M$ and $\lambda_{v,R'}(\mu)=t_W$. Moreover, if $\rho\in R(\mu)$ then for any $m\in M^\star$ unsettled with respect to $(v,R')$ and $(u,R'\cup\{\rho\})$, it can be verified that $\cl(\rho)\cap V(P(m))\subseteq \cl(R')\cap V(P(m))$, in which case $\mu_{R'}(m)=\mu_{R'\cup\{\rho\}}(m)$. Thus, by Observation \ref{obs:focusSettle}, to prove the desired claim it is sufficient to show that for all $m\in M^\star$, it holds that $m$ has the same type with respect to $(v,R')$ and $(u,R')$, and if $\rho\in R(\mu)$, the $m$ also has the same type with respect to $(v,R')$ and $(u,R'\cup\{\rho\})$. The correctness of this statement is given by Lemma \ref{lem:forgetType}.
\end{proof}

Now, for the other direction; suppose that ${\cal S}(u,R',t_M,t_W)\cup{\cal S}(u,R'\cup\{\rho\},t_M,t_W)\neq\emptyset$. Then, there exists $\mu\in{\cal S}(u,R',t_M,t_W)\cup{\cal S}(u,R'\cup\{\rho\},t_M,t_W)$. The conclusion that ${\cal S}(v,R',t_M,t_W)\neq\emptyset$ is implied by the following lemma.

\begin{lemma}
Let $\mu\in{\cal S}(u,R',t_M,t_W)\cup{\cal S}(u,R'\cup\{\rho\},t_M,t_W)$. Then, $\mu\in{\cal S}(v,R',t_M,t_W)$.
\end{lemma}

\begin{proof}
First, assume that $\mu\in{\cal S}(u,R',t_M,t_W)$. Then, $R(\mu)\cap\beta(u)=R'$. Thus, since $\beta(v)=\beta(u)\setminus\{\rho\}$, $R(\mu)\cap\beta(v)=R'$. Therefore, $\mu\in{\cal S}(v,R')$. Next, assume that $\mu\in{\cal S}(u,R'\cup\{\rho\},t)$. Then, $R(\mu)\cap\beta(u)=R'\cup\{\rho\}$ and $\cl(\rho)\cap\beta(u)\subseteq R'\cup\{\rho\}$. Thus, since $\beta(v)=\beta(u)\setminus\{\rho\}$, $R(\mu)\cap\beta(v)=R'$ and $\cl(\rho)\cap\beta(v)\subseteq R'$. Therefore, $\mu\in{\cal S}(v,R')$.

It remains to show that $\alpha_{v,R'}(\mu)=t_M$ and $\lambda_{v,R'}(\mu)=t_W$. Note that if $\mu\in{\cal S}(u,R',t_M,t_W)$ then $\alpha_{u,R'}(\mu)=t_M$ and $\lambda_{u,R'}(\mu)=t_W$, and otherwise $\alpha_{u,R'\cup\{\rho\}}(\mu)=t_M$ and $\lambda_{u,R'\cup\{\rho\}}(\mu)=t_W$. Again, by Observation \ref{obs:focusSettle}, to prove the desired claim it is sufficient to show that for all $m\in M^\star$, it holds that $m$ has the same type with respect to $(v,R')$ and $(u,R')$, and if $\mu\in{\cal S}(u,R'\cup\{\rho\},t_M,t_W)$, then $m$ also has the same type with respect to $(v,R')$ and $(u,R'\cup\{\rho\})$. The correctness of this statement is given by Lemma \ref{lem:forgetType}.
\end{proof}

\medskip
\myparagraph{Introduce Node.}
First, note that if $(\cl(R')\cap\beta(v))\setminus R'\neq\emptyset$, then by Proposition \ref{lem:smCorrRots}, it holds that ${\cal S}(v,R')=\emptyset$. Thus, in this case, the computation is correct. We next assume that it holds that $\cl(R')\cap\beta(v)\subseteq R'$. In particular, this implies that if $\rho\notin R'$ then $\rho\notin\cl(R')$, and if $\rho\in R'$ then $\cl(\rho)\cap\beta(v)\subseteq R'$.
 
We proceed to analyze Cases \ref{case:introEasy} and \ref{case:introSep}.

\begin{lemma}\label{lem:introSubset}
${\cal S}(v,R')\subseteq {\cal S}(u,R'\setminus\{\rho\})$.
\end{lemma}

\begin{proof}
Let $\mu\in{\cal S}(v,R')$. Then, $R(\mu)\cap\beta(v)=R'$. Since $\beta(u)=\beta(v)\setminus\{\rho\}$, we have that $R(\mu)\cap\beta(u)=R'\setminus\{\rho\}$. Thus, ${\cal S}(v,R')\subseteq {\cal S}(u,R'\setminus\{\rho\})$.
\end{proof}

We next show that we can add or remove $\rho$ from the closed set $R(\mu)$ of rotations of some $\mu\in{\cal S}(u,R'\setminus\{\rho\},t_M,t_W)$ and yet obtain another stable matching in ${\cal S}(u,R'\setminus\{\rho\},t_M,t_W)$.

\begin{lemma}\label{lem:introAlmostSubset}
If ${\cal S}(u,R'\setminus\{\rho\},t_M,t_W)\neq\emptyset$, we have that
\begin{itemize}
\itemsep0em 
\item if $\rho\notin R'$, then there exists $\mu\in{\cal S}(u,R',t_M,t_W)$ such that $\rho\notin R(\mu)$, and
\item if $\rho\in R'$, then there exists $\mu\in{\cal S}(u,R'\setminus\{\rho\},t_M,t_W)$ such that $\rho\in R(\mu)$.
\end{itemize}
\end{lemma}

\begin{proof}
Assume that ${\cal S}(u,R'\setminus\{\rho\},t_M,t_W)\neq\emptyset$. Then, there exists $\mu\in {\cal S}(u,R'\setminus\{\rho\},t_M,t_W)$. First, suppose that $\rho\notin R'$. Denote $\widehat{R}=\{\widehat{\rho}\in R: \rho\in\cl(\widehat{\rho})\}$. Then, $R^*=R(\mu)\setminus\widehat{R}$ is a closed set. Thus, by Proposition \ref{lem:smCorrRots}, there exists $\mu^*\in{\cal S}$ such that $R(\mu^*)=R^*$. Since $R(\mu)\cap\beta(u)=R'$, $\beta(u)=\beta(v)\setminus\{\rho\}$, $\cl(R')\cap\beta(v)\subseteq R'$ and $\rho\notin R'$, we have that $R^*\cap\beta(u)=R'$, and therefore $\mu^*\in{\cal S}(u,R')$. It remains to show that $\alpha_{u,R'}(\mu^*)=\alpha_{u,R'}(\mu)$ and $\lambda_{u,R'}(\mu^*)=\lambda_{u,R'}(\mu)$. Note that the type of a man is defined with respect to a state. Thus, by Definition \ref{def:manPartner} and Observation \ref{obs:focusPartner}, it remains to show that for every settled man $m$, it holds that $\mu^*(m)=\mu(m)$. Let $m$ be a settled man. Then, the last vertex on $P(m)$ that belongs to $R(\mu)$, denoted by $\rho'$, also belongs to $\gamma(u)$. To show that $\mu^*(m)=\mu(m)$, it is sufficient to show that $\rho'\in R^*$. To this end, it is sufficient to show that $\rho\notin\cl(\rho')$. Suppose, by way of contradiction, that this claim is false. Then, there exists a directed path in $D_\Pi$ from $\rho$ to $\rho'$. Since $\rho\notin\gamma(u)$ and $\rho'\in\gamma(u)$, Proposition \ref{lem:tree} implies that this path contains a vertex $\widehat{\rho}\in\beta(u)$. Since $\mu\in{\cal S}(u,R')$, we have that $\widehat{\rho}\in R'$. This contradicts the facts that $\cl(R')\cap\beta(v)\subseteq R'$ and $\rho\notin R'$.

Next, suppose that $\rho\in R'$. Moreover, suppose that $\rho\notin R(\mu)$, else we are done. Note that $R^*=R(\mu)\cup\cl(\rho)$ is a closed set. Thus, by Proposition \ref{lem:smCorrRots}, there exists $\mu^*\in{\cal S}$ such that $R(\mu^*)=R^*$. Since $R(\mu)\cap\beta(u)=R'$, $\beta(u)=\beta(v)\setminus\{\rho\}$ and $\cl(\rho)\cap\beta(v)\subseteq R'$, we have that $R^*\cap\beta(u)=R'\setminus\{\rho\}$, and therefore $\mu^*\in{\cal S}(u,R'\setminus\{\rho\})$. It remains to show that $\alpha_{u,R'\setminus\{\rho\}}(\mu^*)=\alpha_{u,R'\setminus\{\rho\}}(\mu)$ and $\lambda_{u,R'\setminus\{\rho\}}(\mu^*)=\lambda_{u,R'\setminus\{\rho\}}(\mu)$. Again, for this purpose, it is sufficient to let $m$ be some settled man, and prove that $\mu^*(m)=\mu(m)$. Let $\rho'$ denote the last vertex on $P(m)$ that belongs to $R^*$. Then, $\rho'\in\gamma(u)$. We need to show that $\rho'\in R(\mu)$. To this end, it is sufficient to show that $\rho'\notin\cl(\rho)\setminus R(\mu)$. Suppose, by way of contradiction, that this claim is false. Then, there exists a directed path in $D_\Pi$ from $\rho'$ to $\rho$. Since $\rho\notin\gamma(u)$ and $\rho'\in\gamma(u)$, Proposition \ref{lem:tree} implies that this path contains a vertex $\widehat{\rho}\in\beta(u)$. Since $\mu^*\in{\cal S}(u,R')$, we have that $\widehat{\rho}\in R'$. Thus, since $\mu\in{\cal S}(u,R')$, we have that $\widehat{\rho}\in R(\mu)$. However, since $\rho'\in\cl(\widehat{\rho})$ and $R(\mu)$ is a closed set, this is impossible.
\end{proof}

We proceed to argue about equivalence between partners of men in $M^\star\setminus M_\rho$.

\begin{lemma}\label{lem:samePartnerNoRho}
Let $m\in M^\star\setminus M_\rho$ and $\mu\in {\cal S}(v,R')\cap{\cal S}(u,R'\setminus\{\rho\})$. Then, the partners of $m$ with respect to $(v,R',\mu)$ and $(u,R'\setminus\{\rho\},\mu)$ are identical.
\end{lemma}

\begin{proof}
Suppose that $m$ has different types with respect to $(v,R',\mu)$ and $(u,R'\setminus\{\rho\},\mu)$, else we are done.
Since $m\notin M_\rho$ and $\beta(v)=\beta(u)\cup\{\rho\}$, we have that $\ell_{v,R'}(m)=\ell_{u,R'\setminus\{\rho\}}(m)$. Then, $f_{v,R'}(m)\neq f_{u,R'\setminus\{\rho\}}(m)$, else $m$ has the same type with respect to $(v,R',\mu)$ and $(u,R'\setminus\{\rho\},\mu)$. In this case, it must hold that $f_{v,R'}(m)=\rho$ precedes $f_{u,R'\setminus\{\rho\}}(m)$ on $P(m)$. By Proposition \ref{lem:tree} and since $\rho\notin\gamma(u)$, we deduce that either $\rho$ is the first vertex on $P(m)$ or $\rho$ is the outgoing neighbor of $\ell_{v,R'}(m)$ on $P(m)$. Overall, this implies that $\mu(m)=\mu_{R'}(m)$. Therefore, regardless of the types of $m$ with respect to $(v,R',\mu)$ and $(u,R'\setminus\{\rho\},\mu)$, it holds that the partners of $m$ with respect to $(v,R',\mu)$ and $(u,R'\setminus\{\rho\},\mu)$ are identical.
\end{proof}

In case $\rho\notin R'$, it holds that $M_\rho=\emptyset$ and $(u,R'\setminus\{\rho\})=(u,R')$. Thus, we obtain the following corollary to Lemma~\ref{lem:samePartnerNoRho}.

\begin{corollary}\label{cor:samePartnerNoRho}
Suppose that $\rho\notin R'$, and let $m\in M^\star$ and $\mu\in {\cal S}(v,R')\cap{\cal S}(u,R')$. Then, the partners of $m$ with respect to $(v,R',\mu)$ and $(u,R',\mu)$ are identical.
\end{corollary}

Next, we prove that the computation for Case \ref{case:introEasy} (i.e when $\rho \notin R'$) is correct.

\begin{lemma}\label{lem:introEasy}
If $\rho\notin R'$, then ${\cal S}(v,R',t_M,t_W)\neq\emptyset$ if and only ${\cal S}(u,R',t_M,t_W)\neq\emptyset$.
\end{lemma}

\begin{proof}
Assume that $\rho\notin R'$. First, suppose ${\cal S}(v,R',t_M,t_W)\neq\emptyset$. Then, there exists $\mu\in{\cal S}(v,R',t_M,t_W)$. By Lemma~\ref{lem:introSubset}, we have that $\mu\in{\cal S}(u,R')$. By Corollary \ref{cor:samePartnerNoRho} and Observation \ref{obs:focusPartner}, we obtain that $\alpha_{u,R'}(\mu)=\alpha_{v,R'}(\mu)$ and $\lambda_{u,R'}(\mu)=\lambda_{v,R'}(\mu)$, and therefore $\mu\in{\cal S}(u,R',t_M,t_W)$.

Now, suppose ${\cal S}(u,R',t_M,t_W)\neq\emptyset$. Then, by Lemma \ref{lem:introAlmostSubset}, there exists $\mu\in{\cal S}(u,R',t_M,t_W)$ such that $\rho\notin R(\mu)$. Since $R(\mu)\cap\beta(u)=R'$, $\rho\notin R(\mu)\cup R'$ and $\beta(v)=\beta(u)\cup\{\rho\}$, we have that $\mu\in{\cal S}(v,R')$. By Corollary \ref{cor:samePartnerNoRho} and Observation \ref{obs:focusPartner}, we obtain that $\alpha_{v,R'}(\mu)=\alpha_{u,R'}(\mu)$ and $\lambda_{v,R'}(\mu)=\lambda_{u,R'}(\mu)$, and therefore $\mu\in{\cal S}(v,R',t_M,t_W)$.
\end{proof}

%We remark that it is not true that ${\cal S}(v,R',t)={\cal S}(u,R',t)$. However, by the inductive hypothesis, having Lemma \ref{lem:introEasy} is sufficient to deduce that the computation in Case \ref{case:introEasy} is correct.

Next, we suppose that $\rho\in R'$, and our objective is to prove that the computation in Case \ref{case:introSep} is correct. For this purpose, by the inductive hypothesis, it is sufficient to show that ${\cal S}(v,R',t_M,t_W)\neq\emptyset$ if and only ${\cal S}(u,R'\setminus\{\rho\},t_M-\diff^\alpha(M_\rho),t_W-\diff^\lambda(M_\rho))\neq\emptyset$. To proceed, we need two arguments concerning equivalence between partners of men in $M_\rho$..

\begin{lemma}\label{lem:introPartnerV}
Let $m\in M_\rho$ and $\mu\in {\cal S}(v,R')$. The partners of $m$ with respect to $(v,R',\mu)$ and $(v,R',\mu_{R'})$ are identical.
\end{lemma}

\begin{proof}
%Since $m\in M_\rho$, it holds that $\ell_{v,R'}(m)=\rho$.
If $\mu(m)=\mu_{R'}(m)$, we necessarily have that $\partner_{v,R',\mu}(m)=\partner_{v,R',\mu_{R'}}(m)$. Moreover, if $m$ is unsettled with respect to $(v,R')$, then again $\partner_{v,R',\mu}(m)=\partner_{v,R',\mu_{R'}}(m)$. Now, we suppose that $\mu(m)\neq\mu_{R'}(m)$ and $m$ is settled (with respect to $(v,R')$), and show that this supposition leads to a contradiction.
The first assumption implies that $(R(\mu)\cap V(P(m)))\setminus\cl(R')\neq\emptyset$. However, since $m$ is settled and $\mu\in{\cal S}(v,R')$, from Proposition \ref{lem:tree} we deduce that $(R(\mu)\cap V(P(m)))\setminus\cl(R')\subseteq\gamma(v)$. Since $m\in M_\rho$, we have that $\ell_{v,R'}(m)\in\cl(\rho_m)\cup\{\nil\}$, and therefore $\rho_m\in\gamma(u)\cup\{\ell_{v,R'}(m)\}$. However, since $m\in M_\rho$, it holds that $\rho_m\in R\setminus\gamma(u)$, and therefore $\rho_m=\ell_{v,R'}(m)$. Since $(R(\mu)\cap V(P(m)))\setminus\cl(R')\neq\emptyset$, $R(\mu)$ contains the outgoing neighbor $\rho'$ of $\rho$ on $P(m)$. Since $\ell_{v,R'}(m)=\rho$, we have that $\rho'\notin\beta(v)$. Since $m$ is settled, we deduce that $\rho'\in\gamma(v)\setminus\beta(v)$. However, since $v$ is an introduce node, from Property \ref{item:twconnected} in Definition \ref{def:treewidth} we conclude that there does not exist $u\in V(T)$ such that $\rho,\rho'\in\beta(u)$. This conclusion contradicts Property \ref{item:twedge} in Definition \ref{def:treewidth} (since $\rho,\rho$ are neighbors in $G_\Pi$).
\end{proof}

\begin{lemma}\label{lem:introPartnerU}
Let $m\in M_\rho$ and $\mu\in {\cal S}(u,R'\setminus\{\rho\})$. The partners of $m$ with respect to $(u,R'\setminus\{\rho\},\mu)$ and $(u,R'\setminus\{\rho\},\mu_{R'\setminus\{\rho\}})$ are identical.
\end{lemma}

\begin{proof}
%Since $m\in M_\rho$, it holds that $\ell_{v,R'}(m)=\rho$.
If $m$ is unsettled with respect to $(u,R'\setminus\{\rho\})$, then we have that $\partner_{u,R'\setminus\{\rho\},\mu}(m)=\partner_{u,R'\setminus\{\rho\},\mu_{R'\setminus\{\rho\}}}(m)$. Now, we suppose that $m$ is settled (with respect to $(u,R'\setminus\{\rho\})$), and show that this supposition leads to a contradiction.
Since $m$ is settled and $\mu\in{\cal S}(v,R')$, from Proposition \ref{lem:tree} we deduce that $(R(\mu)\cap V(P(m)))\setminus\cl(R'\setminus\{\rho\})\subseteq \gamma(u)$.
Moreover, since $m\in M_\rho$, we have that $\rho_m\in R(\mu)$. Thus, $\rho_m\in\gamma(u)$. Since $m\in M_\rho$, we have reached a contradiction.
\end{proof}

As a corollary to Lemma \ref{lem:introPartnerU}, we have

\begin{corollary}\label{cor:introPartnerU}
For any $\mu\in {\cal S}(v,R')\cap {\cal S}(u,R'\setminus\{\rho\})$, it holds that $\displaystyle{\sum_{m\in M_\rho}\alpha_{u,R'\setminus\{\rho\}}(\mu,m)} = \displaystyle{\sum_{m\in M_\rho}\pos_m(\partner_{u,R'\setminus\{\rho\},\mu_{R'\setminus\{\rho\}}}(m))},\ \mathrm{and}\ \displaystyle{\sum_{m\in M_\rho}\lambda_{u,R'\setminus\{\rho\}}(\mu,m)}= \displaystyle{\pos_{\partner_{u,R'\setminus\{\rho\},\mu_{R'\setminus\{\rho\}}}(m)}(m)}$.
%\[\displaystyle{\sum_{m\in M_\rho}\alpha_{u,R'\setminus\{\rho\}}(\mu,m)} = \displaystyle{\sum_{m\in M_\rho}\pos_m(\partner_{u,R'\setminus\{\rho\},\mu_{R'\setminus\{\rho\}}}(m))},\ \mathrm{and}\ \displaystyle{\sum_{m\in M_\rho}\lambda_{u,R'\setminus\{\rho\}}(\mu,m)}= \displaystyle{\pos_{\partner_{u,R'\setminus\{\rho\},\mu_{R'\setminus\{\rho\}}}(m)}(m)}.\]
\end{corollary}

We are now ready to analyze $\alpha_{v,R'}(\mu)$ and $\lambda_{v,R'}(\mu)$ where $\mu\in {\cal S}(v,R')\cap {\cal S}(u,R'\setminus\{\rho\})$.

\begin{lemma}\label{lem:introSameScore}
Let $\mu\in {\cal S}(v,R')\cap {\cal S}(u,R'\setminus\{\rho\})$. Then, $\alpha_{v,R'}(\mu)=\alpha_{u,R'\setminus\{\rho\}}(\mu)+\diff^\alpha(M_\rho)$ and $\lambda_{v,R'}(\mu)=\lambda_{u,R'\setminus\{\rho\}}(\mu)+\diff^\lambda(M_\rho)$.
\end{lemma}

\begin{proof}
First, note that
\[\begin{array}{lll}
\alpha_{v,R'}(\mu) & = & \displaystyle{\sum_{m\in M^\star\setminus M_\rho}\alpha_{v,R'}(\mu,m) + \sum_{m\in M_\rho}\alpha_{v,R'}(\mu,m)}\\

&= & \displaystyle{\sum_{m\in M^\star\setminus M_\rho}\pos_m(\partner_{v,R',\mu}(m))}+ \displaystyle{\sum_{m\in M_\rho}\pos_m(\partner_{v,R',\mu}(m))},
\end{array}\]

\noindent which, by Lemmata \ref{lem:samePartnerNoRho} and \ref{lem:introPartnerV}, equals
\[\begin{array}{ll}
& \displaystyle{\sum_{m\in M^\star\setminus M_\rho}\pos_m(\partner_{u,R'\setminus\{\rho\},\mu}(m))} + \displaystyle{\sum_{m\in M_\rho}\pos_m(\partner_{v,R',\mu_{R'}}(m))}\\

=& \displaystyle{\sum_{m\in M^\star\setminus M_\rho}\alpha_{u,R'\setminus\{\rho\}}(\mu,m)} + \displaystyle{\sum_{m\in M_\rho}\pos_m(\partner_{v,R',\mu_{R'}}(m))},
\end{array}\]

\noindent which, by Corollary \ref{cor:introPartnerU}, equals
\[\begin{array}{l}
\displaystyle{\sum_{m\in M^\star\setminus M_\rho}\alpha_{u,R'\setminus\{\rho\}}(\mu,m) + \sum_{m\in M_\rho}\alpha_{u,R'\setminus\{\rho\}}(\mu,m)}\\

\hspace{2em} + \displaystyle{\sum_{m\in M_\rho}\pos_m(\partner_{v,R',\mu_{R'}}(m))} - \displaystyle{\sum_{m\in M_\rho}\pos_m(\partner_{u,R'\setminus\{\rho\},\mu_{R'\setminus\{\rho\}}}(m))},
\end{array}\]
which is exactly $\alpha_{u,R'\setminus\{\rho\}}(\mu) + \diff^\alpha(M_\rho)$.

Similarly, note that
\[\begin{array}{lll}
\lambda_{v,R'}(\mu) & = & \displaystyle{\sum_{m\in M^\star\setminus M_\rho}\lambda_{v,R'}(\mu,m) + \sum_{m\in M_\rho}\lambda_{v,R'}(\mu,m)}\\

&= & \displaystyle{\sum_{m\in M^\star\setminus M_\rho}\pos_{\partner_{v,R',\mu}(m)}(m)}+ \displaystyle{\sum_{m\in M_\rho}\pos_{\partner_{v,R',\mu}(m)}(m)},
\end{array}\]

\noindent which, by Lemmata \ref{lem:samePartnerNoRho} and \ref{lem:introPartnerV}, equals
\[\begin{array}{ll}
& \displaystyle{\sum_{m\in M^\star\setminus M_\rho}\pos_{\partner_{u,R'\setminus\{\rho\},\mu}(m)}(m)} + \displaystyle{\sum_{m\in M_\rho}\pos_{\partner_{v,R',\mu_{R'}}(m)}(m)}\\

=& \displaystyle{\sum_{m\in M^\star\setminus M_\rho}\lambda_{u,R'\setminus\{\rho\}}(\mu,m)} + \displaystyle{\sum_{m\in M_\rho}\pos_{\partner_{v,R',\mu_{R'}}(m)}(m)},
\end{array}\]

\noindent which, by Corollary \ref{cor:introPartnerU}, equals
\[\begin{array}{l}
\displaystyle{\sum_{m\in M^\star\setminus M_\rho}\lambda_{u,R'\setminus\{\rho\}}(\mu,m) + \sum_{m\in M_\rho}\lambda_{u,R'\setminus\{\rho\}}(\mu,m)}\\

\hspace{2em} + \displaystyle{\sum_{m\in M_\rho}\pos_{\partner_{v,R',\mu_{R'}}(m)}(m)} - \displaystyle{\sum_{m\in M_\rho}\pos_{\partner_{u,R'\setminus\{\rho\},\mu_{R'\setminus\{\rho\}}}(m)}(m)},
\end{array}\]
which is exactly $\lambda_{u,R'\setminus\{\rho\}}(\mu) + \diff^\lambda(M_\rho)$.
\end{proof}

Armed with Lemmata \ref{lem:introSubset}, \ref{lem:introAlmostSubset} and \ref{lem:introSameScore}, we are ready to prove that the computation for Case \ref{case:introSep} (i.e when $\rho\in R'$) is correct.

\begin{lemma}
${\cal S}(v,R',t_M,t_W)\neq\emptyset$ if and only if ${\cal S}(u,R'\setminus\{\rho\},t_M\!-\!\diff^\alpha(M_\rho),t_W\!-\!\diff^\lambda(M_\rho))\neq\emptyset$.
\end{lemma}

\begin{proof}
First, suppose that ${\cal S}(v,R',t_M,t_W)\neq\emptyset$. Then, there exists $\mu\in{\cal S}(v,R',t_M,t_W)$. By Lemma~\ref{lem:introSubset}, we have that $\mu\in{\cal S}(u,R'\setminus\{\rho\})$. Since $\alpha_{v,R'}(\mu)=t_M$ and $\lambda_{v,R'}(\mu)=t_W$, to conclude that $\mu\in{\cal S}(u,R'\setminus\{\rho\},t_M-\diff^\alpha(M_\rho),t_W-\diff^\lambda(M_\rho))$, it remains to show that $\alpha_{u,R'\setminus\{\rho\}}(\mu)=\alpha_{v,R'}(\mu)-\diff^\alpha(M_\rho)$ and $\lambda_{u,R'\setminus\{\rho\}}(\mu)=\lambda_{v,R'}(\mu)-\diff^\lambda(M_\rho)$. By Lemma \ref{lem:introSameScore}, this is~true.

Next, suppose that ${\cal S}(u,R'\setminus\{\rho\},t_M-\diff^\alpha(M_\rho),t_W-\diff^\lambda(M_\rho))\neq\emptyset$. Then, by Lemma \ref{lem:introAlmostSubset}, there exists $\mu\in{\cal S}(u,R'\setminus\{\rho\},t_M-\diff^\alpha(M_\rho),t_W-\diff^\lambda(M_\rho))$ such that $\rho\in R(\mu)$. Then, since $\beta(v)=\beta(u)\cup\{\rho\}$, $\rho\in R'$ and $\rho\in R(\mu)$, we have that $\mu\in{\cal S}(v,R')$. Since $\alpha_{u,R'\setminus\{\rho\}}(\mu)=t_M-\diff^\alpha(M_\rho)$ and $\lambda_{u,R'\setminus\{\rho\}}(\mu)=t_W-\diff^\lambda(M_\rho)$, to conclude that $\mu\in{\cal S}(v,R',t_M,t_W)$, it remains to show that $\alpha_{v,R'}(\mu)=t_M$ and $\lambda_{v,R'}(\mu)=t_W$. By Lemma \ref{lem:introSameScore}, this is~true.
\end{proof}

\medskip
\myparagraph{Join Node.}
By the inductive hypothesis, to prove that our computation is correct, it is sufficient to show that ${\cal S}(v,R',t_M,t_W)\neq\emptyset$ if and only if there exist $\widehat{t}_M,\widehat{t}_W,t^*_M,t^*_W\in[n^2]$ such that $t_M=\widehat{t}_M+t^*_M-\alpha_{v,R'}(\mu_{R'})$, $t_W=\widehat{t}_W+t^*_W-\lambda_{v,R'}(\mu_{R'})$ and ${\cal S}(u,R',\widehat{t}_M,\widehat{t}_W),{\cal S}(w,R',t^*_M,t^*_W)\neq\emptyset$.

\begin{lemma}\label{lem:joinEqualS}
${\cal S}(v,R')={\cal S}(u,R')={\cal S}(w,R')$.
\end{lemma}

\begin{proof}
Note that $\beta(v)=\beta(u)=\beta(w)$. Therefore, for any $x,y\in\{v,u,w\}$, a stable matching $\mu$ satisfies $R(\mu)\cap\beta(x)=R'$ if and only if it satisfies $R(\mu)\cap\beta(y)=R'$. This implies that ${\cal S}(v,R')={\cal S}(u,R')={\cal S}(w,R')$.
\end{proof}

In what follows, we implicitly rely on Lemma~\ref{lem:joinEqualS}. Towards proving the correctness of our computation, we need a lemma specifying some of the partners of the men.

\begin{lemma}\label{lem:joinKnowOne}
Let $\mu\in{\cal S}(v,R')$, $\widehat{\mu}\in {\cal S}(u,R')$ and $\mu^*\in{\cal S}(w,R')$ such that $R(\mu)=R(\widehat{\mu})\cup R(\mu^*)$. Then, at least one of the following conditions is satisfied.
\begin{enumerate}
\item The partners of $m$ with respect to $(v,R',\mu)$ and $(u,R',\widehat{\mu})$ are identical, and the partners of $m$ with respect to $(v,R',\mu)$ and $(w,R',\mu_{R'})$ are identical.
\item The partners of $m$ with respect to $(v,R',\mu)$ and $(w,R',\mu^*)$ are identical, and the partners of $m$ with respect to $(v,R',\mu)$ and $(u,R',\mu_{R'})$ are identical.
\end{enumerate}
\end{lemma}

\begin{proof}
Definitions \ref{def:manPartner} and \ref{def:deltaState} directly imply that $m$'s partners with respect to $(v,R',\mu_{R'})$, $(u,R',\mu_{R'})$ and $(w,R',\mu_{R'})$ are identical. In this proof, we implicitly rely on this observation.
If $R(\mu)\cap P(m)=\emptyset$, then $R(\widehat{\mu})\cap P(m)=R(\mu^*)\cap P(m)=\emptyset$. In this case, or if $m$ is unsettled with respect to $(v,R',\mu)$, $(u,R',\widehat{\mu})$ and $(w,R',\mu^*)$, then $m$'s partners with respect to $(v,R',\mu)$, $(u,R',\widehat{\mu})$ and $(w,R',\mu^*)$ are identical to $\mu_{R'}(m)$. Thus, we next assume that $V(P(m))\cap R(\mu)\cap \gamma(v)\neq\emptyset$, else we are done. Let $\rho$ denote the last vertex on $P(m)$ that belongs to $R(\mu)\cap\gamma(v)$.
Since $\gamma(v)=\gamma(u)\cup\gamma(w)$, we can suppose w.l.o.g.~that $\rho\in \gamma(u)$.

First, suppose that $\ell_{v,R'}(m)=\nil$. Then, by Proposition \ref{lem:tree}, $V(P(m))\subseteq\gamma(u)$ and $V(P(m))\cap\gamma(w)=\emptyset$. Therefore, since $R(\mu)=R(\widehat{\mu})\cup R(\mu^*)$, we have that $R(\mu)\cap V(P(m))=R(\widehat{\mu})\cap V(P(m))$ and $R(\mu^*)\cap V(P(m))=\emptyset$. In this case, $m$ is settled with respect to both $(v,R')$ and $(u,R')$, but unsettled with respect to $(w,R')$. Overall, $\partner_{v,R'}(\mu,m)=\partner_{u,R'}(\widehat{\mu},m)$ and $\partner_{w,R'}(\mu,m)=\mu_{R'}(m)$, which concludes the proof of this case. 

Now, suppose that $\ell_{v,R'}(m)\neq\nil$. Then, by Proposition \ref{lem:tree}, we have that the last vertex on $P(m)$ that belongs to $\gamma(w)$ is $\ell_{v,R'}(m)$. Thus, regardless of the type of $m$ with respect to the three states, the fact that $R(\mu)=R(\widehat{\mu})\cup R(\mu^*)$ implies that $\partner_{v,R'}(\mu,m)=\partner_{u,R'}(\widehat{\mu},m)$ and $\partner_{w,R'}(\mu,m)=\mu_{R'}(m)$, which concludes the proof of this case.
\end{proof}

Finally, we are ready to prove that our computation is correct. Thus, we also conclude that {\sc GSM} is solvable in time $\OO(2^{tw}\cdot n^{10})$.

\begin{lemma}\label{lem:joinEqualFirst}
${\cal S}(v,R',t_M,t_W)\neq\emptyset$ if and only if there exist $\widehat{t}_M,\widehat{t}_W,t^*_M,t^*_W\in[n^2]$ such that $t_M=\widehat{t}_M+t^*_M-\alpha_{v,R'}(\mu_{R'})$, $t_W=\widehat{t}_W+t^*_W-\lambda_{v,R'}(\mu_{R'})$ and ${\cal S}(u,R',\widehat{t}_M,\widehat{t}_W),{\cal S}(w,R',t^*_M,t^*_W)\neq\emptyset$.
\end{lemma}

\begin{proof}
In the first direction, assume that ${\cal S}(v,R',t_M,t_W)\neq\emptyset$, and let $\mu\in{\cal S}(v,R',t_M,t_W)$. Then, $\mu\in{\cal S}(u,R')\cap{\cal S}(w,R')$. By Lemma \ref{lem:joinKnowOne}, for all $m\in M^\star$, at least one of the two following conditions hold.
\begin{itemize}
\item $\alpha_{v,R'}(\mu,m) = \alpha_{u,R'}(\mu,m) + \alpha_{w,R'}(\mu,m) - \alpha_{u,R'}(\mu_{R'},m)$, and $\lambda_{v,R'}(\mu,m) = \lambda_{u,R'}(\mu,m) + \lambda_{w,R'}(\mu,m) - \lambda_{u,R'}(\mu_{R'},m)$.
\item $\alpha_{v,R'}(\mu,m) = \alpha_{u,R'}(\mu,m) + \alpha_{w,R'}(\mu,m) - \alpha_{w,R'}(\mu_{R'},m)$, and $\lambda_{v,R'}(\mu,m) = \lambda_{u,R'}(\mu,m) + \lambda_{w,R'}(\mu,m) - \lambda_{w,R'}(\mu_{R'},m)$.
\end{itemize}
By Definitions \ref{def:manPartner} and \ref{def:deltaState}, it holds that $\alpha_{v,R'}(\mu_{R'},m)=\alpha_{u,R'}(\mu_{R'},m)=\alpha_{w,R'}(\mu_{R'},m)$ and $\lambda_{v,R'}(\mu_{R'},m)=\lambda_{u,R'}(\mu_{R'},m)=\lambda_{w,R'}(\mu_{R'},m)$. Therefore, $\alpha_{v,R'}(\mu) = \alpha_{u,R'}(\mu) + \alpha_{w,R'}(\mu) - \alpha_{v,R'}(\mu_{R'})$ and $\lambda_{v,R'}(\mu) = \lambda_{u,R'}(\mu) + \lambda_{w,R'}(\mu) - \lambda_{v,R'}(\mu_{R'})$. By setting $\widehat{t}_M=\alpha_{u,R'}(\mu)$, $\widehat{t}_W=\lambda_{u,R'}(\mu)$, $t^*_M=\alpha_{w,R'}(\mu)$ and $t^*_W=\lambda_{w,R'}(\mu)$, we get that $\mu\in{\cal S}(u,R',\widehat{t}_M,\widehat{t}_W)\cap{\cal S}(w,R',t^*_M,t^*_W)$.

In the second direction, assume that there exist $\widehat{t}_M,\widehat{t}_W,t^*_M,t^*_W\in[n^2]$ such that $t_M=\widehat{t}_M+t^*_M-\alpha_{v,R'}(\mu_{R'})$, $t_W=\widehat{t}_W+t^*_W-\lambda_{v,R'}(\mu_{R'})$ and ${\cal S}(u,R',\widehat{t}_M,\widehat{t}_W),{\cal S}(w,R',t^*_M,t^*_W)\neq\emptyset$. Let $\widehat{m}\in {\cal S}(u,R',\widehat{t}_M,\widehat{t}_W)$ and $\mu^*\in{\cal S}(w,R',t^*_M,t^*_W)$. By Proposition \ref{lem:smCorrRots} and since $\widehat{m}\in {\cal S}(u,R')$ and $\mu^*\in{\cal S}(w,R')$, there exists $\mu\in{\cal S}(v,R')$ such that $R(\mu)=R(\widehat{m})\cup R(\mu^*)$. Now, by Lemma \ref{lem:joinKnowOne}, for all $m\in M^\star$, at least one of the following conditions hold.
\begin{itemize}
\item $\alpha_{v,R'}(\mu,m) = \alpha_{u,R'}(\widehat{\mu},m) + \alpha_{w,R'}(\mu^*,m) - \alpha_{u,R'}(\mu_{R'},m)$, and $\lambda_{v,R'}(\mu,m) = \lambda_{u,R'}(\widehat{\mu},m) + \lambda_{w,R'}(\mu^*,m) - \lambda_{u,R'}(\mu_{R'},m)$.
\item $\alpha_{v,R'}(\mu,m) = \alpha_{u,R'}(\widehat{\mu},m) + \alpha_{w,R'}(\mu^*,m) - \alpha_{w,R'}(\mu_{R'},m)$, and $\lambda_{v,R'}(\mu,m) = \lambda_{u,R'}(\widehat{\mu},m) + \lambda_{w,R'}(\mu^*,m) - \lambda_{w,R'}(\mu_{R'},m)$.
\end{itemize}
By Definitions \ref{def:manPartner} and \ref{def:deltaState}, $\alpha_{v,R'}(\mu_{R'},m)=\alpha_{u,R'}(\mu_{R'},m)=\alpha_{w,R'}(\mu_{R'},m)$ and $\lambda_{v,R'}(\mu_{R'},m)=\lambda_{u,R'}(\mu_{R'},m)=\lambda_{w,R'}(\mu_{R'},m)$. Therefore, $\alpha_{v,R'}(\mu) = \widehat{t}_M+t^*_M-\alpha_{v,R'}(\mu_{R'})$ and $\lambda_{v,R'}(\mu) = \widehat{t}_W+t^*_W-\lambda_{v,R'}(\mu_{R'})$, which implies that $\mu\in{\cal S}(v,R',t_M,t_W)$.
\end{proof}

\subsection{Sex Equal Stable Marriage}\label{sec:fptSESM}

Let $I=(M,W,\{\pos_m\}_{m\in M},\{\pos_w\}_{w\in W})$ be an instance of {\sc SESM}. First, call \alg\ with $I$ as input. We thus obtain the set of all pairs $(t_M,t_W)$ for which there exists $\mu\in{\cal S}$ such that both $\sat_M(\mu)=t_M$ and $\sat_W(\mu)=t_W$. Then, we return a pair $(t_M,t_W)$ that minimizes $|t_M-t_W|$. Clearly, by the correctness of \alg, we thus solve {\sc SESM} in time $\OO(2^{\tw}\cdot n^{10})$.

To solve {\sc SESM} faster, we define a table {\sf S} as follows. Each entry of the table {\sf S} is of the form {\sf S}$[v,R',d]$, where $v\in V(T)$, $R'\subseteq\beta(v)$ and $d\in \{-n^2,-n^2+1,\ldots,n^2\}$. The following definition addresses the purpose of these entries.

\begin{definition}\label{def:ScomputeCorrectly}
We say that {\sf S} is {\em computed correctly} if for each entry {\sf S}$[v,R',d]$, it holds that 
{\sf S}$[v,R',d]\in\{0,1\}$, and {\sf S}$[v,R',t_M,t_W]=1$ if and only if there exist $t_M,t_W\in[n^2]$ such that $t_M-t_W=d$ and ${\cal S}(v,R',t_M,t_W)\neq\emptyset$.
\end{definition}

However, Definition \ref{def:ScomputeCorrectly} directly implies the correctness of the following observation.

\begin{observation}\label{obs:genToSESM}
Let {\sf N} be a table that is computed correctly by Definition \ref{def:computeCorrectly}. Then, the table {\sf S} is computed correctly by Definition \ref{def:ScomputeCorrectly} if and only if for each entry {\sf S}$[v,R',d]$, it holds that 
{\sf S}$[v,R',d]\in\{0,1\}$, and {\sf S}$[v,R',t_M,t_W]=1$ if and only if there exist $t_M,t_W\in[n^2]$ such that $t_M-t_W=d$ and {\sf N}$[v,R',t_M,t_W]=1$.
\end{observation}

In light of Observation \ref{obs:genToSESM}, it is straightforward to modify \alg\ to solve {\sc SESM} in time $\OO(2^{\tw}\cdot n^{6})$. For the sake of completeness, we briefly present the computation here. We process the entries of {\sf S} by traversing the tree $T$ in post-order. The order in which we process entries corresponding to the same node $v\in V(T)$ is arbitrary. Thus, the basis corresponds to entries {\sf S}$[v,R',d]$ where $v$ is a leaf, and the steps correspond to entries where $v$ is a forget node, an introduce node or a join node.

\medskip
\myparagraph{Leaf Node.} In the basis, where $v$ is a leaf, we have that $\beta(v)=\emptyset$. We consider two cases.
\begin{enumerate}
\itemsep0em 
\item If $d=\alpha_{v,\emptyset}(\mu_{\emptyset})-\lambda_{v,\emptyset}(\mu_{\emptyset})$: {\sf S}$[v,\emptyset,d]=1$.
\item Otherwise: {\sf S}$[v,\emptyset,d]=0$.
\end{enumerate}

\myparagraph{Forget Node.} Let $u$ denote the child of $v$ in $T$, and $\rho$ denote the vertex in $\beta(u)\setminus\beta(v)$. Then, {\sf S}$[v,R',d]=\max\{\mathrm{\sf N}[u,R',d],\mathrm{\sf N}[u,R'\cup\{\rho\},d]\}$.

\medskip
\myparagraph{Introduce Node.} Let $u$ be the child of $v$ in $T$, and $\rho$ be the vertex in $\beta(v)\setminus\beta(u)$. Consider the following cases.
\begin{enumerate}
\item If $(\cl(R')\cap\beta(v))\setminus R'\neq\emptyset$: {\sf S}$[v,R',d]=0$.
\item Else if $\rho\notin R'$: $M_\rho=\emptyset$; {\sf S}$[v,R',d]=\mathrm{\sf S}[u,R',d]$.
\item Otherwise ($\rho\in R'$): For each $m\in M_\rho$, denote $\diff^{SESM}_{v,R'}(M_\rho)=\diff^\alpha(M_\rho)-\diff^{\lambda}(M_\rho)$. Now, {\sf S}$[v,R',d]$ is computed as follows.
\vspace{-0.5em}
\[\mathrm{\sf S}[v,R',d] = \mathrm{\sf S}[u,R'\setminus\{\rho\},d-\diff^{SESM}(M_\rho)].\]
\end{enumerate}

\myparagraph{Join Node.} Let $u$ and $w$ denote the children of $v$ in $T$. For the sake of efficiency, we compute all entries of the form {\sf S}$[v,R',\cdot]$ simultaneously. First, we initialize each such entry to 0. Now, we compute $D(u,R')=\{\widehat{d}\in\{-n^2,-n^2+1,\ldots,n^2-1,n^2\}: \mathrm{\sf S}[u,R',\widehat{d}]=1\}$ and $D(w,R')=\{\widehat{d}\in\{-n^2,-n^2+1,\ldots,n^2-1,n^2\}: \mathrm{\sf S}[w,R',\widehat{d}]=1\}$. Then, for all $\widehat{d}\in D(u,R')$ and $d^*\in D(w,R')$ such that $d=\widehat{d}+d^*-(\alpha_{v,R'}(\mu_{R'})-\lambda_{v,R'}(\mu_{R'}))\in\{-n^2,-n^2+1,\ldots,n^2\}$, we set $\mathrm{\sf S}[v,R',d]=1$.

\subsection{Balanced Stable Marriage}\label{sec:fptBSM}

Let $I=(M,W,\{\pos_m\}_{m\in M},\{\pos_w\}_{w\in W})$ be an instance of {\sc BSM}. First, call \alg\ with $I$ as input. We thus obtain the set of all pairs $(t_M,t_W)$ for which there exists $\mu\in{\cal S}$ such that both $\sat_M(\mu)=t_M$ and $\sat_W(\mu)=t_W$. Then, we return a pair $(t_M,t_W)$ that minimizes $\max\{t_M,t_W\}$. Clearly, by the correctness of \alg, we thus solve {\sc BSM} in time $\OO(2^{\tw}\cdot n^{10})$.

To solve {\sc BSM} faster, we define a table {\sf B} as follows. Each entry of the table {\sf B} is of the form {\sf B}$[v,R',b]$, where $v\in V(T)$, $R'\subseteq\beta(v)$ and $b\in [n^2]\cup\{\nil\}$. The following definition addresses the purpose of these entries.

\begin{definition}\label{def:BcomputeCorrectly}
We say that {\sf B} is {\em computed correctly} if for each entry {\sf B}$[v,R',b]$,
{\sf B}$[v,R',b]\in[n^2]\cup\{\nil\}$, and for all $i\in[n^2]$, {\sf B}$[v,R',b]=i$ if and only if ${\cal S}(v,R',b,i)\neq\emptyset$ and for all $j\in[n^2]$ where $j<i$, ${\cal S}(v,R',b,j)=\emptyset$.
\end{definition}

However, Definition \ref{def:BcomputeCorrectly} directly implies the correctness of the following observation.

\begin{observation}\label{obs:genToBSM}
Let {\sf N} be a table that is computed correctly by Definition \ref{def:computeCorrectly}. Then, the table {\sf B} is computed correctly by Definition \ref{def:BcomputeCorrectly} if and only if for each entry {\sf B}$[v,R',d]$,
{\sf B}$[v,R',d]\in[n^2]\cup\{\nil\}$, and for all $i\in[n^2]$, {\sf B}$[v,R',b]=i$ if and only if ${\sf N}[v,R',b,i]=1$ and for all $j\in[n^2]$ where $j<i$, ${\sf N}[v,R',b,j]=0$.
\end{observation}

In light of Observation \ref{obs:genToBSM}, it is straightforward to modify \alg\ to solve {\sc BSM} in time $\OO(2^{\tw}\cdot n^{6})$. For the sake of completeness, we briefly present the computation here. We process the entries of {\sf B} by traversing the tree $T$ in post-order. The order in which we process entries corresponding to the same node $v\in V(T)$ is arbitrary. Thus, the basis corresponds to entries {\sf B}$[v,R',b]$ where $v$ is a leaf, and the steps correspond to entries where $v$ is a forget node, an introduce node or a join node.

\medskip
\myparagraph{Leaf Node.} In the basis, where $v$ is a leaf, we have that $\beta(v)=\emptyset$. We consider two cases.
\begin{enumerate}
\itemsep0em 
\item If $b=\alpha_{v,\emptyset}(\mu_{\emptyset})$: {\sf B}$[v,\emptyset,b]=\lambda_{v,\emptyset}(\mu_{\emptyset})$.
\item Otherwise: {\sf B}$[v,\emptyset,b]=\nil$.
\end{enumerate}

\myparagraph{Forget Node.} Let $u$ denote the child of $v$ in $T$, and $\rho$ denote the vertex in $\beta(u)\setminus\beta(v)$. Then, {\sf B}$[v,R',b]=\min\{\mathrm{\sf B}[u,R',b],\mathrm{\sf B}[u,R'\cup\{\rho\},b]\}$. Here, we define integers to be smaller than $\nil$.

\medskip
\myparagraph{Introduce Node.} Let $u$ be the child of $v$ in $T$, and $\rho$ be the vertex in $\beta(v)\setminus\beta(u)$. Consider the following cases.
\begin{enumerate}
\item If $(\cl(R')\cap\beta(v))\setminus R'\neq\emptyset$: {\sf B}$[v,R',b]=\nil$.
\item Else if $\rho\notin R'$: $M_\rho=\emptyset$; {\sf B}$[v,R',b]=\mathrm{\sf B}[u,R',b]$.
\item Otherwise ($\rho\in R'$): Compute {\sf B}$[v,R',b]$ as follows. Here, we define the subtraction of an integer from $\nil$ to be $\nil$.
\vspace{-0.5em}
\[\mathrm{\sf B}[v,R',b] = \mathrm{\sf B}[u,R'\setminus\{\rho\},b-\diff^\alpha(M_\rho)]-\diff^\lambda(M_\rho).\]
\end{enumerate}

\myparagraph{Join Node.} Let $u$ and $w$ denote the children of $v$ in $T$. For the sake of efficiency, we compute all entries of the form {\sf B}$[v,R',\cdot]$ simultaneously. First, we initialize each such entry to $\nil$. Now, we compute $B(u,R')=\{\widehat{b}\in[n^2]: \mathrm{\sf B}[u,R',\widehat{b}]\neq \nil\}$ and $B(w,R')=\{\widehat{b}\in[n^2]: \mathrm{\sf B}[w,R',\widehat{b}]\neq\nil\}$. Then, for all $\widehat{b}\in D(u,R')$ and $b^*\in D(w,R')$ such that $b=\widehat{b}+b^*-\alpha_{v,R'}(\mu_{R'})\in[n^2]$, we set $\mathrm{\sf B}[v,R',b]$ to be the minimum between the previous value this entry stored and $\mathrm{\sf B}[u,R',\widehat{b}]+\mathrm{\sf B}[w,R',b^*]-\lambda_{v,R'}(\mu_{R'})$.

%% file: seth.tex
\section{Rotation Digraph: Lower Bounds}\label{sec:seth}

In this section, we prove Theorem \ref{thm:sethIntro}, based on the approach described in Section~\ref{sec:overview}. Throughout this section, the notation $\tw$ would refer to the treewidth of the underlying undirected graph of the rotation digraph. The source of our reductions is the {\sc $s$-Sparse $p$-CNF-SAT} problem, which is the special case of the {\sc $p$-CNF-SAT} problem where the number of clauses is upper bounded by $s\cdot n$. Here, $n$ is the number of variables.

\begin{proposition}[\cite{DBLP:journals/jcss/ImpagliazzoPZ01,DBLP:conf/coco/CalabroIP06,DBLP:journals/talg/CyganDLMNOPSW16}]\label{prop:SparseSAT}
Unless \SETH\ fails, for every fixed $\epsilon<1$, there exist integers $p=p(\epsilon)$ and $s=s(\epsilon)$ such that {\sc $s$-Sparse $p$-CNF-SAT} cannot be solved in time $\OO((2-\epsilon)^{n})$, where $n$ is the number of variables.
\end{proposition}

Given a clause $C$ of a formula $\varphi$ in CNF, we would also treat $C$ as a set whose elements are the literals of the clause $C$.

\input{sethSESM.tex}

\input{sethBSM.tex}

%% file: sethSESM.tex
\subsection{Sex Equal Stable Marriage}\label{sec:sethSESM}

First, we prove that unless \SETH\ fails, {\sc SESM} cannot be solved in time $(2-\epsilon)^{\tw}\cdot n^{\OO(1)}$ for any fixed $\epsilon>0$, where $n$ is the number of agents. Note that this claim is equivalent to the one stating that unless \SETH\ fails, {\sc SESM} cannot be solved in time $2^{\epsilon\tw}\cdot n^{\OO(1)}$ for any fixed $\epsilon<1$. To prove this claim, we suppose, by way of contradiction, that there exist fixed $\epsilon>0$ and $c\geq 1$ as well as an algorithm \algSESM\ such that \algSESM\ solves {\sc SESM} in time $2^{\epsilon\tw}\cdot n^c$.

\subsubsection{Reduction}

Denote $\delta=\epsilon+(1-\epsilon)/2<1$. By Proposition~\ref{prop:SparseSAT}, supposing that \SETH\ is true, there exist integers $p=p(\delta)$ and $s=s(\delta)$ such that {\sc $s$-Sparse $p$-CNF-SAT} cannot be solved in time $\OO(2^{\delta n})$. Let $\varphi=C_1\wedge C_2\wedge\cdots\wedge C_r$ be an instance of {\sc $p$-CNF-SAT}. Denote ${\cal C}=\{C_1,C_2,\ldots,C_r\}$. Note that by the definition of {\sc $s$-Sparse $p$-CNF-SAT}, $r\leq s\cdot n$. In this context, let $X=\{x_1,x_2,\ldots,x_n\}$ denote the set of positive literals (to which we also refer as variables), and let $\overline{X}=\{\overline{x}_1,\overline{x}_2,\ldots,\overline{x}_n\}$ denote the set of negative literals, where for all $t\in[n]$, $\overline{x}_t$ is the negation of $x_t$. We now describe how to construct an instance $\red_{SESM}(\varphi)=(M,W,\{\pos_m\}|_{m\in M},\{\pos_w\}|_{w\in W})$ of {\sc SESM}.

\medskip
\myparagraph{Partial Truth Assignments.} Let $d$ denote the smallest possible integer such that $\delta \geq  \epsilon + \frac{3cs}{d}$, and set $q=\frac{r}{d}$. We assume w.l.o.g.~that $r$ is divisible by $d$. We partition $\cal C$ into $q$ sets as follows: For all $i\in[q]$, we define ${\cal C}^i=\{C_{1+(i-1)d},C_{2+(i-1)d},\ldots,C_{id}\}$. Moreover, for all $i\in[q]$, let $X^i$ be the set of every variable $x_j\in X$ such that at least one among the literals $x_j$ and $\overline{x}_j$ belongs to at least one clause in ${\cal C}^i$. Now, for all $i\in[q]$, let $F^i=\{f^i_1,f^i_2,\ldots,f^i_{|F^i|}\}$ denote the set of every truth assignment to the variables in $X^i$ that satisfies {\em all} of the clauses in ${\cal C}^i$. For all $i\in[q]$ and $j\in[|F^i|]$, we let $P^i_j$ denote the subset of variables of $X^i$ to which $f^i_j$ assigns true, and we let $N^i_j$ denote the subset of variables of $X^i$ to which $f^i_j$ assigns false. Clearly, for all $i\in[q]$ and $j\in[|F^i|]$, $P^i_j\cup N^i_j=X^i$ and $P^i_j\cap N^i_j=\emptyset$. For all $i\in[q]$, denote $a^i=|F^i|$, ${\cal P}^i=\{P^i_1,P^i_2,\ldots,P^i_{a^i}\}$ and ${\cal N}^i=\{N^i_1,N^i_2,\ldots,N^i_{a^i}\}$. Finally, denote $\sum_{i\in[q]}a^i=\widetilde{a}$.

For all $t\in[n]$, denote ${\cal P}(x_t)=\{P^i_j: i\in[q], j\in[a^i], x_t\in P^i_j\}$ and ${\cal N}(x_t)=\{N^i_j: i\in[q], j\in[a^i], x_t\in N^i_j\}$.

\medskip
\myparagraph{Agents.} First, to represent variables, we introduce four sets of new agents:
\begin{itemize}
\item $M_{\var}=\{m_1,m_2,\ldots,m_n\}$.
\item $\widehat{M}_{\var}=\{\widehat{m}_1,\widehat{m}_2,\ldots,\widehat{m}_n\}$.
\item $W_{\var}=\{w_1,w_2,\ldots,w_n\}$.
\item $\widehat{W}_{\var}=\{\widehat{w}_1,\widehat{w}_2,\ldots,\widehat{w}_n\}$.
\end{itemize}

Next, for all $i\in[q]$, to represent the sets in ${\cal P}^i$, we introduce four sets of new agents:
\begin{itemize}
\item $M^i=\{m^i_1,m^i_2,\ldots,m^i_{a^i}\}$.
\item $\widehat{M}^i=\{\widehat{m}^i_1,\widehat{m}^i_2,\ldots,\widehat{m}^i_{a^i}\}$.
\item $W^i=\{w_1,w_2,\ldots,w_{a^i}\}$.
\item $\widehat{W}^i=\{\widehat{w}^i_1,\widehat{w}^i_2,\ldots,\widehat{w}_{a^i}\}$.
\end{itemize}

Similarly, for all $i\in[q]$, to represent the sets in ${\cal P}^i$, we introduce four sets of new agents:
\begin{itemize}
\item $\overline{M}^i=\{\overline{m}^i_1,\overline{m}^i_2,\ldots,\overline{m}^i_{a^i}\}$.
\item $\widehat{\overline{M}}^i=\{\widehat{\overline{m}}^i_1,\widehat{\overline{m}}^i_2,\ldots,\widehat{\overline{m}}^i_{a^i}\}$.
\item $\overline{W}^i=\{\overline{w}^i_1,\overline{w}^i_2,\ldots,\overline{w}^i_{a^i}\}$.
\item $\widehat{\overline{W}}^i=\{\widehat{\overline{w}}^i_1,\widehat{\overline{w}}^i_2,\ldots,\widehat{\overline{w}}^i_{a^i}\}$.
\end{itemize}

Moreover, we introduce $\alpha$ new happy pairs, denoted by $(m^1_{\hap},w^1_{\hap}),(m^2_{\hap},w^2_{\hap}),\ldots,$ $(m^\alpha_{\hap},w^\alpha_{\hap})$, where the definition of $\alpha$ relies on the arguments $\lambda(i)$ ,$\gamma(i)$ and $\tau$, which are as defined later.
\[\begin{array}{ll}
\alpha = \displaystyle{\sum_{i\in[q]}((a^i-1)\lambda(i)-\gamma(i))} + (2q-\widetilde{a})\tau.
\end{array}\]

Finally, we introduce one new man, $m^\star$, and one new woman $w^\star$. We define the preference list of $m^\star$ by setting its domain to contain all happy women and $w^\star$, and defining $\pos_{m^\star}(w^\star)=\alpha+1$ and $\pos_{m^\star}(w^i_{\hap})=i$ for all $i\in[\alpha]$. The preference list of $w^\star$ is simply defined to contain only~$m^\star$.

\medskip
\myparagraph{Variable Selector.} We first describe the preference lists of the agents in the sets $M_{\var}$, $\widehat{M}_{\var}$, $W_{\var}$ and $\widehat{W}_{\var}$. For all $t\in[n]$, we define the preference list of $m_t$ as follows. We set $\domain(\pos_{m_t})=\{w_t\}\cup \{\overline{w}^i_j: N^i_j\in {\cal N}(x_t)\}\cup\{\widehat{w}_t\}$. Then, we let $w_t$ be the woman most preferred by $m_t$, $\widehat{w}_t$ be the woman least preferred by $m_t$, and strictly order all of the other women in the domain arbitrarily. For all $t\in[n]$, the preference list of $\widehat{m}_t$ is simply defined by setting $\domain(\pos_{\widehat{m}_t})=\{\widehat{w}_t,w_t\}$, $\pos_{\widehat{m}_t}(\widehat{w}_t)=1$ and $\pos_{\widehat{m}_t}(w_t)=2$.

Now, for all $t\in[n]$, we define the preference list of $w_t$ as follows. We set $\domain(\pos_{w_t})=\{\widehat{m}_t\}\cup \{m^i_j: P^i_j\in {\cal P}(x_t)\}\cup\{m_t\}$. Then, we let $\widehat{m}_t$ be the man most preferred by $w_t$, $m_t$ be the man least preferred by $w_t$, and strictly order all of the other men in the domain arbitrarily. For all $t\in[n]$, the preference list of $\widehat{w}_t$ is simply defined by setting $\domain(\pos_{\widehat{w}_t})=\{m_t,\widehat{m}_t\}$, $\pos_{\widehat{w}_t}(m_t)=1$ and $\pos_{\widehat{w}_t}(\widehat{m}_t)=2$.

\medskip
\myparagraph{Truth Selector.} We next describe the preference lists of the agents in the sets $M^i$, $\widehat{M}^i$, $W^i$ and $\widehat{W}^i$. For all $i\in[q]$ and $j\in[a^i]$, we define the preference list of $m^i_j$ as follows. We set $\domain(\pos_{m^i_j})=\{w^i_j\}\cup\{w_t: x_t\in P^i_j\}\cup\{\overline{w}^i_k: k\in[a^i], k\neq j\}\cup W'\cup\{\widehat{w}^i_j\}$, where $W'$ is some subset of $\gamma(i)=n^{20}\cdot 2^{i-1}$ arbitrarily chosen happy women. Then, we let $w^i_j$ be the woman most preferred by $m^i_j$, $\widehat{w}^i_j$ be the woman least preferred by $m^i_j$, and strictly order all of the other women in the domain arbitrarily. For all $i\in[q]$ and $j\in[a^i]$, the preference list of $\widehat{m}^i_j$ is simply defined by setting $\domain(\pos_{\widehat{m}^i_j})=\{\widehat{w}^i_j,w^i_j\}$, $\pos_{\widehat{m}^i_j}(\widehat{w}^i_j)=1$ and $\pos_{\widehat{m}^i_j}(w^i_j)=2$.

Now, for all $i\in[q]$ and $j\in[a^i]$, the preference list of $w^i_j$ is defined as follows. We set $\domain(\pos_{w^i_j})=\{\widehat{m}^i_j\}\cup M'\cup\{m^i_j\}$, where $M'$ is some subset of $\lambda(i)=n^{20}\cdot2^{2q-i}$ arbitrarily chosen happy men. Then, we let $\widehat{m}^i_j$ be the man most preferred by $w^i_j$, $m^i_j$ be the man least preferred by $w^i_j$, and strictly order all of the other men in the domain arbitrarily. For all $i\in[q]$ and $j\in[a^i]$, the preference list of $\widehat{w}^i_j$ is simply defined by setting $\domain(\pos_{\widehat{w}^i_j})=\{m^i_j,\widehat{m}^i_j\}$, $\pos_{\widehat{w}^i_j}(m^i_j)=1$ and $\pos_{\widehat{w}^i_j}(\widehat{m}^i_j)=2$.

\medskip
\myparagraph{False Selector.} Finally, we describe the preference lists of the agents in the sets $\overline{M}^i$, $\widehat{\overline{M}}^i$, $\overline{W}^i$ and $\widehat{\overline{W}}^i$. For all $i\in[q]$ and $j\in[a^i]$, the preference list of $\overline{m}^i_j$ is defined as follows. We set $\domain(\pos_{\overline{m}^i_j})=\{\overline{w}^i_j\}\cup W'\cup\{\widehat{\overline{w}}^i_j\}$, where $W'$ is some subset of $\tau=n^{10}$ arbitrarily chosen happy women. Then, we let $\overline{w}^i_j$ be the woman most preferred by $\overline{m}^i_j$, $\widehat{\overline{w}}^i_j$ be the woman least preferred by $\overline{m}^i_j$, and strictly order all of the other women in the domain arbitrarily. For all $i\in[q]$ and $j\in[a^i]$, the preference list of $\widehat{\overline{m}}^i_j$ is simply defined by setting $\domain(\pos_{\widehat{\overline{m}}^i_j})=\{\widehat{\overline{w}}^i_j,\overline{w}^i_j\}$, $\pos_{\widehat{\overline{m}}^i_j}(\widehat{\overline{w}}^i_j)=1$ and $\pos_{\widehat{\overline{m}}^i_j}(\overline{w}^i_j)=2$.

Now, for all $i\in[q]$ and $j\in[a^i]$, the preference list of $\overline{w}^i_j$ is defined as follows. We set $\domain(\pos_{\overline{w}^i_j})=\{\widehat{\overline{m}}^i_j\}\cup \{m_t: x_t\in N^i_j\}\cup\{m^i_k: k\in[a^i], k\neq j\}\cup M'\cup\{\overline{m}^i_j\}$, where $M'$ is some subset of $\tau$ arbitrarily chosen happy men. Then, we let $\widehat{\overline{m}}^i_j$ be the man most preferred by $\overline{w}^i_j$, $\overline{m}^i_j$ be the man least preferred by $\overline{w}^i_j$, and strictly order all of the other men in the domain arbitrarily. For all $i\in[q]$ and $j\in[a^i]$, the preference list of $\widehat{\overline{w}}^i_j$ is simply defined by setting $\domain(\pos_{\widehat{\overline{w}}^i_j})=\{\overline{m}^i_j,\widehat{\overline{m}}^i_j\}$, $\pos_{\widehat{\overline{w}}^i_j}(\overline{m}^i_j)=1$ and $\pos_{\widehat{\overline{w}}^i_j}(\widehat{\overline{m}}^i_j)=2$.

\subsubsection{All Stable Matchings}\label{sec:AllSMs}

We begin our analysis by identifying exactly which matchings are stable matchings. Let us start with the following definition.

\begin{definition}
Let $\varphi$ be an instance of {\sc $s$-Sparse $p$-CNF-SAT}. Then, in the context of $\red_{SESM}(\varphi)$, the matching $\mu_{\emptyset}$ is defined as follows.
\begin{itemize}
\item For all $t\in[n]$: $\mu(m_t)=w_t$ and $\mu(\widehat{m}_t)=\widehat{w}_t$.
\item For all $i\in[q]$ and $j\in[a^i]$: $\mu(m^i_j)=w^i_j$ and $\mu(\widehat{m}^i_j)=\widehat{w}^i_j$.
\item For all $i\in[q]$ and $j\in[a^i]$: $\mu(\overline{m}^i_j)=\overline{w}^i_j$ and $\mu(\widehat{\overline{m}}^i_j)=\widehat{\overline{w}}^i_j$.
\item For all $i\in[\alpha]$: $\mu(m^i_{\hap})=w^i_{\hap}$.
\item $\mu(m^\star)=w^\star$.
\end{itemize}
\end{definition}

Observe that in the matching $\mu_{\emptyset}$, every man except $m^\star$ is matched to the woman he prefers the most. Moreover, the man $m^\star$ cannot belong to a blocking pair since all of the women who he prefers over $w^\star$ are matched to their most preferred men. Thus, we have the following observation.

\begin{observation}\label{obs:manOpt1}
Let $\varphi$ be an instance of {\sc $s$-Sparse $p$-CNF-SAT}. Then, in the context of $\red_{SESM}(\varphi)$, the matching $\mu_{\emptyset}$ is a stable matching.
\end{observation}

Next, we proceed to note that also in the current reduction, all agents are matched.

\begin{lemma}\label{lem:sethSESMAll}
Let $\varphi$ be an instance of {\sc $s$-Sparse $p$-CNF-SAT}. Every stable matching of $\red_{SESM}(\varphi)$ matches all agents.
\end{lemma}

\begin{proof}
Note that $\mu_{\emptyset}$ matches all men. Thus, by Observation \ref{obs:manOpt1} and Proposition \ref{lem:matchSame}, we deduce that every stable matching of $\red_{SESM}(\varphi)$ matches all men. Since the number of men is equal to the number of women, we conclude the correctness of the lemma.
\end{proof}

To formalize the conditions that a stable matching should satisfy, we introduce the following definition.

\begin{definition}\label{def:sethSESMGood}
Let $\varphi$ be an instance of {\sc $s$-Sparse $p$-CNF-SAT}. Let $\mu$ be a matching of $\red_{SESM}(\varphi)$. Then, we say that $\mu$ is {\em good} if it satisfies the following conditions.
\begin{enumerate}
\item\label{good1} For all $t\in[n]$: Either both $\mu(m_t)=w_t$ and $\mu(\widehat{m}_t)=\widehat{w}_t$ or both $\mu(m_t)=\widehat{w}_t$ and $\mu(\widehat{m}_t)=w_t$.
\item\label{good2} For all $i\in[q]$ and $j\in[a^i]$: Either both $\mu(m^i_j)=w^i_j$ and $\mu(\widehat{m}^i_j)=\widehat{w}^i_j$ or both $\mu(m^i_j)=\widehat{w}^i_j$ and $\mu(\widehat{m}^i_j)=w^i_j$.
\item\label{good3} For all $i\in[q]$ and $j\in[a^i]$: Either both $\mu(\overline{m}^i_j)=\overline{w}^i_j$ and $\mu(\widehat{\overline{m}}^i_j)=\widehat{\overline{w}}^i_j$ or both $\mu(\overline{m}^i_j)=\widehat{\overline{w}}^i_j$ and $\mu(\widehat{\overline{m}}^i_j)=\overline{w}^i_j$.
\item\label{good4} For all $i\in[\alpha]$: $\mu(m^i_{\hap})=w^i_{\hap}$.
\item\label{good5} $\mu(m^\star)=w^\star$.
\end{enumerate}
\end{definition}

We start with the following lemma.

\begin{lemma}\label{lem:sethSESMCaptureGood}
Let $\varphi$ be an instance of {\sc $s$-Sparse $p$-CNF-SAT}. Let $\mu$ be a stable matching of $\red_{SESM}(\varphi)$. Then, $\mu$ is a good matching.
\end{lemma}

\begin{proof}
Let $\mu$ be a stable matching of $\red_{SESM}(\varphi)$. The proof of the last two conditions is the same as the one given for Lemma \ref{lem:w1SESMReverseHap}. For completeness, we present them here as well. Since for all $i\in[\alpha]$, $m^i_{\hap}$ and $w^i_{\hap}$ prefer each other over all other agents, they must be matched to one another (by $\mu$), else they form a blocking pair. Then, since apart from happy women, the preference lists of $m^\star$ and $w^\star$ only contain each other, they also must be matched to one another, else they form a blocking pair. Hence, Conditions \ref{good4} and \ref{good5} (in Definition \ref{def:sethSESMGood}) are satisfied.

For all $i\in[q]$ and $j\in[a^i]$, observe that excluding happy agents, the only two men in the preference lists of $w^i_j$ and $\widehat{w}^i_j$ are $m^i_j$ and $\widehat{m}^i_j$, and the only two women in the preference lists of $\overline{m}^i_j$ and $\widehat{\overline{m}}^i_j$ are $\overline{w}^i_j$ and $\widehat{\overline{w}}^i_j$. Thus, by Lemma \ref{lem:sethSESMAll}, we deduce that Conditions \ref{good2} and \ref{good3} are satisfied. Finally, observe that excluding women who have already been determined not to be matched to agents in $M_{\var}\cup\widehat{M}_{\var}$, for all $t\in[n]$, the only two women in the preference lists of $m_t$ and $\widehat{m}_t$ are $w_t$ and $\widehat{w}_t$. Thus, by Lemma \ref{lem:sethSESMAll}, we deduce that Condition \ref{good1} is satisfied as well.
\end{proof}

However, the converse of Lemma \ref{lem:sethSESMCaptureGood} is not true, namely, not every good matching is a stable matching. To capture only stable matchings, we further need to strengthen Definition \ref{def:sethSESMGood}.

\begin{definition}\label{def:sethSESMExcellent}
Let $\varphi$ be an instance of {\sc $s$-Sparse $p$-CNF-SAT}. Let $\mu$ be a matching of $\red_{SESM}(\varphi)$. Then, we say that $\mu$ is {\em excellent} if it good as well as satisfies the two following conditions.
\begin{enumerate}
\item\label{excel1} For all $t\in[n]$ such that $\mu(m_t)=\widehat{w}_t$: For all $i\in[q]$ and $j\in[a^i]$ such that $x_t\in N^i_j$, $\mu(\overline{m}^i_j)=\widehat{\overline{w}}^i_j$. 
\item For all $i\in[q]$ and $j\in[a^i]$ such that $\mu(m^i_j)=\widehat{w}^i_j$, the two following conditions are satisfied.
	\begin{enumerate}
	\item\label{excel2} For all $t\in[n]$ such that $x_t\in P^i_j$, $\mu(m_t)=\widehat{w}_t$.
	\item\label{excel3} For all $k\in[a^i]$ such that $k\neq j$, $\mu(\overline{m}^i_k)=\widehat{\overline{w}}^i_k$. 
	\end{enumerate}
\end{enumerate}
In the context of $\red_{SESM}(\varphi)$, define $\Lambda$ to be the set of all excellent matchings.
\end{definition}

The two following lemmata show that Definition \ref{def:sethSESMExcellent} provides both sufficient and necessary conditions for a matching to be a stable matching.

\begin{lemma}\label{lem:sethSESMCaptureEx1}
Let $\varphi$ be an instance of {\sc $s$-Sparse $p$-CNF-SAT}. Let $\mu$ be a stable matching of $\red_{SESM}(\varphi)$. Then, $\mu$ is an excellent matching.
\end{lemma}

\begin{proof}
By Lemma \ref{lem:sethSESMCaptureGood}, it holds that $\mu$ is a good matching. We first claim that Condition \ref{excel1} is Definition \ref{def:sethSESMExcellent} is satisfied. Suppose, by way of contradiction, that this claim is false. Then, there exist $t\in[n]$, $i\in[q]$ and $j\in[a^i]$ such that $\mu(m_t)=\widehat{w}_t$, $x_t\in N^i_j$, $\mu(\overline{m}^i_j)\neq\widehat{\overline{w}}^i_j$. Since $\mu$ is a good matching, we then have that $\mu(\overline{m}^i_j)\neq\overline{w}^i_j$. However, $m_t$ prefers $\overline{w}^i_j$ over $\widehat{w}_t$, and $\overline{w}^i_j$ prefers $m_t$ over $\overline{m}^i_j$, which contradicts the fact that $\mu$ is stable as $(m_t,\overline{w}^i_j)$ is a blocking pair. Thus, Condition \ref{excel1} is Definition \ref{def:sethSESMExcellent} is satisfied.

Second, we claim that Condition \ref{excel2} is satisfied. Suppose, by way of contradiction, that this claim is false. Then, there exist $t\in[n]$, $i\in[q]$ and $j\in[a^i]$ such that $\mu(m^i_j)=\widehat{w}^i_j$, $x_t\in P^i_j$ and $\mu(m_t)\neq \widehat{w}_t$. Since $\mu$ is a good matching, we then have that $\mu(m_t)=w_t$. However, $w_t$ prefers $m^i_j$ over $m_t$, and $m^i_j$ prefers $w_t$ over $\widehat{w}^i_j$, which contradicts the fact that $\mu$ is stable. Thus, Condition \ref{excel2} is satisfied as well.

Third, we claim that Condition \ref{excel3} is satisfied, which would conclude the proof of the lemma. Suppose, by way of contradiction, that this claim is false. Then, there exist $i\in[q]$ and $j,k\in[a^i]$ such that $j\neq k$, $\mu(m^i_j)=\widehat{w}^i_j$ and $\mu(\overline{m}^i_k)\neq\widehat{\overline{w}}^i_k$. Since $\mu$ is a good matching, we then have that $\mu(\overline{m}^i_k)=\overline{w}^i_k$. However, $\overline{w}^i_k$ prefers $m^i_j$ over $\overline{m}^i_k$, and $m^i_j$ prefers $\overline{w}^i_k$ over $\widehat{w}^i_j$, which contradicts the fact that $\mu$ is stable. Thus, Condition \ref{excel3} is satisfied.
\end{proof}

\begin{lemma}\label{lem:sethSESMCaptureEx2}
Let $\varphi$ be an instance of {\sc $s$-Sparse $p$-CNF-SAT}. Let $\mu$ be an excellent matching of $\red_{SESM}(\varphi)$. Then, $\mu$ is a stable matching.
\end{lemma}

\begin{proof}
Since all happy agents are matched to agents at position 1 in their preference lists, they cannot belong to any blocking pair. Hence, we also have that $m^\star$ and $w^\star$ cannot belong to any blocking pair. For all $t\in[n]$, since $\mu$ is in particular a good matching, either both $m_t$ and $\widehat{m}_t$ are matched to the women at position 1 in their preference lists, in which case they cannot belong to any blocking pair, or $\mu(m_t)=\widehat{w}_t$ and $\mu(\widehat{m}_t)=w_t$. In the latter case, by Condition \ref{excel1} in Definition \ref{def:sethSESMExcellent}, we have that all women that $m_t$ prefers over $\widehat{w}_t$ are matched to men that they prefer over $m_t$, and therefore $m_t$ cannot belong to any blocking pair. Since $\widehat{w}_t$ is matched to the man at position 1 in her preference list and this is the only woman that $\widehat{m}_t$ prefers over $w_t$, we have that $\widehat{m}_t$ cannot belong to any blocking pair. 

Now, consider some $i\in[q]$ and $j\in[a^i]$. Since $\mu$ is in particular a good matching, either both $m^i_j$ and $\widehat{m}^i_j$ are matched to the women at position 1 in their preference lists, in which case they cannot belong to any blocking pair, or $\mu(m^i_j)=\widehat{w}^i_j$ and $\mu(\widehat{m}^i_j)=w^i_j$. In the latter case, by Conditions \ref{excel2} and \ref{excel3} in Definition \ref{def:sethSESMExcellent}, we have that all women that $m^i_j$ prefers over $\widehat{w}_t$ are matched to men that they prefer over $m^i_j$, and therefore $m^i_j$ cannot belong to any blocking pair. Since $\widehat{w}^i_j$ is matched to the man at position 1 in her preference list and this is the only woman that $\widehat{m}^i_j$ prefers over $w^i_j$, we have that $\widehat{m}^i_j$ cannot belong to any blocking pair. Finally, since $\mu$ is in particular a good matching, either both $\overline{m}^i_j$ and $\widehat{\overline{m}}^i_j$ are matched to the women at position 1 in their preference lists, in which case they cannot belong to any blocking pair, or $\mu(\overline{m}^i_j)=\widehat{\overline{w}}^i_j$ and $\mu(\widehat{\overline{m}}^i_j)=w^i_j$. However, the only two non-happy women in the preference lists of $\overline{m}^i_j$ and $\widehat{\overline{m}}^i_j$ are $w^i_j$ and $\widehat{\overline{w}}^i_j$, who are matched to men at position 1 in their preference lists. Thus, neither $\overline{m}^i_j$ nor $\widehat{\overline{m}}^i_j$ can belong to any blocking pair. Since the choices of $i$ and $j$ were arbitrary, we overall conclude that no man belongs to any blocking pair, and therefore $\mu$ is a stable matching.
\end{proof}

From Lemmata \ref{lem:sethSESMCaptureEx1} and \ref{lem:sethSESMCaptureEx2}, we derive the following corollary. Here, recall that $\cal S$ denotes the set of all stable matchings.

\begin{corollary}\label{cor:sethSESMCaptureEx}
Let $\varphi$ be an instance of {\sc $s$-Sparse $p$-CNF-SAT}. Then, in the context of $\red_{SESM}(\varphi)$, ${\cal S}=\Lambda$.
\end{corollary}

To argue about the sex-equality measure of stable matchings, we rely on the following definition.
\begin{definition}\label{def:sethSESMMeasure}
Let $\varphi$ be an instance of {\sc $s$-Sparse $p$-CNF-SAT}, and let $\mu$ be a stable matching of $\red_{SESM}(\varphi)$. For all $i\in[q]$, denote $a(\mu,i)=|\{j\in [a^i]: \mu(m^i_j)=\widehat{w}^i_j\}|$. Moreover, denote $b(\mu)=|(i,j): i\in[q], j\in[a^i], \mu(\overline{m}^i_j)=\widehat{\overline{w}}^i_j\}|$.
\end{definition}

\begin{lemma}\label{lem:sethGenMeasure}
Let $\varphi$ be an instance of {\sc $s$-Sparse $p$-CNF-SAT}, and let $\mu$ be a stable matching of $\red_{SESM}(\varphi)$. Then, there exist $0\leq x,y\leq \displaystyle{100s4^{pd}\cdot n^2}=\OO(n^2)$ such that the two following conditions hold.
\begin{itemize}
\item $\sat_M(\mu)=2\alpha+\displaystyle{\sum_{i\in[q]}a(\mu,i)\gamma(i)} + b(\mu)\tau + x$.
\item $\sat_W(\mu)=\alpha+\displaystyle{\sum_{i\in[q]}(a^i-a(\mu,i))\lambda(i)} + (\widetilde{a}-b(\mu))\tau + y$.
\end{itemize}
\end{lemma}

\begin{proof}
Note that for all $i\in[q]$ and $j\in[a^i]$, $|P^i_j|,|N^i_j|\leq |X^i|\leq \displaystyle{\frac{pr}{q}} =pd$, and that for all $t\in[n]$, $|{\cal N}(x_t)|,|{\cal P}(x_t)|\leq \widetilde{a}\leq\displaystyle{q2^{pd}}$. Thus, on the one hand, by Corollary \ref{cor:sethSESMCaptureEx} and the definition of the preference lists of the men of $\red_{SESM}(\varphi)$, we have that
\[\begin{array}{lll}
\sat_M(\mu) & = & \displaystyle{\sum_{t\in [n], \mu(m_t)=w_t}2 + \sum_{t\in [n], \mu(m_t)=\widehat{w}_t}(4+|{\cal N}(x_t)|)}\\
&& + \displaystyle{\sum_{i\in[q],j\in[a^i],\mu(m^i_j)=w^i_j}2 + \sum_{i\in[q],j\in[a^i],\mu(m^i_j)=\widehat{w}^i_j}(3+|P^i_j|+a^i+\gamma(i))}\\
&& + \displaystyle{\sum_{i\in[q],j\in[a^i],\mu(\overline{m}^i_j)=\overline{w}^i_j}2 + \sum_{i\in[q],j\in[a^i],\mu(\overline{m}^i_j)=\widehat{\overline{w}}^i_j}(4+\tau)} + 2\alpha+1\\

& = & 2\alpha+\displaystyle{\sum_{i\in[q]}a(\mu,i)\gamma(i) + b(\mu)\tau + x,}
\end{array}\]
for some $0\leq x\leq 1 + \displaystyle{(4 + q2^{pd})n + (7+pd+2^{pd})\widetilde{a}}\leq \displaystyle{5q2^{pd}n + 9\cdot2^{pd}\cdot q2^{pd}}\leq \displaystyle{100s4^{pd}\cdot n^2}$.

On the other hand, by Corollary \ref{cor:sethSESMCaptureEx} and the definition of the preference lists of the women of $\red_{SESM}(\varphi)$, we have that
\[\begin{array}{lll}
\sat_W(\mu) & = & \displaystyle{\sum_{t\in [n], \mu(w_t)=m_t}(4+|{\cal P}(x_t)|) + \sum_{t\in [n], \mu(w_t)=\widehat{m}_t}2}\\
&& + \displaystyle{\sum_{i\in[q],j\in[a^i],\mu(w^i_j)=m^i_j}(4+\lambda(i)) + \sum_{i\in[q],j\in[a^i],\mu(w^i_j)=\widehat{m}^i_j}2}\\
&& + \displaystyle{\sum_{i\in[q],j\in[a^i],\mu(\overline{w}^i_j)=\overline{m}^i_j}(3+|N^i_j|+a^i+\tau) + \sum_{i\in[q],j\in[a^i],\mu(\overline{w}^i_j)=\widehat{\overline{m}}^i_j}2} + \alpha+1\\

& = & \alpha+\displaystyle{\sum_{i\in[q]}(a^i-a(\mu,i))\lambda(i) + (\widetilde{a}-b(\mu))\tau + y,}
\end{array}\]
for some $0\leq y\leq \displaystyle{100s4^{pd}\cdot n^2}$.
\end{proof}

\subsubsection{Correctness}\label{sec:sethSESMCor}

\myparagraph{Forward Direction.} We first show how given a truth assignment for an instance $\varphi$ of {\sc $s$-Sparse $p$-CNF-SAT} that satisfies $\varphi$, we can construct a stable matching $\mu$ of $\red_{SESM}(\varphi)$ whose sex-equality measure is at most $100s4^{pd}\cdot n^2$. For this purpose, we introduce the following definition. Here, given a truth assignment $f$, we let $X(f)$ denote the variables to which $f$ assigns true.

\begin{definition}\label{def:sethSESMmu}
Let $\varphi$ be a \yesinstance\ of {\sc $s$-Sparse $p$-CNF-SAT}, and let $f$ a truth assignment that satisfies $\varphi$. Then, the matching $\mu^f_{SESM}$ of $\red_{SESM}(\varphi)$ is defined as follows.
\begin{itemize}
\item For all $t\in[n]$ such that $x_t\in X(f)$: $\mu^f_{SESM}(m_t)=\widehat{w}_t$ and $\mu^f_{SESM}(\widehat{m}_t)=w_t$.
\item For all $t\in[n]$ such that $x_t\notin X(f)$: $\mu^f_{SESM}(m_t)=w_t$ and $\mu^f_{SESM}(\widehat{m}_t)=\widehat{w}_t$.
\item For all $i\in[q]$ and $j\in[a^i]$ such that $P^i_j=X(f)\cap X^i$: $\mu^f_{SESM}(m^i_j)=\widehat{w}^i_j$ and $\mu^f_{SESM}(\widehat{m}^i_j)=w^i_j$.
\item For all $i\in[q]$ and $j\in[a^i]$ such that $P^i_j\neq X(f)\cap X^i$: $\mu^f_{SESM}(m^i_j)=w^i_j$ and $\mu^f_{SESM}(\widehat{m}^i_j)=\widehat{w}^i_j$.
\item For all $i\in[q]$ and $j\in[a^i]$ such that $N^i_j\neq X^i\setminus X(f)$: $\mu^f_{SESM}(\overline{m}^i_j)=\widehat{\overline{w}}^i_j$ and $\mu^f_{SESM}(\widehat{\overline{m}}^i_j)=\overline{w}^i_j$.
\item For all $i\in[q]$ and $j\in[a^i]$ such that $N^i_j=X^i\setminus X(f)$: $\mu^f_{SESM}(\overline{m}^i_j)=\overline{w}^i_j$ and $\mu^f_{SESM}(\widehat{\overline{m}}^i_j)=\widehat{\overline{w}}^i_j$.
\item For all $i\in[\alpha]$: $\mu^f_{SESM}(m^i_{\hap})=w^i_{\hap}$.
\item $\mu^f_{SESM}(m^\star)=w^\star$.
\end{itemize}
\end{definition}

Let us first argue that $\mu^f_{SESM}$ is a stable matching.

\begin{lemma}\label{lem:sethSESMforwardSM}
Let $\varphi$ be a \yesinstance\ of {\sc $s$-Sparse $p$-CNF-SAT}. Let $f$ be a truth assignment that satisfies $\varphi$. Then, $\mu^f_{SESM}$ is an excellent matching of $\red_{SESM}(\varphi)$.
\end{lemma}

\begin{proof}
Definition \ref{def:sethSESMmu} directly implies that $\mu$ is a good matching. Next, we verify that Conditions \ref{excel1}--{excel3} in Definition \ref{def:sethSESMExcellent} are satisfied as well. First, consider some $t\in[n]$ such that $\mu(m_t)=\widehat{w}_t$, $i\in[q]$ and $j\in[a^i]$ such that $x_t\in N^i_j$. Since $\mu(m_t)=\widehat{w}_t$, we have that $x_t\in X(f)$, which implies that $x_t\notin X^i\setminus X(f)$. Thus, $N^i_j\neq X^i\setminus X(f)$, which implies that $\mu(\overline{m}^i_j)=\widehat{\overline{w}}^i_j$. Hence, Condition \ref{excel1} is satisfied. Second, consider some $i\in[q]$, $j\in[a^i]$ such that $\mu(m^i_j)=\widehat{w}^i_j$ and $t\in[n]$ such that $x_t\in P^i_j$. Since $\mu(m^i_j)=\widehat{w}^i_j$, we have that $P^i_j=X(f)\cap X^i$. Thus, $x_t\in X(f)$, which implies that $\mu(m_t)=\widehat{w}_t$. Hence, Condition \ref{excel2} is satisfied. Third, consider some consider some $i\in[q]$, $j\in[a^i]$ such that $\mu(m^i_j)=\widehat{w}^i_j$ and $k\in[a^i]$ such that $k\neq j$. Since $\mu(m^i_j)=\widehat{w}^i_j$, we again have that $P^i_j=X(f)\cap X^i$. Therefore, $N^i_j=X^i\setminus X(f)$. Since $N^i_k\neq N^i_j$, we have that $N^i_k\neq X^i\setminus X(f)$, and thus $\mu(\overline{m}^i_k)=\widehat{\overline{w}}^i_k$. Hence, Condition \ref{excel3} is satisfied. We thus conclude that $\mu^f_{SESM}$ is an excellent matching.
\end{proof}

By Lemma \ref{lem:sethSESMCaptureEx2}, we have the following corollary to Lemma \ref{lem:sethSESMforwardSM}.

\begin{corollary}\label{cor:sethSESMforwardSM}
Let $\varphi$ be a \yesinstance\ of {\sc $s$-Sparse $p$-CNF-SAT}. Let $f$ be a truth assignment that satisfies $\varphi$. Then, $\mu^f_{SESM}$ is a stable matching of $\red_{SESM}(\varphi)$.
\end{corollary}

In light of Corollary \ref{cor:sethSESMforwardSM}, the measure $\delta(\mu^f_{SESM})$ is well defined. We proceed to analyze this measure with the following lemma.

\begin{lemma}\label{lem:sethSESMforwardMeasure}
Let $\varphi$ be a \yesinstance\ of {\sc $s$-Sparse $p$-CNF-SAT}. Let $f$ be a truth assignment that satisfies $\varphi$. Then, $\delta(\mu^f_{SESM})\leq 100s4^{pd}\cdot n^2$.
\end{lemma}

\begin{proof}
By Lemma \ref{lem:sethGenMeasure}, there exist $0\leq x,y\leq \displaystyle{100s4^{pd}\cdot n^2}$ such that the two following conditions hold.
\begin{itemize}
\item $\sat_M(\mu^f_{SESM})=2\alpha+\displaystyle{\sum_{i\in[q]}a(\mu^f_{SESM},i)\gamma(i)} + b(\mu^f_{SESM})\tau + x$.
\item $\sat_W(\mu^f_{SESM})=\alpha+\displaystyle{\sum_{i\in[q]}(a^i-a(\mu^f_{SESM},i))\lambda(i)} + (\widetilde{a}-b(\mu^f_{SESM}))\tau + y$.
\end{itemize}

In the case of $\mu^f_{SESM}$, we have that for all $i\in[q]$, $a(\mu^f_{SESM})=1$, and also $b(\mu^f_{SESM})=\widetilde{a}-q$. Thus, we have that the two following conditions hold.
\begin{itemize}
\item $\sat_M(\mu^f_{SESM})=2\alpha+\displaystyle{\sum_{i\in[q]}\gamma(i)} + (\widetilde{a}-q)\tau + x$.
\item $\sat_W(\mu^f_{SESM})=\alpha+\displaystyle{\sum_{i\in[q]}(a^i-1)\lambda(i)} + q\tau + y$.
\end{itemize}

Now, note that $\alpha = \displaystyle{\sum_{i\in[q]}((a^i-1)\lambda(i)-\gamma(i))} + (2q-\widetilde{a})\tau = \displaystyle{\left(\sum_{i\in[q]}(a^i-1)\lambda(i) + q\tau\right)} - \left(\displaystyle{\sum_{i\in[q]}\gamma(i)} + (\widetilde{a}-q)\tau\right) = (\sat_W(\mu^f_{SESM})-y) - (\sat_M(\mu^f_{SESM})-\alpha-x)$. Thus, we have that
\[\begin{array}{lll}
\delta(\mu^f_{SESM}) & = & |\sat_W(\mu^f_{SESM}) - \sat_M(\mu^f_{SESM})|\\

& = & |y-x-\alpha + \left((\sat_W(\mu^f_{SESM})-y) - (\sat_M(\mu^f_{SESM})-\alpha-x)\right)|\\

& = & |y-x|\leq 100s4^{pd}\cdot n^2.
\end{array}\]

This concludes the proof of the lemma.
\end{proof}

Combining Corollary \ref{cor:sethSESMforwardSM} and Lemma \ref{lem:sethSESMforwardMeasure}, we derive the following corollary.

\begin{corollary}\label{cor:sethSESMforward}
Let $\varphi$ be a \yesinstance\ of {\sc $s$-Sparse $p$-CNF-SAT}. Then, for the instance $\red_{SESM}(\varphi)$ of {\sc SESM}, $\Delta\leq 100s4^{pd}\cdot n^2$. Here, $n$ is the number of variables.
\end{corollary}

This concludes the proof of the forward direction.

\medskip
\myparagraph{Reverse Direction.} Second, we prove that given an instance $\varphi$ of {\sc $s$-Sparse $p$-CNF-SAT}, if for the instance $\red_{SESM}(\varphi)$ of {\sc SESM}, $\Delta$ is ``low'' (namely, at most $100s4^{pd}\cdot n^2$), then we can construct a truth assignment that satisfies $\varphi$. To this end, we first need to analyze the structure of stable matchings of $\red_{SESM}(\varphi)$ whose sex-equality measure is low. Let us begin by proving the following lemma.

\begin{lemma}\label{lem:sethSESMEasy}
Let $\varphi$ be a instance of {\sc $s$-Sparse $p$-CNF-SAT}. Let $\mu$ be a stable matching of $\red_{SESM}(\varphi)$ such that $\delta(\mu)\leq 100s4^{pd}\cdot n^2$. Then, the two following equalities are satisfied.
\begin{itemize}
\item $\displaystyle{\sum_{i\in[q]}\left((1-a(\mu,i))\lambda(i) - (1-a(\mu,i))\gamma(i)\right)=0}$.
\item $\displaystyle{(\widetilde{a}-b(\mu)-q)\tau = 0}$.
\end{itemize}
\end{lemma}

\begin{proof}
Since $\delta(\mu)\leq 100s4^{pd}\cdot n^2$, it holds that $|\sat_W(\mu)-\sat_M(\mu)|\leq 100s4^{pd}\cdot n^2$. Recall that $s,p,d=\OO(1)$. Now, notice that $\tau = n^{10} > 100s4^{pd}\cdot n^2$, else the problem is solvable in polynomial time. Moreover, for all $i\in[q]$, $\lambda(i)$ and $\gamma(i)$ are divisible by $n^{10}$. Thus, Lemma \ref{lem:sethGenMeasure} implies that the following equality is satisfied.
\[\begin{array}{ll}
0 = &\displaystyle{\sum_{i\in[q]}(a^i-a(\mu,i))\lambda(i) + (\widetilde{a}-b(\mu))\tau}\\
 &\displaystyle{- \left(\sum_{i\in[q]}a(\mu,i)\gamma(i) + \sum_{i\in[q]}\left((a^i-1)\lambda(i)-\gamma(i)\right) + b(\mu)\tau + (2q-\widetilde{a})\tau \right)}.
\end{array}\]

The equality above is equivalent to the following equality.
\[\begin{array}{ll}
0 = &\displaystyle{\sum_{i\in[q]}\left((1-a(\mu,i))\lambda(i) - (1-a(\mu,i))\gamma(i)\right)  + 2(\widetilde{a}-b(\mu)-q)\tau.}
\end{array}\]

Next, notice that $0\leq b(\mu)\leq \widetilde{a}$, and hence $-2qn^{10}\leq 2(\widetilde{a}-b(\mu)-q)\tau$ as well as $2(\widetilde{a}-b(\mu)-q)\tau\leq 2(\widetilde{a}-q)n^{10}\leq 2q(2^{pd}-1)n^{10}\leq n^{12}$, where the last inequality is assumed to hold else the problem is solvable in polynomial time. Moreover, for all $i\in[q]$, $\lambda(i),\gamma(i)$ are divisible by $n^{20}$. Thus, we derive that the two equalities given in the statement of the lemma must be satisfied.
\end{proof}

\begin{lemma}\label{lem:sethSESM1PerTruth}
Let $\varphi$ be an instance of {\sc $s$-Sparse $p$-CNF-SAT}. Let $\mu$ be a stable matching of $\red_{SESM}(\varphi)$ such that $\delta(\mu)\leq 100s4^{pd}\cdot n^2$. Then, for all $i\in[q]$, there exists $j\in[a^i]$ such that $\mu(m^i_j)=\widehat{w}^i_j$ and for all $k\neq j$, $\mu(m^i_k)=w^i_k$.
\end{lemma}

\begin{proof}
By Corollary \ref{cor:sethSESMCaptureEx}, the statement of the lemma is equivalent to the statement that for all $i\in[k]$, $a(\mu,i)=1$. By Lemma \ref{lem:sethSESMEasy}, the following equality is satisfied.
\[\displaystyle{\sum_{i\in[q]}\left((1-a(\mu,i))\lambda(i) - (1-a(\mu,i))\gamma(i)\right)=0},\]
which is equivalent to the following one.
\[\displaystyle{\sum_{i\in[q]}a(\mu,i)(\gamma(i)+\lambda(i)) = \sum_{i\in[q]}(\gamma(i)+\lambda(i))}.\]

Substituting $\gamma(i)$ and $\lambda(i)$ for all $i\in[q]$ and dividing both sides by $n^{20}$, we derive that the following equality is satisfied.
 \[\displaystyle{\sum_{i\in[q]}a(\mu,i)(2^{i-1}+2^{2q-i}) = \sum_{i\in[q]}(2^{i-1}+2^{2q-i})}.\]

Since $\displaystyle{\sum_{i\in[q]}(2^{i-1}+2^{2q-i}) = (2^q-1) + 2^q(2^q-1)}$, we overall get that the following equality is satisfied.
\[\displaystyle{\sum_{i\in[q]}a(\mu,i)(2^{i-1}+2^{2q-i}) = 2^{2q}-1}.\]

Let $\psi$ be an assignment to the variables $a(\mu,i)$ that satisfies this condition as well as the two equalities above. We claim that $\psi$ necessarily assigns 1 to all of these variables. This claim can be easily proven by induction on $q$. In the base case, where $q=1$, the first equality directly implies that $a(\mu,1)=1$. Now, suppose that $q\geq 2$ and that the claim holds for $q-1$. Then, first note if $a(\mu,1)=0$, then the left side of the equality would have been an even number while the right side is an odd number, and therefore $a(\mu,1)\geq 1$. Now, observe that if $a(\mu,1)\geq 2$, then $\displaystyle{\sum_{i\in[q]}a(\mu,i)(2^{i-1}+2^{2q-i})}\geq 2+2\cdot 2^{2q-1} > 2^{2q}-1$, which is a contradiction. Thus, we have that $a(\mu,1)=1$. We get that the following equality is then satisfied, where for all $i\in [q-1]$, we denote $\widehat{a}(\mu,i) = a(\mu,i+1)$.
\[\displaystyle{(1+2^{2q-1}) + 2\sum_{i\in[q-1]}\widehat{a}(\mu,i)(2^{i-1}+2^{2q-i}) = 2^{2q}-1},\]
which is equivalent to the following one.
\[\displaystyle{\sum_{i\in[q-1]}\widehat{a}(\mu,i)(2^{i-1}+2^{2q-i}) = 2^{2(q-1)}-1}.\]

By the inductive hypothesis, we have that for all $i\in[q-1]$, it holds that $\widehat{a}(\mu,i)=1$. Therefore, for all $i\in\{2,3,\ldots,q\}$, it holds that $a(\mu,i)=1$. This concludes the proof of the lemma.
\end{proof}

\begin{lemma}\label{lem:sethSESM1PerFalse}
Let $\varphi$ be an instance of {\sc $s$-Sparse $p$-CNF-SAT}. Let $\mu$ be a stable matching of $\red_{SESM}(\varphi)$ such that $\delta(\mu)\leq 100s4^{pd}\cdot n^2$. Then, for all $i\in[q]$, there exists $j\in[a^i]$ such that $\mu(\overline{m}^i_j)=\overline{w}^i_j$ and for all $k\neq j$, $\mu(\overline{m}^i_k)=\widehat{\overline{w}}^i_k$.
\end{lemma}

\begin{proof}
First, notice that by Lemma \ref{lem:sethSESMEasy}, it holds that $\displaystyle{(\widetilde{a}-b(\mu)-q)\tau = 0}$. Therefore, $b(\mu)=\widetilde{a}-q$. Hence, by Corollary \ref{cor:sethSESMCaptureEx}, it is sufficient to show that for all $i\in[q]$, there do not exist distinct $j,k\in[a^i]$ such that both $\mu(\overline{m}^i_j)=\overline{w}^i_j$ and $\mu(\overline{m}^i_k)=\overline{w}^i_k$. Suppose, by way of contradiction, that this claim is false. That is, there exist $i\in[q]$ and distinct $j,k\in[a^i]$ such that both $\mu(\overline{m}^i_j)=\overline{w}^i_j$ and $\mu(\overline{m}^i_k)=\overline{w}^i_k$. Then, by Corollary \ref{cor:sethSESMCaptureEx} and Condition \ref{excel3} in Definition \ref{def:sethSESMExcellent}, there do not exist $\ell\in[a^i]$ such that $\mu(m^i_j)=\widehat{w}^i_j$. However, this contradicts Lemma \ref{lem:sethSESM1PerTruth}, and thus we conclude that the lemma is correct.
\end{proof}

We are now ready to prove the correctness of the reverse direction.

\begin{lemma}\label{lem:sethSESMReverse}
Let $\varphi$ be an instance of {\sc $s$-Sparse $p$-CNF-SAT}. If for the instance $\red_{SESM}(\varphi)$ of {\sc SESM}, $\Delta\leq 100s4^{pd}\cdot n^2$, then $\varphi$ is a \yesinstance\ of {\sc $s$-Sparse $p$-CNF-SAT}.
\end{lemma}

\begin{proof}
Suppose that for the instance $\red_{SESM}(\varphi)$ of {\sc SESM}, $\Delta\leq 100s4^{pd}\cdot n^2$. Then, there exists a stable matching $\mu$ such that $\delta(\mu)\leq 100s4^{pd}\cdot n^2$. Note that by Corollary \ref{cor:sethSESMCaptureEx}, $\mu$ is an excellent matching. Let $f$ denote the truth assignment such that for all $t\in[n]$, $f$ assigns true to $x_t$ if and only if $\mu(m_t)=\widehat{w}_t$. Recall that $X(f)$ denotes the set of variables to which $f$ assigns true. We claim that $f$ satisfies $\varphi$. By the definition of the partial truth assignments, to show that $f$ satisfies $\varphi$, it is sufficient to show that for all $i\in[q]$, there exists $j\in[a^i]$ such that $P^i_j\subseteq X(f)\cap X^i$ and $N^i_j\subseteq X^i\setminus X(f)$ (in which case $P^i_j=X(f)\cap X^i$). Indeed, it then holds that for all $i\in[q]$, $f$ satisfies all of the clauses in ${\cal C}^i$ as when $f$ is restricted to $X^i$, it would be identical to an assignment in $F^i$. Notice that the proof of our claim would also conclude the proof of the lemma, as it implies that $\varphi$ is a \yesinstance\ of {\sc $s$-Sparse $p$-CNF-SAT}.

To prove our claim, we first note that by Lemma \ref{lem:sethSESM1PerTruth}, for all $i\in[q]$, there exists $\ell_i\in[a^i]$ such that $\mu(m^i_{\ell_i})=\widehat{w}^i_j$ and for all $j\neq \ell_i$, $\mu(m^i_j)=w^i_j$. Moreover, by Lemma \ref{lem:sethSESM1PerFalse} and Condition \ref{excel3} in Definition \ref{def:sethSESMExcellent}, for all $i\in[q]$, we have that $\mu(\overline{m}^i_{\ell_i})=\overline{w}^i_{\ell_i}$ and for all $j\neq \ell_i$, $\mu(\overline{m}^i_j)=\widehat{\overline{w}}^i_j$. Now, choose some arbitrary $t\in [n]$ and $i\in[q]$ such that $x_t\in X^i$. We show that if $x_t\in P^i_{\ell_i}$ then $\mu(m_t)=\widehat{w}_t$, and otherwise ($x_t\in N^i_{\ell_i}$) it holds that $\mu(m_t)\neq\widehat{w}_t$. First, suppose that $x_t\in P^i_j$. Then, by  Condition \ref{excel2} in Definition \ref{def:sethSESMExcellent}, it indeed holds that $\mu(m_t)=\widehat{w}_t$. Second, suppose that $x_t\in N^i_{\ell_i}$. Then, by Condition \ref{excel1} in Definition \ref{def:sethSESMExcellent}, it indeed holds that $\mu(m_t)\neq\widehat{w}_t$. As we have argued earlier, this concludes the proof of the lemma.
\end{proof}

\input{sethSESMRot.tex}

\subsubsection{Treewidth}

We are now ready to bound the treewidth of $G_\Pi$, the undirected underlying graph of the rotation digraph $D_\Pi$.

\begin{lemma}\label{lem:twRotSESM}
Let $\varphi$ be an instance of {\sc $s$-Sparse $p$-CNF-SAT}. Then, in the context of the instance $\red_{SESM}(\varphi)$ of {\sc SESM}, the treewidth of $G_\Pi$ is bounded by $n+2\cdot 2^{pd}$.
\end{lemma}

\begin{proof}
By Lemma \ref{lem:rotFinal}, $G_\Pi$ is a subgraph of the underlying undirected graph of $H_\Pi$, which we denote by $\widetilde{H}$. Thus, to prove that the treewidth of $G_\Pi$ is bounded by $n+2^{pd}$, it is sufficient to prove that the treewidth of $\widetilde{H}$ is bounded by $n+2^{pd}$. Let $H'$ denote the graph obtained from $\widetilde{H}$ by the removal of all of vertices in $R_2$. Note that $|R_2|=n$. Hence, to prove that the treewidth of $\widetilde{H}$ is bounded by $n+2^{pd}$, it is sufficient to prove that the treewidth of every connected component of $H'$ is bounded by $2\cdot 2^{pd}$. However, every connected component of $H'$ consists only of the vertices in $\{\overline{\rho}^i_j: j\in[a^i]\}\cup\{\rho^i_j: j\in[a^i]\}$ for some $i\in[q]$. Note that for all $i\in[q]$, $a^i\leq 2^{|X^i|}=2^{pd}$. Thus, the size of every connected component of $H'$ is bounded by $2\cdot 2^{pd}$, and therefore it is clear that the treewidth of every connected component of $H'$ is bounded by $2\cdot 2^{pd}$ as well.
\end{proof}

\subsubsection{Running Time}\label{sec:sethRunningTime}

Let us first bound the number of agents constructed by our reduction, assuming that the size of the input is not bounded by some fixed constant (since $d$ is a fixed constant).

\begin{observation}\label{obs:sethSESMNumAgents}
Let $\varphi$ be an instance of {\sc $s$-Sparse $p$-CNF-SAT}. Then, in the context of the instance $\red_{SESM}(\varphi)$ of {\sc SESM}, the number of agents, $|A|$, is exactly $4(n + 2q2^{pd})+\alpha+1$, which is upper bounded by $8^q$.
\end{observation}

Next, we define an algorithm \algSAT\ for {\sc $s$-Sparse $p$-CNF-SAT} as follows. Given an instance $\varphi$ of {\sc $s$-Sparse $p$-CNF-SAT}, \algSAT\ construct the instance $\red_{SESM}(\varphi)$ of {\sc SESM}. Then, it calls \algSESM\ with $\red_{SESM}(\varphi)$ as input. If \algSESM\ determines that $\Delta\leq 100s4^{pd}\cdot n^2$, then \algSAT\ determines that $\varphi$ is satisfiable, and otherwise it determines that $\varphi$ is not satisfiable.

\begin{lemma}
If \algSESM\ exists, then \algSAT\ solves {\sc $s$-Sparse $p$-CNF-SAT} in time $2^{\delta n}$.
\end{lemma}

\begin{proof}
By Corollary \ref{cor:sethSESMforward} and Lemma \ref{lem:sethSESMReverse}, \algSAT\ solves {\sc $s$-Sparse $p$-CNF-SAT} correctly. By the running time of \algSESM, we note that \algSAT\ runs in time $\OO(2^{\epsilon\tw}|A|^c)$. By Lemma \ref{lem:twRotSESM} and Observation \ref{obs:sethSESMNumAgents}, we have that \algSAT\ runs in time $\displaystyle{\OO(2^{\epsilon(n+2\cdot 2^{pd})}8^{\frac{cs}{d}n})}$. Note that since $d=\OO(1)$, we further have that \algSAT\ runs in time $\displaystyle{\OO(2^{(\epsilon + \frac{3cs}{d})n})}$. Since $\delta \geq  \epsilon + \frac{3cs}{d}$, it holds that \algSAT\ runs in time $\displaystyle{\OO(2^{\delta n})}$.
\end{proof}

Note that given the existence of \algSAT\ that solves {\sc $s$-Sparse $p$-CNF-SAT} in time $2^{\delta n}$, we have that \SETH\ fails. Hence, we conclude that unless \SETH\ fails, {\sc SESM} cannot be solved in time $(2-\epsilon)^{\tw}\cdot n^{\OO(1)}$ for any fixed $\epsilon>0$.

%% file: sethSESMRot.tex
\subsubsection{Rotation Digraph}

By Observation \ref{obs:manOpt1} and Corollary \ref{cor:sethSESMCaptureEx}, we directly derive the following observation.

\begin{observation}\label{obs:manOpt}
Let $\varphi$ be an instance of {\sc $s$-Sparse $p$-CNF-SAT}. Then, the man-optimal stable matching of $\red_{SESM}(\varphi)$ is $\mu_{\emptyset}$.
\end{observation}

Now, we define three sequences of pairs, which we would later prove to capture rotations.

\begin{definition}\label{def:threeRots}
Let $\varphi$ be an instance of {\sc $s$-Sparse $p$-CNF-SAT}. Then, in the context of $\red_{SESM}(\varphi)$, the three sets $R_1,R_2$ and $R_3$ are defined as follows.
\begin{itemize}
\item $R_1=\{\overline{\rho}^i_j=((\overline{m}^i_j,\overline{w}^i_j),(\widehat{\overline{m}}^i_j,\widehat{\overline{w}}^i_j)): i\in[q], j\in[a^i]\}$.
\item $R_2=\{\rho_t=((m_t,w_t),(\widehat{m}_t,\widehat{w}_t)): t\in[n]\}$.
\item $R_3=\{\rho^i_j=((m^i_j,w^i_j),(\widehat{m}^i_j,\widehat{w}^i_j)): i\in[q], j\in[a^i]\}$.
\end{itemize}
\end{definition}

We further define which combinations of subsets of $R_1,R_2$ and $R_3$ would be relevant to us.

\begin{definition}\label{def:legalRots}
Let $\varphi$ be an instance of {\sc $s$-Sparse $p$-CNF-SAT}. Then, in the context of $\red_{SESM}(\varphi)$, a set $R\subseteq R_1\cup R_2\cup R_3$ is {\em legal} if it satisfies the following conditions.
\begin{enumerate}
\item \label{legal1} For all $t\in[n]$ such that $\rho_t\in R$: For all $i\in[q]$ and $j\in[a^i]$ such that $x_t\in N^i_j$, $\overline{\rho}^i_j\in R$. 
\item For all $i\in[q]$ and $j\in[a^i]$ such that $\rho^i_j\in R$, the two following conditions are satisfied.
	\begin{enumerate}
	\item\label{legal2} For all $t\in[n]$ such that $x_t\in P^i_j$, $\rho_t\in R$.
	\item\label{legal3} For all $k\in[a^i]$ such that $k\neq j$, $\overline{\rho}^i_k\in R$. 
	\end{enumerate}
\end{enumerate}
Define ${\cal L}=\{R\subseteq R_1\cup R_2\cup R_2 : R$ is legal$\}$.
\end{definition}

With the lemma, we begin to analyze the relations between ${\cal L}$ and ${\cal S}$.

\begin{lemma}\label{lem:rot1}
Let $\varphi$ be an instance of {\sc $s$-Sparse $p$-CNF-SAT}. Let $R\in {\cal L}$ be a legal set of $\red_{SESM}(\varphi)$. Then, every sequence in $R_1\cap R$ is a $\mu_{\emptyset}$-rotation, every sequence in $R_2\cap R$ is a $\mu_{R_1\cap R}$-rotation, and every sequence in $R_3\cap R$ is a $\mu_{(R_1\cup R_2)\cap R}$-rotation.
\end{lemma}

\begin{proof}
First, note that for all $\overline{\rho}^i_j=((\overline{m}^i_j,\overline{w}^i_j),(\widehat{\overline{m}}^i_j,\widehat{\overline{w}}^i_j))\in R_1$, it holds that $\widehat{\overline{w}}^i_j = s_{\mu_\emptyset}(\overline{m}^i_j)$ and $\overline{w}^i_j = s_{\mu_\emptyset}(\widehat{\overline{m}}^i_j)$. Therefore, by Definition \ref{def:rotation}, every sequence in $R_1\cap R$ is a $\mu_{\emptyset}$-rotation. We thus have that $\mu_{R_1\cap R}$ is well defined as a stable matching. Now, consider some $\rho_t=((m_t,w_t),(\widehat{m}_t,\widehat{w}_t))\in R_2\cap R$. By Condition \ref{legal1} in Definition \ref{def:legalRots}, we have that $\widehat{w}_t= s_{\mu_{R_1\cap R}}(m_t)$ and $w_t= s_{\mu_{R_1\cap R}}(\widehat{m}_t)$. Therefore, by Definition \ref{def:rotation}, every sequence in $R_2\cap R$ is a $\mu_{R_1\cap R}$-rotation. We thus have that $\mu_{(R_1\cup R_2)\cap R}$ is well defined as a stable matching. Finally, consider some $\rho^i_j=((m^i_j,w^i_j),(\widehat{m}^i_j,\widehat{w}^i_j))\in R_3\cap R$. By Conditions \ref{legal2} and \ref{legal3} in in Definition \ref{def:legalRots}, we have that $\widehat{w}^i_j= s_{\mu_{(R_1\cup R_2)\cap R}}(m^i_j)$ and $w^i_j= s_{\mu_{(R_1\cup R_2)\cap R}}(\widehat{m}^i_j)$. Therefore, by Definition \ref{def:rotation}, every sequence in $R_3\cap R$ is a $\mu_{(R_1\cup R_2)\cap R}$-rotation.
\end{proof}

For the complementary direction of Lemma \ref{lem:rot1}, we also require the following definition.
\begin{definition}\label{def:matchRots}
Let $\varphi$ be an instance of {\sc $s$-Sparse $p$-CNF-SAT}, and let $\mu$ be a stable matching of $\red_{SESM}(\varphi)$. Then, the sets $R_1(\mu), R_2(\mu)$ and $R_3(\mu)$ are defined as follows.
\begin{itemize}
\item $R_1(\mu)=\{\overline{\rho}^i_j\in R_1: \mu(\overline{m}^i_j)=\widehat{\overline{w}}^i_j\}$.
\item $R_2(\mu)=\{\rho_t\in R_2: \mu(m_t)=\widehat{w}_t\}$.
\item $R_3(\mu)=\{\rho^i_j\in R_3: \mu(m^i_j)=\widehat{w}^i_j\}$.
\end{itemize}
\end{definition}

In the context of the following lemma, recall that $R(\mu)$ is a notation defined in Section \ref{sec:prelimsRotDig}.

\begin{lemma}\label{lem:rot2}
Let $\varphi$ be an instance of {\sc $s$-Sparse $p$-CNF-SAT}. For any stable matching $\mu$ of $\red_{SESM}(\varphi)$, $R(\mu)=R_1(\mu)\cup R_2(\mu)\cup R_3(\mu)$.
\end{lemma}

\begin{proof}
The correctness of this lemma directly follows from Corollary \ref{cor:sethSESMCaptureEx} and Definition \ref{def:matchRots}.
\end{proof}

From Lemmata \ref{lem:rot1} and \ref{lem:rot2}, we directly derive the following corollary.

\begin{corollary}\label{cor:rot12}
Let $\varphi$ be an instance of {\sc $s$-Sparse $p$-CNF-SAT}. Then, the set of all rotations of $\red_{SESM}(\varphi)$ is $R_1\cup R_2\cup R_2$. 
\end{corollary}

We have thus identified the vertex set of the rotation digraph $D_\Pi$. Let us now identify a superset of the edge set of $D_\Pi$ as well.

\begin{definition}
Let $\varphi$ be an instance of {\sc $s$-Sparse $p$-CNF-SAT}. Then, in the context of $\red_{SESM}(\varphi)$, the sets of ordered pairs of rotations, $E_{12}$, $E_{13}$ and $E_{23}$, are define as follows.
\begin{itemize}
\item $E_{12}=\{(\overline{\rho}^i_j,\rho_t): \overline{\rho}^i_j\in R_1, \rho_t\in R_2\}$.
\item $E_{13}=\{(\overline{\rho}^i_j,\rho^i_j): i\in[q], j\in[a^i]\}$.
\item $E_{23}=\{(\rho_t,\rho^i_j): \rho_t\in R_2, \rho^i_j\in R_3\}$.
\end{itemize}
\end{definition}

We can now easily conclude with the following result.

\begin{lemma}\label{lem:rotFinal}
Let $\varphi$ be an instance of {\sc $s$-Sparse $p$-CNF-SAT}. Then, the rotation digraph $D_\Pi$ of $\red_{SESM}(\varphi)$ is a subgraph of $H_\Pi=(R_1\cup R_2\cup R_2, E_{12}\cup E_{13}\cup E_{23})$.
\end{lemma}

\begin{proof}
Definition \ref{def:legalRots} and Corollary \ref{cor:rot12} directly imply that the transitive closure of $H_\Pi$ is a supergraph of $\Pi$. Thus, by the definition of $D_\Pi$, we conclude the correctness of the lemma.
\end{proof}

%% file: sethBSM.tex
\subsection{Balanced Stable Marriage}\label{sec:sethBSM}

First, we prove that unless \SETH\ fails, {\sc BSM} cannot be solved in time $(2-\epsilon)^{\tw}\cdot n^{\OO(1)}$ for any fixed $\epsilon>0$, where $n$ is the number of agents. Again, this claim is equivalent to the one stating that unless \SETH\ fails, {\sc BSM} cannot be solved in time $2^{\epsilon\tw}\cdot n^{\OO(1)}$ for any fixed $\epsilon<1$. To prove this claim, we suppose, by way of contradiction, that there exist fixed $\epsilon>0$ and $c\geq 1$ as well as an algorithm \algBSM\ such that \algBSM\ solves {\sc BSM} in time $2^{\epsilon\tw}\cdot n^c$.

\subsubsection{Reduction}

Denote $\delta=\epsilon+(1-\epsilon)/2<1$. By Proposition~\ref{prop:SparseSAT}, supposing that \SETH\ is true, there exist integers $p=p(\delta)$ and $s=s(\delta)$ such that {\sc $s$-Sparse $p$-CNF-SAT} cannot be solved in time $\OO(2^{\delta n})$. Let $\varphi$ be an instance of {\sc $p$-CNF-SAT}. We construct an instance $\red_{BSM}(\varphi)=(M,W,\{\pos_m\}|_{m\in M},\{\pos_w\}|_{w\in W})$ of {\sc BSM} in a manner that is identical to the one in which we construct $\red_{SESM}(\varphi)$ except for the following three modifications.
\begin{itemize}
\item For all $i\in[q]$, we replace $\lambda(i)=n^{20}\cdot 2^{2q-i}$ by $\widehat{\lambda}(i)=n^{20}\cdot 4^{i-1}$.

\item We replace $\alpha=\displaystyle{\sum_{i\in[q]}((a^i-1)\lambda(i)-\gamma(i))} + (2q-\widetilde{a})\tau$ by the the following value.
\[\begin{array}{ll}
\widehat{\alpha} & = \displaystyle{\alpha - \sum_{i\in[q]}(a^i-1)\lambda(i) + \sum_{i\in[q]}(a^i-1)\widehat{\lambda}(i)} = \displaystyle{\sum_{i\in[q]}((a^i-1)\widehat{\lambda}(i)-\gamma(i))} + (2q-\widetilde{a})\tau.
\end{array}\]
\end{itemize}

Finally, let us define
\[\begin{array}{ll}
\eta &=  100s4^{pd}\cdot n^2 + \widehat{\alpha} + \displaystyle{\sum_{i\in[q]}(a^i-1)\widehat{\lambda}(i)} + q\tau \\
&= 100s4^{pd}\cdot n^2 + 2\widehat{\alpha}+\displaystyle{\sum_{i\in[q]}\gamma(i)} + (\widetilde{a}-q)\tau.
\end{array}\]

\subsubsection{Proof Modification}

Let us observe that since our modification to the reduction to {\sc SESM} only concern the number of happy pairs in total and is preference lists of some other agents. Hence, the rotation digraphs of $\red_{BSM}(\varphi)$ is identical to the rotation digraph of $\red_{SESM}(\varphi)$. Hence, we immediately derive the following version of Lemma \ref{lem:twRotSESM}.

\begin{lemma}\label{lem:twRotBSM}
Let $\varphi$ be an instance of {\sc $s$-Sparse $p$-CNF-SAT}. Then, in the context of the instance $\red_{BSM}(\varphi)$ of {\sc BSM}, the treewidth of $G_\Pi$ is bounded by $n+2\cdot 2^{pd}$.
\end{lemma}

Furthermore, it still holds that the number of agents in the reduction is upper bounded by $8^q$. Hence, the arguments given in Section \ref{sec:sethRunningTime} imply the if we show that for every instance $\varphi$ of {\sc $s$-Sparse $p$-CNF-SAT}, it holds that $\varphi$ is a \yesinstance\ if and only our reduction is correct in the sense that if the instance $\red_{BSM}(\varphi)$ of {\sc BSM} satisfies $\Bal\leq\eta$, then we would be able to conclude that unless \SETH\ fails, {\sc BSM} cannot be solved in time $(2-\epsilon)^{\tw}\cdot n^{\OO(1)}$ for any fixed $\epsilon>0$. In light of this observation, we next focus only on the proof of correctness of our reduction.

\subsubsection{Correctness}

First, notice that Definitions \ref{def:sethSESMGood}, \ref{def:sethSESMExcellent} and \ref{def:sethSESMMeasure} are also applicable to our current setting, that is, where $\red_{SESM}(\varphi)$ is replaced by $\red_{BSM}(\varphi)$. Hence, exactly as in the case of {\sc SESM} in Section~\ref{sec:AllSMs}, we derive the following results. 

\begin{corollary}\label{cor:sethBSMCaptureEx}
Let $\varphi$ be an instance of {\sc $s$-Sparse $p$-CNF-SAT}. Then, in the context of $\red_{BSM}(\varphi)$, ${\cal S}=\Lambda$.
\end{corollary}

\begin{lemma}\label{lem:sethGenMeasureBSM}
Let $\varphi$ be an instance of {\sc $s$-Sparse $p$-CNF-SAT}, and let $\mu$ be a stable matching of $\red_{BSM}(\varphi)$. Then, there exist $0\leq x,y\leq \displaystyle{100s4^{pd}\cdot n^2}=\OO(n^2)$ such that the two following conditions hold.
\begin{itemize}
\item $\sat_M(\mu)=2\widehat{\alpha}+\displaystyle{\sum_{i\in[q]}a(\mu,i)\gamma(i)} + b(\mu)\tau + x$.
\item $\sat_W(\mu)=\widehat{\alpha}+\displaystyle{\sum_{i\in[q]}(a^i-a(\mu,i))\widehat{\lambda}(i)} + (\widetilde{a}-b(\mu))\tau + y$.
\end{itemize}
\end{lemma}

We next turn to modify the proofs of the forward and reverse directions in Section \ref{sec:sethSESMCor} to handle our current reduction.

\medskip
\myparagraph{Forward Direction.} We first show how given a truth assignment for an instance $\varphi$ of {\sc $s$-Sparse $p$-CNF-SAT} that satisfies $\varphi$, we can construct a stable matching $\mu$ of $\red_{BSM}(\varphi)$ whose balance is at most $100s4^{pd}\cdot n^2$.  Given a satisfying truth assignment $f$ of an instance $\varphi$ of {\sc $s$-Sparse $p$-CNF-SAT}, we define $\mu^f_{BSM}$ exactly as $\mu^f_{SESM}$ (with the modification that we now match a different number of happy agents to one another). Hence, exactly as in the case of {\sc SESM}, we obtain the following corollary.

\begin{corollary}\label{cor:sethBSMforwardSM}
Let $\varphi$ be a \yesinstance\ of {\sc $s$-Sparse $p$-CNF-SAT}. Let $f$ be a truth assignment that satisfies $\varphi$. Then, $\mu^f_{BSM}$ is a stable matching of $\red_{BSM}(\varphi)$.
\end{corollary}

In light of Corollary \ref{cor:sethBSMforwardSM}, the measure $\bal(\mu^f_{BSM})$ is well defined. We proceed to analyze this measure with the following lemma.

\begin{lemma}\label{lem:sethBSMforwardMeasure}
Let $\varphi$ be a \yesinstance\ of {\sc $s$-Sparse $p$-CNF-SAT}. Let $f$ be a truth assignment that satisfies $\varphi$. Then, $\bal(\mu^f_{BSM})\leq \eta$.
\end{lemma}

\begin{proof}
As in the proof of Lemma \ref{lem:sethSESMforwardMeasure}, we obtain that there exist $0\leq x,y\leq \displaystyle{100s4^{pd}\cdot n^2}$ such that the two following conditions are satisfied.
\begin{itemize}
\item $\sat_M(\mu^f_{BSM})=2\widehat{\alpha}+\displaystyle{\sum_{i\in[q]}\gamma(i)} + (\widetilde{a}-q)\tau + x$.
\item $\sat_W(\mu^f_{BSM})=\widehat{\alpha}+\displaystyle{\sum_{i\in[q]}(a^i-1)\widehat{\lambda}(i)} + q\tau + y$.
\end{itemize}

Thus, by the definition of $\eta$, we directly conclude that $\sat_M(\mu^f_{BSM}),\sat_W(\mu^f_{BSM})\leq\eta$. Hence, $\bal(\mu^f_{BSM})\leq \eta$.
\end{proof}

Combining Corollary \ref{cor:sethBSMforwardSM} and Lemma \ref{lem:sethBSMforwardMeasure}, we derive the following corollary.

\begin{corollary}\label{cor:sethBSMforward}
Let $\varphi$ be a \yesinstance\ of {\sc $s$-Sparse $p$-CNF-SAT}. Then, for the instance $\red_{BSM}(\varphi)$ of {\sc BSM}, $\Bal\leq 1\eta$.
\end{corollary}

This concludes the proof of the forward direction.

\medskip
\myparagraph{Reverse Direction.} Second, we prove that given an instance $\varphi$ of {\sc $s$-Sparse $p$-CNF-SAT}, if for the instance $\red_{BSM}(\varphi)$ of {\sc BSM}, $\Bal\leq\eta$, then we can construct a truth assignment that satisfies $\varphi$. We start our analysis with the following lemma.

\begin{lemma}\label{lem:sethBSMEasy}
Let $\varphi$ be a instance of {\sc $s$-Sparse $p$-CNF-SAT}. Let $\mu$ be a stable matching of $\red_{BSM}(\varphi)$ such that $\bal(\mu)\leq \eta$. Then, the following inequalities are satisfied.
\begin{itemize}
\item $\displaystyle{\sum_{i\in[q]}(a(\mu,i)-1)\gamma(i)} \leq 0$.
\item $\displaystyle{\sum_{i\in[q]}(1-a(\mu,i))\widehat{\lambda}(i)} \leq 0$.
\item $\displaystyle{\sum_{i\in[q]}(a(\mu,i)-1)\gamma(i)} + (b(\mu)+q-\widetilde{a})\tau \leq 0$.
\item $\displaystyle{\sum_{i\in[q]}(1-a(\mu,i))\widehat{\lambda}(i)} + (\widetilde{a}-b(\mu)-q)\tau \leq 0$.
\end{itemize}
\end{lemma}

\begin{proof}
Since $\bal(\mu)\leq \eta$, it holds that $\sat_W(\mu),\sat_M(\mu)\leq \eta$. Recall that $s,p,d=\OO(1)$. Now, notice that $\tau = n^{10} > 100s4^{pd}\cdot n^2$, else the problem is solvable in polynomial time. Moreover, for all $i\in[q]$, $\widehat{\lambda}(i)$ and $\gamma(i)$ are divisible by $n^{10}$. Thus, Lemma \ref{lem:sethGenMeasure} implies that the following inequalities are satisfied.
\begin{itemize}
\item $\displaystyle{\sum_{i\in[q]}a(\mu,i)\gamma(i)} + b(\mu)\tau \leq \displaystyle{\sum_{i\in[q]}\gamma(i)} + (\widetilde{a}-q)\tau$.
\item $\displaystyle{\sum_{i\in[q]}(a^i-a(\mu,i))\widehat{\lambda}(i)} + (\widetilde{a}-b(\mu))\tau \leq \displaystyle{\sum_{i\in[q]}(a^i-1)\widehat{\lambda}(i)} + q\tau$.
\end{itemize}

The inequalities above are equivalent to the following equalities.
\begin{itemize}
\item $\displaystyle{\sum_{i\in[q]}(a(\mu,i)-1)\gamma(i)} + (b(\mu)+q-\widetilde{a})\tau \leq 0$.
\item $\displaystyle{\sum_{i\in[q]}(1-a(\mu,i))\widehat{\lambda}(i)} + (\widetilde{a}-b(\mu)-q)\tau \leq 0$.
\end{itemize}

Next, notice that as in the proof of Lemma \ref{lem:sethSESMEasy}, $-2qn^{10}\leq 2(\widetilde{a}-b(\mu)-q)\tau$ as well as $2(\widetilde{a}-b(\mu)-q)\tau\leq n^{12}$. Moreover, for all $i\in[q]$, $\lambda(i),\gamma(i)$ are divisible by $n^{20}$. Thus, we derive that the two first inequalities given in the statement of the lemma must also be satisfied.
\end{proof}

\begin{lemma}\label{lem:sethBSM1PerTruth}
Let $\varphi$ be an instance of {\sc $s$-Sparse $p$-CNF-SAT}. Let $\mu$ be a stable matching of $\red_{BSM}(\varphi)$ such that $\bal(\mu)\leq \eta$. Then, for all $i\in[q]$, there exists $j\in[a^i]$ such that $\mu(m^i_j)=\widehat{w}^i_j$ and for all $k\neq j$, $\mu(m^i_k)=w^i_k$.
\end{lemma}

\begin{proof}
By Corollary \ref{cor:sethBSMCaptureEx}, the statement of the lemma is equivalent to the statement that for all $i\in[k]$, $a(\mu,i)=1$. By Lemma \ref{lem:sethBSMEasy}, the following inequalities are satisfied.
\begin{enumerate}
\item $\displaystyle{\sum_{i\in[q]}(a(\mu,i)-1)\gamma(i)} \leq 0$.
\item $\displaystyle{\sum_{i\in[q]}(1-a(\mu,i))\widehat{\lambda}(i)} \leq 0$.
\end{enumerate}

Substituting $\gamma(i)$ and $\widehat{\lambda}(i)$ for all $i\in[q]$ and dividing both sides by $n^{20}$, we derive that the following equalities are satisfied.
\begin{itemize}
\item $\displaystyle{\sum_{i\in[q]}a(\mu,i)2^{i-1} \leq 2^q-1}$.
\item $\displaystyle{\frac{1}{3}(4^q-1) \leq \sum_{i\in[q]}a(\mu,i)4^{i-1}}$.
\end{itemize}
However, these inequalities are precisely of the form of those obtained in the proof of Lemma \ref{lem:w1BSMReverseTilde}, and hence we again derive that both of them can be satisfied simultaneously only when for all $i\in[k]$, $a(\mu,i)=1$.
\end{proof}

\begin{lemma}\label{lem:sethBSM1PerFalse}
Let $\varphi$ be an instance of {\sc $s$-Sparse $p$-CNF-SAT}. Let $\mu$ be a stable matching of $\red_{BM}(\varphi)$ such that $\bal(\mu)\leq 100s4^{pd}\cdot n^2$. Then, for all $i\in[q]$, there exists $j\in[a^i]$ such that $\mu(\overline{m}^i_j)=\overline{w}^i_j$ and for all $k\neq j$, $\mu(\overline{m}^i_k)=\widehat{\overline{w}}^i_k$.
\end{lemma}

\begin{proof}
First, notice that by Lemma \ref{lem:sethSESMEasy}, it holds that the two following inequalities are satisfied.
\begin{itemize}
\item $\displaystyle{\sum_{i\in[q]}(a(\mu,i)-1)\gamma(i)} + (b(\mu)+q-\widetilde{a})\tau \leq 0$.
\item $\displaystyle{\sum_{i\in[q]}(1-a(\mu,i))\widehat{\lambda}(i)} + (\widetilde{a}-b(\mu)-q)\tau \leq 0$.
\end{itemize}
However, by Lemma \ref{lem:sethBSM1PerTruth}, we have that for all $i\in[q]$, $a(\mu,i)=1$. Thus, the two equalities above imply that $(b(\mu)+q-\widetilde{a})\tau = 0$. Having this equality at hand, the proof proceeds exactly as the proof of Lemma \ref{lem:sethSESM1PerFalse}.
\end{proof}

Having Lemmata \ref{lem:sethBSM1PerTruth} and \ref{lem:sethBSM1PerFalse}, the proof of following lemma is identical to the proof of Lemma \ref{lem:sethSESMReverse}.

\begin{lemma}\label{lem:sethBSMReverse}
Let $\varphi$ be an instance of {\sc $s$-Sparse $p$-CNF-SAT}. If for the instance $\red_{BSM}(\varphi)$ of {\sc BSM}, $\Bal\leq \eta$, then $\varphi$ is a \yesinstance\ of {\sc $s$-Sparse $p$-CNF-SAT}.
\end{lemma}

%% file: conclusion.tex
%!TEX root = mainTW.tex
\section{Conclusion}\label{sec:conclusion}
%\hly{TODO: conclusion.}
In this paper we studied  {\sc Sex-Equal Stable Marriage}, {\sc Balanced Stable Marriage}, {\sc max-Stable Marriage with Ties} and {\sc min-Stable Marriage with Ties}, four of the most central \NPH\ optimization versions of {\sc Stable Marriage}, in the realm of Parameterized Complexity.  We analyzed these problems with respect to the parameter treewidth and 
presented a comprehensive, complete picture of the behavior of central optimization versions of {\sc Stable Marriage} with respect to treewidth. Towards this, we established that all four problems are \WOH. In particular, while all four problems admit algorithms that run in time $n^{\OO(\tw)}$, we proved that all of these algorithms are likely to be essentially optimal. Next, we studied the treewidth $\tw$ of the rotation digraph. For both {\sc SESM} and {\sc BSM}, we designed algorithms that run in time $2^{\tw}n^{\OO(1)}$. Then, for both {\sc SESM} and {\sc BSM}, we also proved that unless \SETH\ is false, algorithms that run in time $(2-\epsilon)^{\tw}n^{\OO(1)}$ do not exist for any fixed $\epsilon>0$. We believe that  our parameterized algorithms, \WO-hardness reductions and \SETH-based reductions will act as  a template to show similar results  for other computational problems arising in Economics and resource allocation.
As a direction for further research, we suggest to conduct a comprehensive study that measures various parameters, with emphasis on treewidth, of instances of {\sc Sex-Equal Stable Marriage}, {\sc Balanced Stable Marriage}, {\sc max-Stable Marriage with Ties} and {\sc min-Stable Marriage with Ties} that arise in real-world applications.